\documentclass{article}
\title{Graphs with no 7-wheel subdivision}
\author{{Rebecca Robinson and Graham Farr}\\
	[1ex] Clayton School of Information Technology\\
	Monash University\\
	Clayton, Victoria, 3800\\
	Australia\\[1ex]
	Rebecca.Robinson@monash.edu\\
	Graham.Farr@monash.edu
	}
\date{June 13, 2012; amended 17 Dec 2013}

\usepackage[dvips]{graphicx}
\usepackage{amsmath, amsthm, amssymb}
\usepackage{url}

\newtheorem{thm}{Theorem}
\newtheorem{lem}{Lemma}
\newtheorem{red}{Reduction}

\begin{document}
\maketitle

\begin{abstract}
The subgraph homeomorphism problem, SHP($H$), has been shown to be polynomial-time solvable for any fixed pattern graph $H$, but practical algorithms have been developed only for a few specific pattern graphs. Among these are the wheels with four, five, and six spokes. This paper examines the subgraph homeomorphism problem where the pattern graph is a wheel with seven spokes, and gives a result that describes graphs with no $W_{7}$-subdivision, showing how they can be built up, using certain operations, from smaller `pieces' that meet certain conditions. We also discuss algorithmic aspects of the problem.

\end{abstract}

\section{Introduction}
\label{intro}

A graph $G$ is said to be a \emph{subdivision} of a graph $H$ if a graph isomorphic to $H$ can be obtained from $G$ by performing a series of edge-contractions on $G$ where the contractions are limited to edges with at least one endvertex of degree exactly 2. In such a situation, $G$ is often referred to as an \emph{$H$-subdivision}. A particularly important application of this concept is found in Kuratowski's famous theorem characterizing planarity, which tells us that a graph is non-planar if and only if it contains either a $K_{5}$-subdivision or a $K_{3,3}$-subdivision \cite{Kuratowski30}.

Characterizations of graphs containing no subdivisions of a particular fixed graph (which we call the \emph{pattern graph}) $H$ are few and far between. Some notable examples are where the pattern graph is $K_{4}$ \cite{Dirac52, Duffin65}, $K_{3,3}$ \cite{Hall43}, $W_{4}$ \cite{Farr88}, $W_{5}$ \cite{Farr88}, and $W_{6}$ \cite{Robinson08}, where $W_{k}$ denotes the wheel with $k$ spokes.

The algorithmic problem of determining whether or not some graph contains a subdivision of a particular pattern graph is known as the \emph{subgraph homeomorphism problem}, or \emph{topological containment}: 

\begin{tabular}{l}
\\
SUBGRAPH HOMEOMORPHISM ($H$) (abbreviated SHP($H$))\\
Instance: Graph $G$.\\
Question: Does $G$ contain a subdivision of $H$?\\
\\
\end{tabular}

Provided the pattern graph $H$ is fixed, it is known from general results of Robertson and Seymour that SHP($H$) can be solved in polynomial time for any $H$ (\cite{R&S95}, recently improved in \cite{Grohe11, Grohe11_2} to establish fixed-parameter tractability, using elements of the Robertson-Seymour approach). However, precise characterizations are only known for a handful of pattern graphs, including those listed above. The difficulty of finding a complete characterization increases very rapidly as the size of $H$ increases, hence few characterizations are known. A more complete review of previous research in this area can be found in \cite{Robinson08}.


This paper builds on the results of \cite{Farr88} and \cite{Robinson08} to give a result pertaining to graphs that do not contain a subdivision of $W_{7}$, the wheel with seven spokes. The result gives a characterization for such graphs, provided they have `pieces' of a particular bounded size (this is defined below in more detail), and leads to an efficient algorithm for solving SHP($W_{7}$).

The main result of this paper is Theorem \ref{theorem}, which characterizes (up to bounded size pieces) graphs that do not contain subdivisions of $W_{7}$. The proof is constructed in a similar way to the proofs of the main theorems in \cite{Robinson08} and \cite{Farr88}, which characterize graphs with no subdivisions of $W_{6}$ and $W_{5}$ respectively. The proof of Theorem \ref{theorem} begins by showing that for some graph $G$ that meets the conditions of the hypothesis, there must exist some $W_{6}$-subdivision $H$ centred on a specific vertex $v_{0}$ of degree $\ge 7$. It is then observed that some neighbour $u$ of $v_{0}$ exists such that $u$ is not a neighbour of $v_{0}$ in $H$, and that, since $G$ is 3-connected, there must be two disjoint paths in $G$ from $u$ to $H$ that do not meet $v_{0}$. The proof examines all possible placings of these paths, and shows that each resulting graph must contain a $W_{7}$-subdivision, if it is to satisfy the conditions of the Theorem. As in \cite{Robinson08}, parts of proofs in this paper requiring exhaustive case analysis depend on results generated by a program written in C. This program automates the construction of the small graphs arising as cases in the proof, and tests each graph for the presence of a $W_{7}$-subdivision. The complete code for the program can be found online at \verb|http://www.csse.monash.edu.au/~rebeccar/wheelcode.html|. For further discussion of this program and the algorithms used, see \cite{Robinson09_2}.

While the proof of Theorem \ref{theorem} uses the same overall method as the main theorems of \cite{Robinson08} and \cite{Farr88}, a key difference is that in this theorem, one of the conditions of the hypothesis is a minimum bound on $|V(G)|$. This has certain implications in designing an algorithm (given in Section \ref{algorithm}) which solves SHP($W_{7}$) for any given input graph, the steps of which follow from each of the restrictions placed on $G$ in Theorem \ref{theorem}. One of the steps in this algorithm involves exhaustive search, for a $W_{7}$-subdivision, in `pieces' of the graph that are of bounded size (where the maximum bound of each such `piece' is less than the minimum bound of $|V(G)|$ in Theorem \ref{theorem}). By contrast, the theorem for $W_{6}$ in \cite{Robinson08} describes completely the structure of a graph with no $W_{6}$-subdivisions, and as such the corresponding algorithm for solving SHP($W_{6}$) does not require an exhaustive search for $W_{6}$-subdivisions at any point.

This paper begins by presenting some definitions of terms used throughout. We then give a short section stating two simple but important lemmas which are proved in \cite{Robinson08}, but also used frequently in this paper. Sections \ref{reductions} and \ref{separatingsets} define various types of reductions and separating sets respectively, each of which is forbidden in a graph meeting the conditions of the main theorem, Theorem \ref{theorem}. With each such definition a corresponding theorem is given, proving that a certain operation can be performed on some input graph $G$ (either performing a reduction, or dividing the input graph into components along the given separating set) without altering the existence or otherwise of a $W_{7}$-subdivision in $G$. 

In Section \ref{6wheelresults}, we give some results on graphs with no 6-wheel subdivisions, building on the main result of \cite{Robinson08}. The final theorem in this section (Theorem \ref{w6cor}) gives the important result that any graph $G$ meeting the conditions of Theorem \ref{theorem} must contain a $W_{6}$-subdivision centred on any given vertex $v_{0}$ of degree $\ge 7$ in $G$. This result is key to proving Theorem \ref{theorem}. Section \ref{supportinglemmas} gives some further lemmas which support the main result, then Section \ref{mainresult} contains the main theorem of the paper, Theorem \ref{theorem}, and its proof. An algorithm which solves SHP($W_{7}$) follows from this result: this is given in Section \ref{algorithm}. Finally, some concluding remarks are given suggesting further work in this area.

\section{Definitions}
\label{definitions}

If $G$ is a graph, and $E'$ is a set of edges in $G$, then $G - E'$ denotes the graph obtained from $G$ after the removal of all the edges of $E'$.

If $X$ is a set of vertices in graph $G$, then $G - X$ denotes the graph obtained from $G$ after the removal of all the vertices of $X$.

A \emph{separating set} $S$ in a graph $G$ is a set of vertices in $G$ whose removal disconnects $G$.

The \emph{neighbourhood} $N_{G}(v)$ of a vertex $v$ in $G$ is the set of vertices which are adjacent to $v$ in $G$.

An \emph{internal 3-edge-cutset} in a graph $G$ is a set $E'$ of at most three edges of $G$ such that $G - E'$ is disconnected with each component having at least two vertices. 

An \emph{internal 4-edge cutset} in a graph $G$ is a set $E'$ of four edges of $G$ such that $G - E'$ is disconnected with each component having at least three vertices, and exactly two edges in $E'$ share an endpoint.

Given a path $P$ where $x, y \in V(P)$, then $xPy$ denotes the subpath of $P$ between $x$ and $y$, including $x$ and $y$.

If $W$ is a set of vertices in graph $G$, then $G|W$ denotes the set of all maximal subsets $U$ of $V(G)$ such that any two vertices of $U$ are joined by a path in $G$ with no internal vertex in $W$. Each element of $G|W$ is refered to as a \emph{bridge} of $G|W$.

The \emph{centre} of a wheel subdivision $W_{n}$ is the vertex of degree $n$ in that wheel subdivision. The \emph{rim} of a wheel subdivision $W_{n}$ is the cycle around the outside of that wheel subdivision (excluding the centre). The \emph{spoke-meets-rim vertices} of a wheel subdivision $W_{n}$ are the $n$ vertices of degree 3 in that wheel subdivision. The \emph{spokes} of a wheel subdivision $W_{n}$ are the $n$ paths from the centre vertex to the spoke-meets-rim vertices in that wheel subdivision. 

If $G$ is a graph, and $S$ is a set of vertices such that $S \in V(G)$, then $\langle S\rangle$ is the subgraph induced by $S$.

$X\setminus S$ denotes the set-theoretic difference of $X$ and $S$.

$X - v$ is equivalent to $X \setminus \{v\}$.

$H\cap X$ is equivalent to $H\cap \langle X\rangle$.

\section{Two important lemmas}
\label{twolemmas}

The proofs of the following two lemmas, Lemma \ref{lemma1} and Lemma \ref{lemma2}, are given in \cite{Robinson08}. These lemmas are used many times throughout the paper to support the proofs of other lemmas and theorems.

\begin{lem}
\label{lemma1}

Let $G$ be a 3-connected graph containing a separating set $S$ such that $|S| = 3$. Let $X$ be some bridge of $G|S$ which contains at least two vertices not in $S$. Suppose there are at least four edges joining the vertices in $S$ to the vertices in $X \setminus S$, and suppose there is some vertex $x \in S$ which has at least two neighbours $y$ and $z$ in $X \setminus S$. Then there exists:
\begin{itemize}
\item a path $P$ in $\langle X\rangle$ such that $P$ has only its endpoints in $S$ but does not meet $x$; and
\item two paths $Q_{1}$ and $Q_{2}$ and two vertices $q_{1}$ and $q_{2}$ such that $q_{1}$ and $q_{2}$ are two distinct vertices on $P$, and $Q_{1}$ and $Q_{2}$ are paths from $x$ to $q_{1}$ and $x$ to $q_{2}$ respectively, which meet only at $x$, and which contain no other vertex of $P$.
\end{itemize}

\end{lem}

\begin{lem}
\label{lemma2}

Let $G$ be a 3-connected graph containing a separating set $S$ such that $|S| = 3$. Let $X$ be some bridge of $G|S$. Suppose there is some vertex $x \in S$ which has three neighbours $w$, $y$ and $z$ in $X \setminus S$. Then there exists:
\begin{itemize}
\item a path $P$ in $\langle X\rangle$ such that $P$ has only its endpoints in $S$, but does not meet $x$; and
\item three paths $Q_{1}$, $Q_{2}$ and $Q_{3}$ and three vertices $q_{1}$, $q_{2}$ and $q_{3}$, such that $q_{1}$, $q_{2}$ and $q_{3}$ are distinct vertices on path $P$, and $Q_{1}$, $Q_{2}$ and $Q_{3}$ are paths from $x$ to $q_{1}$, $x$ to $q_{2}$ and $x$ to $q_{3}$ respectively, that are pairwise vertex-disjoint except at $x$.
\end{itemize}

\end{lem}

\section{Reductions}
\label{reductions}

Each reduction defined in this section is forbidden in a graph that meets the conditions of the main theorem, Theorem \ref{theorem}. For each reduction given, we prove that performing that reduction on $G$ will not alter the presence or otherwise of a $W_{k}$-subdivision in $G$, for some bounded value of $k$. This means that each of the reductions are useful in creating an algorithm to solve SHP($W_{k}$), since they can be performed in polynomial time on the input graph, thus reducing the size of the graph, and modifying it to help meet the conditions of Theorem \ref{theorem}.

Note that Reductions \ref{r1} and \ref{r2} are generalizations of Reductions 1 and 2 in \cite{Robinson08}.

\begin{red}
\label{r1}
Let $G$ be a 3-connected graph containing a set $S = \{u, v, w\}$ of vertices. Suppose there are at least three bridges $X, Y, Z$ of $G|S$, such that each of the bridges $Y$ and $Z$ contains a subdivision of $X$. Suppose that $v$ and $w$ are either adjacent or joined by a path in some fourth bridge $A$ of $G|S$. Call this path (or edge) $P_{w}$. Suppose also that $v$ and $u$ are either adjacent or joined by a path in some bridge $B$ of $G|S$ other than $X$, $Y$, $Z$ or $A$ (if $A$ exists). Call this path (or edge) $P_{u}$.

Form $G'$ from $G$ by removing $X \setminus S$ and adding a single edge from $u$ to $w$, if such an edge does not already exist.
\end{red}

\begin{thm}
\label{reduction1}
Let $G$ be some 3-connected graph on which Reduction \ref{r1} can be performed. Let $G'$ be the resulting graph after Reduction \ref{r1} has been performed on $G$. Then $G$ contains a $W_{k}$-subdivision if and only if $G'$ contains a $W_{k}$-subdivision, where $k \ge 4$.
\end{thm}

\begin{proof}
It is obvious that if $G'$ contains a $W_{k}$-subdivision, then $G$ will also, since $G$ contains a subdivision of $G'$.

Assume then that $G$ contains a $W_{k}$-subdivision, $H$. The centre of $H$ must either be in $G - X$, in $S$, or in $X \setminus S$.

Let us consider these three possibilities.

\textbf{(a)} The centre of $H$ is in $G - X$. Without loss of generality, assume that $\langle Y \setminus S \rangle$ contains the centre of $H$.

Suppose firstly that $H$ is contained in two bridges of $G|S$ (that is, $Y$ and some other bridge). If $X$ is not one of these two bridges, then removing $X \setminus S$ will have no effect on the existence of the $W_{k}$-subdivision. If, however, part of $H$ is in $\langle X \setminus S \rangle$, then another $W_{k}$-subdivision which does not pass through $\langle X \setminus S \rangle$ can be formed using parts of $\langle Z \setminus S \rangle$, since $Z$ contains a subdivision of a structure isomorphic to $X$. Thus, $X \setminus S$ can again be removed from the graph without altering the existence or otherwise of a $W_{k}$-subdivision.

Suppose now that $H$ is contained in at least three bridges of $G|S$ (that is, $Y$ and at least two other bridges). Assume without loss of generality that one of these bridges is $Z$. If $X$ is not another of the three bridges, then again, the removal of $X \setminus S$ will have no effect on the existence of $H$. Suppose then that $\langle X \setminus S\rangle$ is used in forming $H$. All spoke-meets-rim vertices in $H$ must be contained in $Y$ in this situation. Thus $X \setminus S$ and $Z \setminus S$ can only contain vertices of degree 2 in $H$, and no bridge of $G|S$ other than $X$, $Y$, and $Z$ can contain any part of $H$ at all (except for vertices of $S$). Thus the part of $\langle X \setminus S\rangle$ used in $G$ to form $H$ can be replaced by one of $P_{w}$, $P_{u}$, or $uw$ in $G'$.

\textbf{(b)} The centre of $H$ is in $S$. There are two possibilities:

\textbf{(b)(i)} Vertex $v$ forms the centre of $H$.

In this case, the rim of $H$ must be contained in at most two of the bridges of $G|S$. Assume without loss of generality that neither of these two bridges are $A$ or $B$. If the rim does not pass through $\langle X \setminus S \rangle$, then $X \setminus S$ can only be used in $H$ to form part of a spoke. However, any spoke passing through $X \setminus S$ can be replaced either by $P_{w}$ or $P_{u}$ in $G'$, so $X \setminus S$ is not necessary in forming a $W_{k}$-subdivision in $G'$.

Suppose now that the rim of $H$ passes through $\langle X \setminus S\rangle$. Without loss of generality, assume the rim of $H$ is contained in $\langle X \cup Y\rangle$. Then the remaining bridges of $G|S$ can only be used to form (at most two) spokes in $H$. If $\langle Z \setminus S \rangle$ is used to form any spokes in $H$, then instead use one or both of $P_{w}$ and $P_{u}$ to create the corresponding spokes in $G'$. Use $\langle Z \rangle$ in $G'$ to replace that portion of $H$ contained in $\langle X \rangle$ in $G$.

\textbf{(b)(ii)} Either $w$ or $u$ forms the centre of $H$ (assume $u$ without loss of generality).

Again, the rim of $H$ must be contained in at most two of the bridges of $G|S$. Suppose firstly that the rim of $H$ does not pass through $\langle X \setminus S \rangle$. Without loss of generality, assume that the rim is contained in $\langle Y \cup Z\rangle$. Again, if $X \setminus S$ is used in $H$, it must be to form part of a spoke, joining the rim of $H$ either at vertex $v$ or vertex $w$. This path can be replaced by one of $uw$ or $P_{u}$ in $G'$.

Suppose now that the rim of $H$ does pass through $\langle X \setminus S\rangle$. Assume without loss of generality that the rim is contained in $\langle X \cup Y\rangle$. Again, the other bridges of $G|S$ can only be used to form spokes in $H$. If $\langle Z \setminus S \rangle$ is used to form any spokes in $H$, then instead use one or both of $P_{u}$ and $uw$ to create these spokes in $G'$. In any case, use $\langle Z \rangle$ in $G'$ to replace that portion of $H$ contained in $\langle X \rangle$ in $G$.

\textbf{(c)} The centre of $H$ is in $X \setminus S$.

If $H$ is entirely contained in $X$, then in $G'$, either $Y$ or $Z$ can be used to create a $W_{k}$-subdivision, since each of these bridges contains a subdivision of $X$. Similiarly, if $H$ is contained in $X$ and only one other bridge, then at least one of $Y$ or $Z$ can still be used to replace the parts of $H$ contained in $X$ in $G$.

Suppose then that $H$ is contained in at least three bridges of $G|S$ (that is, $X$ and at least two other bridges). If $Y$ and $Z$ are not two of these bridges, then at least one of $Y$ or $Z$ can be used to replace the parts of $H$ contained in $X$ in $G$. Assume then that the two of the bridges are $Y$ and $Z$. All spoke-meets-rim vertices in $H$ must be contained in $X$ in this situation. Thus $Y \setminus S$ and $Z \setminus S$ can only contain vertices of degree 2 in $H$, and no bridge of $G|S$ other than $X$, $Y$, and $Z$ can contain any part of $H$ at all (except for vertices of $S$). In $G'$, then, replace the part of $Y \setminus S$ used in $G$ to form $H$ with one of $P_{w}$, $P_{u}$, or $uw$, and instead use $Y$ to form those parts of $H$ contained in $X$ in $G$.
\end{proof}

Reductions \ref{r1}A, \ref{r1}B, and \ref{r1}C, which follow, are special cases of Reduction \ref{r1}, where the bridge to be removed from the graph is limited in size.

\subsubsection*{Reduction \ref{r1}A}
Let $G$ be a 3-connected graph containing a set $S = \{u, v, w\}$ of vertices. Suppose there are at least three bridges $X$, $Y$, and $Z$ of $G|S$, such that $X \setminus S$ contains a single vertex $x$. Suppose that $v$ and $w$ are either adjacent or joined by a path in some fourth bridge $A$ of $G|S$. Call this path (or edge) $P_{w}$. Suppose also that $v$ and $u$ are either adjacent or joined by a path in some bridge $B$ of $G|S$ other than $X$, $Y$, $Z$ or $A$ (if $A$ exists). Call this path (or edge) $P_{u}$.

\begin{figure}[!h]
\begin{center}
\includegraphics[width=0.8\textwidth]{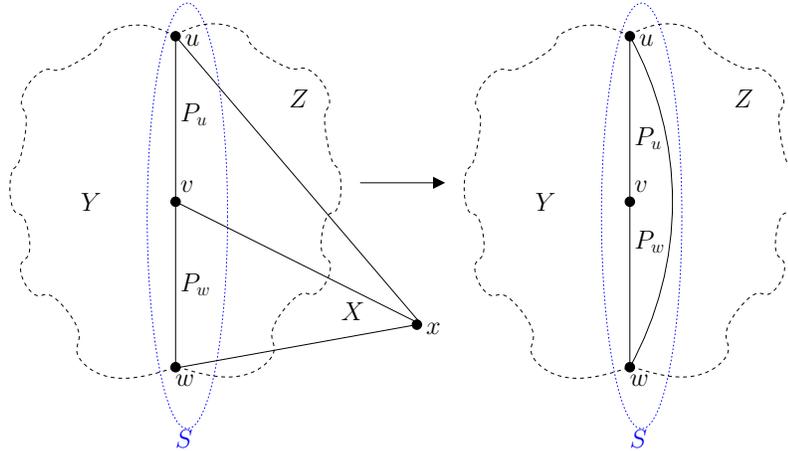}
\caption[Reduction \ref{r1}A]{Reduction \ref{r1}A: $G$ and $G'$}
\label{rt1}
\end{center}
\end{figure}

Form $G'$ from $G$ by removing vertex $x$ and adding a single edge from $u$ to $w$, as in Figure \ref{rt1}.

\subsubsection*{Reduction \ref{r1}B}
Let $G$ be a 3-connected graph containing a set $S = \{u, v, w\}$ of vertices. Suppose there are at least three bridges $X$, $Y$, and $Z$ of $G|S$, such that $X \setminus S$ contains exactly two vertices, and each of the bridges $Y$ and $Z$ contains a subdivision of $X$. Suppose that $v$ and $w$ are either adjacent or joined by a path in some fourth bridge $A$ of $G|S$. Call this path (or edge) $P_{w}$. Suppose also that $v$ and $u$ are either adjacent or joined by a path in some bridge $B$ of $G|S$ other than $X$, $Y$, $Z$ or $A$ (if $A$ exists). Call this path (or edge) $P_{u}$.

Form $G'$ from $G$ by removing $X \setminus S$ and adding a single edge from $u$ to $w$, if such an edge does not already exist.

\subsubsection*{Reduction \ref{r1}C}
Let $G$ be a 3-connected graph containing a set $S = \{u, v, w\}$ of vertices. Suppose there are at least three bridges $X$, $Y$, and $Z$ of $G|S$, such that $X \setminus S$ contains exactly three vertices, and each of the bridges $Y$ and $Z$ contains a subdivision of $X$. Suppose that $v$ and $w$ are either adjacent or joined by a path in some fourth bridge $A$ of $G|S$. Call this path (or edge) $P_{w}$. Suppose also that $v$ and $u$ are either adjacent or joined by a path in some bridge $B$ of $G|S$ other than $X$, $Y$, $Z$ or $A$ (if $A$ exists). Call this path (or edge) $P_{u}$.

Form $G'$ from $G$ by removing $X \setminus S$ and adding a single edge from $u$ to $w$, if such an edge does not already exist.

\begin{red}
\label{r2}
Let $G$ be a 3-connected graph containing a set $S = \{u, v, w\}$ of vertices such that both $u$ and $w$ have degree $< k$, where $k \ge 5$. Suppose there are three bridges of $G|S$, namely $X$, $Y$ and $Z$, such that each of the bridges $Y$ and $Z$ contains a subdivision of $X$. Suppose that $v$ and $w$ are either adjacent or joined by a path in some fourth bridge $A$ of $G|S$. Call this path (or edge) $P_{w}$.

Form $G'$ from $G$ by removing $X \setminus S$ and adding a single edge from $v$ to $u$. 
\end{red}

\begin{thm}
\label{reduction2}
Let $G$ be a graph on which Reduction \ref{r2} can be performed. Let $G'$ be the resulting graph after Reduction \ref{r2} has been performed on $G$, with $k \ge 5$. Then $G$ contains a $W_{k}$-subdivision if and only if $G'$ contains a $W_{k}$-subdivision.
\end{thm}

\begin{proof}
It is obvious that if $G'$ contains a $W_{k}$-subdivision, then $G$ will also, since $G$ contains a subdivision of $G'$.

Assume then that $G$ contains a $W_{k}$-subdivision, $H$. The centre of $H$ must either be in $G - X$, in $S$, or in $X \setminus S$.

Consider the three possibilities.

\textbf{(a)} The centre of $H$ is in $G - X$. Without loss of generality, assume that $\langle Y \setminus S \rangle$ contains the centre of $H$.

Suppose firstly that $H$ is contained in at most two bridges of $G|S$ (that is, $Y$ and some other bridge). If $X$ is not one of these two bridges, then removing $X \setminus S$ will have no effect on the existence of the $W_{k}$-subdivision. If, however, part of $H$ is in $\langle X \setminus S \rangle$, then another $W_{k}$-subdivision which does not pass through $\langle X \setminus S \rangle$ can be formed using parts of $\langle Z \setminus S \rangle$, since $Z$ contains a subdivision of a structure isomorphic to $X$. Thus, $X \setminus S$ can again be removed from the graph without altering the existence or otherwise of a $W_{k}$-subdivision.

Suppose now that $H$ is contained in at least three bridges of $G|S$ (that is, $Y$ and at least two other bridges). Assume without loss of generality that one of these bridges is $Z$. If $X$ is not another of the three bridges, then again, the removal of $X \setminus S$ will have no effect on the existence of $H$. Suppose then that $\langle X \setminus S\rangle$ is used in forming $H$. All spoke-meets-rim vertices in $H$ must be contained in $Y$ in this situation. Thus $X \setminus S$ and $Z \setminus S$ can only contain vertices of degree 2 in $H$, and no bridge of $G|S$ other than $X$, $Y$, and $Z$ can contain any part of $H$ at all (except for vertices of $S$). The part of $\langle X \setminus S\rangle$ used in $G$ to form $H$ must be either a path from $u$ to $v$, from $v$ to $w$, or from $u$ to $w$. In the first two cases, this path can be replaced in $G'$ by $uv$ or $P_{w}$ respectively. In the third case, observe that the part of $\langle Z \setminus S\rangle$ used in $G$ to form $H$ must be either a path from $u$ to $v$ or from $v$ to $w$. This path, then, is replaced in $G'$ by either $uv$ or $P_{w}$, and $\langle Z \setminus S\rangle$ is instead used in $G'$ to create the path from $u$ to $w$.

\textbf{(b)} The centre of $H$ is in $S$. Since vertices $w$ and $u$ each have degree $< k$, $H$ must be centred on $v$.

In this case, the rim of $H$ must be contained in at most two of the bridges of $G|S$. Assume without loss of generality that neither of these two bridges are $A$. If the rim does not pass through $\langle X \setminus S \rangle$, then $X \setminus S$ can only be used in $H$ to form part of a spoke. However, any spoke passing through $X \setminus S$ can be replaced either by $P_{w}$ or $uv$ in $G'$, so $X \setminus S$ is not necessary in forming a $W_{k}$-subdivision in $G'$.

Suppose now that the rim of $H$ passes through $\langle X \setminus S\rangle$. Without loss of generality, assume the rim of $H$ is contained in $\langle X \cup Y\rangle$. Then the remaining bridges of $G|S$ can only be used to form (at most two) spokes in $H$.

If $\langle Z \setminus S \rangle$ is used to form any spokes in $H$, then instead use one or both of $P_{w}$ and $uv$ to create the corresponding spokes in $G'$. Use $\langle Z \rangle$ in $G'$ to replace that portion of $H$ contained in $\langle X \rangle$ in $G$.

\textbf{(c)} The centre of $H$ is in $X \setminus S$.

If $H$ is entirely contained in $X$, then in $G'$, either $Y$ or $Z$ can be used to create a $W_{k}$-subdivision, since these bridges each contain a subdivision of $X$. Similiarly, if $H$ is contained in $X$ and only one other bridge, then at least one of $Y$ or $Z$ can still be used to replace the parts of $H$ contained in $X$ in $G$.

Suppose then that $H$ is contained in at least three bridges of $G|S$ (that is, $X$ and at least two other bridges). If $Y$ and $Z$ are not two of these bridges, then at least one of $Y$ or $Z$ can be used to replace the parts of $H$ contained in $X$ in $G$. Assume then that two of the bridges are $Y$ and $Z$. All spoke-meets-rim vertices in $H$ must be contained in $X$ in this situation. Thus $Y \setminus S$ and $Z \setminus S$ can only contain vertices of degree 2 in $H$, and no bridge of $G|S$ other than $X$, $Y$, and $Z$ can contain any part of $H$ at all (except for vertices of $S$). The part of $\langle Y \setminus S\rangle$ used in $G$ to form $H$ must be either a path from $u$ to $v$, from $v$ to $w$, or from $u$ to $w$. In the first two cases, this path can be replaced in $G'$ by $uv$ or $P_{w}$ respectively, while $Y$ can instead be used to form those parts of $H$ contained in $X$ in $G$. In the third case, that part of $\langle Z \setminus S\rangle$ used in $G$ to form $H$ must be either a path from $u$ to $v$ or from $v$ to $w$. Replace this path in $G'$ by either $uv$ or $P_{w}$, then use $Z$ to form those parts of $H$ contained in $X$ in $G$.
\end{proof}

Reduction \ref{r2}A and Reduction \ref{r2}B, which follow, are special cases of Reduction \ref{r2}, where the bridge to be removed from the graph is limited in size.

\subsubsection*{Reduction \ref{r2}A}
Let $G$ be a 3-connected graph containing a set $S = \{u, v, w\}$ of vertices such that both $u$ and $w$ have degree $< k$, where $k \ge 5$. Suppose there are three bridges of $G|S$, namely $X$, $Y$ and $Z$, such that $X \setminus S$ contains a single vertex $x$. Suppose that $v$ and $w$ are either adjacent or joined by a path in some fourth bridge of $G|S$. Call this path (or edge) $P_{w}$.

\begin{figure}[!ht]
\begin{center}
\includegraphics[width=0.8\textwidth]{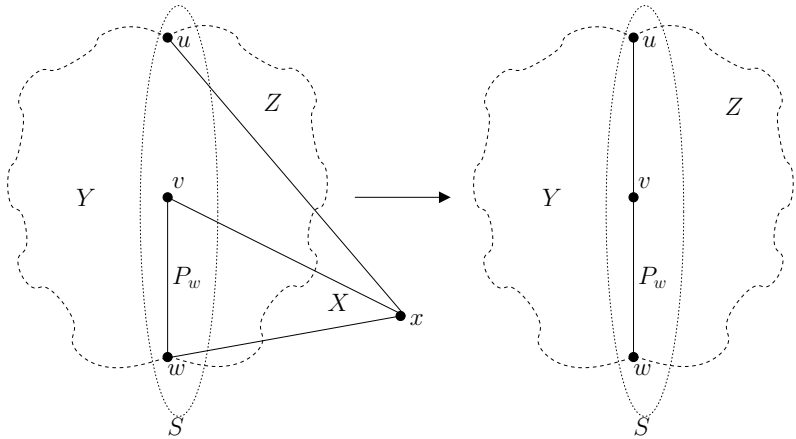}
\caption[Reduction \ref{r2}A]{Reduction \ref{r2}A: $G$ and $G'$}
\label{rt3}
\end{center}
\end{figure}

Form $G'$ from $G$ by removing vertex $x$ and adding a single edge from $v$ to $u$, as in Figure \ref{rt3}.

\subsubsection*{Reduction \ref{r2}B}
Let $G$ be a 3-connected graph containing a set $S = \{u, v, w\}$ of vertices such that both $u$ and $w$ have degree $< k$, where $k \ge 5$. Suppose there are three bridges of $G|S$, namely $X$, $Y$ and $Z$, such that $X \setminus S$ contains exactly two vertices, and each of the bridges $Y$ and $Z$ contains a subdivision of $X$. Suppose that $v$ and $w$ are either adjacent or joined by a path in some fourth bridge of $G|S$. Call this path (or edge) $P_{w}$.

Form $G'$ from $G$ by removing $X \setminus S$ and adding a single edge from $v$ to $u$.

\begin{red}
\label{r6}
Let $k$ be some integer $\ge 7$. Let $G$ be a 3-connected graph containing a set $S = \{t, u, v, w\}$ of vertices, such that there exists some bridge $X$ of $G|S$, where $|X| \le k$, and $X\cap S = \{u, v, w\}$. Suppose that for each vertex $i\in \{u, v, w\}$, either $i$ has degree $< k$, or every bridge of $G|S$ contains at most one neighbour of $i$ not in $S$. Suppose that $v$ and $u$ are adjacent, and that $v$ and $w$ are adjacent.

Form $G'$ from $G$ by removing $X \setminus S$ and adding an edge from $w$ to $u$, if such an edge does not already exist (as in Figure \ref{rf6}). 
\end{red}

\begin{figure}[!ht]
\begin{center}
\includegraphics[width=0.8\textwidth]{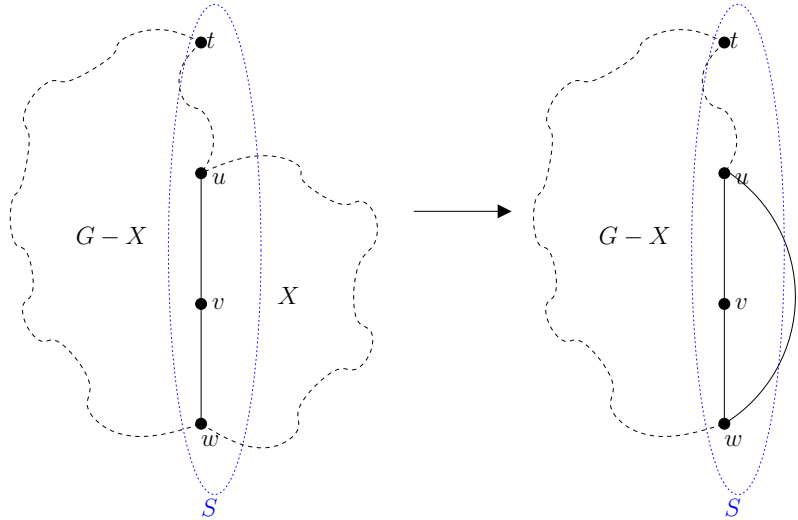}
\caption[Reduction \ref{r6}]{Reduction \ref{r6}: $G$ and $G'$}
\label{rf6}
\end{center}
\end{figure}

\begin{thm}
\label{reduction6}
Let $G$ be a graph on which Reduction \ref{r6} can be performed. Let $G'$ be the resulting graph after Reduction \ref{r6} has been performed on $G$, with $k \ge 7$. Then $G$ contains a $W_{k}$-subdivision if and only if $G'$ contains a $W_{k}$-subdivision.
\end{thm}

\begin{proof}
It is obvious that if $G'$ contains a $W_{k}$-subdivision, then $G$ will also, since $G$ contains a subdivision of $G'$.

Assume then that $G$ contains a $W_{k}$-subdivision, $H$.

Since $|X| \le k$, the maximum degree of any vertex in $X\setminus S$ is $k - 1$. Therefore, the centre of $H$ cannot be in $X \setminus S$.

Suppose $H$ is centred on some vertex $v_{0} \in \{u, v, w\}$. Since $v_{0}$ must then have degree $\ge k$, each bridge of $G|S$ must contain at most one neighbour of $v_{0}$ not in $S$, by the hypothesis of the theorem. Since there are only three vertices in $S$ other than $v_{0}$, the rim can pass through $S$ at most three times, and thus can be contained in at most three bridges of $G|S$. Since $S$ can contain at most three spoke-meets-rim vertices of $H$, and each of the three bridges containing parts of $H$ can have at most one neighbour of $v_{0}$ not in $S$, $H$ can have at most 6 spokes. This is a contradiction, since $k \ge 7$. Therefore, the centre of $H$ cannot be in $\{u, v, w\}$.

Suppose then that the centre of $H$ is in $G - X$. If $\langle X\setminus S\rangle$ does not contain part of $H$, then removing $X \setminus S$ will have no effect on the existence of the $W_{k}$-subdivision. Suppose then that part of $H$ is contained in $\langle X \setminus S\rangle$.

Suppose firstly that $\langle X\setminus S\rangle$ contains only a single path belonging to $H$. Call this path $Q$. If $Q$ runs from $u$ to $w$, then replace $Q$ with the edge $uw$ in $G'$. If $Q$ runs from $u$ to $v$, then the edge $uv$ cannot be used as part of $H$ in $G$. Thus, use $uv$ to replace $Q$ in $G'$. Similarly, replace $Q$ with the edge $vw$ in $G'$ if $Q$ runs from $v$ to $w$.

Suppose then that $X\setminus S$ contains spoke-meets-rim vertices belonging to $H$. Since the centre of $H$ is not in $X$, $X\setminus S$ can contain only one spoke-meets-rim vertex of $H$. Two of the vertices in $\{u, v, w\}$ must then lie on the rim of $H$, while the third lies on a spoke. Thus, neither $uv$ nor $vw$ can be used to form part of $H$ in $G$. Use two of $uv$, $vw$, and $uw$, then, to form the required paths in $G'$, so that the vertex in $\{u, v, w\}$ that was previously on a spoke of $H$ now forms the required spoke-meets-rim vertex.

Thus, Reduction \ref{r6} can be performed on $G$ without altering the existence or otherwise of a $W_{k}$-subdivision.
\end{proof}

\begin{red}
\label{r7}
Let $k$ be some integer $\ge 7$. Let $G$ be a 3-connected graph containing a set $S = \{t, u, v, w\}$ of vertices, such that there exist at least four bridges of $G|S$: $W$, $X$, $Y$, and $Z$. Suppose that each of these four bridges contains all vertices of $S$. Suppose that either:

\begin{itemize}
\item[(i)] $|X| \le k$; or
\item[(ii)] $|X| = k+1$, and there are exactly four edges joining $S$ to $X\setminus S$.
\end{itemize}

Suppose that for each vertex $i\in S$, either $i$ has degree $< k$, or every bridge of $G|S$ contains at most one neighbour of $i$ not in $S$. Suppose that $v$ and $w$ are either adjacent or joined by a path in some fifth bridge $A$ of $G|S$. Call this path (or edge) $P_{w}$. Suppose also that $v$ and $u$ are either adjacent or joined by a path in some bridge $B$ of $G|S$ other than $W$, $X$, $Y$, $Z$ or $A$ (if $A$ exists). Call this path (or edge) $P_{u}$. Suppose also that $v$ and $t$ are either adjacent or joined by a path in some bridge $C$ of $G|S$ other than $W$, $X$, $Y$, $Z$, $A$, or $B$ (if $A$ and $B$ exist). Call this path (or edge) $P_{t}$.

Form $G'$ from $G$ by removing $X \setminus S$, as in Figure \ref{rf7}.
\end{red}

\begin{figure}[!ht]
\begin{center}
\includegraphics[width=0.8\textwidth]{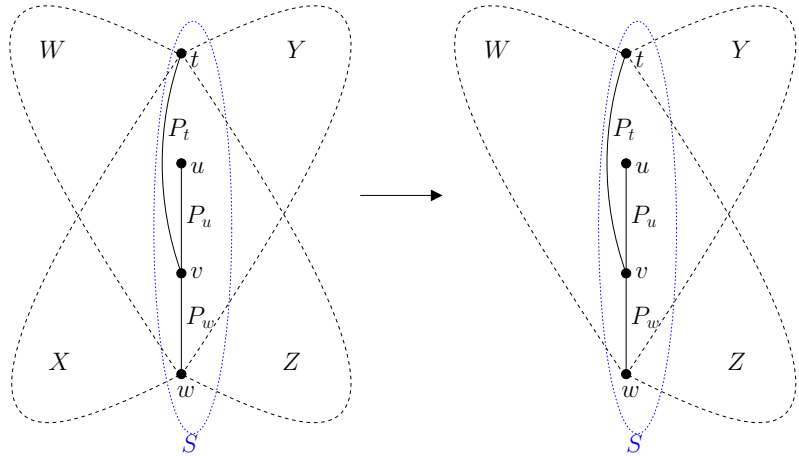}
\caption[Reduction \ref{r7}]{Reduction \ref{r7}: $G$ and $G'$}
\label{rf7}
\end{center}
\end{figure}

\begin{thm}
\label{reduction7}
Let $G$ be a graph on which Reduction \ref{r7} can be performed. Let $G'$ be the resulting graph after Reduction \ref{r7} has been performed on $G$, with $k \ge 7$. Then $G$ contains a $W_{k}$-subdivision if and only if $G'$ contains a $W_{k}$-subdivision.
\end{thm}

\begin{proof}
It is obvious that if $G'$ contains a $W_{k}$-subdivision, then $G$ will also, since $G$ contains a subdivision of $G'$.

Assume then that $G$ contains a $W_{k}$-subdivision, $H$.

By the same reasoning used in Theorem \ref{reduction6} for Reduction \ref{r6}, $H$ cannot be centred in $S$.

Suppose that $H$ is centred in $X\setminus S$. Thus, there exists some vertex $v_{0}$ in $X\setminus S$ with degree $\ge 7$. $G$ must then fall into case (ii) as described in the definition of Reduction \ref{r7}, where $|X| = k+1$, and there are exactly four edges joining $S$ to $X\setminus S$. Since there are exactly $k$ vertices in $X - v_{0}$, $v_{0}$ must be adjacent to every vertex in $X - v_{0}$. Therefore, each of the four edges joining $X$ to $X\setminus S$ has $v_{0}$ as an endpoint. The removal of $v_{0}$ will disconnect $X\setminus (S\cup\{v_{0}\})$ from the rest of the graph, then, thus violating 3-connectivity.

Suppose then that the centre of $H$ is in $G - X$. Let $U$ be the bridge of $G|S$ containing the centre of $H$. Note that $U$ could be any one of $W$, $Y$, $Z$, or some other bridge (other than $X$).

Since the rim of $H$ can pass through $S$ at most four times, $H$ can be contained in the union of at most four bridges of $G|S$ (that is, $U$ and at most three other bridges). If $X\setminus S$ does not contain part of $H$, then removing $X \setminus S$ will have no effect on the existence of the $W_{k}$-subdivision. Suppose then that part of $H$ is contained in $X \setminus S$.

Suppose $H$ is contained in exactly four bridges of $G|S$ (including $X$). Without loss of generality, assume these bridges are $W$, $X$, $Y$, and $Z$. Recall that $U$ is one of the bridges of $G|S$ other than $X$, that is, $U \in \{W, Y, Z\}$. Assume without loss of generality that $U = W$ --- in other words, $H$ is centred in $W$. Apart from $W$, each of the bridges containing parts of $H$ can contain only a single path belonging to $H$. In $X$, call this path $Q$. If $v$ is an endpoint of $Q$, replace it with one of $P_{t}$, $P_{u}$, or $P_{w}$ in $G'$. If not, then there exists some other bridge $U'$ (where $U' \in \{Y, Z\}$) that contains only a single path $Q'$ belonging to $H$, such that $v$ is an endpoint of $Q'$. One of $P_{t}$, $P_{u}$, or $P_{w}$ can be used to replace $Q'$ in $G'$, leaving $U'$ free. Parts of $U'$ can then be used to replace $Q$ in $G$, so that $X$ is no longer required.

Suppose now that $H$ is contained in at most three bridges of $G|S$. If the part of $H$ contained in $X\setminus S$ is only a single path, it can be replaced as in the previous paragraph. Suppose then that $X$ contains a single spoke-meets-rim vertex and the three paths meeting this vertex. Then, since at least one of $W$, $Y$, $Z$ is not used to form $H$ in $G$, use this bridge in $G'$ to form the parts of $H$ previously contained in $X$.

Thus, Reduction \ref{r7} can be performed on $G$ without altering the existence or otherwise of a $W_{k}$-subdivision.
\end{proof}

\begin{red}
\label{r8}
Let $k$ be some integer $\ge 7$. Let $G$ be a 3-connected graph containing a set $S = \{t, u, v, w\}$ of vertices, such that there exist at least three bridges of $G|S$: $X$, $Y$, and $Z$. Suppose that $|X| \le k$, and $X\cap S = \{u, v, w\}$, and that $Y$ and $Z$ also contain the vertices $u, v, w$. Suppose that for each vertex $i\in S$, either $i$ has degree $< k$, or every bridge of $G|S$ contains at most one neighbour of $i$ not in $S$. Suppose also that $v$ and $u$ are either adjacent, or joined by a path $P_{u}$ in some fourth bridge $A$ of $G|S$.

Form $G'$ from $G$ by removing $X \setminus S$ and adding an edge from $v$ to $w$, if such an edge does not already exist (as in Figure \ref{rf8}). 
\end{red}

\begin{figure}[!ht]
\begin{center}
\includegraphics[width=0.8\textwidth]{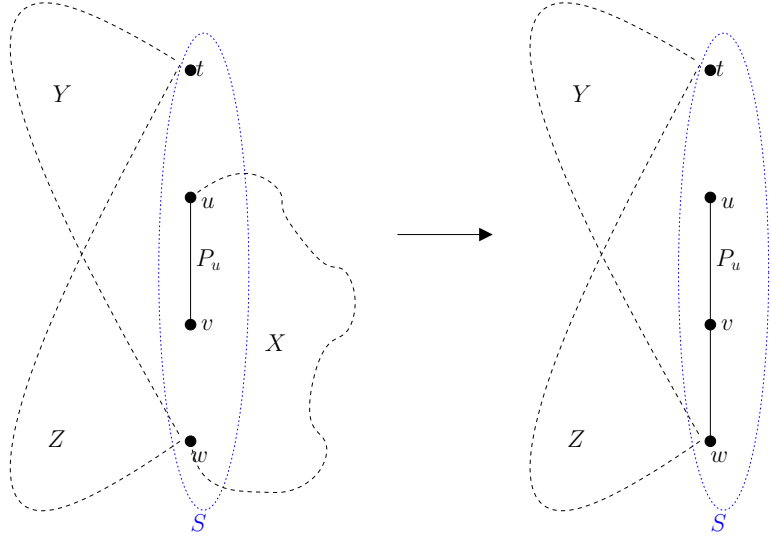}
\caption[Reduction \ref{r8}]{Reduction \ref{r8}: $G$ and $G'$}
\label{rf8}
\end{center}
\end{figure}

\begin{thm}
\label{reduction8}
Let $G$ be a graph on which Reduction \ref{r8} can be performed. Let $G'$ be the resulting graph after Reduction \ref{r8} has been performed on $G$, with $k \ge 7$. Then $G$ contains a $W_{k}$-subdivision if and only if $G'$ contains a $W_{k}$-subdivision.
\end{thm}

\begin{proof}
It is obvious that if $G'$ contains a $W_{k}$-subdivision, then $G$ will also, since $G$ contains a subdivision of $G'$.

Assume then that $G$ contains a $W_{k}$-subdivision, $H$.

Since $|X| \le k$, the maximum degree of any vertex in $X\setminus S$ is $k - 1$. Therefore, the centre of $H$ cannot be in $X \setminus S$.

Suppose $H$ is centred on some vertex $v_{0} \in S$. Since $v_{0}$ must have degree $\ge k$, each bridge of $G|S$ must contain at most one neighbour of $v_{0}$ not in $S$, by the hypothesis of the theorem. Since there are only three vertices in $S$ other than $v_{0}$, the rim can pass through $S$ at most three times, and thus can be contained in at most three bridges of $G|S$. Since $S$ can contain at most three spoke-meets-rim vertices of $H$, and each of the three bridges containing parts of $H$ can have at most one neighbour of $v_{0}$ not in $S$, $H$ can have at most 6 spokes. This is a contradiction, since $k \ge 7$. Therefore, the centre of $H$ cannot be in $S$.

Suppose then that the centre of $H$ is in $Z\setminus S$. If $\langle X\setminus S\rangle$ does not contain part of $H$, then removing $X \setminus S$ will have no effect on the existence of the $W_{k}$-subdivision. Suppose then that part of $H$ is contained in $\langle X \setminus S\rangle$.

\textbf{1.} Suppose firstly that $\langle X\setminus S\rangle$ contains only a single path belonging to $H$. Call this path $Q$.

Suppose $Q$ runs from $v$ to $u$. In this situation, $H$ cannot contain $P_{u}$, since $P_{u}$ would meet $Q$ at both its endpoints. Thus, $Q$ can be replaced with $P_{u}$ in $G'$.

Suppose $Q$ runs from $v$ to $w$. $Q$ can then be replaced with $vw$ in $G'$.

Assume then that $Q$'s endpoints are $u$ and $w$.

If $\langle Y\setminus S\rangle$ is not used to form any part of $H$, then replace $Q$ in $G'$ with some path in $\langle Y\rangle$. Suppose then that part of $\langle Y\setminus S\rangle$ is used to form part of $H$ in $G$.

\textbf{1.1.} Suppose firstly that $Y\setminus S$ contains spoke-meets-rim vertices of $H$.

Since the centre of $H$ is not in $Y$, and since two vertices in $S$ are already used in $H$ as endpoints of the path $Q$, $Y\setminus S$ can contain only a single spoke-meets-rim vertex of $H$ (else it is routine to show that $H$ cannot be a subdivision of $W_{k}$). Call this spoke-meets-rim vertex $y$. Recall that $Q$ runs from $u$ to $w$. Thus, $H\cap \langle Y\rangle$ must meet $S$ at the vertices $t$, $v$, and some vertex $x_{1}$, where $x_{1} \in \{u, w\}$. Let $x_{2}$ be the vertex in $\{u, w\}$ such that $x_{1} \neq x_{2}$.
There are three possibilities:

\begin{itemize}
\item[(i)] $tHy$ is part of a spoke of $H$, while $x_{1}Hy$ and $vHy$ are part of the rim;
\item[(ii)] $x_{1}Hy$ is part of a spoke of $H$ (and $Q$ is part of the same spoke), while $tHy$ and $vHy$ are part of the rim; or 
\item[(iii)] $vHy$ is part of a spoke of $H$, while $tHy$ and $x_{1}Hy$ are part of the rim.
\end{itemize}

Suppose (i) is true. If $Q$ also forms part of the rim of $H$, then $Q$ and $H\cap \langle Y\rangle$ can be replaced in $G'$ with a path $P_{t}$ in $\langle Y\rangle$ from $t$ to $x_{2}$, a path in $\langle Y \rangle $ from $x_{1}$ to an internal vertex $y_{1}$ of $P_{t}$, and \emph{either} the path $P_{u}$ (if $x_{1} = u$), \emph{or} the edge $vw$ (if $x_{1} = w$). The part of $H$'s rim previously formed by $x_{1}Hy$, $vHy$, and $Q$ is now formed by $x_{2}P_{t}y_{1}$, the path from $x_{1}$ to $y_{1}$, and either $P_{u}$ or $vw$. If $Q$ forms part of a spoke of $H$, then $Q$ and $H\cap \langle Y \rangle $ can be replaced in $G'$ with a path $P_{t}$ in $\langle Y \rangle $ from $t$ to $x_{2}$, a path in $\langle Y \rangle $ from $x_{1}$ to an internal vertex $y_{1}$ of $P_{t}$, and \emph{either} the edge $vw$ (if $x_{1} = u$), \emph{or} the path $P_{u}$ (if $x_{1} = w$). In both cases, the part of spoke previously formed by $tHy$ is now formed by $tP_{t}y_{1}$. 

If (ii) is true, $Q$ and $H\cap \langle Y \rangle $ can be replaced in $G'$ with a path $P_{t}$ in $\langle Y \rangle $ from $t$ to $x_{2}$, and either the edge $vw$ or the path $P_{u}$ . Then the part of $H$'s rim previously formed by $tHy$ and $vHy$ is now formed by $P_{t}$ and either $vw$ or $P_{u}$. The part of spoke that was formed by $Q$ and $x_{1}Hy$ is no longer needed, as the rim now meets $x_{2}$, making $x_{2}$ a spoke-meets-rim vertex where previously it was not.

Suppose (iii) is true. If $Q$ also forms part of the rim of $H$, then $Q$ and $vHy$ can be replaced in $G'$ with the path $P_{u}$ and the edge $vw$. The part of $H$'s rim previously formed by $Q$ is now formed by $P_{u}$ and $vw$ (note that the paths $tHy$ and $x_{1}Hy$ are still used). The part of spoke previously formed by $vHy$ is no longer needed, as $v$ is now a spoke-meets-rim vertex. If $Q$ forms part of a spoke of $H$, then $Q$ and $H\cap \langle Y \rangle $ can be replaced in $G'$ with the path $P_{u}$, the edge $vw$, and a path in $\langle Y\rangle$ from $t$ to $x_{2}$. The parts of spokes previously formed by $vHy$ and $Q$ are no longer needed, as $v$ and $x_{2}$ are now spoke-meets-rim vertices.

\textbf{1.2.} Assume now that $\langle Y\setminus S\rangle$ contains only a single path belonging to $H$. Call this path $R$.

Since $Q$ already forms a path from $u$ to $w$, one of the following must hold:

\begin{itemize}
\item[(i)] $R$ has $v$ as an endpoint; or
\item[(ii)] $R$ forms a path from $x$ to $t$, where $x\in \{u, w\}$.
\end{itemize}

Suppose (i) holds. If $R$ runs from $v$ to $u$, replace it with $P_{u}$ in $G'$. If $R$ runs from $v$ to $w$, replace it with $vw$ in $G'$. In each case, use part of $\langle Y\rangle$ to replace $Q$ in $G'$, so that $X\setminus S$ is no longer needed. If $R$ runs from $v$ to $t$, then regardless of whether $Q$ and $R$ form parts of spokes of $H$ or of the rim of $H$, they can be replaced in $G'$ \emph{either} by the edge $vw$ and a path in $\langle Y\setminus S\rangle$ from $t$ to $u$, \emph{or} by the path $P_{u}$ and a path in $\langle Y\setminus S\rangle$ from $t$ to $w$. Figure \ref{rf8_12} illustrates some examples of different configurations that may occur in this case. Note that in some cases, parts of $H\cap \langle Z \rangle $ that were previously part of the rim of $H$ may become part of a spoke in $G'$, and vice versa.

\begin{figure}[ht]
\begin{center}
\includegraphics[width=0.8\textwidth]{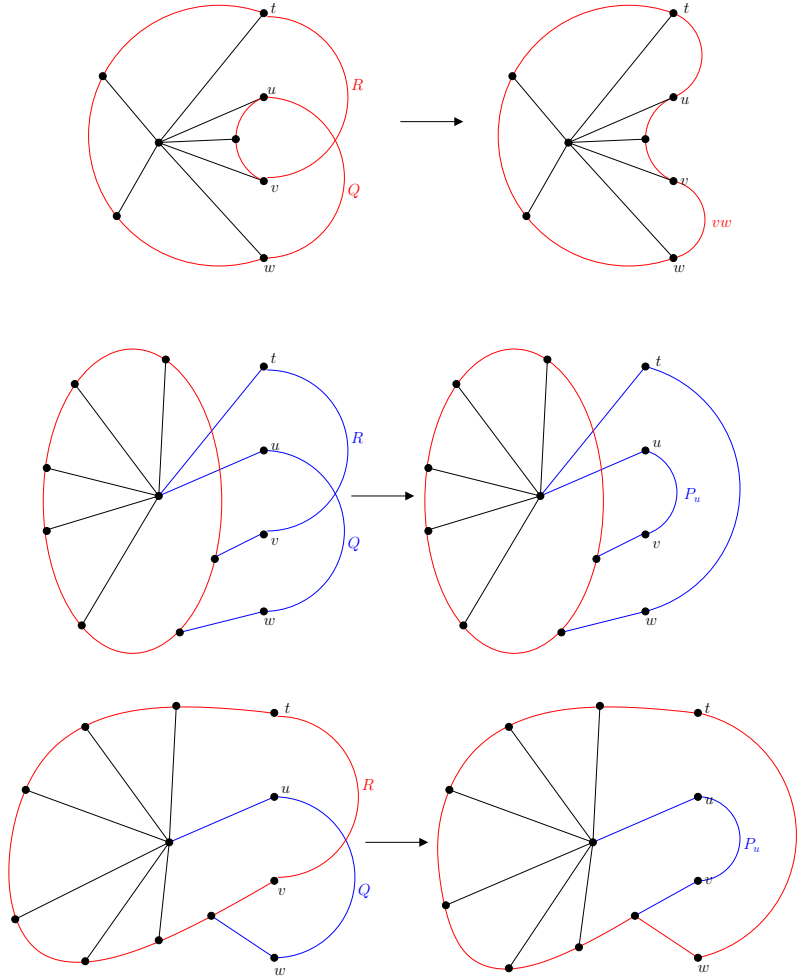}
\caption[Theorem \ref{reduction8}, Case 1.2]{Theorem \ref{reduction8}, Case 1.2: Replacing $Q$ and $R$ with different paths in $G'$ when $R$ runs from $v$ to $t$.}
\label{rf8_12}
\end{center}
\end{figure}

Suppose then that (ii) holds. Replace $Q$ in $G'$ with a path $P_{R}$ in $\langle Y\rangle$ that runs from $R$ to $w$ (if $x = u$), or $R$ to $u$ (if $x = w$). Let $p$ be the endpoint of $P_{R}$ that lies on $R$. If $x$ is not a spoke-meets-rim vertex in $H$, then $P_{R}$ also replaces the path $xRp$ in $G'$. If $x$ \emph{is} a spoke-meets-rim vertex in $H$, then $p$ replaces it as a spoke-meets-rim vertex in $G'$.

\textbf{2.} Suppose now that $X\setminus S$ contains spoke-meets-rim vertices belonging to $H$.

Since the centre of $H$ is not in $X$, $X\setminus S$ can contain only one spoke-meets-rim vertex of $H$.

If $v$ is not a spoke-meets-rim vertex in $H$, then replace $H\cap \langle X \rangle $ in $G'$ with $P_{u}$ and $vw$. Suppose then that $v$ is a spoke-meets-rim vertex in $G$.

If $\langle Y\setminus S\rangle$ is not used to form any part of $H$ in $G$, then use part of $\langle Y\setminus S\rangle$ in $G'$ to replace the part of $H$ previously contained in $\langle X\setminus S\rangle$. Suppose then that part of $\langle Y\setminus S\rangle$ is used to form part of $H$ in $G$. Only a single path belonging to $H$ can be contained in $\langle Y\rangle$. Call this path $Q$. The path $Q$ has $t$ as one endpoint, and some vertex $x$ as the other endpoint, where $x \in \{u, v, w\}$. 

Suppose $x = v$. If $Q$ forms part of a spoke of $H$, and the rim of $H$ in $\langle X\rangle$ runs from $u$ to $v$, then replace $H\cap \langle X \rangle $ in $G'$ with the edge $vw$ and a path in $\langle Y\rangle$ from $Q$ to $u$. If $Q$ forms part of a spoke of $H$, and the rim of $H$ in $\langle X\rangle$ runs from $v$ to $w$, then replace $H\cap \langle X \rangle $ in $G'$ with the path $P_{u}$ and a path in $\langle Y\rangle$  from $Q$ to $w$. If $Q$ forms part of the rim of $H$, and the rim of $H$ in $\langle X\rangle$ runs from $u$ to $v$, then replace $H\cap \langle X \rangle $ in $G'$ with the path $P_{u}$ and a path in $\langle Y\rangle$  from $Q$ to $w$.  If $Q$ forms part of the rim of $H$, and the rim of $H$ in $\langle X\rangle$ runs from $v$ to $w$, then replace $H\cap \langle X \rangle $ in $G'$ with the edge $vw$ and a path in $\langle Y\rangle$ from $Q$ to $u$.

Suppose then that $x \in \{u, w\}$.  Let $x_{2}$ be the other vertex in $\{u, w\}$, such that $x \neq x_{2}$. If $Q$ forms part of the rim of $H$, replace $H\cap \langle X \rangle $ in $G'$ with a path in $\langle Y\rangle$  from $Q$ to $x_{2}$, and \emph{either}  the path $P_{u}$ (if $x = u$), \emph{or} the edge $vw$ (if $x = w$). If $Q$ forms part of a spoke of $H$, replace $Q$ and $H\cap \langle X \rangle $ in $G'$ with a path $P_{t}$ in $Y$ from $t$ to $x_{2}$, and a path in $Y$ from $v$ to some internal vertex of $P_{t}$. 

Thus, if there exists a $W_{k}$-subdivision $H$ in $G$ that is centred in $Z\setminus S$, then a $W_{k}$-subdivision also exists in $G'$. If $H$ is instead centred in $Y\setminus S$, or in some other bridge of $G|S$ (other than $X$, which has already been dealt with), then essentially the same arguments used for the $Z\setminus S$ case can be used to show that a $W_{k}$-subdivision can still be formed in $G'$. 

Thus, Reduction \ref{r8} can be performed on $G$ without altering the existence or otherwise of a $W_{k}$-subdivision.
\end{proof}

\begin{red}
\label{r1_big}
Let $k$ be some integer $\ge 5$. Let $G$ be a 3-connected graph containing a separating set $S = \{t, u, v, w\}$ of vertices. Suppose there exist at least three bridges of $G|S$, $X$, $Y$, and $Z$, such that $X\cap S = \{u, v, w\}$, $Y\cap S = \{u, v, w\}$, $|X| \le k$, and $Z$ also contains $\{u, v, w\}$. Suppose that for each vertex $i\in \{u, v, w\}$, either $i$ has degree $< k$, or $X\setminus S$ contains at most one neighbour of $i$. Suppose that $v$ and $w$ are adjacent. Suppose also that $v$ and $u$ are either adjacent or joined by a path $P_{u_{1}}$ in some bridge $A$ of $G|S$ other than $X$, $Y$, or $Z$, and by a path $P_{u_{2}}$ in some bridge $B$ of $G|S$ other than $X$, $Y$, $Z$, or $A$.

Form $G'$ from $G$ by removing $X \setminus S$ and adding a single edge from $u$ to $w$, if such an edge does not already exist.
\end{red}

\begin{thm}
\label{reduction1big}
Let $G$ be some 3-connected graph on which Reduction \ref{r1_big} can be performed. Let $G'$ be the resulting graph after Reduction \ref{r1_big} has been performed on $G$. Then $G$ contains a $W_{k}$-subdivision if and only if $G'$ contains a $W_{k}$-subdivision, where $k \ge 5$.
\end{thm}

\begin{proof}
It is obvious that if $G'$ contains a $W_{k}$-subdivision, then $G$ will also, since $G$ contains a subdivision of $G'$.

Assume then that $G$ contains a $W_{k}$-subdivision, $H$.

Since $|X| \le k$, the maximum degree of any vertex in $X\setminus S$ is $k - 1$. Therefore, the centre of $H$ cannot be in $X \setminus S$.

\textbf{(a)} Suppose $H$ is centred on some vertex $v_{0} \in \{u, v, w\}$.

The rim of $H$ can pass through at most three bridges of $G|S$. Assume without loss of generality that $B$ is not one of these bridges. If $X$ is also not one of these bridges, then $X\setminus S$ can only contain part or all of a single spoke of $H$. This can be replaced in $G'$ by one of $P_{u_{2}}$, $vw$, or $uw$.

Suppose then that $X$ contains part of the rim of $H$. Since $v_{0}$ must then have degree $\ge k$, $v_{0}$ can have at most one neighbour in $X\setminus S$, by the hypothesis of the theorem. Thus, $X\setminus S$ contains at most one spoke-meets-rim vertex of $H$.

If $Y$ also contains part of the rim of $H$, then, since $X\cap S = Y\cap S = \{u, v, w\}$ the rim of $H$ must be entirely contained in $\langle X\cup Y\rangle$. Thus, at most two other bridges can contain parts of $H$: each of these bridges may contain a single spoke from $v_{0}$ to some vertex in $\{u, v, w\} - v_{0}$. Assume without loss of generality that $B$ is not one of these bridges. Then these two spokes can be replaced in $G'$ by two of $P_{u_{2}}$, $vw$, and $uw$.

Suppose then that $Y$ does not contain part of the rim of $H$. Then $Y$ can only be used to form some spoke $P$ in $H$ from $v_{0}$ to some vertex in $\{u, v, w\} - v_{0}$. If $P$ is a path from $u$ to $w$, replace $P$ with $uw$ in $G'$. If $P$ is a path from $u$ to $v$, replace $P$ with $P_{u_{2}}$ in $G'$ (since it is assumed that $B$ is not used to form any part of $H$ in $G$). If $P$ is a path from $v$ to $w$, then the edge $vw$ cannot be used to form part of $H$ in $G$ --- thus, use $vw$ to replace $P$ in $G'$. The bridge $Y$ can then be used in $G'$ to replace that part of $H$ formed by $X$ in $G$.

\textbf{(b)} Suppose then that $H$ is centred in $G - X$. Let $U$ be the bridge of $G|S$ containing the centre of $H$.

The rim of $H$ must be contained in $U$ and at most three other bridges of $G|S$.

Suppose the rim of $H$ is contained in four bridges of $G|S$. Then no other bridge of $G|S$ can contain any part of $H$, since every spoke of $H$ must be contained in $\langle U\rangle$. Thus, if the rim of $H$ does not pass through $X\setminus S$, the removal of $X\setminus S$ will have no effect on the existence of $H$. Assume then that the rim of $H$ passes through $X\setminus S$. Without loss of generality, assume that $B$ is not one of the four bridges containing the rim of $H$. Then that part of the rim of $H$ contained in $\langle X\rangle$ can be replaced in $G'$ by one of $P_{u_{2}}$, $vw$, or $uw$.

Suppose then that the rim of $H$ is contained in at most three bridges of $G|S$. Without loss of generality, assume that $B$ is not one of these bridges. Note that $\langle X\setminus S\rangle$ can contain either a single path belonging to $H$, or a single spoke-meets-rim vertex of $H$ and the three paths that meet this vertex. If the former is true, replace this path in $G'$ with one of $P_{u_{2}}$, $vw$, or $uw$. If the latter is true, replace these paths in $G'$ with two of $P_{u_{2}}$, $vw$, and $uw$, so that some vertex in $\{u, v, w\}$ becomes the spoke-meets-rim vertex previously contained in $X\setminus S$.
\end{proof}

\section{Separating sets}
\label{separatingsets}

For each type of separating set $S$ defined in this section, we prove a theorem. Each such theorem supposes the existence of some graph $G$ containing such a set $S$, and forms two smaller graphs from $G$ by performing the following steps:

\begin{itemize}
\item[(a)] Divide $G$ along $S$ into two components $G_{1}$ and $G_{2}$.
\item[(b)] Add some small structure $X$ to $G_{1}$ to form $G'_{1}$, where $|V(X)| < |V(G_{2})|$.
\item[(c)] Add some small structure $Y$ to $G_{2}$ to form $G'_{2}$, where $|V(Y)| < |V(G_{1})|$.
\end{itemize}

It is then shown that if $G$ contains a $W_{7}$-subdivision, then either $G'_{1}$ or $G'_{2}$ will also, and if $G$ does not contain a $W_{7}$-subdivision, then neither $G'_{1}$ nor $G'_{2}$ will.

These theorems enable such separating sets to be used in an algorithm for SHP($W_{7}$), in that the input graph $G$ can be separated along any existing separating set of an appropriate type, and each resulting component then examined using a divide-and-conquer method. This method is also used in \cite{Farr88, Robinson08} with internal 3-edge-cutsets and in \cite{Robinson08} with internal 4-edge-cutsets, both of which are also used in the $W_{7}$ algorithm outlined here in Section \ref{algorithm}. The types of separating sets defined in this section are more complex, however, and fall into two distinct categories: the edge-vertex-cutsets, of which there are eight different types; and the internal $\{1,1,1,1\}$-cutset, which consists of four disjoint edges.

\vspace{0.1in}
\noindent \textbf{Definition}

A \emph{type 1 edge-vertex-cutset} in a graph $G$ is a set $S = \{e_{1}, e_{2}, v\}$ of two edges $e_{1}, e_{2}$ of $G$ and one vertex $v$ of $G$ such that $G - S$ is disconnected, with each component having at least four vertices.

\begin{figure}[!ht]
\begin{center}
\includegraphics[width=0.5\textwidth]{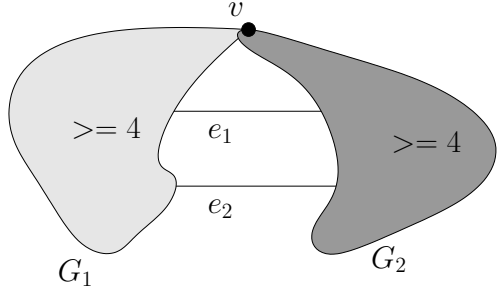}
\caption{Type 1 edge-vertex-cutset}
\label{evc1}
\end{center}
\end{figure}

\begin{thm}
\label{edgevertex}

Let $G$ be a 3-connected graph which contains a type 1 edge-vertex-cutset $S = \{e_{1}, e_{2}, v\}$. Let $G_{1}, G_{2}$ be the components of $G - S$. Let $u_{1}, \ldots , u_{i}$ be the neighbours of $v$ in $G_{1}$, and let $v_{1}, \ldots, v_{j}$ be the neighbours of $v$ in $G_{2}$.

Form $G'_{1}$ from $G$ by replacing $G_{2}$ with the subgraph $X$, where:

\begin{itemize}
\item if $j < 3$, $X$ contains only the vertices $v_{1}, \ldots, v_{j}$, all of which are made adjacent if they were not already, such that $v_{1}$ is an endpoint of $e_{1}$ and $v_{j}$ is an endpoint of $e_{2}$;
\item if $j \ge 3$, $X$ contains only the vertices $v_{1}, v_{2}, v_{3}$, such that $v_{1}$ is an endpoint of $e_{1}$, $v_{3}$ is an endpoint of $e_{2}$, and $v_{2}$ is adjacent to both $v_{1}$ and $v_{3}$.
\end{itemize}

Form $G'_{2}$ from $G$ by replacing $G_{1}$ with the subgraph $Y$, where:

\begin{itemize}
\item if $i < 3$, $Y$ contains only the vertices $u_{1}, \ldots, u_{i}$, all of which are made adjacent if they were not already, such that $u_{1}$ is an endpoint of $e_{1}$ and $u_{i}$ is an endpoint of $e_{2}$;
\item if $i \ge 3$, $Y$ contains only the vertices $u_{1}, u_{2}, u_{3}$, such that $u_{1}$ is an endpoint of $e_{1}$, $u_{3}$ is an endpoint of $e_{2}$, and $u_{2}$ is adjacent to both $u_{1}$ and $u_{3}$.
\end{itemize}

Then $G$ contains a $W_{7}$-subdivision if and only if at least one of $G'_{1}$ and $G'_{2}$ contains a $W_{7}$-subdivision.
\end{thm}

\begin{proof}

$(\Rightarrow )$ Suppose firstly that $G$ contains a $W_{7}$-subdivision $H$.

If $H$ is entirely contained in $G_{1}$ or $G_{2}$, then $H$ will also be contained in $G'_{1}$ or $G'_{2}$ respectively. Suppose then that $H$ contains parts of both $G_{1}$ and $G_{2}$.

Suppose firstly that $H$ is centred in $G - v$. Without loss of generality, suppose $H$ is centred in $G_{1}$. Then one of the following must be true:

\begin{itemize}
\item[(A)] $G_{2}$ contains a single path of $H$;
\item[(B)] $v$ forms a spoke-meets-rim vertex of $H$, and two separate paths of $H$ enter $G_{2}$ at this vertex and leave $G_{2}$ via $e_{1}$ and $e_{2}$; or
\item[(C)] $G_{2}$ contains a single spoke-meets-rim vertex of $H$.
\end{itemize}

Each of the three possibilities can also be formed by $X$ in $G'_{1}$. (If $H$ is instead centred in $G_{2}$, then similarly $H$ will also be contained in $G'_{2}$.)

Suppose now that $H$ is centred on $v$. Since $|N_{H}(v)| = 7$, one of $G_{1}$, $G_{2}$ can contain at most three of the vertices in $N_{H}(v)$. Assume without loss of generality that this is $G_{2}$. Then, since $G_{2}$ can contain at most three spoke-meets-rim vertices of $H$, $X$ can be used in $G'_{2}$ to replace the parts of $H$ previously contained in $G_{2}$. (If instead $G_{1}$ contains no more than three vertices in $N_{H}(v)$, then the same argument applies to show that $H$ is also contained in $G'_{1}$.)

Thus, whenever $G$ contains a $W_{7}$-subdivision, at least one of $G'_{1}$ or $G'_{2}$ does also.

$(\Leftarrow )$ Suppose now that either $G'_{1}$ or $G'_{2}$ contains a $W_{7}$-subdivision --- assume $G'_{1}$ without loss of generality. If $H$ is entirely contained in $G'_{1} - X$, then $H$ is also contained in $G$. Suppose then that $H$ contains parts of $X$. One of the following must hold:

\begin{itemize}
\item[(a)] $X$ contains a single path of $H$;
\item[(b)] $v$ forms a spoke-meets-rim vertex of $H$, and two separate paths of $H$ enter $X$ at this vertex and leave $X$ via $e_{1}$ and $e_{2}$;
\item[(c)] $X$ contains a single spoke-meets-rim vertex of $H$; or
\item[(d)] $H$ is centred on $v$, and $X$ contains two or three spoke-meets-rim vertices of $H$.
\end{itemize}

If (a), (b), or (c) hold, then by 3-connectivity of $G$, the required paths can be formed in $G_{2}$. If (d) holds, then by Lemma \ref{lemma1} or Lemma \ref{lemma2}, the required structure can also be formed in $G_{2}$.

Thus, whenever $G'_{1}$ contains a $W_{7}$-subdivision, $G$ does also. (By the same argument, whenever $G'_{2}$ contains a $W_{7}$-subdivision, $G$ does also.)
\end{proof}

\noindent \textbf{Definition}

A \emph{type 1a edge-vertex-cutset} in a graph $G$ is a set $S = \{e_{1}, e_{2}, v\}$ of two edges $e_{1}, e_{2}$ of $G$ and one vertex $v$ of $G$ such that $G - S$ is disconnected, with each component having at least three vertices, and $v$ has degree $< 7$.

\begin{thm}
\label{edgevertex1a}

Let $G$ be a 3-connected graph which contains a type 1a edge-vertex-cutset $S = \{e_{1}, e_{2}, v\}$. Let $G_{1}, G_{2}$ be the components of $G - S$. Let $u_{1}, \ldots , u_{i}$ be the neighbours of $v$ in $G_{1}$, and let $v_{1}, \ldots, v_{j}$ be the neighbours of $v$ in $G_{2}$.

Form $G'_{1}$ from $G$ by replacing $G_{2}$ with the subgraph $X$, where $X$ contains only the two adjacent vertices $x_{1}, x_{2}$, both of which are made adjacent to $v$, such that $x_{1}$ is an endpoint of $e_{1}$ and $x_{2}$ is an endpoint of $e_{2}$.

Form $G'_{2}$ from $G$ by replacing $G_{1}$ with the subgraph $Y$, where $Y$ contains only the two adjacent vertices $y_{1}, y_{2}$, both of which are made adjacent to $v$, such that $y_{1}$ is an endpoint of $e_{1}$ and $y_{2}$ is an endpoint of $e_{2}$.

Then $G$ contains a $W_{7}$-subdivision if and only if at least one of $G'_{1}$ and $G'_{2}$ contains a $W_{7}$-subdivision.
\end{thm}

\begin{proof}

$(\Rightarrow )$ Suppose firstly that $G$ contains a $W_{7}$-subdivision $H$.

If $H$ is entirely contained in $G_{1}$ or $G_{2}$, then $H$ will also be contained in $G'_{1}$ or $G'_{2}$ respectively. Suppose then that $H$ contains parts of both $G_{1}$ and $G_{2}$.

Since $v$ has degree $< 7$, $H$ must be centred in $G - v$. Without loss of generality, suppose $H$ is centred in $G_{1}$. Then one of the following must be true:

\begin{itemize}
\item[(A)] $G_{2}$ contains a single path of $H$;
\item[(B)] $v$ forms a spoke-meets-rim vertex of $H$, and two separate paths of $H$ enter $G_{2}$ at this vertex and leave $G_{2}$ via $e_{1}$ and $e_{2}$; or
\item[(C)] $G_{2}$ contains a single spoke-meets-rim vertex of $H$.
\end{itemize}

Each of the three possibilities can also be formed by $X$ in $G'_{1}$. (If $H$ is instead centred in $G_{2}$, then similarly $H$ will also be contained in $G'_{2}$.)

Thus, whenever $G$ contains a $W_{7}$-subdivision, at least one of $G'_{1}$ or $G'_{2}$ does also.

$(\Leftarrow )$ Suppose now that either $G'_{1}$ or $G'_{2}$ contains a $W_{7}$-subdivision --- assume $G'_{1}$ without loss of generality. If $H$ is entirely contained in $G'_{1} - X$, then $H$ is also contained in $G$. Suppose then that $H$ contains parts of $X$. One of the following must hold:

\begin{itemize}
\item[(a)] $X$ contains a single path of $H$;
\item[(b)] $v$ forms a spoke-meets-rim vertex of $H$, and two separate paths of $H$ enter $X$ at this vertex and leave $X$ via $e_{1}$ and $e_{2}$; or
\item[(c)] $X$ contains a single spoke-meets-rim vertex of $H$.
\end{itemize}

By 3-connectivity of $G$, the required paths can be formed in $G_{2}$ for any of these three cases. Thus, whenever $G'_{1}$ contains a $W_{7}$-subdivision, $G$ does also. (By the same argument, whenever $G'_{2}$ contains a $W_{7}$-subdivision, $G$ does also.)
\end{proof}

\vspace{0.1in}
\noindent \textbf{Definition}

A \emph{type 2 edge-vertex-cutset} in a graph $G$ is a set $S = \{e, v_{1}, v_{2}\}$ of two vertices $v_{1}, v_{2}$ of $G$ and one edge $e$ of $G$ such that $G - S$ is disconnected, with each component having at least four vertices.

\begin{thm}
\label{edgevertex2}

Let $G$ be a 3-connected graph with no type 1 edge-vertex-cutsets, but which contains a type 2 edge-vertex-cutset $S = \{e, v_{1}, v_{2}\}$. Let $G_{1}, G_{2}$ be the components of $G - S$.

Form $G'_{1}$ from $G$ by replacing $G_{2}$ with the subgraph $X$, where:

\begin{itemize}
\item[(i)] if $v_{1}$ and $v_{2}$ both have fewer than three neighbours in $G_{2}$, $X$ contains two adjacent vertices, $x_{1}$ and $x_{2}$, each of which is adjacent to both $v_{1}$ and $v_{2}$, such that $x_{1}$ is an endpoint of $e$;
\item[(ii)] if both $v_{1}$ and $v_{2}$ have $\ge 3$ neighbours in $G_{2}$, $X$ contains three pairwise-adjacent vertices, $x_{1}, x_{2}, x_{3}$, each of which is adjacent to both $v_{1}$ and $v_{2}$, such that $x_{1}$ is an endpoint of $e$;
\item[(iii)] if only one of $v_{1}$ and $v_{2}$ has $\ge 3$ neighbours in $G_{2}$, $X$ is formed as in (ii), but with no edge between $x_{1}$ and the member of $\{v_{1}, v_{2}\}$ with fewer than three neighbours in $G_{2}$.
\end{itemize}

Form $G'_{2}$ from $G$ by replacing $G_{1}$ with the subgraph $Y$ in the same manner.

Then $G$ contains a $W_{7}$-subdivision if and only if at least one of $G'_{1}$ and $G'_{2}$ contains a $W_{7}$-subdivision.
\end{thm}

\begin{proof}

$(\Rightarrow )$ Suppose firstly that $G$ contains a $W_{7}$-subdivision $H$.

If $H$ is entirely contained in $G_{1}$ or $G_{2}$, then $H$ will also be contained in $G'_{1}$ or $G'_{2}$ respectively. Suppose then that $H$ contains parts of both $G_{1}$ and $G_{2}$.

Suppose firstly that $H$ is centred in $G \setminus \{v_{1}, v_{2}\}$. Without loss of generality, suppose $H$ is centred in $G_{1}$. Then one of the following must be true:

\begin{itemize}
\item[(A)] $G_{2}$ contains a single path of $H$;
\item[(B)] one of $v_{1}, v_{2}$ forms a spoke-meets-rim vertex of $H$, and two separate paths of $H$ enter $G_{2}$ at this vertex and leave $G_{2}$ via the other two members of $S$; or
\item[(C)] $G_{2}$ contains a single spoke-meets-rim vertex of $H$.
\end{itemize}

Each of the three possibilities can easily be formed by $X$ in $G'_{1}$. (If $H$ is instead centred in $G_{2}$, then similarly $H$ will also be contained in $G'_{2}$.)

Suppose now that $H$ is centred on either $v_{1}$ or $v_{2}$ --- assume $v_{1}$ without loss of generality. By the same argument used in Theorem \ref{edgevertex}, some component $G_{x} \in \{G_{1}, G_{2}\}$ contains at most three members of $N_{H}(v_{1})$, and thus contains at most three spoke-meets-rim vertices of $H$. Suppose $G_{x} = G_{2}$. If $G_{2}$ contains three spoke-meets-rim vertices of $H$, then $G_{2}$ must contain at least three neighbours of $v_{1}$, thus $X$ is formed as in either case (ii) or (iii) of the hypothesis, and vertices $x_{1}$, $x_{2}$ and $x_{3}$ can be used to form the three spoke-meets-rim vertices previously contained in $G_{2}$. If $G_{2}$ contains fewer than three spoke-meets-rim vertices of $H$, then one or both of $x_{1}$ and $x_{2}$ can be used in $X$ to form these spoke-meets-rim vertices.

If $G_{x} = G_{1}$, $Y$ can be used in $G'_{1}$ in the same manner to replace the parts of $H$ previously contained in $G_{1}$. 

Thus, whenever $G$ contains a $W_{7}$-subdivision, at most one of $G'_{1}$ or $G'_{2}$ does also.

$(\Leftarrow )$ Suppose now that either $G'_{1}$ or $G'_{2}$ contains a $W_{7}$-subdivision --- assume $G'_{1}$ without loss of generality. If $H$ is entirely contained in $G'_{1} - X$, then $H$ is also contained in $G$. Suppose then that $H$ contains parts of $X$. One of the following must hold:

\begin{itemize}
\item[(a)] $X$ contains a single path of $H$;
\item[(b)] one of $v_{1}, v_{2}$ forms a spoke-meets-rim vertex of $H$, and two separate paths of $H$ enter $G_{2}$ at this vertex and leave $G_{2}$ via the other two members of $S$;
\item[(c)] $X$ contains a single spoke-meets-rim vertex of $H$;
\item[(d)] $H$ is centred on $v_{1}$ or $v_{2}$, and $X$ contains two spoke-meets-rim vertices of $H$; or
\item[(e)] $H$ is centred on $v_{1}$ or $v_{2}$, and $X$ contains three spoke-meets-rim vertices of $H$.
\end{itemize}

If (a), (b), or (c) hold, then by 3-connectivity of $G$, the required paths can be formed in $G_{2}$.

Suppose (d) holds. Without loss of generality, suppose $H$ is centred on $v_{1}$. If $v_{1}$ has only one neighbour in $G_{2}$, then a type 1 edge-vertex-cutset exists in $G$. Assume then that $G_{2}$ contains $\ge 2$ neighbours of $v_{1}$. By Lemma \ref{lemma1}, then, the required structure can be formed in $G_{2}$.

Suppose (e) holds. Since $X$ contains three vertices in this case, $X$ must have been formed using case (ii) or (iii), thus $G_{2}$ must contain at least three neighbours of $v_{1}$. By Lemma \ref{lemma2}, then, the required structure can be formed in $G_{2}$.

Thus, whenever $G'_{1}$ contains a $W_{7}$-subdivision, $G$ does also. (By the same argument, whenever $G'_{2}$ contains a $W_{7}$-subdivision, $G$ does also.)
\end{proof}

\vspace{0.1in}
\noindent \textbf{Definition}

A \emph{type 2a edge-vertex-cutset} in a graph $G$ is a set $S = \{e, v_{1}, v_{2}\}$ of two vertices $v_{1}, v_{2}$ of $G$ and one edge $e$ of $G$ such that $G - S$ is disconnected with each component having at least three vertices, and for each vertex $v_{i}$, $1 \le i \le 2$, either:

\begin{itemize}
\item $v_{i}$ has at most two neighbours in one of the components of $G - S$; or
\item $v_{i}$ has degree $< 7$.
\end{itemize}

\begin{thm}
\label{edgevertex2a}
Let $G$ be a 3-connected graph with no type 1 edge-vertex-cutsets, but which contains a type 2a edge-vertex-cutset $S = \{e, v_{1}, v_{2}\}$. Let $G_{1}, G_{2}$ be the components of $G - S$.

Form $G'_{1}$ from $G$ by replacing $G_{2}$ with the subgraph $X$, where $X$ contains two adjacent vertices, $x_{1}$ and $x_{2}$, each of which is adjacent to both $v_{1}$ and $v_{2}$, such that $x_{1}$ is an endpoint of $e$.

Form $G'_{2}$ from $G$ by replacing $G_{1}$ with the subgraph $Y$ in the same manner.

Then $G$ contains a $W_{7}$-subdivision if and only if at least one of $G'_{1}$ and $G'_{2}$ contains a $W_{7}$-subdivision.
\end{thm}

\begin{proof}

$(\Rightarrow )$ Suppose firstly that $G$ contains a $W_{7}$-subdivision $H$.

If $H$ is entirely contained in $G_{1}$ or $G_{2}$, then $H$ will also be contained in $G'_{1}$ or $G'_{2}$ respectively. Suppose then that $H$ contains parts of both $G_{1}$ and $G_{2}$. If $H$ is centred in $G \setminus \{v_{1}, v_{2}\}$, then by the same arguments used in Theorem \ref{edgevertex2}, $H$ is also contained in either $G'_{1}$ or $G'_{2}$.

Suppose now that $H$ is centred on $v_{1}$ or $v_{2}$ --- assume $v_{1}$ without loss of generality. Since this means that $v_{1}$ has degree $\ge 7$, some component $G_{x} \in \{G_{1}, G_{2}\}$ must contain at most two neighbours of $v_{1}$, and thus contains at most two spoke-meets-rim vertices of $H$. By Lemma \ref{lemma1}, this structure can be replaced by $X$ if $G_{x} = G_{2}$, or $Y$ if $G_{x} = G_{1}$.

Thus, whenever $G$ contains a $W_{7}$-subdivision, at least one of $G'_{1}$ or $G'_{2}$ does also.

$(\Leftarrow )$ Suppose now that either $G'_{1}$ or $G'_{2}$ contains a $W_{7}$-subdivision --- assume $G'_{1}$ without loss of generality. If $H$ is entirely contained in $G'_{1} - X$, then $H$ is also contained in $G$. Suppose then that $H$ contains parts of $X$. One of the following must hold:

\begin{itemize}
\item[(a)] $X$ contains a single path of $H$;
\item[(b)] one of $v_{1}, v_{2}$ forms a spoke-meets-rim vertex of $H$, and two separate paths of $H$ enter $G_{2}$ at this vertex and leave $G_{2}$ via the other two members of $S$;
\item[(c)] $X$ contains a single spoke-meets-rim vertex of $H$; or
\item[(d)] $H$ is centred on $v_{1}$ or $v_{2}$, and $X$ contains two spoke-meets-rim vertices of $H$.
\end{itemize}

By the same arguments used in Theorem \ref{edgevertex2}, whenever $G'_{1}$ contains a $W_{7}$-subdivision, $G$ does also. (Similarly, whenever $G'_{2}$ contains a $W_{7}$-subdivision, $G$ does also.)
\end{proof}

\vspace{0.1in}
\noindent \textbf{Definition}

A \emph{type 3 edge-vertex-cutset} in a graph $G$ is a set $S = \{v, e_{1}, e_{2}, e_{3}, e_{4}\}$ of one vertex $v$ of $G$ and four edges $e_{1}, \ldots, e_{4}$ of $G$ such that $G - S$ is disconnected, with each component having at least four vertices, and with one of the components containing exactly two vertices incident with $e_{1}, \ldots , e_{4}$.

\begin{thm}
\label{edgevertex3}

Let $G$ be a 3-connected graph with no type 1 or 2 edge-vertex-cutsets, but which contains a type 3 edge-vertex-cutset $S = \{v, e_{1}, e_{2}, e_{3}, e_{4}\}$. Let $G_{1}$ be the component of $G - S$ that contains exactly two vertices, say $v_{1}$ and $v_{2}$, incident with $e_{1}, \ldots , e_{4}$, and let $G_{2}$ be the other component of $G - S$.

Form $G'_{1}$ from $G$ by replacing $G_{2}$ with the subgraph $X$, where:

\begin{itemize}
\item[(i)] if $v$ has fewer than three neighbours in $G_{2}$, $X$ contains two adjacent vertices, $x_{1}$ and $x_{2}$, each of which is adjacent to $v$, such that each of $x_{1}, x_{2}$ forms an endpoint of exactly two edges in $e_{1}, \ldots, e_{4}$;
\item[(ii)] if $v$ has $\ge 3$ neighbours in $G_{2}$, $X$ contains three pairwise-adjacent vertices, $x_{1}, x_{2}, x_{3}$, each of which is adjacent to $v$, such that each of $x_{1}$ and $x_{3}$ forms an endpoint of exactly one edge in $e_{1}, \ldots, e_{4}$, while $x_{2}$ forms an endpoint of the two remaining edges in $e_{1}, \ldots, e_{4}$.
\end{itemize}

Form $G'_{2}$ from $G$ by replacing $G_{1}$ with the subgraph $Y$, where:

\begin{itemize}
\item[(i)] if $v$ has fewer than three neighbours in $G_{1}$, $Y$ contains only the two vertices $v_{1}$ and $v_{2}$, and the edges $vv_{1}$, $vv_{2}$, and $v_{1}v_{2}$;
\item[(ii)] if $v$ has $\ge 3$ neighbours in $G_{1}$, $Y$ contains the vertices $v_{1}$ and $v_{2}$ and a third vertex, $y$, such that these three vertices are pairwise-adjacent and are each adjacent to $v$.
\end{itemize}

Then $G$ contains a $W_{7}$-subdivision if and only if at least one of $G'_{1}$ and $G'_{2}$ contains a $W_{7}$-subdivision.
\end{thm}

\begin{proof}

If either $v_{1}$ or $v_{2}$ has only one incident edge in $e_{1}, \ldots, e_{4}$, then a type 2 edge-vertex-cutset exists in $G$. Assume then that this is not the case; thus, $v_{1}$ and $v_{2}$ each have exactly two incident edges in $e_{1}, \ldots, e_{4}$. Without loss of generality, assume that $v_{1}$ is incident with $e_{1}$ and $e_{2}$, and that $v_{2}$ is incident with $e_{3}$ and $e_{4}$.

$(\Rightarrow )$ Suppose firstly that $G$ contains a $W_{7}$-subdivision $H$.

If $H$ is entirely contained in $G_{1}$ or $G_{2}$, then $H$ will also be contained in $G'_{1}$ or $G'_{2}$ respectively. Suppose then that $H$ contains parts of both $G_{1}$ and $G_{2}$.

The centre of $H$ can be either in $G_{1}$, in $G_{2}$, or $v$. Consider each of the cases.

\vspace{0.1in}
\noindent \textbf{Case 1}

Suppose $H$ is centred in $G_{1}$. Then one of the following must be true:

\begin{itemize}
\item[(A)] $G_{2}$ contains a single path of $H$;
\item[(B)] one of $v$, $v_{1}$, or $v_{2}$ is contained in $H$, and two separate paths of $H$ leave $G_{1}$ at this vertex and return to $G_{1}$ via other members of $S$, such that these paths are vertex-disjoint within $G_{2}$;
\item[(C)] $G_{2}$ contains a single spoke-meets-rim vertex of $H$; or
\item[(D)] $H$ is centred on one of $v_{1}, v_{2}$, and $G_{2}$ contains two spoke-meets-rim vertices of $H$.
\end{itemize}

In each case, the required structure can easily be formed in $X$, giving a $W_{7}$-subdivision in $G'_{1}$.

\vspace{0.1in}
\noindent \textbf{Case 2}

Suppose now $H$ is centred in $G_{2}$. Then one of the following must be true:

\begin{itemize}
\item[(A)] $G_{1}$ contains a single path of $H$;
\item[(B)] $G_{1}$ contains two disjoint paths of $H$, one being just a single vertex from $\{v_{1}, v_{2}\}$;
\item[(C)] $G_{1}$ contains a single spoke-meets-rim vertex of $H$;
\item[(D)] $G_{1}$ contains two spoke-meets-rim vertices of $H$; or
\item[(E)] $v_{1}$, $v_{2}$, and some third vertex in $G_{1}$ form three spoke-meets-rim vertices of $H$.
\end{itemize}

In each case, the required structure can easily be formed in $Y$, giving a $W_{7}$-subdivision in $G'_{2}$.

\vspace{0.1in}
\noindent \textbf{Case 3}

Suppose now that $H$ is centred on $v$. Since $|N_{H}(v)| = 7$, some component $G_{x} \in \{G_{1}$, $G_{2}\}$ contains at most three of the vertices in $N_{H}(v)$. If $G_{x} = G_{2}$, then, since $G_{2}$ can contain at most three spoke-meets-rim vertices of $H$, $X$ can be used in $G'_{2}$ to replace the parts of $H$ previously contained in $G_{2}$. If $G_{x} = G_{1}$, $Y$ can be used in $G'_{1}$ to replace the parts of $H$ previously contained in $G_{1}$. 

Thus, whenever $G$ contains a $W_{7}$-subdivision, at least one of $G'_{1}$ or $G'_{2}$ does also.

$(\Leftarrow )$ Suppose now that either $G'_{1}$ or $G'_{2}$ contains a $W_{7}$-subdivision.

\vspace{0.1in}
\noindent \textbf{Case 1}

Suppose firstly that $G'_{1}$ contains a $W_{7}$-subdivision. If $H$ is entirely contained in $G'_{1} - X$, then $H$ is also contained in $G$. Suppose then that $H$ contains parts of $X$. One of the following must hold:

\begin{itemize}
\item[(a)] $X$ contains a single path of $H$;
\item[(b)] one of $v$, $v_{1}$, or $v_{2}$ is contained in $H$, and two separate paths of $H$ leave $G_{1}$ at this vertex and return to $G_{1}$ via other members of $S$, such that these paths are vertex-disjoint within $X$;
\item[(c)] $X$ contains a single spoke-meets-rim vertex of $H$;
\item[(d)] $H$ is centred on one of $v$, $v_{1}$, or $v_{2}$, and $X$ contains two spoke-meets-rim vertices of $H$; or
\item[(e)] $H$ is centred on $v$, and $X$ contains three spoke-meets-rim vertices of $H$.
\end{itemize}

If (a) or (c) hold, then by 3-connectivity of $G$, the required paths can be formed in $G_{2}$.

Suppose (b) holds. The two paths in $G'_{1}$ must meet at either $v_{1}$ or $v_{2}$ --- suppose $v_{1}$ without loss of generality. Thus, one of the paths contains $e_{1}$, while the other contains $e_{2}$. By 3-connectivity, there must be two paths in $G_{2}\cup S\cup \{v_{1}, v_{2}\}$ joining $v_{1}$ to $\{v, v_{2}\}$ such that these paths are vertex-disjoint except at $v_{1}$, otherwise the removal of $v_{1}$ and some other vertex in $G_{2}$ will disconnect the graph. Use these two paths to replace the original paths in $G'_{1}$.

Suppose (d) holds. If $v$ has only one neighbour in $G_{2}$, then a type 2 edge-vertex-cutset exists in $G$. Assume then that $G_{2}$ contains at least two neighbours of $v$. Thus, if $H$ is centred on $v$, by Lemma \ref{lemma1} the required structure can also be formed in $G_{2}$. Similarly, if $H$ is centred on $v_{1}$ or $v_{2}$, then by Lemma \ref{lemma1} the required structure can be formed in $G_{2}$, since each of these vertices has two neighbours in $G_{2}$.

Suppose (e) holds.  Since $X$ contains three vertices in this case, $X$ must have been formed using case (ii) in the Theorem, thus $G_{2}$ must contain at least three neighbours of $v$. By Lemma \ref{lemma2}, then, the required structure can be formed in $G_{2}$.

\vspace{0.1in}
\noindent \textbf{Case 2}

Suppose now that $G'_{2}$ contains a $W_{7}$-subdivision. If $H$ is entirely contained in $G'_{2} - X$, then $H$ is also contained in $G$. Suppose then that $H$ contains parts of $Y$. One of the following must hold:

\begin{itemize}
\item[(a)] $Y$ contains a single path of $H$;
\item[(b)] $Y$ contains two disjoint paths of $H$;
\item[(c)] $Y$ contains a single spoke-meets-rim vertex of $H$;
\item[(d)] $Y$ contains two spoke-meets-rim vertices of $H$; or
\item[(e)] $Y$ contains three spoke-meets-rim vertices of $H$.
\end{itemize}

If (a), (b), (c), or (d) hold, then by 3-connectivity of $G$, the required paths can be formed in $G_{1}$.

Suppose (e) holds. By 3-connectivity, there exists a path $P_{1}$ in $\langle V(G_{1})\cup \{v\}\rangle$ from $v_{1}$ to $v$, and a path $P_{2}$ in $\langle V(G_{1})\cup \{v\}\rangle$ from $v_{2}$ to $v$. Let $p$ be the vertex closest to $v_{1}$ along $P_{1}$ where these paths meet. Then the paths $P_{2}$ and $v_{1}P_{1}p$ form the required structure in $G_{1}$, with $v_{1}$, $v_{2}$, and $p$ forming the three spoke-meets-rim vertices.                                                                                                                                                                                                                                                                                                                                                                                                                                                                                                                                                                                                                                                                                                                                                                                                                                                                                                                                                                                                                                                                                                                                                                                                                                                                                                                                                                                                                                                                                                                                                                                                                                                                                                                                                                                                                                                                                                                                                                                                                                                                                                                                                                                                                                                                                                                                                                                                                                                                                                                   

Thus, whenever one of $G'_{1}$ or $G'_{2}$ contains a $W_{7}$-subdivision, $G$ does also.
\end{proof}

\vspace{0.1in}
\noindent \textbf{Definition}

A \emph{type 3a edge-vertex-cutset} in a graph $G$ is a set $S = \{v, e_{1}, e_{2}, e_{3}, e_{4}\}$ of one vertex $v$ of $G$ and four edges $e_{1}, \ldots, e_{4}$ of $G$ such that $G - S$ is disconnected with each component having at least three vertices, $v$ either has degree $< 7$, or has at most two neighbours in one of the components of $G - S$, and one of the components of $G - S$ contains exactly two vertices, $v_{1}$ and $v_{2}$, incident with $e_{1}, \ldots , e_{4}$, such that $v_{1}$ is incident with $e_{1}$ and $e_{2}$, and $v_{2}$ is incident with $e_{3}$ and $e_{4}$.

\begin{thm}
\label{edgevertex3a}

Let $G$ be a 3-connected graph with no type 1 or 2 edge-vertex-cutsets, but which contains a type 3a edge-vertex-cutset $S = \{v, e_{1}, e_{2}, e_{3}, e_{4}\}$. Let $G_{1}$ be the component of $G - S$ that contains exactly two vertices, say $v_{1}$ and $v_{2}$, incident with $e_{1}, \ldots , e_{4}$, and let $G_{2}$ be the other component of $G - S$.

Form $G'_{1}$ from $G$ by replacing $G_{2}$ with the subgraph $X$, where $X$ contains two adjacent vertices, $x_{1}$ and $x_{2}$, each of which is adjacent to $v$, such that each of $x_{1}, x_{2}$ forms an endpoint of exactly two edges in $e_{1}, \ldots, e_{4}$.

Form $G'_{2}$ from $G$ by replacing $G_{1}$ with the subgraph $Y$, where $Y$ contains only the two vertices $v_{1}$ and $v_{2}$, and the edges $vv_{1}$, $vv_{2}$, and $v_{1}v_{2}$.

Then $G$ contains a $W_{7}$-subdivision if and only if at least one of $G'_{1}$ and $G'_{2}$ contains a $W_{7}$-subdivision.
\end{thm}

\begin{proof}

Without loss of generality, assume that $v_{1}$ is incident with $e_{1}$ and $e_{2}$, and that $v_{2}$ is incident with $e_{3}$ and $e_{4}$.

$(\Rightarrow )$ Suppose firstly that $G$ contains a $W_{7}$-subdivision $H$.

If $H$ is entirely contained in $G_{1}$ or $G_{2}$, then $H$ will also be contained in $G'_{1}$ or $G'_{2}$ respectively. Suppose then that $H$ contains parts of both $G_{1}$ and $G_{2}$. If $H$ is centred in $G \setminus \{v\}$, then by the same arguments used in Theorem \ref{edgevertex3}, $H$ is also contained in either $G'_{1}$ or $G'_{2}$.

Suppose now that $H$ is centred on $v$. Since this means that $v$ has degree $\ge 7$, some component $G_{x} \in \{G_{1}$, $G_{2}\}$ must contain at most two neighbours of $v$, and thus contains at most two spoke-meets-rim vertices of $H$. This structure can be replaced by $X$ if $G_{x} = G_{2}$, or $Y$ if $G_{x} = G_{1}$.

Thus, whenever $G$ contains a $W_{7}$-subdivision, at least one of $G'_{1}$ or $G'_{2}$ does also.

$(\Leftarrow )$ Suppose now that either $G'_{1}$ or $G'_{2}$ contains a $W_{7}$-subdivision.

\vspace{0.1in}
\noindent \textbf{Case 1}

Suppose firstly that $G'_{1}$ contains a $W_{7}$-subdivision. If $H$ is entirely contained in $G'_{1} - X$, then $H$ is also contained in $G$. Suppose then that $H$ contains parts of $X$. One of the following must hold:

\begin{itemize}
\item[(a)] $X$ contains a single path of $H$;
\item[(b)] one of $v$, $v_{1}$, or $v_{2}$ is contained in $H$, and two separate paths of $H$ leave $G_{1}$ at this vertex and return to $G_{1}$ via other members of $S$, such that these paths are vertex-disjoint within $X$;
\item[(c)] $X$ contains a single spoke-meets-rim vertex of $H$; or
\item[(d)] $H$ is centred on one of $v$, $v_{1}$, or $v_{2}$, and $G_{2}$ contains two spoke-meets-rim vertices of $H$.
\end{itemize}

By the same arguments used in Theorem \ref{edgevertex3}, whenever $G'_{1}$ contains a $W_{7}$-subdivision, $G$ does also.

\vspace{0.1in}
\noindent \textbf{Case 2}

Suppose now that $G'_{2}$ contains a $W_{7}$-subdivision. If $H$ is entirely contained in $G'_{2} - X$, then $H$ is also contained in $G$. Suppose then that $H$ contains parts of $Y$. One of the following must hold:

\begin{itemize}
\item[(a)] $Y$ contains a single path of $H$;
\item[(b)] $Y$ contains two disjoint paths of $H$;
\item[(c)] $Y$ contains a single spoke-meets-rim vertex of $H$; or
\item[(d)] $Y$ contains two spoke-meets-rim vertices of $H$.
\end{itemize}

Again, by the same arguments used in Theorem \ref{edgevertex3}, whenever $G'_{2}$ contains a $W_{7}$-subdivision, $G$ does also.

Thus, whenever one of $G'_{1}$ or $G'_{2}$ contains a $W_{7}$-subdivision, $G$ does also.
\end{proof}

\vspace{0.1in}
\noindent \textbf{Definition}

A \emph{type 4 edge-vertex-cutset} in a graph $G$ is a set $S = \{v_{1}, v_{2}, e_{1}, e_{2}\}$ of two vertices $v_{1}, v_{2}$ of $G$ and two edges $e_{1}, e_{2}$ of $G$ such that $G - S$ is disconnected, with each component having at least four vertices, and with one of the components containing exactly one vertex incident with $e_{1}$ and $e_{2}$.

\begin{thm}
\label{edgevertex4}

Let $G$ be a 3-connected graph with no type 1, 2, or 3 edge-vertex-cutsets, but which contains a type 4 edge-vertex-cutset $S = \{v_{1}, v_{2}, e_{1}, e_{2}\}$. Let $G_{1}$ be the component of $G - S$ that contains exactly one vertex incident with $e_{1}, e_{2}$, and let $G_{2}$ be the other component of $G - S$.

Form $G'_{1}$ from $G$ by replacing $G_{2}$ with the subgraph $X$, where $X$ contains three pairwise-adjacent vertices, $x_{1}, x_{2}, x_{3}$, each of which is adjacent to both $v_{1}$ and $v_{2}$, such that $x_{1}$ and $x_{2}$ form endpoints of $e_{1}$ and $e_{2}$ respectively.

Form $G'_{2}$ from $G$ by replacing $G_{1}$ with the subgraph $Y$, where $Y$ contains three pairwise-adjacent vertices, $y_{1}, y_{2}, y_{3}$, each of which is adjacent to both $v_{1}$ and $v_{2}$, such that $y_{1}$ forms an endpoint of both $e_{1}$ and $e_{2}$.

Then $G$ contains a $W_{7}$-subdivision if and only if at least one of $G'_{1}$ and $G'_{2}$ contains a $W_{7}$-subdivision.
\end{thm}

\begin{proof}

Let $u$ be the vertex incident with both $e_{1}$ and $e_{2}$ in $G_{1}$.

$(\Rightarrow )$ Suppose firstly that $G$ contains a $W_{7}$-subdivision $H$.

If $H$ is entirely contained in $G_{1}$ or $G_{2}$, then $H$ will also be contained in $G'_{1}$ or $G'_{2}$ respectively. Suppose then that $H$ contains parts of both $G_{1}$ and $G_{2}$.

The centre of $H$ can be either in $G_{1}$, in $G_{2}$, or on one of $v_{1}, v_{2}$. Consider each of the cases.

\vspace{0.1in}
\noindent \textbf{Case 1}

Suppose $H$ is centred in $G_{1}$. Then one of the following must be true:

\begin{itemize}
\item[(A)] $G_{2}$ contains a single path of $H$;
\item[(B)] one of $v_{1}$, $v_{2}$, or $u$ is contained in $H$, and two separate paths of $H$ leave $G_{1}$ at this vertex and return to $G_{1}$ via other members of $S$, such that these paths are vertex-disjoint within $G_{2}$;
\item[(C)] $G_{2}$ contains a single spoke-meets-rim vertex of $H$; or
\item[(D)] $H$ is centred on $u$, and $G_{2}$ contains two spoke-meets-rim vertices of $H$.
\end{itemize}

In each case, the required structure can easily be formed in $X$.

\vspace{0.1in}
\noindent \textbf{Case 2}

Suppose now $H$ is centred in $G_{2}$. Then one of the following must be true:

\begin{itemize}
\item[(A)] $G_{1}$ contains a single path of $H$;
\item[(B)] $G_{1}$ contains two disjoint paths of $H$;
\item[(C)] $G_{1}$ contains a single spoke-meets-rim vertex of $H$; or
\item[(D)] $u$ and some other vertex in $G_{1}$ form two spoke-meets-rim vertices of $H$.
\end{itemize}

In each case, the required structure can easily be formed in $Y$.

\vspace{0.1in}
\noindent \textbf{Case 3}

Suppose now that $H$ is centred on either $v_{1}$ or $v_{2}$ --- assume $v_{1}$ without loss of generality. Since $|N_{H}(v_{1})| = 7$, some component $G_{x} \in \{G_{1}$, $G_{2}\}$ contains at most three of the vertices in $N_{H}(v)$. If $G_{x} = G_{2}$, then, since $G_{2}$ can contain at most three spoke-meets-rim vertices of $H$, $X$ can be used in $G'_{2}$ to replace the parts of $H$ previously contained in $G_{2}$. If $G_{x} = G_{1}$, $Y$ can be used in $G'_{1}$ to replace the parts of $H$ previously contained in $G_{1}$. 

Thus, whenever $G$ contains a $W_{7}$-subdivision, at least one of $G'_{1}$ or $G'_{2}$ does also.

$(\Leftarrow )$ Suppose now that either $G'_{1}$ or $G'_{2}$ contains a $W_{7}$-subdivision.

\vspace{0.1in}
\noindent \textbf{Case 1}

Suppose firstly that $G'_{1}$ contains a $W_{7}$-subdivision. If $H$ is entirely contained in $G'_{1} - X$, then $H$ is also contained in $G$. Suppose then that $H$ contains parts of $X$. One of the following must hold:

\begin{itemize}
\item[(a)] $X$ contains a single path of $H$;
\item[(b)] one of $v_{1}$, $v_{2}$, or $u$ is contained in $H$, and two separate paths of $H$ leave $G_{1}$ at this vertex and return to $G_{1}$ via other members of $S$, such that these paths are vertex-disjoint within $X$;
\item[(c)] $X$ contains a single spoke-meets-rim vertex of $H$;
\item[(d)] $H$ is centred on one of $v_{1}$, $v_{2}$, or $u$, and $G_{2}$ contains two spoke-meets-rim vertices of $H$; or
\item[(e)] $H$ is centred on $v_{1}$ or $v_{2}$, and $X$ contains three spoke-meets-rim vertices of $H$.
\end{itemize}

If (a), (b), or (c) hold, then by 3-connectivity of $G$, the required paths can be formed in $G_{2}$. 

Suppose (d) holds. If either of $v_{1}$ or $v_{2}$ has only one neighbour in $G_{2}$, then a type 2 edge-vertex-cutset exists in $G$. Assume then that $G_{2}$ contains at least two neighbours of $v_{1}$ and at least two neighbours of $v_{2}$. Thus, if $H$ is centred on $v_{1}$ or $v_{2}$, by Lemma \ref{lemma1} the required structure can also be formed in $G_{2}$. Similarly, if $H$ is centred on $u$, then by Lemma \ref{lemma1} the required structure can be formed in $G_{2}$, since $u$ has two neighbours in $G_{2}$.

Suppose (e) holds.  If either of $v_{1}$ or $v_{2}$ has fewer than three neighbours in $G_{2}$, then a type 3 edge-vertex-cutset exists in $G$. Assume then that $G_{2}$ contains at least three neighbours of $v_{1}$ and at least three neighbours of $v_{2}$. Thus, by Lemma \ref{lemma2} the required structure can also be formed in $G_{2}$.

\vspace{0.1in}
\noindent \textbf{Case 2}

Suppose now that $G'_{2}$ contains a $W_{7}$-subdivision. If $H$ is entirely contained in $G'_{2} - X$, then $H$ is also contained in $G$. Suppose then that $H$ contains parts of $Y$. One of the following must hold:

\begin{itemize}
\item[(a)] $Y$ contains a single path of $H$;
\item[(b)] $Y$ contains two disjoint paths of $H$;
\item[(c)] $Y$ contains a single spoke-meets-rim vertex of $H$; or
\item[(d)] $Y$ contains two spoke-meets-rim vertices of $H$.
\end{itemize}

By 3-connectivity of $G$, the required paths can be formed in $G_{1}$ for any of these cases.

Thus, whenever one of $G'_{1}$ or $G'_{2}$ contains a $W_{7}$-subdivision, $G$ does also.
\end{proof}

\vspace{0.1in}
\noindent \textbf{Definition}

A \emph{type 4a edge-vertex-cutset} in a graph $G$ is a set $S = \{v_{1}, v_{2}, e_{1}, e_{2}\}$ of two vertices $v_{1}, v_{2}$ of $G$ and two edges $e_{1}, e_{2}$ of $G$ such that $G - S$ is disconnected with each component having at least three vertices, with one of the components containing exactly one vertex incident with $e_{1}$ and $e_{2}$, and for each vertex $v_{i}$, $1\le i\le 2$, either:

\begin{itemize}
\item $v_{i}$ has at most two neighbours in one of the components of $G - S$; or
\item $v_{i}$ has degree $< 7$.
\end{itemize}

\begin{thm}
\label{edgevertex4a}

Let $G$ be a 3-connected graph with no type 1, 2, 2a, or 3 edge-vertex-cutsets, but which contains a type 4a edge-vertex-cutset $S = \{v_{1}, v_{2}, e_{1}, e_{2}\}$. Let $G_{1}$ be the component of $G - S$ that contains exactly one vertex incident with $e_{1}, e_{2}$, and let $G_{2}$ be the other component of $G - S$.

Form $G'_{1}$ from $G$ by replacing $G_{2}$ with the subgraph $X$, where $X$ contains two adjacent vertices, $x_{1}$ and $x_{2}$, each of which is adjacent to both $v_{1}$ and $v_{2}$, such that $x_{1}$ and $x_{2}$ form endpoints of $e_{1}$ and $e_{2}$ respectively.

Form $G'_{2}$ from $G$ by replacing $G_{1}$ with the subgraph $Y$, where $Y$ contains two adjacent vertices, $y_{1}$ and $y_{2}$, each of which is adjacent to both $v_{1}$ and $v_{2}$, such that $y_{1}$ forms an endpoint of both $e_{1}$ and $e_{2}$.

Then $G$ contains a $W_{7}$-subdivision if and only if at least one of $G'_{1}$ and $G'_{2}$ contains a $W_{7}$-subdivision.
\end{thm}

\begin{proof}

Let $u$ be the vertex incident with both $e_{1}$ and $e_{2}$ in $G_{1}$.

$(\Rightarrow )$ Suppose firstly that $G$ contains a $W_{7}$-subdivision $H$.

If $H$ is entirely contained in $G_{1}$ or $G_{2}$, then $H$ will also be contained in $G'_{1}$ or $G'_{2}$ respectively. Suppose then that $H$ contains parts of both $G_{1}$ and $G_{2}$.

If $H$ is centred in $G\setminus \{v_{1}, v_{2}\}$, then by the same arguments used in Theorem \ref{edgevertex4}, $H$ is also contained in either $G'_{1}$ or $G'_{2}$.

Suppose then that $H$ is centred on $v_{1}$ or $v_{2}$ --- assume $v_{1}$ without loss of generality. Since this means that $v_{1}$ has degree $\ge 7$, some component $G_{x} \in \{G_{1}, G_{2}\}$ must contain at most two neighbours of $v_{1}$, and thus contains at most two spoke-meets-rim vertices of $H$. By Lemma \ref{lemma1}, this structure can be replaced by $X$ if $G_{x} = G_{2}$, or $Y$ if $G_{x} = G_{1}$.

Thus, whenever $G$ contains a $W_{7}$-subdivision, at least one of $G'_{1}$ or $G'_{2}$ does also.

$(\Leftarrow )$ Suppose now that either $G'_{1}$ or $G'_{2}$ contains a $W_{7}$-subdivision.

\vspace{0.1in}
\noindent \textbf{Case 1}

Suppose firstly that $G'_{1}$ contains a $W_{7}$-subdivision. If $H$ is entirely contained in $G'_{1} - X$, then $H$ is also contained in $G$. Suppose then that $H$ contains parts of $X$. One of the following must hold:

\begin{itemize}
\item[(a)] $X$ contains a single path of $H$;
\item[(b)] one of $v_{1}$, $v_{2}$, or $u$ is contained in $H$, and two separate paths of $H$ leave $G_{1}$ at this vertex and return to $G_{1}$ via other members of $S$, such that these paths are vertex-disjoint within $X$;
\item[(c)] $X$ contains a single spoke-meets-rim vertex of $H$; or
\item[(d)] $H$ is centred on one of $v_{1}$, $v_{2}$, or $u$, and $G_{2}$ contains two spoke-meets-rim vertices of $H$.
\end{itemize}

By the same arguments used in Theorem \ref{edgevertex4}, the required paths can be formed in $G_{2}$ for any of these cases.

\vspace{0.1in}
\noindent \textbf{Case 2}

Suppose now that $G'_{2}$ contains a $W_{7}$-subdivision. If $H$ is entirely contained in $G'_{2} - X$, then $H$ is also contained in $G$. Suppose then that $H$ contains parts of $Y$. One of the following must hold:

\begin{itemize}
\item[(a)] $Y$ contains a single path of $H$;
\item[(b)] $Y$ contains two disjoint paths of $H$;
\item[(c)] $Y$ contains a single spoke-meets-rim vertex of $H$; or
\item[(d)] $Y$ contains two spoke-meets-rim vertices of $H$.
\end{itemize}

By the same arguments used in Theorem \ref{edgevertex4}, the required paths can be formed in $G_{1}$ for any of these cases.                                                                                                                                                                                                                                                                                                                                                                                                                                                                                                                                                                                                                                                                                                                                                                                                                                                                                                                                                                                                                                                                                                                                                                                                                                                                                                                                                                                                                                                                                                                                                                                                                                                                                                                                                                                                                                                                                                                                                                                                                                                                                                                                                                                                                                                                                                                                                                                                                                                                                                               

Thus, whenever one of $G'_{1}$ or $G'_{2}$ contains a $W_{7}$-subdivision, $G$ does also.
\end{proof}

\vspace{0.1in}
\noindent \textbf{Definition}

An \emph{internal $(1, 1, 1, 1)$-cutset} in a graph $G$ is a set $E'$ of four disjoint edges of $G$ such that $G - E'$ is disconnected, with each component having at least five vertices.

\begin{thm}
\label{bigfourcutset}

Let $G$ be a 3-connected graph with no internal 4-edge-cutsets and no type 1 or 1a edge-vertex-cutsets, which contains an internal $(1,1,1,1)$-cutset $E' = \{e_{1}, e_{2}, e_{3}, e_{4}\}$. Let $G_{1}, G_{2}$ be the components of $G-E'$. Let the endpoints of $e_{1}, \ldots , e_{4}$ in $G_{1}$ be the four distinct vertices $u_{1}, u_{2}, u_{3}, u_{4}$, and let the endpoints of $e_{1}, \ldots , e_{4}$ in $G_{2}$ be the four distinct vertices $v_{1}, v_{2}, v_{3}, v_{4}$. Form $G'_{1}$ from $G$ by replacing $G_{2}$ with the subgraph $X$, where $X$ contains only the four vertices $v_{1}, v_{2}, v_{3}, v_{4}$, all of which are made adjacent to one another if they were not already.

Then $G$ contains a $W_{k}$-subdivision centred in $G_{1}$ if and only if $G'_{1}$ contains a $W_{k}$-subdivision, where $k \ge 5$.
\end{thm}

\begin{proof}
$(\Rightarrow )$ Suppose firstly that $G$ contains a $W_{k}$-subdivision $H$ centred in $G_{1}$. If $H$ is entirely contained in $G_{1}$ then $H$ is also contained in $G'_{1}$.

If however $H$ is not entirely contained in $G_{1}$, but rather contains edges of $E'$, then one of the following statements must be true:

\begin{itemize}
\item[(A)] a single path of $H$ leaves and returns to $G_{1}$ via $E'$; 
\item[(B)] two disjoint paths of $H$ leave and return to $G_{1}$ via $E'$; 
\item[(C)] there is a single spoke-meets-rim vertex of $H$ in $G_{2}$; or
\item[(D)] $G_{2}$ contains two spoke-meets-rim vertices belonging to $H$ such that the portion of rim between these vertices is entirely contained in $G_{2}$.
\end{itemize}

If (A) holds, the portion of $H$ in $G_{2}$ can be replaced by a single edge in $X$.

If (B) holds, the portion of $H$ in $G_{2}$ can be replaced by two edges in $X$.

Suppose (C) holds. Any vertex in $X$ can be used as the spoke-meets-rim vertex of $H$ that was previously contained in $G_{2}$.

Suppose (D) holds. In $G$, two of the vertices in $u_{1}, \ldots , u_{4}$ must lie on spokes of $H$, while the other two lie on the rim of $H$. Without loss of generality, assume that $u_{1}$ and $u_{2}$ lie on spokes, and $u_{3}$ and $u_{4}$ lie on the rim. Then in $G'_{1}$, use $v_{1}$ and $v_{2}$ to form the two spoke-meets-rim vertices that were previously in $G_{2}$.

Thus whenever $G$ contains a $W_{k}$-subdivision centred in $G_{1}$, $G'_{1}$ must also.

$(\Leftarrow )$ Assume now that $G'_{1}$ contains a $W_{k}$-subdivision. If none of the edges in $E'$ are used to form $H$ in $G'_{1}$, then $H$ is also contained in $G$. 

Suppose then that $H$ contains edges of $E'$. One of the following must hold:

\begin{itemize}
\item[(a)] a single path of $H$ leaves and returns to $G_{1}$ via $E'$;
\item[(b)] two paths of $H$, say $P_{1}$ and $P_{2}$, leave and return to $G_{1}$ via $E'$;
\item[(c)] one of the vertices in $X$ serves as a spoke-meets-rim vertex belonging to $H$; or
\item[(d)] two of the vertices in $X$ serve as two spoke-meets-rim vertices belonging to $H$, such that the portion of rim between these vertices is entirely contained in $X$.
\end{itemize}

By the 3-connectivity of $G$ we know that there must exist paths in $G_{2} \cup E'$ between each pair of the vertices $u_{1}, u_{2}, u_{3}, u_{4}$. Thus if either (a) or (c) is true, we can use one or parts of two of these paths to form the required paths in $H$.

Suppose (b) is true. Without loss of generality, suppose that $P_{1}$ enters $X$ at $v_{1}$ and leaves at $v_{2}$; and that $P_{2}$ enters $X$ at $v_{3}$ and leaves at $v_{4}$. There are three possibilities:

\begin{itemize}
\item[(i)] $P_{1}$ and $P_{2}$ are both parts of spokes of $H$;
\item[(ii)] $P_{1}$ and $P_{2}$ are both parts of the rim of $H$; or
\item[(iii)] one of the paths forms part of a spoke of $H$, and the other forms part of the rim of $H$.
\end{itemize}

Suppose that (i) holds. Let $S_{1}$ be the spoke of $H$ containing $P_{1}$, and let $S_{2}$ be the spoke of $H$ containing $P_{2}$. Let $v$ be the centre of $H$, and let $s_{1}$ and $s_{2}$ be the points at which $S_{1}$ and $S_{2}$ respectively meet the rim of $H$. Without loss of generality, suppose that $v_{1}$ and $v_{3}$ are closer to $v$ along $S_{1}$ and $S_{2}$ respectively than $v_{2}$ and $v_{4}$. If $P_{1}$ and $P_{2}$ are both part of the same spoke of $H$ (that is, $S_{1} = S_{2}$), then a single path in $G_{2}$ from $v_{1}$ to $v_{4}$ can be used in $G$ to form that part of the spoke between $v_{1}$ and $v_{4}$. (By 3-connectivity, such a path must exist.) Assume then that $S_{1}$ and $S_{2}$ are separate spokes of $H$.

Suppose that $\{v_{1}, v_{3}\}$ and $\{v_{2}, v_{4}\}$ can be separated in $G_{2}$ by a single vertex, $q$. Let $W$ be the component of $G_{2} - q$ containing $v_{1}$ and $v_{3}$, and let $Z$ be the component of $G_{2}-q$ containing $v_{2}$ and $v_{4}$. If $W$ contains more than three vertices, then a type 1 edge-vertex-cutset can be formed by the vertex $q$ and the edges $e_{1}$ and $e_{3}$. If $W$ contains fewer than three vertices (which implies $V(W) = \{v_{1}, v_{3}\}$ and that $q$ is adjacent to $v_{1}$ and $v_{3}$), an internal 4-edge-cutset can be formed by the edges $e_{2}$, $e_{4}$, $qv_{1}$, and $qv_{3}$. Assume then that $W$ contains exactly three vertices. If $Z$ contains more than three vertices, then a type 1 edge-vertex-cutset can be formed by the vertex $q$ and the edges $e_{2}$ and $e_{4}$. If $Z$ contains fewer than three vertices, an internal 4-edge-cutset can be formed by the edges $e_{1}$, $e_{3}$, $qv_{2}$, and $qv_{4}$. Assume then that $Z$ contains exactly three vertices. Then, since $q$ can have degree at most six, a type 1a edge-vertex-cutset can be formed by the vertex $q$ and either the edges $e_{1}$ and $e_{3}$ or the edges $e_{2}$ and $e_{4}$. Therefore, such a vertex $q$ cannot exist in $G_{2}$. Thus, there must be at least two disjoint paths in $G_{2}$ joining $\{v_{1}, v_{3}\}$ to $\{v_{2}, v_{4}\}$. Call these paths $P'_{1}$ and $P'_{2}$.

If $P'_{1}$ and $P'_{2}$ run from $v_{1}$ to $v_{2}$ and from $v_{3}$ to $v_{4}$, then they can be used to replace $P_{1}$ and $P_{2}$ in $G$. Suppose then that $P'_{1}$ is a path from $v_{1}$ to $v_{4}$, and $P'_{2}$ is a path from $v_{3}$ to $v_{2}$. Then in $G$, the two spokes $S_{1}$ and $S_{2}$ can be replaced by the paths $vS_{2}v_{3}P'_{2}v_{2}S_{1}s_{1}$ and $vS_{1}v_{1}P'_{1}v_{4}S_{2}s_{2}$.

Suppose (ii) holds. Without loss of generality, suppose that the portion of the rim of $H$ in $G'_{1} - X$ consists of a path from $v_{1}$ to $v_{3}$ and a path from $v_{2}$ to $v_{4}$. By the same argument used in (i), there must exist two disjoint paths in $G_{2}$, $P'_{1}$ and $P'_{2}$, joining $\{v_{1}, v_{3}\}$ to $\{v_{2}, v_{4}\}$. In $G$, $P'_{1}$ and $P'_{2}$ can be used to replace the parts of the rim not in $G_{1} \cup E'$, regardless of whether they run from $v_{1}$ to $v_{2}$ and $v_{3}$ to $v_{4}$, or from $v_{1}$ to $v_{4}$ and $v_{3}$ to $v_{2}$.

Suppose (iii) holds. Without loss of generality, suppose that $P_{1}$ forms part of the rim of $H$, and $P_{2}$ forms part of a spoke $S_{1}$ of $H$ such that $v_{3}$ is closer to the centre of $H$ along $P_{2}$ than $v_{4}$. Let $s_{1}$ be the vertex at which $S_{1}$ meets the rim of $H$. By the 3-connectivity of $G$, there must be some path $P'_{1}$ from $v_{1}$ to $v_{2}$ in $G_{2}$. Use this to form that part of the rim formed by $P_{1}$ in $G'_{1}$. There must also be some path $P'_{2}$ in $G_{2}$ that runs from $v_{3}$ to some vertex $q$ on $P'_{1}$, such that $P'_{2}$ meets $P'_{1}$ only at $q$. Let $P'_{2}$ replace that part of $H$ formed by $v_{3}S_{1}s_{1}$ in $G'_{1}$, so that $q$ becomes a spoke-meets-rim vertex in $G$ instead of $s_{1}$, and a new spoke is formed by the path $vS_{1}v_{3}P'_{2}q$.

Suppose (d) is true. Without loss of generality, suppose $v_{1}$ and $v_{2}$ are the two spoke-meets-rim vertices in $X$. By the 3-connectivity of $G$ we know that there must exist a path $P$ in $G_{2}$ from $v_{1}$ to $v_{2}$, otherwise the removal of either $u_{1}$ or $u_{2}$ and some other vertex in $G_{2}$ will disconnect the graph, placing $v_{1}$ and $v_{2}$ in separate components. Suppose that $P$ can be separated from $v_{3}, v_{4}$ in $G_{2}$ by the removal of a single vertex, $q$. Let $V_{1}$ be the component of $G_{2} - q$ containing $v_{1}$ and $v_{2}$, and let $V_{3}$ be the component of $G_{2} - q$ containing $v_{3}$ and $v_{4}$. If $V_{3}$ contains more than three vertices, then a type 1 edge-vertex-cutset can be formed by the vertex $q$ and the edges $e_{3}$ and $e_{4}$. If $V_{3}$ contains fewer than three vertices, then an internal 4-edge-cutset can be formed by the edges $e_{1}$, $e_{2}$, $qv_{3}$, and $qv_{4}$. Assume then that $V_{3}$ contains exactly three vertices. By the same argument, $V_{1}$ must also contain exactly three vertices. Then, since $q$ can have degree at most six, a type 1a edge-vertex-cutset can be formed by the vertex $q$ and either the edges $e_{1}$ and $e_{2}$ or the edges $e_{3}$ and $e_{4}$. Therefore, such a vertex $q$ cannot exist in $G_{2}$. Thus, there must be at least two disjoint paths in $G_{2}$ joining $P$ to $v_{3}$ and $v_{4}$. These two paths and the path $P$ form the required structure in $H$.

Thus whenever $G'_{1}$ contains a $W_{k}$-subdivision, $G$ must also.
\end{proof}

\section{Results on graphs with no 6-wheel subdivisions}
\label{6wheelresults}

The following two theorems build directly on the results of \cite{Robinson08}, and relate specifically to $W_{6}$-subdivisions, rather than $W_{7}$-subdivisions. However, they are key to proving Theorem \ref{theorem} in this paper, particularly Theorem \ref{w6cor}, which strengthens the main result of \cite{Robinson08}.

Theorem \ref{w6thm}, given below, follows from the main theorem of \cite{Robinson08} (titled Theorem 4 in that paper).

\begin{thm}
\label{w6thm}

Let $G$ be a 3-connected graph with at least 12 vertices. Suppose $G$ has no internal 3-edge-cutsets, no internal 4-edge-cutsets, and is a graph on which neither Reduction \ref{r1}A nor Reduction \ref{r2}A can be performed.

Then $G$ has a $W_{6}$-subdivision if and only if $G$ contains some vertex $v_{0}$ of degree at least 6.

\end{thm}

\begin{proof}
Suppose that $G$ is not topologically contained in Graph $A$, shown in Figure \ref{graphA}. Then Theorem 4 of \cite{Robinson08} applies to $G$, thus proving the hypothesis. If however $G$ is topologically contained in Graph $A$, then $G$ can contain at most 11 vertices (since $|V(A)| = 11$), which contradicts the original assumption that $|V(G)| \ge 12$.
\end{proof}

\begin{figure}[!ht]
\begin{center}
\includegraphics[width=0.35\textwidth]{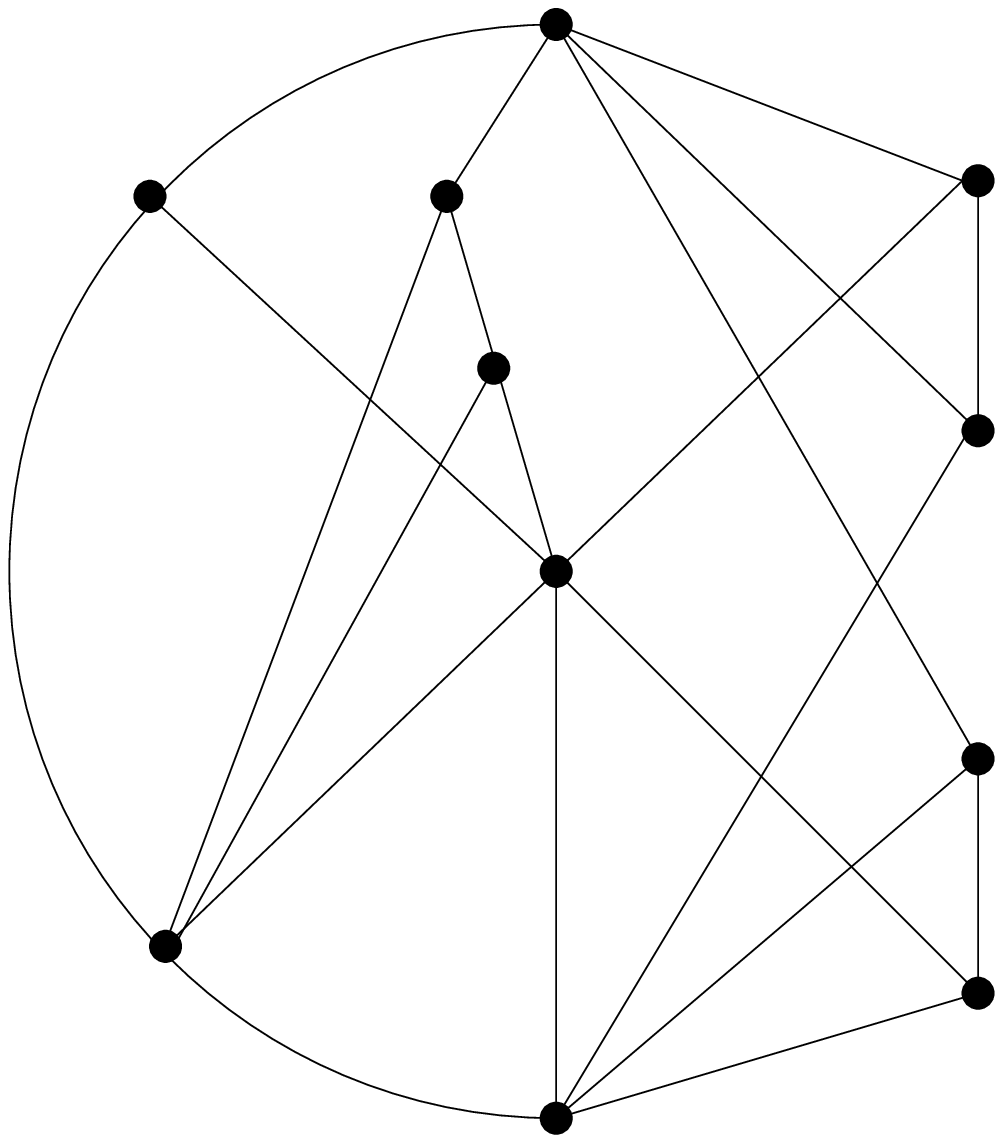}
\caption{Graph $A$}
\label{graphA}
\end{center}
\end{figure}

Theorem \ref{w6cor}, which follows, builds on the previous theorem characterizing graphs that do not contain a $W_{6}$-subdivision. The extra conditions forbidding certain types of edge-vertex-cutsets and additional types of reductions allow for a strengthened result in the case where the $W_{6}$-subdivision found is not centred on the given vertex of degree $\ge 6$.

\begin{thm}
\label{w6cor}

Let $G$ be a 3-connected graph with at least 14 vertices. Suppose $G$ has no type 1, 2, 3, or 4 edge-vertex-cutsets, and is a graph on which neither Reduction \ref{r1}A, Reduction \ref{r1}B, Reduction \ref{r2}A, nor Reduction \ref{r2}B can be performed, for $k = 7$. Let $v_{0}$ be a vertex of degree $\ge 6$ in $G$. Then either $G$ has a $W_{6}$-subdivision centred on $v_{0}$, or $G$ has a $W_{6}$-subdivision centred on some vertex $v_{1}$ of degree $\ge 7$.

\end{thm}

\begin{proof}

Suppose $G$ and $v_{0}$ satisfy the hypotheses of the theorem. From Theorem \ref{w6thm}, $G$ contains a $W_{6}$-subdivision, $H$, and from the proof of this theorem, Theorem 4 of \cite{Robinson08} also applies to $G$. Referring to the proof of Theorem 4 in \cite{Robinson08}, we know that in all cases other than (b)(ii), $H$ must be centred on $v_{0}$.

Looking more closely at the proof of (b)(ii), where $G$ contains the structure illustrated in Figure \ref{b2}, there are a number of ways in which $W_{6}$-subdivisions are formed. Firstly, all possible ways of adding a single path $Q$ to the structure in Figure \ref{b2} are tested, excluding those cases where the path added meets internally either the path from $v_{0}$ to $v_{1}$, or the path from $v_{0}$ to $v_{3}$. Each of the resulting graphs is found to contain a $W_{6}$-subdivision. Since the original structure contains no vertices of degree $\ge 5$ other than $v_{0}$, there can be no vertices with degree $\ge 6$ other than $v_{0}$ in the resulting graph once $Q$ is added. Thus, the $W_{6}$-subdivision created in each case must be centred on $v_{0}$.

\begin{figure}[!ht]
\begin{center}
\includegraphics[width=0.35\textwidth]{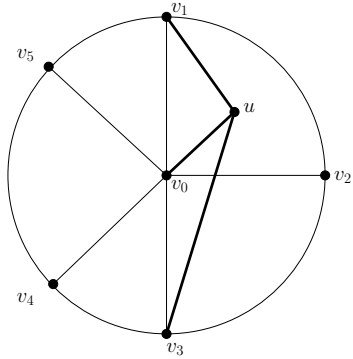}
\caption[\textbf{(b)(ii)}: Starting graph]{\textbf{(b)(ii)}: Starting graph}
\label{b2}
\end{center}
\end{figure}

\vspace{0.2in}
\noindent \textbf{Case 1: Existence of path $R$}

The next part of the proof looks at the addition of some new path $R$ to the graph of Figure \ref{b2}, such that $R$ runs from $v_{4}$ to some point on the path from $v_{0}$ to $v_{1}$, as in Figure \ref{pathr}. All ways of adding $R$ to this structure are tested for a $W_{6}$-subdivision. Again, since $v_{0}$ is the only vertex in the graph of Figure \ref{pathr} of degree $\ge 5$, any $W_{6}$-subdivision created by adding an edge to this graph must be centred on $v_{0}$. The same argument follows in the next part of the proof, where single edges are added to the graph of Figure \ref{pathr2} and the resulting graph tested for the presence of a $W_{6}$-subdivision.

\begin{figure}[!ht]
\begin{center}
\includegraphics[width=0.4\textwidth]{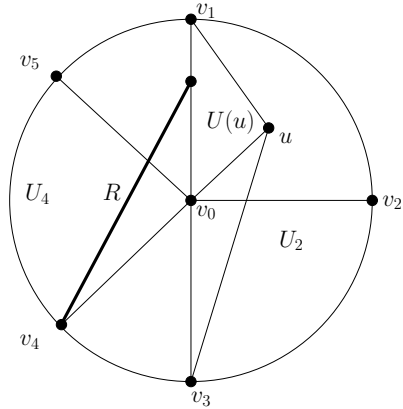}
\caption[Configuration for path $R$]{Configuration for path $R$}
\label{pathr}
\end{center}
\end{figure}

\begin{figure}[!ht]
\begin{center}
\includegraphics[width=0.4\textwidth]{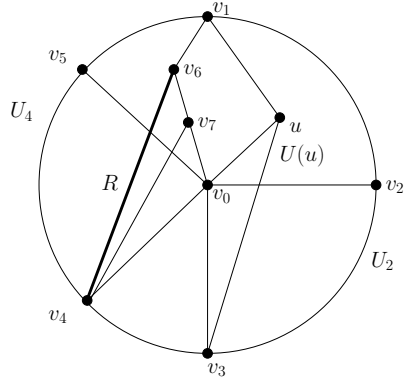}
\caption[More complex configuration involving path $R$]{More complex configuration involving path $R$}
\label{pathr2}
\end{center}
\end{figure}

The next part of Case 1 assumes that $\langle U_{4} \rangle$ in $G$ is isomorphic to that part of the graph illustrated in either Figure \ref{pathr} or Figure \ref{pathr2}, and looks further at the structure of the bridges $U_{2}$ and $U(u)$. As with the original proof, suppose firstly that $U_{2}$ and $U(u)$ each contain at most two vertices not in $W$. There must then exist some fourth bridge of $G|W$, $U^{*}$, since $G$ has at least 14 vertices. Each bridge of $G|W$ must contain at least 2 vertices not in $W$, otherwise Reduction \ref{r1}A can be performed on $G$, and thus to avoid internal 3-edge-cutsets, there must be at least 4 edges joining $W$ to $U_{2} \setminus W$, likewise joining $W$ to $U(u) \setminus W$ and to $U^{*} \setminus W$. Thus some vertex $x \in W$ has two neighbours in $U_{2}$, some vertex $y \in W$ has two neighbours in $U(u)$, and some vertex $z \in W$ has at least two neighbours in $U^{*}$. Note that $U_{2}$ and $U(u)$ each contain exactly two vertices not in $W$. To avoid the possibility of Reduction \ref{r1}B being performed on $G$, then, vertices $x$, $y$, and $z$ must all be distinct. If this were not the case, then either $U_{2}$ or $U(u)$ would be contained as a subdivision in at least two other bridges of $G|W$. Thus $v_{0} \in \{x, y, z\}$, so $v_{0}$ has two neighbours in some bridge other than $U_{4}$. By Lemma \ref{lemma1}, this bridge along with $U_{4}$ can be used to form a $W_{6}$-subdivision centred on $v_{0}$.

Following the original proof again, we now suppose at least one of $U_{2}$ and $U(u)$ (assume $U_{2}$) has more than two vertices not in $W$. To avoid an internal 4-edge-cutset, there must be at least five edges connecting $U_{2} \setminus W$ to $W$. There are two cases:

\begin{itemize}
\item[(i)] there will be two vertices $x$ and $y$ in $W$ each with two neighbours in $U_{2} \setminus W$; or
\item[(ii)] one vertex in $W$ will have three neighbours in $U_{2} \setminus W$.
\end{itemize}

\vspace{0.1in}
\noindent \textbf{Case 1.1}

Suppose firstly that (i) is true. If $v_{0} \in \{x, y\}$, then a $W_{6}$-subdivision can be formed centred on $v_{0}$. Suppose then that $\{x,y\} = \{v_{1}, v_{3}\}$, and that $v_{0}$ has exactly one neighbour in $U_{2}\setminus W$, say $v'_{0}$. If $|U_{2}\setminus W| > 3$, then, a type 2 edge-vertex-cutset can be formed from $v_{1}$, $v_{3}$, and the edge $v_{0}v'_{0}$. Assume then that $|U_{2}\setminus W| = 3$.

If $U(u)\setminus W$ contains more than one neighbour of $v_{0}$, then a $W_{6}$-subdivision can be formed centred on $v_{0}$. Assume then that $U(u)\setminus W$ contains exactly one neighbour of $v_{0}$, say $v''_{0}$. Then if $|U(u)\setminus W| > 3$, a type 2 edge-vertex-cutset can be formed from $v_{1}$, $v_{3}$, and the edge $v_{0}v''_{0}$. Assume then that $|U(u)\setminus W| \le 3$.

$U(u)$ must have at least two vertices not in $W$, otherwise Reduction \ref{r1}A or Reduction \ref{r2}A can be performed. Thus, $U(u)\setminus W$ contains either two or three vertices.

Suppose firstly that $U(u)$ has exactly two vertices not in $W$. There must then exist some fourth bridge of $G|W$, $U^{*}$, since $G$ has at least 14 vertices. $U^{*}$ must have at least two vertices not in $W$ to avoid Reduction \ref{r1}A. Thus, some vertex $z_{1} \in W$ must have two neighbours in $U(u) \setminus W$, and some vertex $z_{2} \in W$ must have two neighbours in $U^{*}$. One of $z_{1}$ or $z_{2}$ must be $v_{0}$, otherwise $U(u)$ will be contained as a subdivision in at least two other bridges of $G|W$, thus allowing Reduction \ref{r1}B to be performed. Thus, using Lemma \ref{lemma1}, a $W_{6}$-subdivision can be formed centred on $v_{0}$ using parts of either $U^{*}$ or $U(u)$, and $U_{4}$.

Suppose now that $U(u)$ has exactly three vertices not in $W$. To avoid internal 4-edge-cutsets, there must be five edges joining $W$ to $U(u) \setminus W$. If $v_{0}$ has two or more neighbours in $U(u) \setminus W$, then a $W_{6}$-subdivision can be formed centred on $v_{0}$. If either $v_{1}$ or $v_{3}$ has three neighbours in $U(u) \setminus W$, then a $W_{6}$-subdivision can be formed centred on that vertex, which is of degree $\ge 7$. Assume then that each of $v_{1}$ and $v_{3}$ has exactly two neighbours in $U(u) \setminus W$, while $v_{0}$ has exactly one neighbour in $U(u) \setminus W$. Since $v_{1}$ and $v_{3}$ each have two neighbours in $U_{2} \setminus W$ as well, it is apparent that $W_{6}$-subdivisions can still be formed centred on these two vertices. Also, since $|U_{4}\cup U_{2}\cup U(u)| \le 13$, but $|V(G)| \ge 14$, there must exist some fourth bridge of $G|W$ other than $U_{4}$, $U_{2}$, and $U(u)$. Thus, $v_{1}$ and $v_{3}$ each have degree $\ge 7$.

\vspace{0.1in}
\noindent \textbf{Case 1.2}

Suppose now that (ii) is true. By Lemma \ref{lemma2}, a $W_{6}$-subdivision can be formed centred on the vertex with three neighbours in $U_{2} \setminus W$, and by the argument used in the previous case, if that vertex is either $v_{1}$ or $v_{3}$, then it must be of degree $\ge 7$.

\vspace{0.2in}
\noindent \textbf{Case 2: No path $R$}

In the second part of the proof, it is assumed that none of the bridges $U_{2}$, $U_{4}$, or $U(u)$ meet internally either $P_{1}$ or $P_{3}$. 

By the arguments in the original proof, we know that there must exist a $W_{6}$-subdivision $H$ centred on some vertex $x \in W$. If $x = v_{0}$, we have nothing more to prove. Assume instead then that $x$ is either $v_{1}$ or $v_{3}$, and we must show $x$ to have degree $\ge 7$.

\vspace{0.1in}
\noindent \textbf{Case 2.1}

Firstly, suppose that there are at least four bridges of $G|W$. The vertices in $N_{H}(x) \setminus W$, of which there must be at least four, are contained in at most two bridges of $W$. These four neighbours of $x$, plus the neighbours in the remaining two bridges and the edge $xv_{0}$, mean that $x$ has degree $\ge 7$.

\vspace{0.1in}
\noindent \textbf{Case 2.2}

Suppose now that there are only three bridges of $G|W$. Since $|V(G)| \ge 14$, then, one of the following must be true:

\begin{itemize}
\item[(i)] two of these bridges must each contain at least four vertices not in $W$, while the third must contain at least three vertices not in $W$; or
\item[(ii)] one bridge contains at least five vertices in $W$.
\end{itemize}

Suppose firstly that (i) is true. Let $A$ and $B$ be the two bridges with at least four vertices not in $W$, and let $C$ be the bridge with at least three vertices not in $W$.

If any two bridges have two neighbours of $v_{0}$ not in $W$, then a $W_{6}$-subdivision can be formed centred on $v_{0}$. Assume then that two of the bridges, call them $U^{*}$ and $U^{**}$, have one neighbour of $v_{0}$ not in $W$. At least one of these two bridges, say $U^{*}$, must be $A$ or $B$. Thus, a type 2 edge-vertex-cutset can be formed from $v_{1}$, $v_{3}$, and the edge joining $v_{0}$ to $U^{*}\setminus W$.

Suppose now that (ii) is true. Let $A$ be the bridge with at least five vertices in $W$. If $v_{0}$ has more than two neighbours in $A \setminus W$, then by Lemma \ref{lemma2} there exists a $W_{6}$-subdivision centred on $v_{0}$. Assume then that $v_{0}$ has at most two neighbours in $A \setminus W$. Then a type 4 edge-vertex-cutset can be formed from $v_{1}$, $v_{3}$, and the two edges joining $v_{0}$ to $U^{*}\setminus W$.
\end{proof}

\section{Supporting lemmas}
\label{supportinglemmas}

The following lemmas are used in support of the main theorem of this paper, Theorem \ref{theorem}. The first, Lemma \ref{lemmaW7}, follows easily from Lemmas \ref{lemma1} and \ref{lemma2}, and is used often throughout the paper in showing the existence of a $W_{7}$-subdivision.

Lemmas \ref{twonbrs} to \ref{twotwotwonbrs} and Lemma \ref{fournbrs} are all similar in nature: each requires a graph $G$ containing some separating set $S$ such that $|S| = 4$ and $S$ contains some vertex $v$ of degree $\ge 7$. By then imposing certain conditions on the neighbours of $v$, it is shown that a $W_{7}$-subdivision must exist centred on $v$. These lemmas are all used in the main theorem, Theorem \ref{theorem}.

Lemmas \ref{threebridges} and \ref{lemma3} both handle situations which often arise in the main theorem.

\begin{lem}
\label{lemmaW7}
Let $G$ be a 3-connected graph containing a separating set $S = \{u, v, w\}$. Let $X$ and $Y$ be two distinct bridges of $G|W$. Suppose that $v$ and $w$ are either adjacent or joined by a path $P_{w}$ in some third bridge $A$ of $G|S$. Suppose also that $v$ and $u$ are either adjacent or joined by a path $P_{u}$ in some bridge of $G|S$ other than $X$, $Y$, or $A$ (if $A$ exists). If $v$ has at least three neighbours in $X \setminus S$, and at least two neighbours in $Y \setminus S$, then $v$ has a $W_{7}$-subdivision centred on it.
\end{lem}

\begin{proof}
By Lemma \ref{lemma1} and Lemma \ref{lemma2}, a $W_{5}$-subdivision can easily be formed centred on $v$ using parts of $X$ and $Y$, using the other two vertices in $S$ as parts of the rim, but not as spoke-meets-rim vertices. $P_{w}$ and $P_{u}$ form the two extra spokes required to make a $W_{7}$-subdivision.
\end{proof}

\begin{lem}
\label{twonbrs}
Let $G$ be a 3-connected graph containing a separating set $S = \{t, u, v, w\}$. Suppose $v$ has degree $\ge 7$, and suppose there are at least three bridges of $G|S$, $X$, $Y$, and $Z$, such that $Y$ contains all four vertices in $S$, while $X$ and $Z$ each contain $v$. Suppose also that one of the following holds:

\begin{itemize}
\item either one of $X$, $Z$ contains all four vertices in $S$; or
\item there exists some fourth bridge $W$ of $G|S$ such that $W$ also contains $v$.
\end{itemize}

Suppose for each $U_{i}, U_{j}$, where $U_{i}, U_{j} \in \{W, X, Z\}$ and $U_{i} \neq U_{j}$, that one of the following holds:

\begin{itemize}
\item either $U_{i}\cap S \neq U_{j}\cap S$; or
\item $U_{i}\cap S = U_{j}\cap S = S$.
\end{itemize}

Suppose that $v$ and $u$ are either adjacent or joined by a path in some bridge $A$ of $G|S$ other than $X$, $Y$, $Z$, or $W$ (if $W$ exists). Call this path (or edge) $P_{u}$. Suppose also that $v$ and $w$ are either adjacent or joined by a path in some bridge $B$ of $G|S$ other than $X$, $Y$, $Z$, $W$, or $A$ (if $W$ and $A$ exist). Call this path (or edge) $P_{w}$. Suppose also that $v$ and $t$ are either adjacent or joined by a path in some bridge $C$ of $G|S$ other than $X$, $Y$, $Z$, $W$, $A$, or $B$ (if $W$, $A$, and $B$ exist). Call this path (or edge) $P_{t}$.

Then if any bridge of $G|S$ contains more than one neighbour of $v$ not in $S$, $G$ contains a $W_{7}$-subdivision centred on $v$.
\end{lem}

\begin{proof}

Suppose there exists some bridge $M$ of $G|S$ such that $M\setminus S$ contains at least two neighbours of $v$.

(a) Suppose firstly that $M = Y$.

It cannot be the case that $X\cap S = Z\cap S$ unless $X\cap S = Z\cap S = S$. Therefore, since $v$ is contained in $X\cap S$, $Y\cap S$ and $Z\cap S$, there exists some cycle $C$ disjoint from $v$ that passes through $\langle X\setminus S\rangle$, $\langle Y\setminus S\rangle$, and $\langle Z\setminus S\rangle$. Since $X$, $Y$, and $Z$ each contain $v$, they must each contain at least one neighbour not in $S$ of $v$. Therefore, there exists a $W_{6}$ subdivision $H$ centred on $v$ such that the rim is formed from $C$, three of the spokes are formed from $P_{u}$, $P_{w}$, and $P_{t}$, one spoke lies in $\langle (X\setminus S)\cup \{v_{0}\}\rangle$, one spoke lies in $\langle (Y\setminus S)\cup \{v_{0}\}\rangle$, and one spoke lies in $\langle (Z\setminus S)\cup \{v_{0}\}\rangle$. 

We know that $Y\setminus S$ contains some vertex $y$ adjacent to $v$, such that $y\notin N_{H}(v)$. Let $v_{1}$ be the spoke-meets-rim vertex of $H$ in $Y\setminus S$. Let $P_{1}$ be the spoke of $H$ from $v$ to $v_{1}$. Since $y$ is contained in the bridge $Y$, there must be some path $Q_{1}$ in $\langle Y\setminus S\rangle$ joining $y$ to $H\cap \langle Y\setminus S\rangle$ that first meets $H\cap \langle Y\setminus S\rangle$ at some vertex $p$. 

\textbf{1.} Suppose firstly that $p$ does not lie on $P_{1}$. Then form a $W_{7}$-subdivision in $G$ from $H\cup vyQ_{1}p$.

\textbf{2.} Suppose now that $p$ lies on $P_{1}$. Without loss of generality, let $H$, $y$, and $Q_{1}$ be chosen to minimise the distance between $p$ and $v_{1}$ along $P_{1}$.

Suppose there exists some path $Q_{2}$ from $vP_{1}p - p$ to $(H\cap \langle Y\rangle) - vP_{1}p$ that meets $H$ only at its endpoints. Such a path cannot meet $v_{1}P_{1}p$, or the distance between $p$ and $v_{1}$ is no longer minimal with regards to $H$, $y$, and $Q_{1}$. Suppose then that $Q_{2}$ runs from $vP_{1}p - p$ to $(H\cap \langle Y\rangle) - P_{1}$. Then a $W_{7}$-subdivision can be formed using parts of $H$ as well as $Q_{1}$ and $Q_{2}$. Suppose then that no such path exists. If $Q_{1}$ is trivial, that is, $y = p$, then the 3-connectivity of $G$ is violated, since the removal of $y$ and $v$ will disconnect the graph. Assume then that $y \neq p$. Then by 3-connectivity, there must be some path $Q_{3}$ in $\langle Y\rangle$ joining $Q_{1} - p$ to $(H\cap \langle Y\rangle) - P_{1}$, such that $Q_{3}$ meets $Q_{1}$ only at one endpoint, say $p_{1}$, and meets $H$ only at the other endpoint, say $q$. If such a path does not exist, the removal of $p$ and $v$ will disconnect the graph. If $q \in \langle Y\setminus S\rangle$, form a $W_{7}$-subdivision from $H\cup yv \cup yQ_{1}p_{1}Q_{3}q$. Suppose then that $q\in S$. Without loss of generality, suppose $q = w$.

The portion of the rim of $H$ that is contained in $\langle Y\rangle$ consists of a path, say $R$, and a single vertex in $S$ disjoint from this path.

\textbf{2.1.} Suppose that $w$ forms one of the endpoints of $R$. Then form a $W_{7}$-subdivision from $H$ by replacing the portion of rim formed by $v_{1}Rw$ with the path $v_{1}P_{1}pQ_{1}p_{1}Q_{3}w$, so that $vyQ_{1}p_{1}$ becomes a spoke, and $y$ and $p$ both become spoke-meets-rim vertices.

\textbf{2.2.} Suppose then that $u$ and $t$ form the two endpoints of $R$. Thus, one of $\langle X\rangle$, $\langle Z\rangle$ contains a path in $H$ from $w$ to $t$, while the other contains a path in $H$ from $w$ to $u$.

\textbf{2.2.1.} Suppose one of $X$, $Z$ contains all four vertices in $S$. Without loss of generality, suppose $X\cap S = S$. Since the rim of $H$ in $Y$ runs from $u$ to $t$, $\langle X\rangle$ must contain a path $R_{1}$ in $H$ either from $w$ to $t$ or from $w$ to $u$, such that $R_{1}$ meets $S$ only at its endpoints. Assume without loss of generality that $R_{1}$ is a path from $w$ to $t$. Since $X$ contains all four vertices in $S$, there exists some neighbour $u_{1}$ of $u$ in $X\setminus S$, and since $u_{1}$ is contained in the bridge $X$, there exists a path $R_{2}$ in $\langle X\setminus S\rangle$ from $u_{1}$ to $R_{1}$ that meets $R_{1}$ only at some vertex $r$. Let $R_{x} = uu_{1}R_{2}rR_{1}t$. Since $v \in X$, there exists some path $P_{2}$ in $X$ from $v$ to $R_{x}$ that meets $S$ only at $v$ and meets $R_{x}$ only at some point $r_{x}$. Form a $W_{7}$-subdivison from $H$ by replacing the parts of $H$ in $X$ with the paths $R_{x}$ and $P_{2}$, so that $r_{x}$ becomes a spoke-meets-rim vertex, and by replacing the path $v_{1}Ru$ in $\langle Y\rangle$ with the path $v_{1}P_{1}pQ_{1}p_{1}Q_{3}w$, so that $vyQ_{1}p_{1}$ again forms a new spoke.

\textbf{2.2.2.} Suppose now that $|X\cap S| = 3$ and $|Z\cap S| = 3$. Then, by the hypothesis of the Lemma, there exists some fourth bridge $W$ of $G|S$ which also contains $v$. If $|W\cap S| = 4$, then the arguments used in case 2.2.1 can be applied to show that a $W_{7}$-subdivision can be formed in $G$. Assume then that $|W\cap S| = 3$.

By the hypothesis of the Lemma, no two of $W\cap S$, $X\cap S$, and $Z\cap S$ can be the same. Since one of $\langle X\rangle$, $\langle Z\rangle$ contains a path in $H$ from $w$ to $t$, and the other contains a path in $H$ from $w$ to $u$, we know that $\{(X\cap S), (Z\cap S)\} = \{\{t, v, w\}, \{u, v, w\}\}$. Thus, $W\cap S = \{t, u, v\}$. Therefore, there exists a path $R_{2}$ in $\langle W\rangle$ from $t$ to $u$ that meets $S$ only at its endpoints, and a path $P_{2}$ in $\langle W\rangle$ from $v$ to $R_{2}$ that meets $S$ only at $v$ and meets $R_{2}$ only at some vertex $r$. Suppose without loss of generality that the part of the rim of $H$ in $\langle X\rangle$ consists of a path from $t$ to $w$. Then form a $W_{7}$-subdivison from $H$ by removing the portion of $H$ contained in $\langle X\setminus S\rangle$ and instead using the paths $P_{2}$ and $R_{2}$ in $\langle W\rangle$, so that $r$ becomes a spoke-meets-rim vertex, and by replacing the path $v_{1}Ru$ in $\langle Y\rangle$ with the paths $v_{1}P_{1}pQ_{1}p_{1}Q_{3}w$, so that $vyQ_{1}p_{1}$ again forms a new spoke.

(b) Suppose now that $M \neq Y$.

\textbf{1.} Suppose that $|M\cap S| = S$. Then the same arguments used in (a) can be used to show that a $W_{7}$-subdivision exists in $G$ centred on $v$. (If $M$ is some bridge $U'$ where $U' \in \{X, Z, A, B, C\}$, then at the points in the proof where $U'$ is required, instead use the structure contained in $Y$.)

\textbf{2.} Suppose now that $|M\cap S| = 3$. Let $x_{1}$, $x_{2}$ be the two neighbours of $v$ in $M\setminus S$.

Since $v\in M\cap S$, suppose without loss of generality that $M\cap S = \{u, v, w\}$. We know that $G - \{u, v, w\}$ can have at most two components, $\langle M\setminus \{u, v, w\}\rangle$ and $G - M$. Since $\langle M\setminus \{u, v, w\}\rangle$ contains at least two neighbours of $v$ ($x_{1}$ and $x_{2}$), and $G - M$ contains at least three neighbours of $v$ (one neighbour lies along the path $P_{t}$, and two neighbours lie in two other bridges of $G|S$, since at least three bridges of $G|S$ contain $v$), by Lemma \ref{lemmaW7}, a $W_{7}$-subdivision exists centred on $v$.

Thus, whenever any bridge of $G|S$ contains more than one neighbour of $v$ not in $S$, $G$ contains a $W_{7}$-subdivision centred on $v$.
\end{proof}

\begin{lem}
\label{threenbrs}
Let $G$ be a 3-connected graph containing a separating set $S = \{t, u, v, w\}$. Suppose $v$ has degree $\ge 7$, and suppose there are at least three bridges of $G|S$, $X$, $Y$, and $Z$, such that $Y$ contains all four vertices in $S$, while $X$ and $Z$ each contain $v$. Suppose also that one of the following holds:

\begin{itemize}
\item either one of $X$, $Z$ contains all four vertices in $S$; or
\item there exists some fourth bridge $W$ of $G|S$ such that $W$ also contains $v$.
\end{itemize}

Suppose for each $U_{i}, U_{j}$, where $U_{i}, U_{j} \in \{W, X, Z\}$ and $U_{i} \neq U_{j}$, that one of the following holds:

\begin{itemize}
\item either $U_{i}\cap S \neq U_{j}\cap S$; or
\item $U_{i}\cap S = U_{j}\cap S = S$.
\end{itemize}

Suppose that $v$ and $u$ are either adjacent or joined by a path in some bridge $A$ of $G|S$ other than $X$, $Y$, $Z$, or $W$ (if $W$ exists). Call this path (or edge) $P_{u}$. Suppose also that $v$ and $w$ are either adjacent or joined by a path in some bridge $B$ of $G|S$ other than $X$, $Y$, $Z$, $W$, or $A$ (if $W$ and $A$ exist). Call this path (or edge) $P_{w}$.

Then if any bridge of $G|S$ contains more than two neighbours of $v$ not in $S$, $G$ contains a $W_{7}$-subdivision centred on $v$.
\end{lem}

\begin{proof}
Suppose there exists some bridge $M$ of $G|S$ such that $M\setminus S$ contains at least three neighbours of $v$.

(a) Suppose firstly that $M = Y$.

By the same arguments used in Lemma \ref{twonbrs}, there exists some $W_{6}$ subdivision $H$ centred on $v$ such that two of the spokes are formed from $P_{u}$ and $P_{w}$, two spokes lie in $\langle (Y\setminus S)\cup \{v_{0}\}\rangle$, and two spokes lie in two other bridges of $G|S$. Let $v_{1}$ and $v_{2}$ be the spoke-meets-rim vertices of $H$ in $Y\setminus S$. Let $P_{1}$ and $P_{2}$ be the spokes of $H$ from $v$ to $v_{1}$ and from $v$ to $v_{2}$ respectively.

We know that $Y\setminus S$ contains some vertex $y$ adjacent to $v$, such that $y\notin N_{H}(v)$. Since $y$ is contained in the bridge $Y$, there must exist some path $Q_{1}$ in $\langle Y\setminus S\rangle$ from $y$ to $H\cap \langle Y\setminus S\rangle$, such that $Q_{1}$ first meets $H\cap \langle Y\setminus S\rangle$ at some vertex $p$.

\textbf{1.} Suppose firstly that $p$ does not lie on $P_{1}$ or $P_{2}$. Then form a $W_{7}$-subdivision in $G$ from $H\cup vyQ_{1}p$.

\textbf{2.} Suppose now that $p$ lies on $P_{1}$ or $P_{2}$. Without loss of generality, let $p$ lie on $P_{1}$, and let $H$, $y$, and $Q_{1}$ be chosen to minimise the distance between $p$ and $v_{1}$ along $P_{1}$.

Suppose there exists some path $Q_{2}$ from $vP_{1}p - p$ to $(H\cap \langle Y\rangle) - vP_{1}p$ that meets $H$ only at its endpoints. Such a path cannot meet $v_{1}P_{1}p$, or the distance between $p$ and $v_{1}$ is no longer minimal with respect to $H$, $y$, and $Q_{1}$. Suppose then that $Q_{2}$ runs from $vP_{1}p - p$ to $(H\cap \langle Y\rangle) - P_{1}$. Then a $W_{7}$-subdivision can be formed using parts of $H$, $Q_{1}$ and $Q_{2}$. Suppose then that no such path exists. If $Q_{1}$ is trivial, that is, $y = p$, then the 3-connectivity of $G$ is violated, since the removal of $y$ and $v$ will disconnect the graph. Assume then that $y \neq p$. Then by 3-connectivity, there must be some path $Q_{3}$ in $\langle Y\rangle$ joining $Q_{1} - p$ to $(H\cap \langle Y\rangle) - P_{1}$, such that $Q_{3}$ meets $Q_{1}$ only at one endpoint, say $p_{1}$, and meets $H$ only at the other endpoint, say $q$. If such a path does not exist, the removal of $p$ and $v$ will disconnect the graph. If $q \in \langle Y\setminus S\rangle$, form a $W_{7}$-subdivision from parts of $H$, $Q_{1}$, and $Q_{3}$. Suppose then that $q\in S$. Without loss of generality, suppose $q = w$.

The portion of the rim of $H$ that is contained in $\langle Y\rangle$ consists of a path, say $R$, and a single vertex in $S$ disjoint from this path.

\textbf{2.1.} Suppose that $w$ forms one of the endpoints of $R$. Then form a $W_{7}$-subdivision from $H$ by replacing the portion of rim formed by $v_{1}Rw$ with the path $v_{1}P_{1}pQ_{1}p_{1}Q_{3}w$, so that $vyQ_{1}p_{1}$ becomes a spoke, and $p$ becomes a spoke-meets-rim vertex instead of $v_{1}$. (If $v_{2}$ is on the path $R$, then extend the spoke formed by $P_{2}$ along $R$ so that $v_{1}$ also becomes a spoke-meets-rim vertex.)

\textbf{2.2.} Suppose then that $u$ and $t$ form the two endpoints of $R$. Thus, one of $\langle X\rangle$, $\langle Z\rangle$ contains a path in $H$ from $w$ to $t$, while the other contains a path in $H$ from $w$ to $u$.

By the same arguments used in case 2.2 of Lemma \ref{twonbrs}, a $W_{7}$-subdivision can be formed from $H$ by replacing parts of the rim so that the rim in $\langle Y\rangle$ runs from $u$ to $w$, and the three spokes in $Y$ are formed from the paths $vP_{1}p$, $vyQ_{1}p_{1}$, and either $P_{2}$ or $P_{2}\cup v_{2}Rv_{1}$ (depending on whether or not $v_{2}$ still lies on the new rim).

(b) Suppose now that $M \neq Y$.

\textbf{1.} Suppose that $|M\cap S| = S$. Then the same arguments used in (a) can be used to show that a $W_{7}$-subdivision exists in $G$ centred on $v$. (If $M$ is some bridge $U'$ where $U' \in \{X, Z, A, B\}$, then at the points in the proof where $U'$ is required, instead use the structure contained in $Y$.)

\textbf{2.} Suppose now that $|M\cap S| = 3$. Let $x_{1}$, $x_{2}, x_{3}$ be the three neighbours of $v$ in $M\setminus S$.

Let $S_{1} = M\cap S$. We know that $G - S_{1}$ can have at most two components, $\langle M\setminus S_{1}\rangle$ and $G - M$. Since $\langle M\setminus S_{1}\rangle$ contains at least three neighbours of $v$ ($x_{1}$, $x_{2}$, and $x_{3}$), and $G - M$ contains at least two neighbours of $v$ (in two other bridges of $G|S$, since at least three bridges of $G|S$ contain $v$), then if $S_{1} = \{u, v, w\}$, by Lemma \ref{lemmaW7}, a $W_{7}$-subdivision exists centred on $v$. Suppose then that one of $u, w$ is not in $S_{1}$. Without loss of generality, suppose $S_{1} = \{t, u, v\}$. Then, since $P_{w}$ is in $G - M$ (except for $v$), there are now at least three neighbours of $v$ in $G - M$ (in two of $X\setminus S$, $Y\setminus S$, and $Z\setminus S$, and along the path $P_{w}$). Thus, by applying Lemma \ref{lemma2} to both $\langle M\setminus S_{1}\rangle$ and $G - M$, and using the path $P_{u}$, a $W_{7}$-subdivision can still be formed centred on $v$.

Thus, whenever any bridge of $G|S$ contains more than two neighbours of $v$ not in $S$, $G$ contains a $W_{7}$-subdivision centred on $v$.
\end{proof}

\begin{lem}
\label{twotwonbrs}
Let $G$ be a 3-connected graph containing a separating set $S = \{t, u, v, w\}$ such that the removal of any smaller subset of $S$ will separate $G$ into at most three components. Suppose $v$ has degree $\ge 7$, and suppose there are at least three bridges of $G|S$, $X$, $Y$, and $Z$, such that $Y$ contains all four vertices in $S$, while $X$ and $Z$ each contain $v$. Suppose also that one of the following holds:

\begin{itemize}
\item either one of $X$, $Z$ contains all four vertices in $S$; or
\item there exists some fourth bridge $W$ of $G|S$ such that $W$ also contains $v$.
\end{itemize}

Suppose for each $U_{i}, U_{j}$, where $U_{i}, U_{j} \in \{W, X, Z\}$ and $U_{i} \neq U_{j}$, that one of the following holds:

\begin{itemize}
\item either $U_{i}\cap S \neq U_{j}\cap S$; or
\item $U_{i}\cap S = U_{j}\cap S = S$.
\end{itemize}

Suppose that $v$ and $u$ are either adjacent or joined by a path in some bridge $A$ of $G|S$ other than $X$, $Y$, $Z$, or $W$ (if $W$ exists). Call this path (or edge) $P_{u}$. Suppose also that $v$ and $w$ are either adjacent or joined by a path in some bridge $B$ of $G|S$ other than $X$, $Y$, $Z$, $W$, or $A$ (if $W$ and $A$ exist). Call this path (or edge) $P_{w}$.

Then if any two bridges of $G|S$ each contain more than one neighbour of $v$ not in $S$, $G$ contains a $W_{7}$-subdivision centred on $v$.
\end{lem}

\begin{proof}
Suppose there exist bridges $M_{1}$ and $M_{2}$ of $G|S$ such that $M_{1}\setminus S$ and $M_{2}\setminus S$ each contain at least two neighbours of $v$.

By the same kind of arguments used early in the proof of Lemma \ref{twonbrs}, there exists some $W_{6}$ subdivision $H_{1}$ centred on $v$ such that two of the spokes are formed from $P_{u}$ and $P_{w}$, two spokes lie in $\langle (M_{1}\setminus S)\cup \{v_{0}\}\rangle$, and two spokes lie in two other bridges of $G|S$. There also exists some other $W_{6}$ subdivision $H_{2}$ centred on $v$ such that two of the spokes are formed from $P_{u}$ and $P_{w}$, two spokes lie in $\langle (M_{2}\setminus S)\cup \{v_{0}\}\rangle$, and two spokes lie in two other bridges of $G|S$.

Suppose the rim of $H_{1}$ in $\langle M_{1}\rangle$ runs from $x_{1}$ to $x_{2}$, where $x_{1} \neq x_{2}$ and $\{x_{1}, x_{2}\} \subseteq \{t, u, w\}$. Let $x_{3}$ refer to the third vertex in $\{t, u, w\}$, such that $x_{3} \notin \{x_{1}, x_{2}\}$.

\textbf{1.} Suppose the rim of $H_{2}$ in $\langle M_{2}\rangle$ does not run from $x_{1}$ to $x_{2}$.

Without loss of generality, suppose the rim of $H_{2}$ in $M_{2}$ runs from $x_{2}$ to $x_{3}$. Let $U$ be the bridge of $G|S$ such that $\langle U\rangle$ contains a path from $x_{2}$ to $x_{3}$ that forms part of the rim of $H_{1}$. Then form a $W_{7}$-subdivision from $H_{1}$ by replacing $H_{1}\cap \langle U\rangle$ with $H_{2}\cap \langle M_{2}\rangle$.

\textbf{2.} Suppose then that the rim of $H_{2}$ in $\langle M_{2}\rangle$ runs from $x_{1}$ to $x_{2}$.

\textbf{2.1.} Suppose firstly that $M_{1}, M_{2} \in \{W, X, Y, Z\}$.

Since $M_{1}$ and $M_{2}$ each contain $x_{1}$, $v$, and $x_{2}$, and $M_{1}\cap S \neq M_{2}\cap S$ unless $M_{1}\cap S = M_{2}\cap S = S$, one of $M_{1}$, $M_{2}$ must contain all of $S$. Suppose without loss of generality that $M_{2}\cap S = S$. 

There exists a path $R$ in $H_{2}\cap \langle M_{2}\rangle$ from $x_{1}$ to $x_{2}$ that meets $S$ only at its endpoints, such that $R$ forms part of the rim of $H_{2}$. Since $M_{2}$ contains all four vertices in $S$, there exists some neighbour $x'_{3}$ of $x_{3}$ in $M_{2}\setminus S$, and since $x'_{3}$ is contained in the bridge $M_{2}$, there exists a path $R'$ in $\langle M_{2}\setminus S\rangle$ from $x'_{3}$ to $R$ that meets $R$ only at some vertex $r$. Let $R_{1} = x_{3}x'_{3}R'rRx_{1}$, and let $R_{2} = x_{3}x'_{3}R'rRx_{2}$. Let $U_{1}$ be the bridge of $G|S$ such that $\langle U_{1}\rangle$ contains a path from $x_{1}$ to $x_{3}$ that forms part of the rim of $H_{1}$. Let $U_{2}$ be the bridge of $G|S$ such that $\langle U_{2}\rangle$ contains a path from $x_{2}$ to $x_{3}$ that forms part of the rim of $H_{1}$. If both of the spoke-meets-rim vertices in $H_{2}\cap \langle M_{2}\rangle$ lie on $R_{1}$, form a $W_{7}$-subdivision from $H_{1}$ by replacing $H_{1}\cap \langle U_{1}\rangle$ with $R_{1}$ and the two paths that form spokes in $H_{2}\cap \langle M_{2}\rangle$. If both of the spoke-meets-rim vertices in $H_{2}\cap \langle M_{2}\rangle$ lie on $R_{2}$, form a $W_{7}$-subdivision from $H_{1}$ by replacing $H_{1}\cap \langle U_{2}\rangle$ with $R_{2}$ and the two paths that form spokes in $H_{2}\cap \langle M_{2}\rangle$. Suppose then that one spoke-meets-rim vertex in $H_{2}\cap\langle M_{2}\rangle$, say $v_{1}$, lies on $R_{1} - r$, and the other, say $v_{2}$, lies on $R_{2} - r$. Then form a $W_{7}$-subdivision from $H_{1}$ by removing $H_{1}\cap \langle U_{1}\rangle$, and adding the path $R_{1}$, the two paths that form spokes in $H_{2}\cap \langle M_{2}\rangle$, and the path $v_{2}R_{2}r$.

\textbf{2.2.} Suppose now that one of $M_{1}$, $M_{2}$ is not in $\{W, X, Y, Z\}$.

If either of $M_{1}$ or $M_{2}$ contain all of $S$, then the same argument used in case 2.1 can be applied. Assume then that $M_{1}\cap S = M_{2}\cap S = \{x_{1}, v, x_{2}\}$.

By 3-connectivity, there are two disjoint paths $P_{1}$ and $P_{2}$ in $\langle M_{2}\rangle$, such that these paths run from $v$ to $x_{1}$ and $v$ to $x_{2}$ respectively, and meet $S$ only at their endpoints. By using these two paths, and by applying Lemma \ref{lemma1} to $M_{1}$ and Lemma \ref{lemma2} to the component of $G - \{x_{1}, v, x_{2}\}$ which contains $X$, $Y$, and $Z$, a $W_{7}$-subdivision can be formed centred on $v$.

Thus, whenever any two bridges of $G|S$ each contain more than one neighbour of $v$ not in $S$, $G$ contains a $W_{7}$-subdivision centred on $v$.
\end{proof}

\begin{lem}
\label{threetwonbrs}
Let $G$ be a 3-connected graph containing a separating set $S = \{t, u, v, w\}$ such that the removal of any smaller subset of $S$ will separate $G$ into at most three components. Suppose $v$ has degree $\ge 7$, and suppose there are at least three bridges of $G|S$, $X$, $Y$, and $Z$, such that $Y$ contains all four vertices in $S$, while $X$ and $Z$ each contain $v$. Suppose also that one of the following holds:

\begin{itemize}
\item either one of $X$, $Z$ contains all four vertices in $S$; or
\item there exists some fourth bridge $W$ of $G|S$ such that $W$ also contains $v$.
\end{itemize}

Suppose for each $U_{i}, U_{j}$, where $U_{i}, U_{j} \in \{W, X, Z\}$ and $U_{i} \neq U_{j}$, that one of the following holds:

\begin{itemize}
\item either $U_{i}\cap S \neq U_{j}\cap S$; or
\item $U_{i}\cap S = U_{j}\cap S = S$.
\end{itemize}

Suppose that $v$ and $u$ are either adjacent or joined by a path in some bridge $A$ of $G|S$ other than $X$, $Y$, $Z$, or $W$ (if $W$ exists). Call this path (or edge) $P_{u}$.

Then if some bridge of $G|S$ contains more than one neighbour of $v$ not in $S$, and some other bridge of $G|S$ contains more than two neighbours of $v$ not in $S$, $G$ contains a $W_{7}$-subdivision centred on $v$.
\end{lem}

\begin{proof}
Suppose there exist bridges $M_{1}$ and $M_{2}$ of $G|S$ such that $M_{1}\setminus S$ contains at least three neighbours of $v$, and $M_{2}\setminus S$ contains at least two neighbours of $v$.

By the same arguments used in Lemma \ref{twotwonbrs}, there exists some $W_{6}$ subdivision $H$ centred on $v$ such that exactly two spokes of $H$ lie in $\langle M_{1}\rangle$.

We know that $M_{1}\setminus S$ contains some vertex $x$ adjacent to $v$, such that $x\notin N_{H}(v)$. By the same arguments used in Lemma \ref{threenbrs}, there exists some path or paths from $x$ to $H\cap \langle M_{1}\rangle$ which can be used to construct a $W_{7}$-subdivision centred on $v$.
\end{proof}

\begin{lem}
\label{twotwotwonbrs}
Let $G$ be a 3-connected graph containing a separating set $S = \{t, u, v, w\}$ such that the removal of any smaller subset of $S$ will separate $G$ into at most three components. Suppose $v$ has degree $\ge 7$, and suppose there are at least three bridges of $G|S$, $X$, $Y$, and $Z$, such that $Y$ contains all four vertices in $S$, while $X$ and $Z$ each contain $v$. Suppose also that one of the following holds:

\begin{itemize}
\item either one of $X$, $Z$ contains all four vertices in $S$; or
\item there exists some fourth bridge $W$ of $G|S$ such that $W$ also contains $v$.
\end{itemize}

Suppose for each $U_{i}, U_{j}$, where $U_{i}, U_{j} \in \{W, X, Z\}$ and $U_{i} \neq U_{j}$, that one of the following holds:

\begin{itemize}
\item either $U_{i}\cap S \neq U_{j}\cap S$; or
\item $U_{i}\cap S = U_{j}\cap S = S$.
\end{itemize}

Suppose that $v$ and $u$ are either adjacent or joined by a path in some bridge $A$ of $G|S$ other than $X$, $Y$, $Z$, or $W$ (if $W$ exists). Call this path (or edge) $P_{u}$.

Then if there exist three bridges of $G|S$ that each contain more than one neighbour of $v$ not in $S$, $G$ contains a $W_{7}$-subdivision centred on $v$.
\end{lem}

\begin{proof}
Suppose there exist bridges $M_{1}$, $M_{2}$, and $M_{3}$ of $G|S$ such that each contain at least two neighbours of $v$ not in $S$.

By the same arguments used in Lemma \ref{twotwonbrs}, there exists some $W_{6}$ subdivision $H_{1}$ centred on $v$ such that one spoke is formed from $P_{u}$, two spokes lie in $\langle (M_{1}\setminus S)\cup \{v_{0}\}\rangle$, two spokes lie in $\langle (M_{2}\setminus S)\cup \{v_{0}\}\rangle$, and one spoke lies in some other bridge of $G|S$. There also exists some other $W_{6}$ subdivision $H_{2}$ centred on $v$ such that one spoke is formed from $P_{u}$, two spokes lie in $\langle (M_{2}\setminus S)\cup \{v_{0}\}\rangle$, two spokes lie in $\langle (M_{3}\setminus S)\cup \{v_{0}\}\rangle$, and one spoke lies in some other bridge of $G|S$. Again by using the arguments in Lemma \ref{twotwonbrs}, a $W_{7}$-subdivision can be formed centred on $v$ using parts of $H_{1}$ and $H_{2}$.
\end{proof}

\begin{lem}
\label{threebridges}
Let $G$ be a 3-connected graph with at least 19 vertices, containing no internal 3-edge-cutsets, no type 1, 2, 2a, or 3 edge-vertex-cutsets, and on which Reductions \ref{r1}A, \ref{r1}B, and \ref{r1}C cannot be performed. Suppose $G$ contains a separating set $S = \{u, v, w\}$, such that there exist at least three bridges, $X$, $Y$, $Z$, of $G|S$. Suppose that $Z\setminus S$ contains at least three neighbours of $v$. Suppose that $v$ is either adjacent to $u$ or joined by a path in some fourth bridge $A$ of $G|S$ other than $X$, $Y$, or $Z$. Call this path (or edge) $P_{u}$. Suppose also that $v$ is either adjacent to $w$ or joined by a path in some bridge $B$ of $G|S$ other than $X$, $Y$, $Z$, or $A$ (if $A$ exists). Call this path (or edge) $P_{w}$. Then there exists a $W_{7}$-subdivision in $G$.
\end{lem}

\begin{proof}
If any bridge of $G|W$ other than $Z$ contains more than one neighbour of $v$ not in $S$, then by Lemma \ref{lemmaW7} a $W_{7}$-subdivision can be formed centred on $v$. Assume then that this is not the case. Thus, if any bridge other than $Z$ contains more than three vertices not in $S$, a type 2 edge-vertex-cutset can be formed from $u$, $w$, and an edge incident with $v$ in that bridge. Assume then that every bridge of $G|S$ other than $Z$ contains at most three vertices not in $S$.

$X\setminus S$ must contain more than one vertex, otherwise Reduction \ref{r1}A can be performed on $G$. Thus, to prevent an internal 3-edge-cutset, there must exist some fourth edge joining $X\setminus S$ to $S$. Some member $x$ of $S$ then must have at least two neighbours in $X\setminus S$. Since $v$ can have at most one neighbour in $X\setminus S$, we can suppose that $x \in \{u, w\}$. Assume $x = u$ without loss of generality.

$Y\setminus S$ must also contain more than one vertex, again to avoid enabling Reduction \ref{r1}A. By the same argument, some member $y \in \{u, w\}$ must have at least two neighbours in $Y\setminus S$. Consider the two options:

(a) Suppose $y = w$.

If $u$ has more than two neighbours in $Z\setminus S$ then by Lemma \ref{lemmaW7} a $W_{7}$-subdivision can be formed centred on $u$. Similarly, if $w$ has more than two neighbours in $Z\setminus S$ then a $W_{7}$-subdivision can be formed centred on $w$. Assume then that $u$ and $w$ each have at most two neighbours in $Z\setminus S$. Then if $Z\setminus S$ contains more than three vertices, there exists a type 3 edge-vertex-cutset in $G$. Assume then that $|Z\setminus S| \le 3$.

Since each bridge of $G|S$ contains at most three vertices not in $S$, and $|V(G)| \ge 19$, there must exist at least three more bridges of $G|S$ other than $X$, $Y$, and $Z$. By the same arguments used for $X$ and $Y$, these bridges must each contain at least two but no more than three vertices not in $S$, and can contain no more than one neighbour of $v$ not in $S$.

Suppose every bridge of $G|S$ other than $Z$ contains only two vertices not in $S$. Then, since there are at least five such bridges of $G|S$, there must exist some bridge that is contained as a subdivision in at least two other bridges, thus enabling Reduction \ref{r1}B. Suppose then that this is not the case; thus, there exists some bridge $U'$ of $G|S$ other than $Z$ such that $|U'\setminus S| = 3$.

If any bridge of $G|S$ other than $Y$ has three neighbours of $w$ not in $S$, then by Lemma \ref{lemmaW7} a $W_{7}$-subdivision can be formed. Suppose then that every bridge of $G|S$ other than $Y$ has at most two neighbours of $w$ not in $S$. Similarly, if any bridge of $G|S$ other than $X$ has three neighbours of $u$ not in $S$, then by Lemma \ref{lemmaW7} a $W_{7}$-subdivision can be formed. Suppose then that every bridge of $G|S$ other than $X$ has at most two neighbours of $u$ not in $S$. 

Suppose $U' \notin \{X, Y\}$. Then $u$ and $w$ each have at most two neighbours in $U'\setminus S$. Since $U' \neq Z$, $U'\setminus S$ contains only one neighbour of $v$, say $v'$. Thus, a type 2a edge-vertex-cutset can be formed from $u$, $w$, and $vv'$.

Suppose then that $U' \in \{X, Y\}$. Without loss of generality, suppose that $U' = X$, thus, $|X\setminus S| = 3$. If $u$ has at most two neighbours in $X\setminus S$, then a type 2a edge-vertex-cutset can be formed, as in the previous paragraph. Suppose then that $u$ has three neighbours in $X\setminus S$. If any bridge of $G|S$ other than $X$ has more than one neighbour of $u$ not in $S$, then by Lemma \ref{lemmaW7} a $W_{7}$-subdivision can be formed centred on $u$. Suppose then that each bridge of $G|S$ other than $X$ contains at most one neighbour of $u$ not in $S$. Let $\mathcal{B}$ be the set of bridges of $G|S$ other than $X, Y, Z$. For each bridge $B_{i} \in \mathcal{B}$, $B_{i}\setminus S$ contains exactly one neighbour of $u$, one neighbour of $v$, and two neighbours of $w$ (since otherwise an internal 3-edge-cutset would exist). Since there are only four edges joining $S$ to each $B_{i}\setminus S$, there can be at most two vertices in each $B_{i}\setminus S$, otherwise an internal 4-edge-cutset exists. Thus, each of the bridges in $\mathcal{B}$ are isomorphic. Since $|\mathcal{B}| \ge 3$, Reduction 1B can be performed on $G$.

(b) Suppose now that $y = u$, and that $w$ has only one neighbour in $Y\setminus S$.

If $u$ has more than two neighbours in $Y\setminus S$, then by Lemma \ref{lemmaW7} a $W_{7}$-subdivision can be formed from bridges $X$ and $Y$. Assume then that $u$ has exactly two neighbours in $Y\setminus S$. Thus, there are exactly four edges joining $S$ to $Y\setminus S$. To avoid an internal 4-edge-cutset, then, $|Y\setminus S| = 2$.

Suppose $u$ has at least two neighbours in $Z\setminus S$. Then Reduction \ref{r1}B can be performed, since $\langle Y\rangle$ is contained as a subdivision in both $\langle X\rangle$ and $\langle Z\rangle$. Suppose then that $u$ has only one neighbour, say $u'$, in $Z\setminus S$. Then, to avoid creating a type 2 edge-vertex-cutset from $v$, $w$, and the edge $uu'$, $|Z\setminus S| \le 3$. Thus, since $|V(G)| \ge 19$, there must be at least three more bridges of $G|S$ other than $X$, $Y$, and $Z$.

If any bridge of $G|S$ other than $X$ and $Y$ contains more than one neighbour of $u$ not in $S$, then $\langle Y\rangle$ is contained as a subdivision in the induced subgraph formed by that bridge, as well as in $\langle X\rangle$, and so Reduction \ref{r1}B can be performed. Thus, there exist at least three bridges other than $Z$ which contain only one neighbour of $u$ not in $S$. Since each of these bridges also contains only one neighbour of $v$ not in $S$, and has at most three vertices not in $S$, at least one of the bridges must be contained as a subdivision in the other two. Thus, either Reduction \ref{r1}B or Reduction \ref{r1}C can be performed on $G$.
\end{proof}

\begin{lem}
\label{fournbrs}
Let $G$ be a 3-connected graph with at least 19 vertices, containing no internal 3-edge-cutsets, no type 1, 2, 2a, or 3 edge-vertex-cutsets, and on which Reductions \ref{r1}A, \ref{r1}B, and \ref{r1}C cannot be performed. Suppose $G$ contains a separating set $S = \{t, u, v, w\}$. Suppose $v$ has degree $\ge 7$, and suppose there are at least three bridges of $G|S$, $X$, $Y$, and $Z$, such that $Y$ contains all four vertices in $S$, $X$ and $Z$ each contain $v$, and $Y\setminus S$ contains at most one neighbour of $u$. Suppose also that one of the following holds:

\begin{itemize}
\item either one of $X$, $Z$ contains all four vertices in $S$; or
\item there exists some fourth bridge $W$ of $G|S$ such that $W$ also contains $v$.
\end{itemize}

Suppose that one of the following holds:

\begin{itemize}
\item either $X\cap S \neq Z\cap S$; or
\item $X\cap S = Z\cap S = S$.
\end{itemize}

Suppose that $v$ and $u$ are either adjacent or joined by a path in some bridge $A$ of $G|S$ other than $X$, $Y$, $Z$, or $W$ (if $W$ exists). Call this path (or edge) $P_{u}$.

Then if $Y\setminus S$ contains more than three neighbours of $v$, $G$ contains a $W_{7}$-subdivision.
\end{lem}

\begin{proof}
Suppose $Y\setminus S$ contains at least four neighbours of $v$.

By the same arguments used in Lemma \ref{threenbrs}, there exists some $W_{6}$ subdivision $H$ centred on $v$ such that one of the spokes is formed from $P_{u}$, three spokes lie in $\langle (Y\setminus S)\cup \{v_{0}\}\rangle$, and two spokes lie in two other bridges of $G|S$.

If there exists some bridge of $G|S$ that contains $v$ other than $X$, $Y$, $Z$, or $A$, then this bridge can be used to form a path from $v$ to $w$ or from $v$ to $t$, thus forming a $W_{7}$-subdivision. Assume then that no other bridges of $G|S$ exist.

Let $v_{1}$, $v_{2}$, $v_{3}$ be the spoke-meets-rim vertices of $H$ in $Y\setminus S$, in order around the rim of $H$. Let $P_{1}$, $P_{2}$, $P_{3}$ be the spokes of $H$ from $v$ to $v_{1}$, $v$ to $v_{2}$, and $v$ to $v_{3}$ respectively. Let $H$ be chosen to minimise the sum of the lengths of the paths $P_{1}$, $P_{2}$, $P_{3}$.

We know that $Y\setminus S$ contains some vertex $y$ adjacent to $v$, such that $y\notin N_{H}(v)$. 

Suppose firstly that $y$ is some vertex in $H$. If $y \in H - (P_{1}\cup P_{2}\cup P_{3})$, then $H\cup vy$ is a $W_{7}$-subdivision. Suppose then that $y \in P_{i}$, where $i \in \{1, 2, 3\}$. Then the path $vyP_{i}v_{i}$ forms a shorter path than $P_{i}$ from $v$ to $v_{i}$, such that this path is still vertex-disjoint from $H - P_{i}$, so the sum of the lengths of $P_{1}$, $P_{2}$, $P_{3}$ is no longer minimal.

Suppose then that $y \notin H$. By the 3-connectivity of $G$, there must be two distinct vertices $q_{1}$ and $q_{2}$ in $H$, and two paths $Q_{1}$ and $Q_{2}$ in $\langle Y\rangle$, from $y$ to $q_{1}$ and $y$ to $q_{2}$ respectively, such that $Q_{1}$ and $Q_{2}$ are vertex-disjoint except at $y$ and are disjoint from $H$ except at their endpoints. (Note that if the rim of $H$ in $\langle Y\rangle$ meets $u$, then neither $Q_{1}$ nor $Q_{2}$ can meet $u$, since by the conditions of the hypothesis, $u$ has only one neighbour in $Y\setminus S$, and this neighbour lies in $H$.) For most placements of $q_{1}$ and $q_{2}$, it is straightforward to check that a $W_{7}$-subdivision can be formed centred on $v$. The only situation where this is not the case is for $\{q_{1}, q_{2}\} = \{v_{1}, v_{3}\}$. Suppose then that $q_{1} = v_{1}$ and $q_{2} = v_{3}$.

Let $S_{1} = \{q_{1}, v, q_{2}\}$. Suppose there exists some path from $y$ to $H - S_{1}$ such that $y$, $v_{2}$, and $u$ are not each in three separate bridges of $G|S_{1}$. It is straightforward to check that a $W_{7}$-subdivision exists in the resulting graph. Suppose then that $S_{1}$ forms a separating set, the removal of which places $y$, $v_{2}$, and $u$ in three separate components.

Call $B$ the bridge of $G|S_{1}$ containing $y$. Call $C$ the bridge of $G|S_{1}$ containing $v_{2}$. Call $D$ the bridge of $G|S_{1}$ containing $u$ and $w$.

Suppose there exists some internal vertex on one of the paths $P_{1}$ or $P_{3}$ such that this vertex is contained in one of the bridges $B$, $C$, or $D$. It is straightforward to check that the existence of such a vertex will result in a $W_{7}$-subdivision. Assume then that if such a vertex exists, it is contained in some fourth bridge $E$ of $G|S_{1}$.

Suppose firstly that such a bridge $E$ exists, and contains internal vertices of both the paths $P_{1}$ and $P_{3}$. Then by Lemma \ref{lemmaW7}, a $W_{7}$-subdivision exists in $G$.

Assume then that if the paths $P_{1}$ and $P_{3}$ both contain internal vertices, that the internal vertices of $P_{1}$ are contained in some bridge $E$ such that $E\notin \{B,C,D\}$, while the internal vertices of $P_{2}$ are contained in some other bridge $F$ such that $F\notin \{B,C,D,E\}$. By Lemma \ref{threebridges}, then, a $W_{7}$-subdivision exists in $G$.

Thus, whenever $Y\setminus S$ contains more than three neighbours of $v$, $G$ contains a $W_{7}$-subdivision.
\end{proof}

\begin{lem}
\label{lemma3}

Let $G$ be a 3-connected graph with at least 19 vertices. Suppose $G$ has no internal 3 or 4-edge-cutsets, no type 1, 2, 2a, 3, 3a, or 4 edge-vertex-cutsets, and is a graph on which none of Reductions \ref{r1}A, \ref{r1}B, \ref{r1}C, \ref{r2}A, and \ref{r6} can be performed. Let $S = \{u, v, w\}$ be a separating set of vertices in $G$ such that $v$ is adjacent to both $u$ and $w$, and such that there are exactly two bridges, $X$ and $Y$, of $G|S$. Suppose that $v$ has at least four neighbours in $X\setminus S$. Then $G$ contains a $W_{7}$-subdivision.
\end{lem}

\begin{proof}

By Lemma \ref{lemma2}, $\langle X\rangle$ contains a structure $H$, which consists of the following:

\begin{itemize}
\item a path $P_{H}$ from $u$ to $w$, that meets $S$ only at its endpoints; and
\item three paths $Q_{H1}, Q_{H2}, Q_{H3}$ and three vertices $q_{H1}$, $q_{H2}$, $q_{H3}$, such that $q_{H1}$, $q_{H2}$, $q_{H3}$ are distinct vertices in order on the path $P_{H}$, and $Q_{H1}, Q_{H2}, Q_{H3}$ are paths from $v$ to $q_{H1}$, $v$ to $q_{H2}$, and $v$ to $q_{H3}$ respectively, that are pairwise vertex-disjoint except at $v$ and meet $S$ only at $v$, such that each $Q_{Hi}$ is disjoint from $P_{H}$ except at $q_{Hi}$, $1 \le i \le 3$.
\end{itemize}

Let $H$ be chosen to minimise $|E(H)|$.

If $v$ contains more than one neighbour in $Y\setminus S$, then by Lemma \ref{lemmaW7}, a $W_{7}$-subdivision exists centred on $v$. Assume then that $v$ has exactly one neighbour $v'$ in $Y\setminus S$. Then if $|Y\setminus S| > 3$, a type 2 edge-vertex-cutset can be formed from $u$, $w$, and the edge $vv'$. Assume then that $|Y\setminus S| \le 3$.

There exists some vertex $a \in N_{\langle X\setminus \{u, w\}\rangle}(v) \setminus N_{H}(v)$, since $v$ has at least four neighbours in $X\setminus S$.

Suppose firstly that $a \in H$. If $a \in H - (Q_{H1}\cup Q_{H2}\cup Q_{H3})$, then a $W_{7}$-subdivision exists in $G$ centred on $v$, namely $H\cup \{v_{0}\}$. Suppose then that $a \in Q_{Hi}$, where $1 \le i \le 3$. Then the path $vaQ_{Hi}q_{Hi}$ forms a shorter path than $Q_{Hi}$ from $v$ to $q_{Hi}$, such that this path is still vertex-disjoint from $H - v - q_{Hi}$, so $|E(H)|$ is no longer minimal.

Suppose then that $a \notin H$. By the 3-connectivity of $G$, there must be two distinct vertices in $H$, $p_{1}$ and $p_{2}$, and two paths in $\langle X\rangle$, $P_{1}$ and $P_{2}$, from $a$ to $p_{1}$ and $a$ to $p_{2}$ respectively, such that $P_{1}$ and $P_{2}$ are vertex-disjoint except at $a$ and are disjoint from $H$ except at their endpoints. For most placements of $p_{1}$ and $p_{2}$, it is straightforward to check that a $W_{7}$-subdivision can be formed centred on $v$. The situations where this is not the case are:

\begin{itemize}
\item[1.] $\{p_{1}, p_{2}\} = \{u, q_{H2}\}$
\item[2.] $\{p_{1}, p_{2}\} = \{w, q_{H2}\}$
\item[3.] $\{p_{1}, p_{2}\} = \{q_{H1}, q_{H3}\}$
\item[4.] $\{p_{1}, p_{2}\} = \{u, q_{H3}\}$
\item[5.] $\{p_{1}, p_{2}\} = \{w, q_{H1}\}$
\end{itemize}

Each of these cases are addressed below.

\vspace{0.2in}
\noindent \textbf{Cases 1 and 2: $\{p_{1}, p_{2}\} = \{u, q_{H2}\}$ or $\{p_{1}, p_{2}\} = \{w, q_{H2}\}$}
 
Without loss of generality, let $p_{1} = u$ and $p_{2} = q_{H2}$.

Let $S_{1} = \{u, v, q_{H2}\}$.

(a) Suppose firstly there exists some path $P_{A}$ in $\langle X\rangle$ from $P_{1}\cup P_{2}$ to $H - S_{1}$, such that the removal of $S_{1}$ does not separate $a$ from $H$. Then either the graph is equivalent to one of those with the placement of $p_{1}$ and $p_{2}$ mentioned above, where a $W_{7}$-subdivision is formed, or $P_{A}$ meets $H$ only at $w$. Suppose the latter holds. Let $p_{a}$ be the vertex along $P_{A}$ closest to $w$ where $P_{A}$ meets $P_{1}\cup P_{2}$. There are three possibilites: $p_{a}$ is on $P_{1} - a$, $p_{a}$ is on $P_{2} - a$, or $p_{a} = a$.

Let $W = \{u, v, w, q_{H2}\}$. Assume that the removal of $W$ separates $a$ from $H - W$, since otherwise a $W_{7}$-subdivision exists in $G$.

Let $A$ be the bridge of $G|W$ containing $a$. Let $B' = uP_{H}q_{H2} \cup Q_{H1}$.

\textbf{1.} Suppose that $S_{1}$ is not a separating set of $G$, that is, there exists some path $P_{B'}$ disjoint from $S_{1}$ joining $B'$ to $(H\cup \langle A\rangle) - B'$. Such a path either results in a $W_{7}$-subdivision, or meets $(H\cup \langle A\rangle) - B'$ only at $w$. Suppose the latter holds. Let $p_{b}$ be the vertex along $P_{B'}$ closest to $w$ where $P_{B'}$ meets $B'$. There are four possibilites: $p_{b}$ is on $Q_{H1} - q_{H1}$, $p_{b}$ is on $uP_{H}q_{H1} - q_{H1}$, $p_{b}$ is on $q_{H1}P_{H}q_{H2} - q_{H1}$, or $p_{b} = q_{H1}$. For each of the four placements of $p_{b}$, there are three possible placements of $p_{a}$, so $G$ contains one of twelve possible structures.

Let $C' = wP_{H}q_{H2} \cup Q_{H3}$.

\textbf{1.1.} Suppose $S_{2} = \{w, v, q_{H2}\}$ is not a separating set, but rather there exists some path $P_{C'}$ disjoint from $S_{2}$ joining $C'$ to $(H\cup \langle A\rangle) - C'$. Such a path either results in a $W_{7}$-subdivision, or meets $(H\cup \langle A\rangle) - C'$ only at $u$. Suppose the latter holds. Let $p_{c}$ be the vertex along $P_{C'}$ closest to $u$ where $P_{C'}$ meets $C'$. There are four possibilites: $p_{c}$ is on $Q_{H3} - q_{H3}$, $p_{c}$ is on $q_{H2}P_{H}q_{H3} - q_{H3}$, $p_{c}$ is on $q_{H3}P_{H}w - q_{H3}$, or $p_{c} = q_{H3}$. For each of the four placements of $p_{c}$, there are twelve possible placements of $p_{a}$ and $p_{b}$, so $G$ contains one of forty-eight possible structures.

Let $B = V(B'\cup P_{B'})$, and let $C = V(C'\cup P_{C'})$.

\textbf{1.1.1.} Suppose $A$, $B$, $C$, and $Y$ are not all separate bridges of $G|W$, but rather, some path $Q$ exists that prevents the removal of $W$ from placing each of $A\setminus W$, $B\setminus W$, $C\setminus W$ and $Y\setminus W$ in separate components. The program was used to generate each of the forty-eight possible structures that $G$ contains, and new graphs were generated from each of these by adding such a path $Q$. Each possible placement of $Q$ was then tested for the presence of a $W_{7}$-subdivision. In every case, a $W_{7}$-subdivision was found to exist.

\textbf{1.1.2.} Suppose then that $A$, $B$, $C$, and $Y$ all form separate bridges of $G|W$. 

Suppose there exists some vertex $v_{0} \in \{u, v, w\}$ with degree $\ge 7$ such that some bridge of $G|W$ contains at least three neighbours of $v_{0}$ not in $W$. Then by Lemma \ref{threenbrs}, there exists a $W_{7}$-subdivision centred on that vertex. Suppose then that no such vertex exists in $W$.

Suppose there exists some vertex $v_{0} \in \{u, v, w\}$ with degree $\ge 7$ such that two bridges of $G|W$ each contain two neighbours of $v_{0}$ not in $W$. Then by Lemma \ref{twotwonbrs}, there exists a $W_{7}$-subdivision centred on that vertex. Suppose then that no such vertex exists in $W$.

Suppose then there exists some vertex $v_{0}$ in $\{u, v, w\}$ with degree $\ge 7$ such that some bridge of $G|W$ contains two neighbours of $v_{0}$ not in $W$. If $v_{0} = v$, then by Lemma \ref{twonbrs}, there exists a $W_{7}$-subdivision in $G$. If $v_{0} \in \{u, w\}$, then since $v_{0}$ has only six neighbours in $A\cup B\cup C\cup Y$, there must be some fifth bridge $D$ of $G|W$ that contains $v_{0}$. Since there are only two bridges of $G|\{u,v,w\}$, $D$ must also contain $q_{H2}$, and therefore must contain a path from $v_{0}$ to $q_{H2}$ that meets $W$ only at its endpoints. Thus, Lemma \ref{twonbrs} can again be applied to show there exists a $W_{7}$-subdivision in $G$. Suppose then that no such vertex exists in $W$.

Since each vertex in $\{u, v, w\}$ with degree $\ge 7$ has no more than one neighbour not in $W$ in each bridge of $G|W$, and since $|Y\setminus S| < 7$, Reduction \ref{r6} can be performed on $G$.

\textbf{1.2.} Suppose then that $S_{2}$ is a separating set, that is, no such path $P_{C'}$ exists. Thus, $C'$ now forms a bridge of $G|S_{2}$ and of $G|W$.

Suppose there exists some bridge $U$ of $G|W$ such that $U\setminus W$ contains more than one neighbour of $v$. Then by Lemma \ref{twonbrs}, there exists a $W_{7}$-subdivision centred on $v$. Suppose then that no such bridge exists, that is, each bridge of $G|W$ contains at most one neighbour of $v$ not in $W$.

Suppose there exists some bridge $U$ of $G|W$ such that $U\cap W = W$, and $U\setminus W$ contains at least four neighbours of $u$. Then by Lemma \ref{fournbrs}, there exists a $W_{7}$-subdivision centred on $u$. Suppose then that no such bridge $U$ exists, that is, any bridge containing all vertices in $W$ contains at most three neighbours of $u$ not in $W$.

Suppose there exists some vertex $v_{0} \in \{u, w\}$ with degree $\ge 7$ such that some bridge $U$ of $G|W$ contains at least three neighbours of $v_{0}$ not in $W$. If $v_{0} = w$, then by Lemma \ref{threenbrs}, there exists a $W_{7}$-subdivision centred on $v_{0}$. Suppose then that $v_{0} = u$. If some bridge of $G|W$ other than $U$ contains at least two neighbours of $u$ not in $W$, then by Lemma \ref{threetwonbrs} a $W_{7}$-subdivision can be formed centred on $u$. Suppose then that all bridges of $G|W$ other than $U$ contain at most one neighbour of $u$ not in $W$. Then, since $u$ has degree $\ge 7$, there must either exist some fourth bridge other than $A$, $B$, and $Y$ which contains $u$, or $u$ must be adjacent to at least one of $w$, $q_{H2}$. Thus, Lemma \ref{threenbrs} can be applied to show that a $W_{7}$-subdivision exists centred on $u$. Suppose then that no such vertex $v_{0}$ exists in $W$.

Suppose then there exists some vertex $v_{0} \in \{u, w\}$ with degree $\ge 7$ such that two bridges of $G|W$, say $U_{1}$ and $U_{2}$, each contain two neighbours of $v_{0}$ not in $W$. Again, if $v_{0} = w$, then by Lemma \ref{twotwonbrs}, there exists a $W_{7}$-subdivision centred on that vertex. Suppose then that $v_{0} = u$. If some bridge of $G|W$ other than $U_{1}$ and $U_{2}$ contains at least two neighbours of $u$ not in $W$, then by Lemma \ref{twotwotwonbrs} a $W_{7}$-subdivision can be formed centred on $u$. Suppose then that all bridges of $G|W$ other than $U_{1}$ and $U_{2}$ contain at most one neighbour of $u$ not in $W$. Then, since $u$ has degree $\ge 7$, there must either exist some fourth bridge other than $A$, $B$, and $Y$ which contains $u$, or $u$ must be adjacent to at least one of $w$, $q_{H2}$. Thus, Lemma \ref{twotwonbrs} can again be applied to show that a $W_{7}$-subdivision exists centred on $u$. Suppose then that no such vertex $v_{0}$ exists in $W$.

Suppose then there exists some vertex $v_{0} \in \{u, w\}$ with degree $\ge 7$ such that some bridge $U$ of $G|W$ contains two neighbours of $v_{0}$ not in $W$. Since $v_{0}$ has either five (if $v_{0} = u$) or six (if $v_{0} = w$) neighbours in $A\cup B\cup C\cup Y$, there must be some fifth bridge $D$ of $G|W$ that contains $v_{0}$, and if $v_{0} = u$, some sixth bridge $E$ of $G|W$ that also contains $v_{0}$. Since there are only two bridges of $G|\{u,v,w\}$, $D$ must also contain $q_{H2}$, and therefore must contain a path from $v_{0}$ to $q_{H2}$ that meets $W$ only at its endpoints. Thus, Lemma \ref{twonbrs} can be applied to show there exists a $W_{7}$-subdivision in $G$. Suppose then that no such vertex exists in $W$.

Thus, since each vertex in $\{u, v, w\}$ with degree $\ge 7$ has no more than one neighbour not in $W$ in each bridge of $G|W$, and since $|Y\setminus S| < 7$, Reduction \ref{r6} can be performed on $G$.

\textbf{2.} Suppose now that $S_{1}$ is a separating set, that is, no such path $P_{B'}$ exists joining $B'$ to $(H\cup \langle A\rangle) - B'$. Call $B$ the bridge of $G|S_{1}$ and of $G|W$ that contains $B'$.

Let $C' = wP_{H}q_{H2}\cup Q_{H3}$. Suppose that $S_{2} = \{w, v, q_{H2}\}$ is not a separating set, but rather there exists some path $P_{C'}$ disjoint from $S_{2}$ joining $C'$ to $(H\cup \langle A\rangle) - C'$. Such a path either results in a $W_{7}$-subdivision, or meets $(H\cup \langle A\rangle) - C'$ only at $u$. If the latter holds, then by symmetry of the graph, the same arguments used in case 1.2 above can be applied to show that $G$ contains a $W_{7}$-subdivision. Assume then that no such path $P_{C'}$ exists. Call $C$ the bridge of $G|S_{2}$ and $G|W$ that contains $C'$.

Suppose there exists some bridge $U$ of $G|W$ such that $U\setminus W$ contains more than one neighbour of $v$. Then by Lemma \ref{twonbrs}, there exists a $W_{7}$-subdivision centred on $v$. Suppose then that no such bridge exists, that is, each bridge of $G|W$ contains at most one neighbour of $v$ not in $W$.

Suppose there exists some vertex $v_{0} \in \{u, w\}$, and some bridge $U$ of $G|W$ such that $U\cap W = W$, and $U\setminus W$ contains at least four neighbours of $v_{0}$. Then by Lemma \ref{fournbrs}, there exists a $W_{7}$-subdivision centred on $v_{0}$. Suppose then that no such bridge $U$ exists, that is, any bridge containing all vertices in $W$ contains at most three neighbours not in $W$ of $u$ and $w$.

Suppose now that $B\setminus S_{1}$ contains at least four neighbours of $u$. Then, since $|B\setminus S_{1}| > 3$ and there is only one neighbour of $v$ in $B\setminus S_{1}$, a type 2 edge-vertex-cutset can be formed from $u$, $q_{H2}$, and the edge incident with $v$ in $\langle B\rangle$. Suppose then that $B\setminus S_{1}$ contains at most three neighbours of $u$. By the same argument, $C\setminus S_{2}$ can contain at most three neighbours of $w$, or a type 2 edge-vertex-cutset is created.

Suppose there exist at least four bridges of $G|S_{1}$, that is, at least two bridges of $G|S_{1}$ exist other than $B$ and the bridge containing $w$. Then by Lemma \ref{threebridges}, a $W_{7}$-subdivision exists in $G$ (since at least three bridges of $G|S_{1}$ exist that do not contain internal vertices on the path $Q_{H2}$). Suppose then there are at most three bridges of $G|S_{1}$.

Suppose there exist at least four bridges of $G|S_{2}$, that is, at least two bridges of $G|S_{2}$ exist other than $C$ and the bridge containing $u$. Then by Lemma \ref{threebridges}, a $W_{7}$-subdivision exists in $G$. Suppose then there are at most three bridges of $G|S_{2}$.

Suppose then there exists some vertex $v_{0}$ in $\{u, w\}$ with degree $\ge 7$ such that some bridge $U$ of $G|W$ contains at least three neighbours of $v_{0}$ not in $W$. If some bridge of $G|W$ other than $U$ contains at least two neighbours of $v_{0}$ not in $W$, then by Lemma \ref{threetwonbrs} a $W_{7}$-subdivision can be formed centred on $v_{0}$. Suppose then that all bridges of $G|W$ other than $U$ contain at most one neighbour of $v_{0}$ not in $W$. Then, since $v_{0}$ has degree $\ge 7$, there must either exist some fourth bridge other than $A$, $B$, and $Y$ which contains $v_{0}$, or $v_{0}$ must be adjacent to some vertex in $W$ other than $v$. Thus, Lemma \ref{threenbrs} can be applied to show that a $W_{7}$-subdivision exists centred on $v_{0}$. Suppose then that no such vertex $v_{0}$ exists in $W$.

Suppose then there exists some vertex $v_{0}$ in $\{u, w\}$ with degree $\ge 7$ such that two bridges of $G|W$, $U_{1}$ and $U_{2}$, each contain two neighbours of $v_{0}$ not in $W$. If some bridge of $G|W$ other than $U_{1}$ and $U_{2}$ contains at least two neighbours of $v_{0}$ not in $W$, then by Lemma \ref{twotwotwonbrs} a $W_{7}$-subdivision can be formed centred on $v_{0}$. Suppose then that all bridges of $G|W$ other than $U_{1}$ and $U_{2}$ contain at most one neighbour of $v_{0}$ not in $W$. Then, since $v_{0}$ has degree $\ge 7$, there must either exist some fourth bridge other than $A$, $B$, and $Y$ which contains $v_{0}$, or $u$ must be adjacent to some vertex in $W$ other than $v$. Thus, Lemma \ref{twotwonbrs} can again be applied to show that a $W_{7}$-subdivision exists centred on $v_{0}$. Suppose then that no such vertex $v_{0}$ exists in $W$.

Suppose then there exists some vertex $v_{0}$ in $\{u, w\}$ with degree $\ge 7$ such that some bridge $U$ of $G|W$ contains two neighbours of $v_{0}$ not in $W$. Since $v_{0}$ has at most five neighbours in $A\cup B\cup C\cup Y$, there must exist fifth and sixth bridges $D$ and $E$ of $G|W$, each of which contain $v_{0}$. Since there are only two bridges of $G|S$, $D$ must also contain $q_{H2}$, and therefore must contain a path from $v_{0}$ to $q_{H2}$ that meets $W$ only at its endpoints. Thus, Lemma \ref{twonbrs} can be applied to show there exists a $W_{7}$-subdivision in $G$. Suppose then that no such vertex exists in $W$.

Thus, since each vertex in $\{u, v, w\}$ with degree $\ge 7$ has no more than one neighbour not in $W$ in each bridge of $G|W$, and since $|Y\setminus S| < 7$, Reduction \ref{r6} can be performed on $G$.

(b) Suppose now that the removal of $S_{1}$ separates $a$ from $H$, that is, no such path $P_{A}$ exists. Let $A$ be the bridge of $G|S_{1}$ containing $a$. Let $B$ be the bridge of $G|S_{1}$ containing $q_{H1}$. Let $C$ be the bridge of $G|S_{1}$ containing $w$ and $q_{H3}$.

\textbf{1.} Suppose there exists some internal vertex $q$ on the path $Q_{H2}$.

By 3-connectivity, there must be some path $Q$ contained in $\langle X\rangle$ that joins $q$ to $H - Q_{H2}$. It is straightforward to check that the existence of such a path will result in a $W_{7}$-subdivision, unless $Q$ first meets $H - Q_{H2}$ at either $u$ or $w$. Suppose then that this is the case. If every path $Q$ from $Q_{H2} - \{v, q_{H2}\}$ to $H - Q_{H2}$ first meets $H - Q_{H2}$ at $u$, then $Q$ is contained in a separate bridge from $A$, $B$, or $C$, and thus Lemma \ref{threebridges} can be applied to show that a $W_{7}$-subdivision exists in $G$. Suppose then that some such path $Q$ first meets $H - Q_{H2}$ at $w$. Thus, $q$ is contained in the bridge $C$.

Suppose there exists some fourth bridge of $G|S_{1}$. Then by Lemma \ref{threebridges}, a $W_{7}$-subdivision exists in $G$. Suppose then that $A$, $B$, and $C$ are the only three bridges of $G|S_{1}$.

Suppose $u$ has at most two neighbours in $C\setminus S_{1}$. Then, unless $|(A\cup B)\setminus \{v, q_{H2}\}| = 3$, a type 2 or 4 edge-vertex-cutset can be formed from $v$, $q_{H2}$, and the edge or edges joining $u$ to $C\setminus S_{1}$ (since $C\setminus S_{1}$ contains at least four vertices: $q$, $q_{H3}$, $w$, and at least one vertex in $Y\setminus S$).

Assume then that $|(A\cup B)\setminus \{v, q_{H2}\}| = 3$. Thus, there must be only one vertex in $A\setminus S_{1}$, and one vertex in $B\setminus S_{1}$. Then, since $|Y\cup A\cup B| \le 9$, and $|V(G)| \ge 19$, there must be at least 10 vertices in $C\setminus (Y \cup \{q_{H2}\})$. Since $|(A\cup B)\setminus S_{1}| = 2$, there must be only two edges joining $q_{H2}$ to $(A\cup B)\setminus S_{1}$. These two edges and the vertices $w$ and $v$ form a type 4 edge-vertex-cutset. 

Suppose then that $u$ has three neighbours in $Y\setminus S$. If $u$ also has at least three neighbours in $X\setminus S$, then by applying Lemma \ref{lemma2} to $X$ and $Y$ and using the edge $vu$, a $W_{7}$-subdivision can  be formed centred on $u$. Assume then that $u$ has only two neighbours in $X\setminus S$, say $u_{1}$ and $u_{2}$. Then a type 4 edge-vertex-cutset is formed from $w$, $v$, $uu_{1}$ and $uu_{2}$. (Since $u$ has three neighbours in $Y\setminus S$, we know that $|Y \setminus \{w, v\}| = 4$).

\textbf{2.} Assume then that no such vertex $q$ exists --- that is, $v$ is adjacent to $q_{H2}$, and $Q_{H2}$ is a single edge. Then by Lemma \ref{threebridges}, a $W_{7}$-subdivision exists in $G$.

\vspace{0.2in}
\noindent \textbf{Case 3: $\{p_{1}, p_{2}\} = \{q_{H1}, q_{H3}\}$}

Without loss of generality, suppose $p_{1} = q_{H1}$ and $p_{2} = q_{H3}$. Let $W = \{q_{H1}, v, q_{H3}\}$.

Suppose that $a$, $q_{H2}$, and $u$ are not each in three separate bridges of $G|W$. Therefore, there must exist some path either from $a$ to $H - W$, or from $wP_{H}q_{H3}\cup uP_{H}q_{1} - q_{1} - q_{3}$ to $Q_{H2}\cup q_{H1}P_{H}q_{H3} - q_{1} - q_{3}$. It is straightforward to check that if such a path exists in $G$, then a $W_{7}$-subdivision also exists in $G$. Suppose then that $W$ forms a separating set, the removal of which places $a$, $q_{H2}$, and $u$ in three separate components.

Let $A$ be the bridge of $G|W$ containing $a$. Let $B$ be the bridge of $G|W$ containing $q_{H2}$. Let $C$ be the bridge of $G|W$ containing $u$ and $w$.

\textbf{1.} Suppose there exists some internal vertex $q$ on one of the paths $Q_{H1}$ or $Q_{H3}$. Without loss of generality, suppose $q$ lies on $Q_{H1}$. By 3-connectivity, there must be some path $Q$ contained in $\langle X\rangle$ that joins $q$ to $H - Q_{H1}$. It is straightforward to check that the existence of such a path will result in a $W_{7}$-subdivision, unless $Q$ first meets $H - Q_{H1}$ at $w$ or at $q_{H3}$. Suppose then that this is the case. Let $q'$ be the point at which $Q$ first meets $H - Q_{H1}$.

\textbf{1.1.} Suppose firstly that $q' = q_{H3}$. Suppose there exists some path $Q'$ joining $(Q\cup Q_{H1})\setminus W$ to $(H\cup \langle A\rangle) - Q_{H1}$. Such a path will result in the existence of a $W_{7}$-subdivision. Suppose then that no such path $Q'$ exists. Then $q$ is contained in some fourth bridge of $U_{q}$ of $G|W$ such that $U_{q} \notin \{A, B, C\}$.

\textbf{1.1.1.} Suppose there also exists some internal vertex $r$ on $Q_{H3}$. By 3-connectivity, there must be some path $R$ contained in $\langle X\rangle$ that joins $r$ to $H - Q_{H3}$. Let $r'$ be the point at which $R$ first meets $H - Q_{H3}$. It is straightforward to check that the existence of such a path will result in a $W_{7}$-subdivision, unless $r' \in \{q_{H1}, u\}$. Suppose then that this is the case.

\textbf{1.1.1.1.} Suppose $r' = q_{H1}$.

Suppose there exists some path $R'$ joining $(R\cup Q_{H3})\setminus W$ to $(H\cup \langle A\rangle \cup \langle U_{q}\rangle) - Q_{H3}$. Such a path will result in the existence of a $W_{7}$-subdivision. Suppose then that no such path $R'$ exists. Then $r$ is contained in some fifth bridge of $U_{r}$ of $G|W$ such that $U_{r} \notin \{A, B, C, U_{q}\}$. Thus, Lemma \ref{threebridges} can be applied to show that a $W_{7}$-subdivision exists in $G$.

\textbf{1.1.1.2.} Suppose $r' = u$.

Suppose that $q_{H1}$ has at least two neighbours in $C\setminus W$. Thus, some neighbour $p'$ of $q_{H1}$ exists in $C\setminus W$ such that $p' \notin N_{H}(q_{H1})$.

Suppose firstly that $p' \in H$. Thus, either $p' \in q_{H3}P_{H}w - q_{H3}$, or $p' \in q_{H1}P_{H}u - q_{H1}$. If the former holds, then a $W_{7}$-subdivision can be found in $G$. Suppose then that the latter holds. Since $p' \notin N_{H}(q_{H1})$, the path $q_{H1}p'P_{H}u$ is a shorter path from $q_{H1}$ to $u$ than the path $q_{H1}P_{H}u$. Thus, $|E(H)|$ is no longer minimal.

Suppose then that $p' \notin H$. Then by 3-connectivity, there must be some path in $\langle C\cap X\rangle$ joining $p'$ to $q_{H3}P_{H}w$. Such a path will create a $W_{7}$-subdivision in $G$.

Suppose then that no such vertex $p'$ exists, that is, $C\setminus W$ contains at most one neighbour, say $p'_{1}$, of $q_{H1}$. Then a type 2 edge-vertex-cutset can be formed from $v$, $q_{H3}$, and the edge $q_{H1}p'_{1}$ (since $C\setminus W$ contains at least four vertices, and the other side of the cutset contains at least the vertices $q_{H1}, q, a, q_{H2}$). 

\textbf{1.1.2.} Suppose then that no such vertex $r$ exists, that is, $Q_{H3}$ is a single edge. Then Lemma \ref{threebridges} can be applied to show that a $W_{7}$-subdivision exists in $G$.

\textbf{1.2.} Suppose now that $q' = w$. Thus, $q \in C$.

\textbf{1.2.1.} Suppose there also exists some internal vertex $r$ on $Q_{H3}$. By 3-connectivity, there must be some path $R$ contained in $X$ that joins $r$ to $H - Q_{H3}$. Let $r'$ be the point at which $R$ first meets $H - Q_{H3}$. It is straightforward to check that the existence of such a path will result in a $W_{7}$-subdivision, unless $r' \in \{q_{H1}, u\}$.

Suppose firstly that $r' = q_{H1}$.

Suppose there exists some path $R'$ joining $(R\cup Q_{H3})\setminus W$ to $(H\cup\langle A\rangle) - Q_{H3}$. Such a path will result in the existence of a $W_{7}$-subdivision. Suppose then that no such path exists. Then $r$ is contained in some fourth bridge $U_{r}$ of $G|W$ such that $U_{r} \notin \{A, B, C\}$. By symmetry, then, the same arguments applied in case 1.1.1.2 can be applied here to show that $G$ must contain a $W_{7}$-subdivision.

Suppose then that $r' = u$. A $W_{7}$-subdivision can then be found in $G$.

\textbf{1.2.2.} Suppose then that no such vertex $r$ exists, that is, $Q_{H3}$ is a single edge.

Suppose that $q_{H3}$ has at least two neighbours in $C\setminus W$. Thus, some neighbour $p'$ of $q_{H3}$ exists in $C\setminus W$ such that $p' \notin N_{H}(q_{H3})$. By symmetry, the same arguments used in case 1.1.1.2 can be applied to show that $G$ must contain a $W_{7}$-subdivision.

Suppose then that no such vertex $p'$ exists, that is, $C\setminus W$ contains at most one neighbour, say $p'_{2}$, of $q_{H3}$. Then unless $|V(G - (C - q_{H3}))| \le 3$, a type 2 edge-vertex-cutset can be formed from $v$, $q_{H1}$, and the edge $q_{H3}p'_{2}$ (since $C\setminus W$ contains at least four vertices). Suppose then that $|V(G - (C - q_{H3}))| \le 3$. Thus, $|A\setminus W| = 1$ and $|B\setminus W| = 1$, and there are no bridges of $G|W$ other than $A$, $B$, and $C$.

\textbf{1.2.2.1.} Suppose that $q_{H1}$ has degree $\ge 7$.

Since $q_{H1}$ has only two neighbours in $A\cup B$, $q_{H1}$ must have at least five neighbours in $(C\cap X)\setminus W$. Thus, at least three neighbours, say $x_{1}$, $x_{2}$, $x_{3}$, of $q_{H1}$ exist in $(C\cap X)\setminus W$ such that $x_{1}$, $x_{2}$, $x_{3} \notin N_{H}(q_{H1})$.

By the 3-connectivity of $G$, there must be at least two disjoint paths in $\langle C\cap X\rangle$, say $P_{x1}$ and $P_{x2}$, joining $\{x_{1}, x_{2}, x_{3}\}$ to $(H\cap \langle C\rangle)\setminus \{q_{H1}\}$. Let $y_{1}$ and $y_{2}$ be the two vertices in $(H\cap \langle C\rangle)\setminus \{q_{H1}\}$ where $P_{x1}$ and $P_{x2}$ first meet $(H\cap \langle C\rangle)\setminus \{q_{H1}\}$ respectively. Without loss of generality, suppose that $x_{1}$ is an endpoint of $P_{x1}$, and $x_{2}$ is an endpoint of $P_{x2}$.

If there are two vertices $y'_{1}$ and $y'_{2}$ such that $\{y'_{1}, y'_{2}, q_{H1}\}$ separates $\{x_{1}, x_{2}, x_{3}\}$ from $H\cap \langle C\rangle$, then Lemma \ref{lemma2} can be applied to the bridge of $G|Z$ containing $x_{1}$, $x_{2}$, and $x_{3}$, and thus a $W_{7}$-subdivision can be formed centred on $q_{H1}$. Suppose then that $\{x_{1}, x_{2}\} = \{y_{1}, y_{2}\}$. By 3-connectivity, there must exist two paths from $x_{3}$ to $\{y_{1}, y_{2}\}$ such that these paths are disjoint except at $x_{3}$, and meet $H$ only at $y_{1}$ and $y_{2}$. These paths allow a $W_{7}$-subdivision to be formed centred on $q_{H1}$.

Suppose then that there exists some path $P_{x3}$ from $x_{3}$ to $(H\cap \langle C\rangle)\setminus \{q_{H1}\}$, such that $P_{x3}$ is disjoint from $P_{x1}$ and $P_{x2}$. Let $y_{3}$ be the vertex closest to $x_{3}$ along $P_{x3}$ where $P_{x3}$ meets $(H\cap \langle C\rangle)\setminus \{q_{H1}\}$.

Suppose that each of $y_{1}$, $y_{2}$, and $y_{3}$ lie on one of the paths $q_{H3}P_{H}w$, $qQ_{H1}v$, or $qQw$. Then a $W_{7}$-subdivision exists in $G$.

Suppose then that one of $y_{1}$, $y_{2}$, $y_{3}$ --- assume $y_{1}$ without loss of generality --- does not lie on $q_{H3}P_{H}w$, $qQ_{H1}v$, or $qQw$. Thus, $y_{1}$ lies on either $q_{H1}P_{H}u$, or $q_{H1}Q_{H1}q$.

Suppose $y_{1} = x_{1}$, that is, the path $P_{x1}$ is only a single vertex. If $x_{1}$ lies on $q_{H1}P_{H}u$, then the path $q_{H1}x_{1}u$ is a shorter path from $q_{H1}$ to $u$ than the path $q_{H1}P_{H}u$. If $x_{1}$ lies on $q_{H1}Q_{H1}q$, then the path $q_{H1}x_{1}Q_{H1}q$ is a shorter path from $q_{H1}$ to $q$ than $q_{H1}Q_{H1}q$. Thus, if $y_{1} = x_{1}$, $|E(H)|$ is no longer minimal. Assume then that $y_{1} \neq x_{1}$, that is, the path $P_{x1}$ is not trivial.

By 3-connectivity, then, there must be some path in $\langle C\cap X\rangle$ disjoint from $P_{x1}$ that joins $x_{1}$ to $(H\cap \langle C\rangle)\setminus \{q_{H1}\}$. Call this path $Q_{x1}$. Let $z_{1}$ be the point closest to $x_{1}$ along $Q_{x1}$ where $Q_{x1}$ meets $(H\cap \langle C\rangle)\setminus \{q_{H1}\}$. If $z_{1}$ and $y_{1}$ both lie on $q_{H1}P_{H}u$, or if $z_{1}$ and $y_{1}$ both lie on $q_{H1}Q_{H1}q$, then by 3-connectivity there must be some other path in $\langle C\cap X\rangle$ disjoint from $P_{x1}$ that joins $x_{1}$ to $(H\cap \langle C\rangle)\setminus \{q_{H1}\}$. Assume then that $z_{1}$ and $y_{1}$ do not both lie on $q_{H1}P_{H}u$ and do not both lie on $q_{H1}Q_{H1}q$.

Suppose $y_{2}$ and $y_{3}$ each lie on one of the paths $q_{H3}P_{H}w$, $qQ_{H1}v$, or $qQw$. Then it is straightforward to check that a $W_{7}$-subdivision can be formed in $G$, regardless of the position of $y_{1}$ and $z_{1}$.

Suppose then that one of $y_{2}$, $y_{3}$ --- assume $y_{2}$ without loss of generality --- lies on either $q_{H1}P_{H}u$ or $q_{H1}Q_{H1}q$. By the same argument used above for $y_{1}$, assume that $P_{x2}$ is not trivial, that is, $y_{2} \neq x_{2}$. Thus, by 3-connectivity, there must be some path $Q_{x2}$ in $\langle C\cap X\rangle$ from $x_{2}$ to $(H\cap \langle C\rangle)\setminus \{q_{H1}\}$ such that this path is disjoint from $P_{x2}$. Let $z_{2}$ be the point closest to $x_{2}$ along $Q_{x2}$ where $Q_{x2}$ meets $(H\cap \langle C\rangle)\setminus \{q_{H1}\}$. By the same argument used above for $y_{1}$ and $z_{1}$, assume that $y_{2}$ and $z_{2}$ do not both lie on $q_{H1}P_{H}u$ and do not both lie on $q_{H1}Q_{H1}q$.

Suppose $\{y_{1}, z_{1}\} = \{y_{2}, z_{2}\}$, that $\{y_{1}, z_{1}, q_{H1}\}$ forms a separating set in $G$, and that $x_{1}$ and $x_{2}$ are in separate bridges of $G|\{y_{1}, z_{1}, q_{H1}\}$. Since $q_{H1}$ has at least three neighbours in a third bridge of $G|\{y_{1}, z_{1}, q_{H1}\}$, by Lemma \ref{threebridges}, a $W_{7}$-subdivision exists in $G$. Assume then that this is not the case.

Suppose $y_{3}$ lies on one of the paths $q_{H3}P_{H}w$, $qQ_{H1}v$, or $qQw$. Then it is straightforward to check that a $W_{7}$-subdivision can be formed in $G$, regardless of the positions of $y_{1}$, $z_{1}$, $y_{2}$, and $z_{2}$.

Suppose then that $y_{3}$ lies on either $q_{H1}P_{H}u$ or $q_{H1}Q_{H1}q$. By the same argument used above for $y_{1}$, assume that $P_{x3}$ is not trivial, that is, $y_{3} \neq x_{3}$. Thus, by 3-connectivity, there must be some path $Q_{x3}$ in $\langle C\cap X\rangle$ from $x_{3}$ to $(H\cap \langle C\rangle)\setminus \{q_{H1}\}$ such that this path is disjoint from $P_{x3}$. Let $z_{3}$ be the point closest to $x_{3}$ along $Q_{x3}$ where $Q_{x3}$ meets $(H\cap \langle C\rangle)\setminus \{q_{H1}\}$. By the same argument used above for $y_{1}$ and $z_{1}$, assume that $y_{3}$ and $z_{3}$ do not both lie on $q_{H1}P_{H}u$ or on $q_{H1}Q_{H1}q$.

By the same argument used above, if $\{y_{3}, z_{3}\} = \{y_{1}, z_{1}\}$, or if $\{y_{3}, z_{3}\} = \{y_{2}, z_{2}\}$, and if $\{y_{3}, z_{3}, q_{H1}\}$ forms a separating set in $G$ such that $x_{3}$ is in a separate bridge of  $G|\{y_{3}, z_{3}, q_{H1}\}$ from either $x_{1}$ or $x_{2}$, then by Lemma \ref{threebridges} a $W_{7}$-subdivision can be formed in $G$. Assume then that this is not the case.

It can be seen then that any possible placement of $y_{1}$, $z_{1}$, $y_{2}$, $z_{2}$, $y_{3}$, and $z_{3}$ results in the existence of a $W_{7}$-subdivision in $G$.

\textbf{1.2.2.2.} Assume then that $q_{H1}$ has degree $< 7$. Reduction \ref{r2}A can be performed on $G$.

\textbf{2.} Assume then that no such vertex $q$ exists --- that is, $v$ is adjacent to both $q_{H1}$ and $q_{H3}$, and both $Q_{H1}$ and $Q_{H3}$ are single edges. Then by Lemma \ref{threebridges}, a $W_{7}$-subdivision exists in $G$.

\vspace{0.2in}
\noindent \textbf{Cases 4 and 5: $\{p_{1}, p_{2}\} = \{u, q_{H3}\}$ or $\{p_{1}, p_{2}\} = \{w, q_{H1}\}$}

Without loss of generality, let $p_{1} = u$ and $p_{2} = q_{H3}$. Let $W = \{u, v, q_{H3}\}$.

Suppose that $a$, $q_{H2}$, and $u$ are not each in three separate bridges of $G|W$. Therefore, there must exist some path either from $a$ to $H - W$, or from $wP_{H}q_{H3} - q_{H3}$ to $Q_{H2} \cup Q_{H1}\cup uP_{H}q_{H3} - q_{H3}$. It is straightforward to check that such a path results in the existence of a $W_{7}$-subdivision in $G$. Suppose then that $W$ forms a separating set, the removal of which places $a$, $q_{H2}$, and $w$ in three separate components.

Let $A$ be the bridge of $G|W$ containing $a$. Let $B$ be the bridge of $G|W$ containing $q_{H1}$ and $q_{H2}$. Let $C$ be the bridge of $G|W$ containing $w$ and $Y$.

If any bridge of $G|W$ contains more than two neighbours of $v$ not in $W$, then by Lemma \ref{lemmaW7}, a $W_{7}$-subdivision can be formed. Suppose then that each bridge of $G|W$ contains at most two neighbours of $v$ not in $W$. Then, if any bridge of $G|W$ contains more than three vertices not in $W$, a type 2 or 4 edge-vertex-cutset can be formed from $u$, $q_{H3}$, and one or two of the edges incident with $v$. Assume then that each bridge of $G|W$ contains at most three vertices not in $W$. Thus, since $|V(G)| \ge 19$, there must be at least six bridges of $G|W$.

If any bridge of $G|W$ contains only one vertex not in $W$, then Reduction \ref{r1}A can be performed on $G$. Assume then that each bridge has at least two but no more than three vertices not in $W$.

Suppose firstly that each bridge of $G|W$ contains only two vertices not in $W$. Then, since $|V(G)| \ge 19$, there must exist at least eight bridges of $G|W$. Thus, there must be at least one bridge of $G|W$ which is contained as a subdivision in two others, and so Reduction \ref{r1}B can be performed on $G$.

Suppose then there exists some bridge $U$ of $G|W$ such that $|U\setminus W| = 3$. There must be at least five edges joining $U\setminus W$ to $W$, to avoid an internal 4-edge-cutset. Suppose that each vertex in $W$ has at most two neighbours in $U\setminus W$. Then either a type 2a or type 3a edge-vertex-cutset exists in $G$. Assume then that at least one vertex in $W$ has three neighbours in $U\setminus W$. We know this vertex is not $v$ (since this will result in a $W_{7}$-subdivision, using Lemma \ref{lemmaW7}) --- assume then without loss of generality that $u$ has three neighbours in $U\setminus W$.

If any bridge of $G|W$ other than $U$ has more than one neighbour of $u$ not in $W$, then by Lemma \ref{lemmaW7}, a $W_{7}$-subdivision exists centred on $u$. Assume then that each bridge of $G|W$ other than $u$ contains at most one neighbour of $u$ not in $W$.

If each bridge of $G|W$ other than $U$ has only two vertices not in $W$, then Reduction \ref{r1}B can be performed on $G$. Assume then that there exists some bridge other than $U$, say $U'$, which contains three vertices not in $W$. By the same argument used above for $U$, we can assume that $U'\setminus W$ contains either three neighbours of $u$ or three neighbours of $q_{H2}$. Since we have already assumed that $U'\setminus W$ contains at most one neighbour of $u$, we can assume now that $U'\setminus W$ contains three neighbours of $q_{H2}$. Thus, if any bridge other than $U'$ contains more than one neighbour of $q_{H2}$ not in $W$, a $W_{7}$-subdivision can be formed centred on $q_{H2}$. Assume then that this is not the case.

Each bridge of $G|W$ other than $U$ and $U'$, then, contains at most one neighbour of $u$ and of $q_{H2}$, and at most two neighbours of $v$. Thus, each such bridge can contain at most two vertices not in $W$. Therefore, since $|V(G)| \ge 19$ and $|U\cup U'| = 9$, there must be at least five bridges of $G|W$ other than $U$ and $U'$, each of which are identical. Thus, Reduction \ref{r1}B can be performed on $G$.
\end{proof}

\section{Main result}
\label{mainresult}

The following theorem is the main result of this paper. It allows graphs with no $W_{7}$-subdivisions to be characterized up to bounded size pieces.          

\begin{thm}
\label{theorem}

Let $G$ be a 3-connected graph with at least 38 vertices. Suppose $G$ has no internal 3 or 4-edge-cutsets, no internal $(1,1,1,1)$-cutsets, no type 1, 1a, 2, 2a, 3, 3a, 4, or 4a edge-vertex-cutsets, and is a graph on which Reductions \ref{r1}A, \ref{r1}B, \ref{r1}C, \ref{r2}A, \ref{r2}B, \ref{r6}, \ref{r7}, \ref{r8}, and \ref{r1_big} cannot be performed, for $k = 7$. 

Then $G$ has a $W_{7}$-subdivision if and only if $G$ contains some vertex $v_{0}$ of degree at least 7.

\end{thm}

\begin{proof}
\label{proof}

Let $G$ be a graph that meets the conditions of the Theorem. The forward implication is trivial.

Suppose then that $G$ has a vertex $v_{0}$ of degree at least 7. By Theorem \ref{w6cor}, $G$ must contain a $W_{6}$-subdivision. Furthermore, we can assume by Theorem \ref{w6cor} that either some $W_{6}$-subdivision is centred on $v_{0}$, or that some $W_{6}$-subdivision in $G$ is centred on some other vertex of degree $\ge 7$ in $G$. If the latter is true, take this new vertex as $v_{0}$. Let $H$ be this $W_{6}$-subdivision. Let $v_{1}, v_{2}, v_{3}, v_{4}, v_{5}, v_{6}$ be the six spoke-meets-rim vertices of $H$, in order around the rim of $H$. For each $i$, $1 \le i \le 6$, let $P_{i}$ be the spoke from $v_{0}$ to $v_{i}$ in $H$. Let $C$ be the rim of $H$.

There are three possibilities.

\textbf{(a)} There is a vertex $u_{1}$ on the rim of $H$ such that $u_{1} \notin \{v_{1}, \ldots , v_{6}\}$, and there is a path from $v_{0}$ to $u_{1}$ that does not meet $H$ except at its endpoints. This path together with $H$ gives a $W_{7}$-subdivision.

\textbf{(b)} $G$ has a vertex $u \in N_{G}(v_{0})\setminus N_{H}(v_{0})$ such that the bridge of $G|V(H)$ containing $u$ also contains two vertices, $u_{1}$ and $u_{2}$, on two separate spokes of $H$. Assume without loss of generality that $u_{1}$ is on $P_{1}$, and $u_{2}$ is on either $P_{2}$, $P_{3}$, or $P_{4}$. It is routine to verify that all instances in this case except for three result in the presence of a $W_{6}$-subdivision. The three specific instances are as follows:

\textbf{(b)(i)} $u_{1} = v_{1}, u_{2} = v_{3}$.

\textbf{(b)(ii)} $u_{1} = v_{1}, u_{2} = v_{4}$.

\textbf{(b)(iii)} $u_{1} = v_{1}$, $u_{2} \in P_{4} \setminus \{v_{0}, v_{4}\}$.

Dealing with these three cases takes up most of the proof, and we return to them shortly.

\textbf{(c)} There is a vertex $u_{1}$ on one of the spokes of $H$, and there is a path from $v_{0}$ to $u_{1}$ that does not meet $H$ except at its end points. This case is dealt with in the same way as in \cite{Farr88}, where it is shown that in order to preserve 3-connectivity, the graph must fall into one of the two previous cases.

We return now to the three subcases in \textbf{(b)}.

For each of these three subcases, let $U(u)$ be the bridge of $G|V(H)$ containing $u$. Recall that this bridge also contains the vertices $u_{1}$ and $u_{2}$. Thus, there exists some path $P'_{u_{1}}$ joining $u_{1}$ to $u$ such that $P'_{u_{1}}$ is contained in $\langle U(u)\rangle$, and some path $P'_{u_{2}}$ joining $u_{2}$ to $u$ such that $P'_{u_{2}}$ is contained in $\langle U(u)\rangle$. Let $u'$ be the vertex closest to $u_{1}$ along $P'_{u_{1}}$ where $P'_{u_{1}}$ and $P'_{u_{2}}$ meet. (Note that it is possible that $u = u'$.)

Denote by $P_{u_{1}}$ the path $u_{1}P'_{u_{1}}u'$. Denote by $P_{u_{2}}$ the path $u_{2}P'_{u_{2}}u'$. Denote by $P_{u'}$ the path $v_{0}u\cup uP'_{u_{2}}u'$. Note that the three paths $P_{u_{1}}$, $P_{u_{2}}$, $P_{u'}$ meet only at $u'$ (see Figure \ref{w7b}).

\begin{figure}[htbp]
\begin{center}
\includegraphics[width=0.5\textwidth]{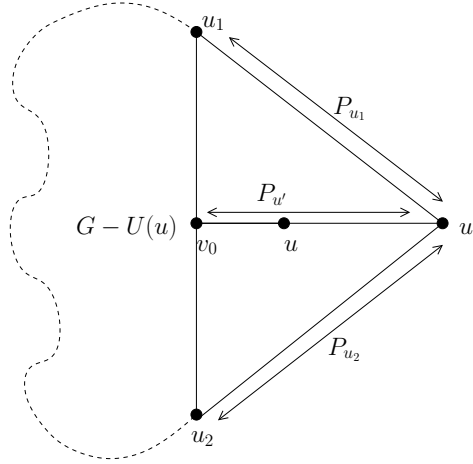}
\caption{Case (b): $P_{u_{1}}$, $P_{u_{2}}$, and $P_{u'}$.}
\label{w7b}
\end{center}
\end{figure}

\vspace{0.2in}
\noindent \textbf{Case (b)(i): $u_{1} = v_{1}, u_{2} = v_{3}$} (Figure \ref{w7case1})

\begin{figure}[htbp]
\begin{center}
\includegraphics[width=0.5\textwidth]{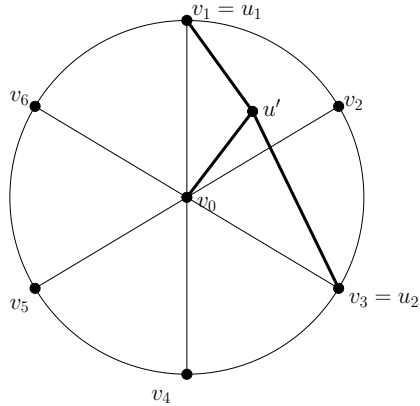}
\caption{Case (b)(i): $u_{1} = v_{1}, u_{2} = v_{3}$}
\label{w7case1}
\end{center}
\end{figure}

Let $W = \{v_{0}, v_{1}, v_{3}\}$. Let $H_{2}$ be the subgraph consisting of the path from $v_{1}$ to $v_{3}$ that passes through $v_{2}$, not including endpoints, and all of $P_{2}$ except for $v_{0}$. Let $H_{4}$ be the subgraph consisting of the path from $v_{1}$ to $v_{3}$ that passes through $v_{4}, v_{5}, v_{6}$, not including endpoints, and all of $P_{4}$, $P_{5}$, and $P_{6}$ except for $v_{0}$.

\textbf{1.} Suppose there exists some path $Q$ from some point in $H_{2}$ to some point in $H_{4}$.

Using the program, all possible configurations of such a path were tested for the presence of a $W_{7}$-subdivision. All but two were found to contain a $W_{7}$-subdivision: the two exceptions are shown in Figure \ref{w7case1_Q1}. Suppose $G$ contains the structure shown in one of these two graphs.

\begin{figure}[htbp]
\begin{center}
\includegraphics[width=0.8\textwidth]{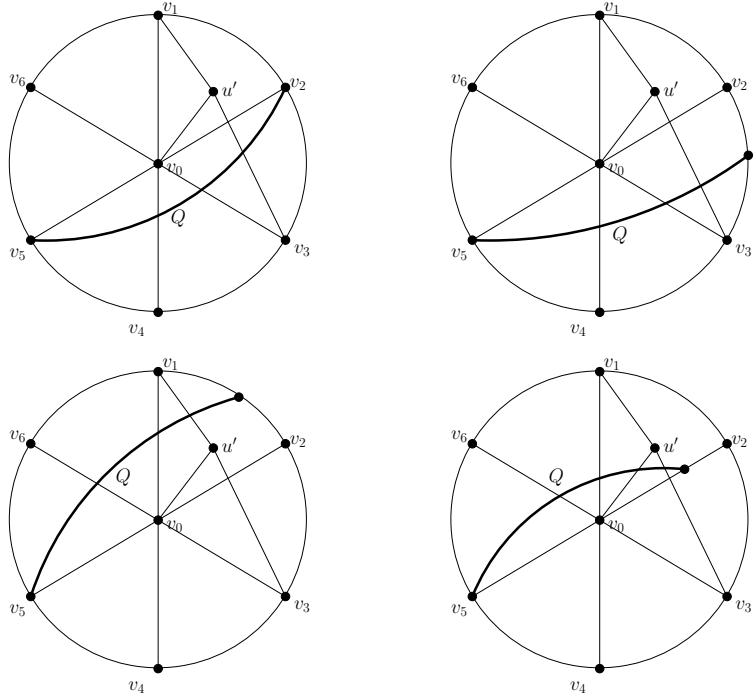}
\caption{Case (b)(i), path $Q$ from $H_{2}$ to $H_{4}$.}
\label{w7case1_Q1}
\end{center}
\end{figure}

\textbf{1.1.} Suppose $W$ is not a separating set of $G$. Then, there exists a path $R$ in $G$ such that $V(H_{2}\cup H_{4})$ is contained in the bridge $U(u)$. Testing all possible configurations of $R$ results in eight graphs that do not contain a $W_{7}$-subdivision, all of which are shown in Figure \ref{w7case1_Q1E1}. Suppose $G$ contains the structure shown in one of these graphs.

\begin{figure}[htbp]
\begin{center}
\includegraphics[width=0.9\textwidth]{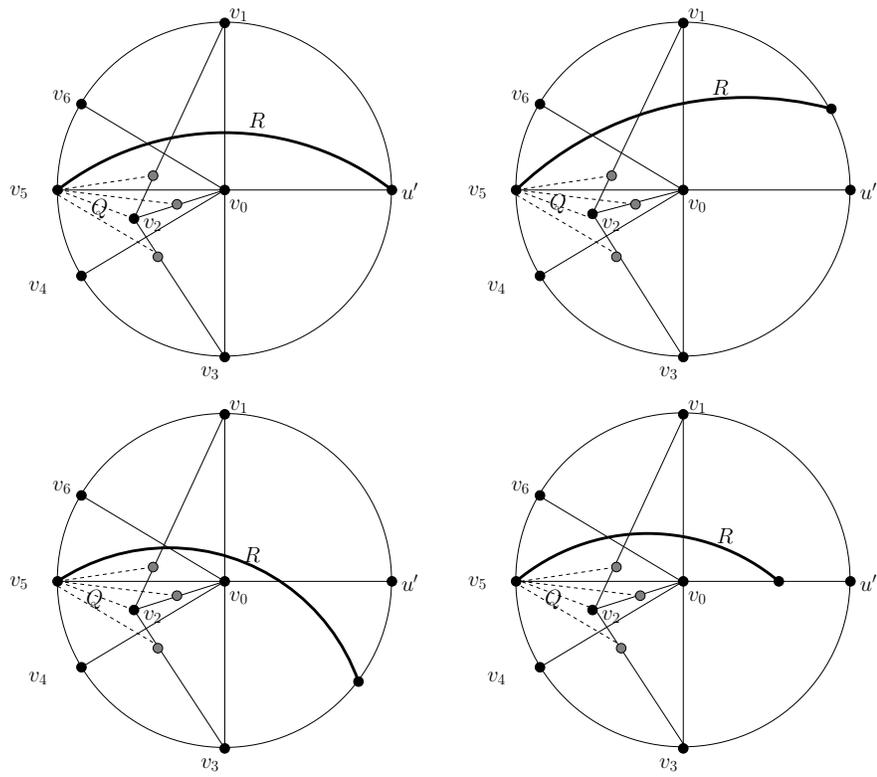}
\caption{Case (b)(i), Path $R$, Case 1.1.}
\label{w7case1_Q1E1}
\end{center}
\end{figure}

\textbf{1.1.1.} Consider the set $S_{1} = \{v_{0}, v_{1}, v_{5}\}$ in each of the graphs of Figure \ref{w7case1_Q1E1}. Suppose this is not a separating set, but rather there exists some path $R_{1}$ that prevents the removal of $S_{1}$ from separating the graph. Figure \ref{w7case1_Q1E1_1} shows the graphs not containing a $W_{7}$-subdivision that can result from such a path.  Suppose $G$ contains the structure shown in one of these graphs.

\begin{figure}[htbp]
\begin{center}
\includegraphics[width=0.9\textwidth]{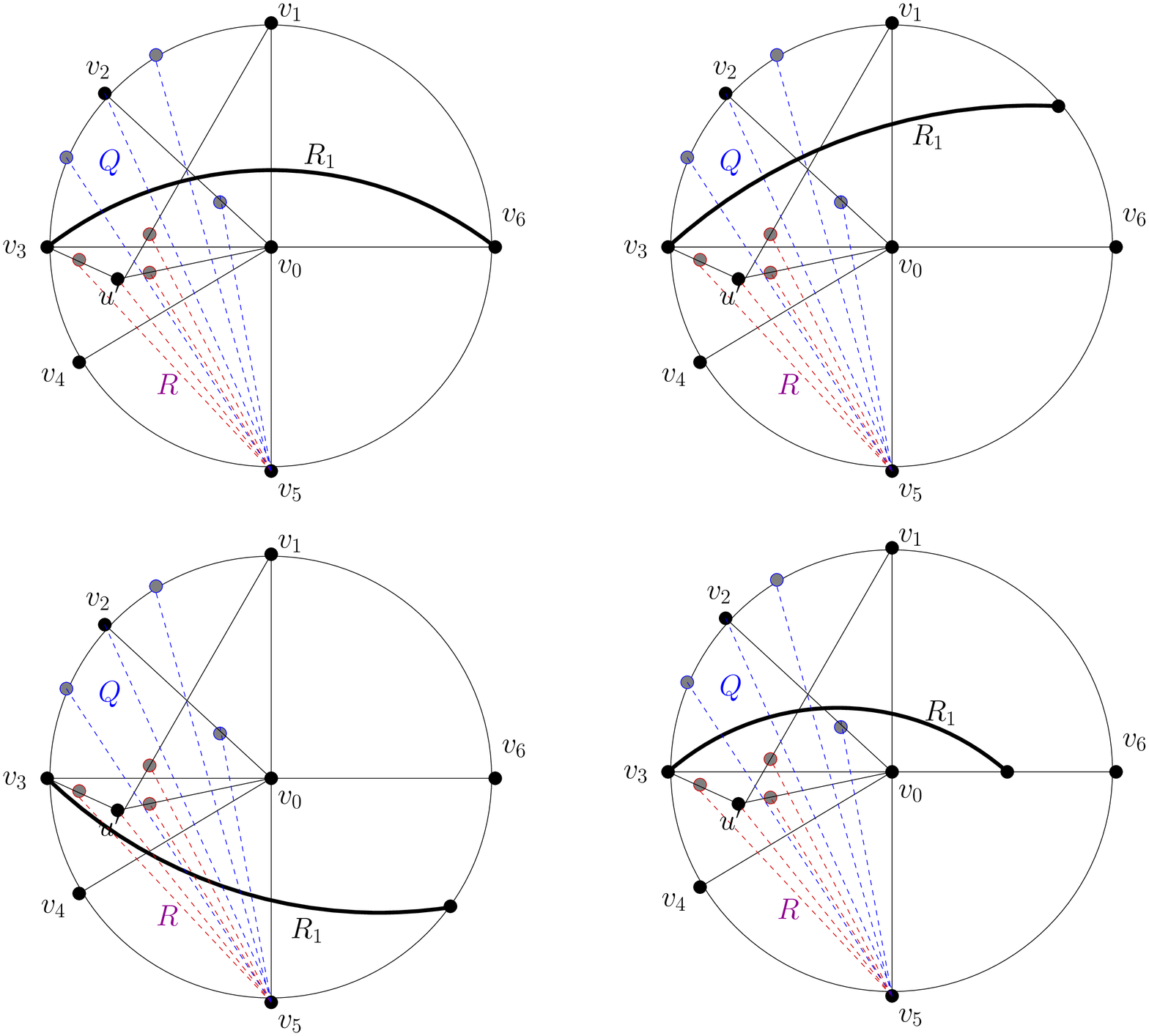}
\caption{Case (b)(i), Path $R_{1}$, Case 1.1.1.}
\label{w7case1_Q1E1_1}
\end{center}
\end{figure}

\textbf{1.1.1.1.} Consider now the set $S_{2} = \{v_{0}, v_{3}, v_{5}\}$ in the graph of Figure \ref{w7case1_Q1E1_1}. Suppose this is not a separating set, but rather there exists some path $R_{2}$ that prevents the removal of $S_{2}$ from separating the graph. Figure \ref{w7case1_Q1E1_2} shows the graphs not containing a $W_{7}$-subdivision that can result from such a path.  Suppose $G$ contains the structure shown in one of these graphs.

\begin{figure}[htbp]
\begin{center}
\includegraphics[width=0.9\textwidth]{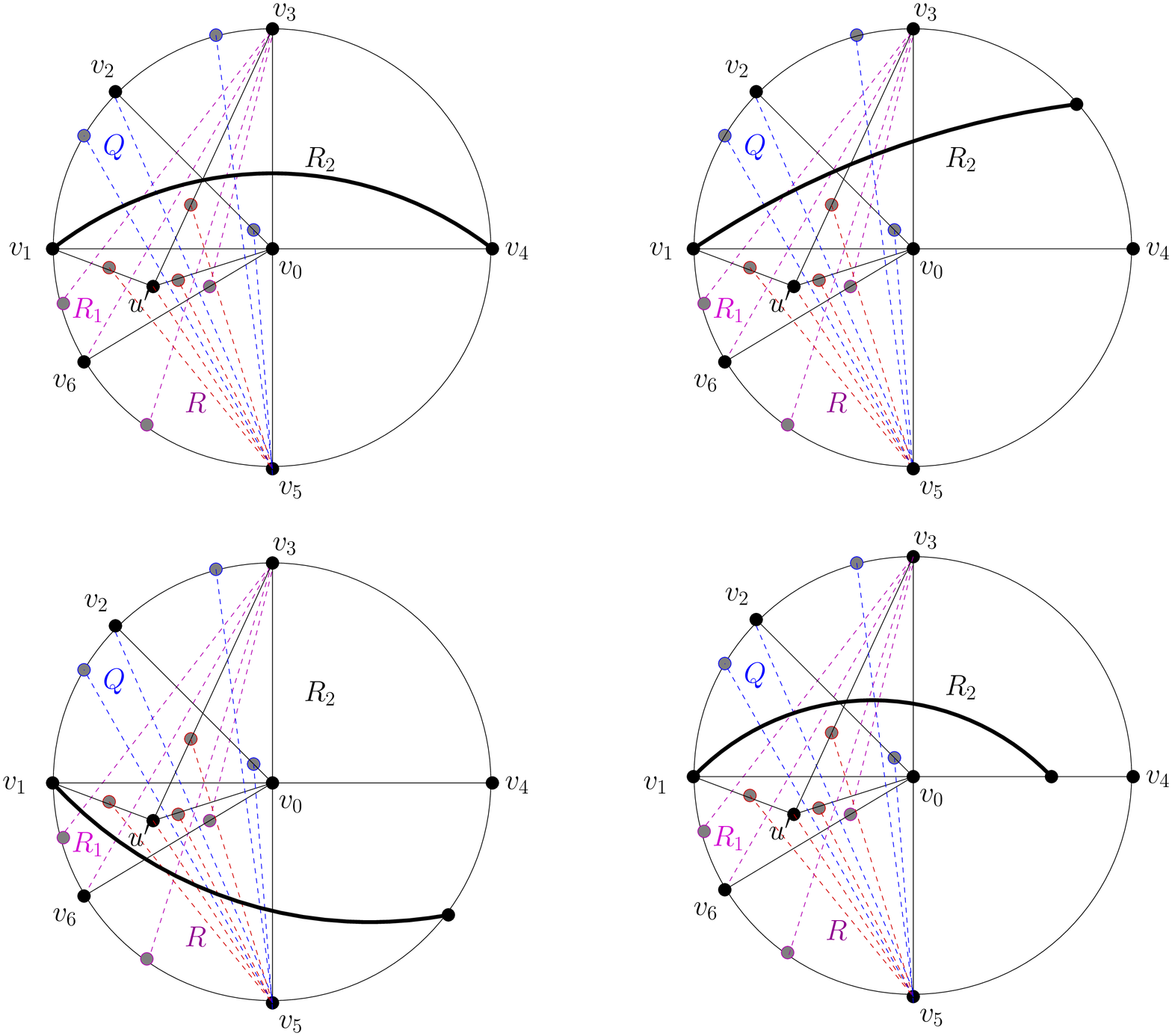}
\caption{Case (b)(i), Path $R_{2}$, Exception 1.1.1.1.}
\label{w7case1_Q1E1_2}
\end{center}
\end{figure}

\textbf{1.1.1.1.1.} Consider now the set $S_{3} = \{v_{0}, v_{1}, v_{3}, v_{5}\}$ in the graph of Figure \ref{w7case1_Q1E1_2}. Suppose this is not a separating set, but rather there exists some path $R_{3}$ that prevents the removal of $S_{3}$ from separating the graph. Searching and checking using the program shows that such a path will always result in the creation of a $W_{7}$-subdivision in $G$. Figure \ref{w7case1_Q1E1_3} shows an example.

\begin{figure}[htbp]
\begin{center}
\includegraphics[width=0.8\textwidth]{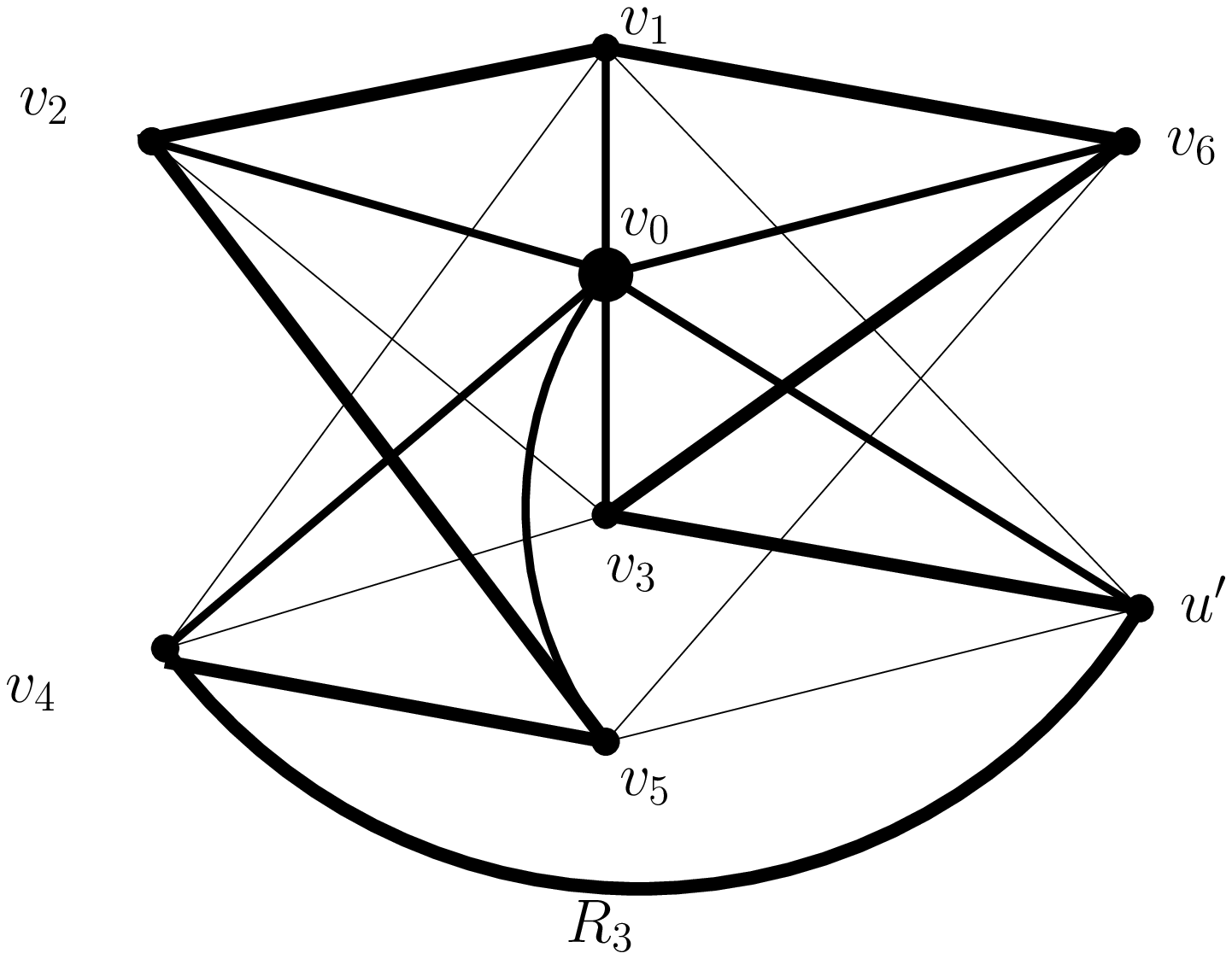}
\caption{Case (b)(i), $W_{7}$-subdivision created by path $R_{3}$}
\label{w7case1_Q1E1_3}
\end{center}
\end{figure}

\textbf{1.1.1.1.2.} Suppose then there is no such path $R_{3}$. Thus, $S_{3}$ forms a separating set of size 4 in $G$. Let $U_{2}$ be the bridge of $G|S_{3}$ containing $v_{2}$; $U_{4}$ be the bridge of $G|S_{3}$ containing $v_{4}$; $U_{6}$ be the bridge of $G|S_{3}$ containing $v_{6}$; and $U_{u'}$ be the bridge of $G|S_{3}$ containing $u'$. (See Figure \ref{w7case1_11112}.)

\begin{figure}[htbp]
\begin{center}
\includegraphics[width=0.8\textwidth]{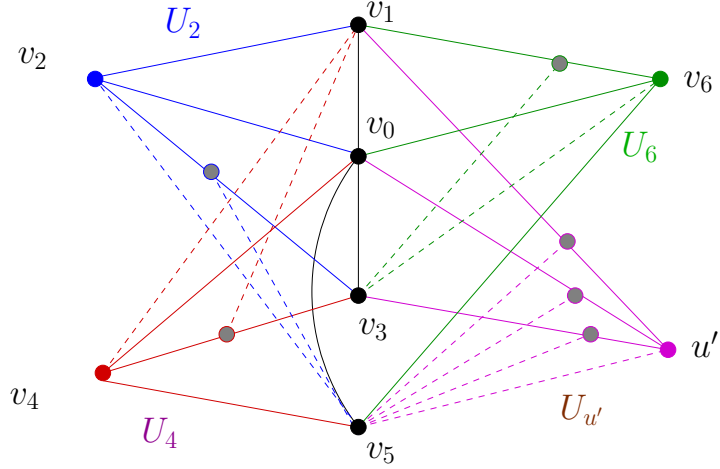}
\caption{Case (b)(i), 1.1.1.1.2.}
\label{w7case1_11112}
\end{center}
\end{figure}

Let $S_{4} = \{v_{1}, v_{3}, v_{5}\}$.

(A) Suppose $S_{4}$ forms a separating set, such that $G|S_{4}$ has at least four bridges, say, $T_{1}$, $T_{2}$, $T_{3}$, and $T_{4}$.

Without loss of generality, suppose that $T_{1}$ contains the bridges $U_{2}$, $U_{4}$, $U_{6}$, and $U_{u'}$ (since while these are separate bridges of $G|S_{3}$, they must all be contained in the one bridge of $G|S_{4}$). Thus, each of $T_{2}$, $T_{3}$, and $T_{4}$ forms a bridge of $G|S_{3}$ as well as $G|S_{4}$. Therefore, there exist three disjoint paths, $P_{t1}$, $P_{t2}$, and $P_{t3}$, from $v_{1}$ to $v_{3}$, $v_{1}$ to $v_{5}$, and $v_{3}$ to $v_{5}$ respectively, such that each of these paths are in a separate bridge of $G|S_{3}$, and none of these paths are in the bridges $U_{2}$, $U_{4}$, $U_{6}$, or $U_{u'}$.

Suppose there exists some vertex $i\in S_{3}$ with degree $\ge 7$ such that some bridge of $G|S_{3}$ contains at least two neighbours of $i$ not in $S_{3}$. By Lemma \ref{twonbrs}, then, there exists a $W_{7}$-subdivision centred on $i$. Table \ref{t_w7i11112a} shows how Lemma \ref{twonbrs} can be applied.

\begin{table}[!h]
\begin{tabular}{p{0.35\textwidth}|p{0.3\textwidth}|p{0.3\textwidth}}
\hline
\textbf{Required in Lemma \ref{twonbrs}} & $S = \{t, u, v, w\}$ & Bridges $X$, $Y$, $Z$ of $G|S$ \\
\hline
\textbf{Equivalent construct in $G$} & $S_{3} = \{v_{0}, v_{1}, v_{3}, v_{5}\}$ & Any three of the bridges $U_{2}$, $U_{4}$, $U_{6}$, $U_{u'}$ \\
\hline \hline
\textbf{Required in Lemma \ref{twonbrs}} & $v$ & $P_{u}$, $P_{w}$, $P_{t}$ \\
\hline
\textbf{Equivalent construct in $G$} & $i$ & The three paths in $\{P_{1}, P_{3}, P_{5}, P_{t1}, P_{t2}, P_{t3}\}$ with $i$ as an endpoint \\
\hline

\end{tabular}
\caption{Case (b)(i), 1.1.1.1.2.(A): Applying Lemma \ref{twonbrs} to $G$.}
\label{t_w7i11112a}
\end{table}

Assume then that no such vertex $i$ exists in $S_{3}$. Then, since each vertex in $S_{3}$ is known to have degree $\ge 7$, because of bridges $T_{2}$, $T_{3}$, $T_{4}$, as well as the structures shown in Figure \ref{w7case1_Q1E1_2}, each bridge of $G|S_{3}$ must contain at most one neighbour not in $S_{3}$ of each vertex in $S_{3}$. Therefore, for each bridge $U$ of $G|S_{3}$, there are at most four edges joining $S_{3}$ to $U\setminus S_{3}$. Thus, each of the bridges of $G|S_{3}$ that contains all of $S_{3}$ can contain at most four vertices not in $S_{3}$, otherwise an internal $(1,1,1,1)$-cutset exists. Reduction \ref{r7} can therefore be performed on $G$.

(B) Assume then that if $\{v_{1}, v_{3}, v_{5}\}$ forms a separating set, its removal separates $G$ into at most three components.

Suppose there exists some vertex $i\in S_{3}$ with degree $\ge 7$ such that some bridge of $G|S_{3}$ contains at least three neighbours of $i$ not in $S_{3}$. By Lemma \ref{threenbrs}, then, there exists a $W_{7}$-subdivision centred on $i$. Table \ref{t_w7i11112b_1} shows how Lemma \ref{threenbrs} can be applied to $G$.

\begin{table}[!h]
\begin{tabular}{p{0.35\textwidth}|p{0.3\textwidth}|p{0.3\textwidth}}
\hline
\textbf{Required in Lemma \ref{threenbrs}} & $S = \{t, u, v, w\}$ & Bridges $X$, $Y$, $Z$ of $G|S$ \\
\hline
\textbf{Equivalent construct in $G$} & $S_{3} = \{v_{0}, v_{1}, v_{3}, v_{5}\}$ & Any three of the bridges $U_{2}$, $U_{4}$, $U_{6}$, $U_{u'}$ \\
\hline \hline
\textbf{Required in Lemma \ref{threenbrs}} & $v$ & $P_{u}$, $P_{w}$ \\
\hline
\textbf{Equivalent construct in $G$} & $i$ & Path or paths in $\{P_{1}, P_{3}, P_{5}\}$ with $i$ as an endpoint, and if required, a second path contained in whichever of the four bridges $U_{2}$, $U_{4}$, $U_{6}$, $U_{u'}$ has not been used to form $X$, $Y$, or $Z$. \\
\hline

\end{tabular}
\caption{Case (b)(i), 1.1.1.1.2.(B): Applying Lemma \ref{threenbrs} to $G$.}
\label{t_w7i11112b_1}
\end{table}

Assume then that no such vertex $i$ exists in $S_{3}$.

Suppose there exists some vertex $i\in S_{3}$ with degree $\ge 7$ such that two bridges of $G|S_{3}$ each contain two neighbours of $i$ not in $S_{3}$. By Lemma \ref{twotwonbrs}, then, there exists a $W_{7}$-subdivision centred on $i$. Assume then that no such vertex $i$ exists in $S_{3}$.

Suppose there exists some vertex $i\in S_{3}$ with degree $\ge 7$ such that some bridge $X$ of $G|S_{3}$ contains two neighbours of $i$ not in $S_{3}$. If $i=v_{0}$, then by Lemma \ref{twonbrs} there exists a $W_{7}$-subdivision in $G$. (Table \ref{t_w7i11112b_2} shows how Lemma \ref{twonbrs} can be applied.) Suppose then that $i\in \{v_{1}, v_{3}, v_{5}\}$. By the assumptions already made in this case (B), $X\setminus S_{3}$ can contain no more than two neighbours of $i$ (or Lemma \ref{threenbrs} would apply), and each bridge of $G|S_{1}$ other than $X$ contains at most one neighbour of $i$ not in $S_{3}$ (or Lemma \ref{twotwonbrs} would apply). Thus, $i$ can have no more than six neighbours in $U_{2}\cup U_{4}\cup U_{6}\cup U_{u'}$. Since $i$ has degree $\ge 7$, there must be some fifth bridge $A$ of $G|S_{3}$ that contains $i$. Thus, Lemma \ref{twonbrs} can again be applied to show that there exists a $W_{7}$-subdivision in $G$.

\begin{table}[!h]
\begin{tabular}{p{0.35\textwidth}|p{0.3\textwidth}|p{0.3\textwidth}}
\hline
\textbf{Required in Lemma \ref{twonbrs}} & $S = \{t, u, v, w\}$ & Bridges $X$, $Y$, $Z$ of $G|S$ \\
\hline
\textbf{Equivalent construct in $G$} & $S_{3} = \{v_{0}, v_{1}, v_{3}, v_{5}\}$ & Any three of the bridges $U_{2}$, $U_{4}$, $U_{6}$, $U_{u'}$ \\
\hline \hline
\textbf{Required in Lemma \ref{twonbrs}} & $v$ & $P_{u}$, $P_{w}$, $P_{t}$ \\
\hline
\textbf{Equivalent construct in $G$} & $v_{0}$ & $P_{1}$, $P_{3}$, $P_{5}$ \\
\hline

\end{tabular}
\caption{Case (b)(i), 1.1.1.1.2.(B): Applying Lemma \ref{twonbrs} to $G$, where $i = v_{0}$.}
\label{t_w7i11112b_2}
\end{table}

\begin{table}[!h]
\begin{tabular}{p{0.35\textwidth}|p{0.3\textwidth}|p{0.3\textwidth}}
\hline
\textbf{Required in Lemma \ref{twonbrs}} & $S = \{t, u, v, w\}$ & Bridges $X$, $Y$, $Z$ of $G|S$ \\
\hline
\textbf{Equivalent construct in $G$} & $S_{3} = \{v_{0}, v_{1}, v_{3}, v_{5}\}$ & Any three of the bridges $U_{2}$, $U_{4}$, $U_{6}$, $U_{u'}$ \\
\hline \hline
\textbf{Required in Lemma \ref{twonbrs}} & $v$ & $P_{u}$, $P_{w}$, $P_{t}$ \\
\hline
\textbf{Equivalent construct in $G$} & $i$ & One path in $\{P_{1}, P_{3}, P_{5}\}$ with $i$ as an endpoint; one path in bridge $A$; and one path in whichever of the four bridges $U_{2}$, $U_{4}$, $U_{6}$, $U_{u'}$ has not been used to form $X$, $Y$, or $Z$. \\
\hline

\end{tabular}
\caption{Case (b)(i), 1.1.1.1.2.(B): Applying Lemma \ref{twonbrs} to $G$, where $i\in \{v_{1}, v_{3}, v_{5}\}$.} \label{t_w7i11112b_3}
\end{table}

Assume then that there is no such vertex $i$ in $S_{3}$. Thus, each bridge has at most one neighbour of any vertex in $S_{3}$ with degree $\ge 7$. The following points then hold:

\begin{itemize}
\item Since we know that $v_{0}$ has degree $\ge 7$, $v_{0}$ must have at most one neighbour not in $S_{3}$ in each bridge of $G|S_{3}$.
\item For each vertex $i\in \{v_{1}, v_{3}, v_{5}\}$, we know that $i$ has degree $\ge 5$. Thus, if there exist at least two bridges that each contain more than one neighbour of $i$, then $i$ has degree $\ge 7$. However, we have already assumed that such a vertex cannot exist. Thus, for each $i$, there can be at most one bridge of $G|S_{3}$ that contains more than one neighbour of $i$ not in $S_{3}$.
\end{itemize}

Let $\mathcal{B}$ be the set of bridges of $G|S_{3}$ such that for each bridge $U \in \mathcal{B}$, there exists some vertex $i\in S_{3}$ which has at least two neighbours in $U\setminus S_{3}$. Given the two points above, we know that $|\mathcal{B}| \le 3$. Thus, there exists some bridge $X \in \{U_{2}, U_{4}, U_{6}, U_{u'}\}$ such that $X \notin \mathcal{B}$. In other words, there are only four edges joining $S_{3}$ to $X\setminus S_{3}$. If $|X\setminus S_{3}| \ge 5$, then, an internal $(1,1,1,1)$-cutset exists in $G$. Assume then that $|X\setminus S_{3}| < 5$.  Then Reduction \ref{r7} can be performed on $G$.

\textbf{1.1.1.2.} Suppose now that no such path $R_{2}$ exists, that is, $S_{2}$ forms a separating set in $G$. Denote by $U_{4}$ the bridge of $G|S_{2}$ containing $v_{4}$. Denote by $U_{2}$ the bridge of $G|S_{2}$ containing $v_{2}$. (See Figure \ref{w7case1_1112}.)

\begin{figure}[htbp]
\begin{center}
\includegraphics[width=0.8\textwidth]{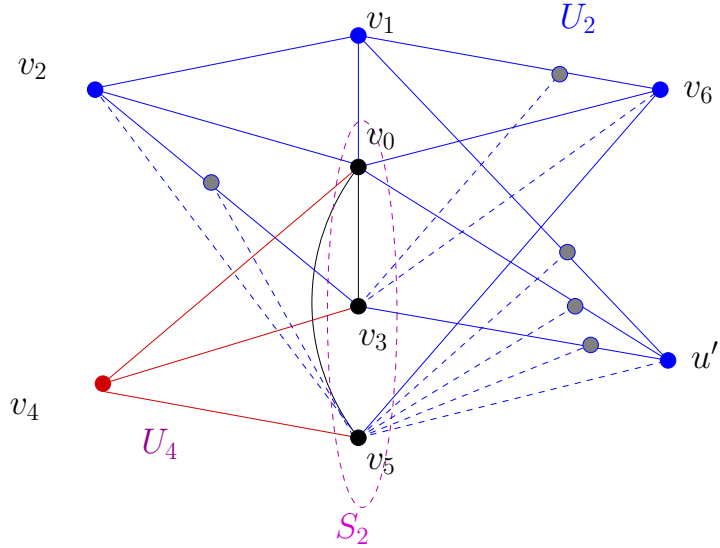}
\caption{Case (b)(i), 1.1.1.2.}
\label{w7case1_1112}
\end{center}
\end{figure}

Suppose there exists some internal vertex on one of the paths $P_{3}$ or $P_{5}$ such that this vertex is contained in either $U_{2}$ or $U_{4}$. It is routine to check that, given 3-connectivity, the existence of such a vertex will result in a $W_{7}$-subdivision, regardless of which other vertices it is adjacent to in its containing bridge. (See Figure \ref{w7case1_1112_1} for an example of such a situation.)

\begin{figure}[htbp]
\begin{center}
\includegraphics[width=0.8\textwidth]{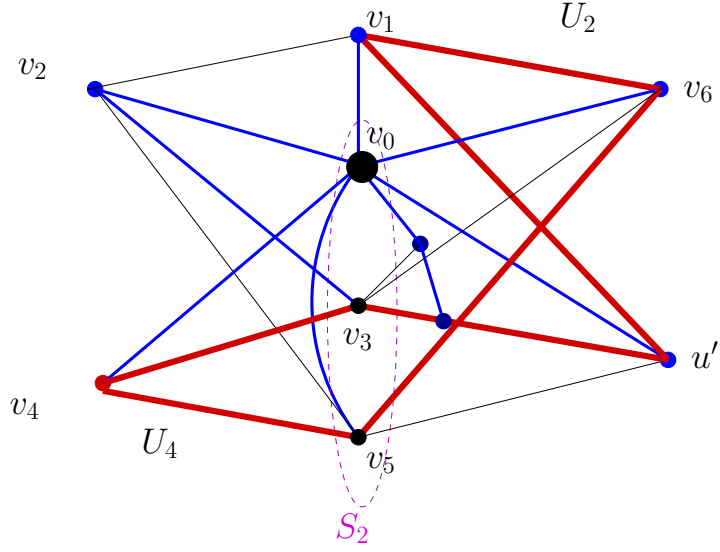}
\caption{Case (b)(i), 1.1.1.2: Example of an internal vertex on $P_{3}$ contained in $U_{2}$ resulting in a $W_{7}$-subdivision.}
\label{w7case1_1112_1}
\end{center}
\end{figure}

Suppose then that no internal vertices of $P_{3}$ or $P_{5}$ are contained in either $U_{2}$ or $U_{4}$.

(A) Suppose there exists some third bridge $A$ of $G|S_{2}$.

Suppose that $A$ contains internal vertices on both of the paths $P_{3}$ and $P_{5}$. Then it is straightforward to check that a $W_{7}$-subdivision exists in $G$. Assume then that $A$ contains internal vertices of at most one of these paths.

(A)(i) Suppose firstly that such vertices are on $P_{3}$, if they exist.

If there are any bridges of $G|S_{2}$ other than $A$, $U_{2}$, or $U_{4}$, then by Lemma \ref{threebridges}, a $W_{7}$-subdivision exists in $G$. Table \ref{t_w7i1112_1} shows how Lemma \ref{threebridges} can be applied.

\begin{table}[!h]
\begin{tabular}{p{0.35\textwidth}|p{0.3\textwidth}|p{0.3\textwidth}}
\hline
\textbf{Required in Lemma \ref{threebridges}} & $S = \{u, v, w\}$ & Bridges $X$, $Y$ of $G|S$ \\
\hline
\textbf{Equivalent construct in $G$} & $S_{2} = \{v_{0}, v_{3}, v_{5}\}$  & $U_{4}$ and $A$ \\
\hline \hline
\textbf{Required in Lemma \ref{threebridges}}  & Bridge $Z$ of $G|S$ containing $\ge$ 3 neighbours of $v$ not in $S$ & $P_{u}$, $P_{w}$ \\
\hline
\textbf{Equivalent construct in $G$}  & $U_{2}$ contains $\ge 3$ neighbours of $v_{0}$ not in $S_{2}$ & $P_{5}$; path from $v_{0}$ to $v_{3}$ in some fourth bridge of $G|S_{2}$ other than $A$, $U_{2}$, or $U_{4}$. \\
\hline

\end{tabular}
\caption{Case (b)(i), 1.1.1.2: Applying Lemma \ref{threebridges} to $G$, where there are at least four bridges of $G|S_{2}$.}
\label{t_w7i1112_1}
\end{table}

Assume then that only three bridges of $G|S_{2}$ exist: $A$, $U_{2}$, and $U_{4}$. Since no internal vertices of $P_{3}$ or $P_{5}$ are contained in either $U_{2}$ or $U_{4}$, and $A$ may contain internal vertices of $P_{3}$ but not $P_{5}$, it can be assumed that $P_{5}$ is a single edge.

If either $A\setminus S_{2}$ or $U_{4}\setminus S_{2}$ contains more than one neighbour of either $v_{0}$ or $v_{5}$, then by Lemma \ref{lemmaW7}, a $W_{7}$-subdivision exists in $G$ (see Table \ref{t_w7i1112_2}). Assume then that $A\setminus S_{2}$ and $U_{4}\setminus S_{2}$ each contain at most one neighbour of $v_{0}$ and at most one neighbour of $v_{5}$.

\begin{table}[!h]
\begin{tabular}{p{0.35\textwidth}|p{0.3\textwidth}|p{0.3\textwidth}}
\hline
\textbf{Required in Lemma \ref{lemmaW7}} & Set $S = \{u, v, w\}$ & Bridge $X$ of $G|S$ such that $X\setminus S$ has $\ge 3$ neighbours of $v$\\
\hline
\textbf{Equivalent construct in $G$} & Set $S_{2} = \{v_{0}, v_{3}, v_{5}\}$ & $U_{2}\setminus S_{2}$ contains $\ge 3$ neighbours of both $v_{0}$ and $v_{5}$\\
\hline \hline
\textbf{Required in Lemma \ref{lemmaW7}}  & Bridge $Y$ of $G|S$ such that $Y\setminus S$ has $\ge 2$ neighbours of $v$  & Paths $P_{w}$ and $P_{u}$ \\
\hline
\textbf{Equivalent construct in $G$}  & Some bridge $Y' \in \{U_{4}, A\}$ contains $\ge 2$ neighbours of either $v_{0}$ or $v_{5}$ & $P_{5}$ forms one of the required paths; the other is contained in the bridge other than $Y'$ in $\{U_{4}, A\}$. \\
\hline

\end{tabular}
\caption{Case (b)(i), 1.1.1.2: Applying Lemma \ref{lemmaW7} to $G$, where $A\setminus S_{2}$ or $U_{4}\setminus S_{2}$ contains more than one neighbour of $v_{0}$ or $v_{5}$.}
\label{t_w7i1112_2}
\end{table}

If either $A\setminus S_{2}$ or $U_{4}\setminus S_{2}$ contains more than two neighbours of $v_{3}$, then by applying Lemma \ref{lemma2} to that bridge and to $U_{2}$, a $W_{7}$-subdivision can be formed. Assume then that each of these bridges contain at most two neighbours of $v_{3}$ not in $S_{2}$. Thus, since there are at most four edges joining $S_{2}$ to $A\setminus S_{2}$, and at most four edges joining $S_{2}$ to $U_{4}\setminus S_{2}$, there can be at most two vertices in each of $A\setminus S_{2}$ and $U_{4}\setminus S_{2}$, otherwise an internal 4-edge-cutset exists in $G$. Therefore, if $A$ contains no internal vertices on the path $P_{3}$, either Reduction \ref{r1}A or \ref{r1}B can be performed (see Figure \ref{w7case1_1112_2} for an example). Assume then that $A$ contains some such vertex.

\begin{figure}[htbp]
\begin{center}
\includegraphics[width=\textwidth]{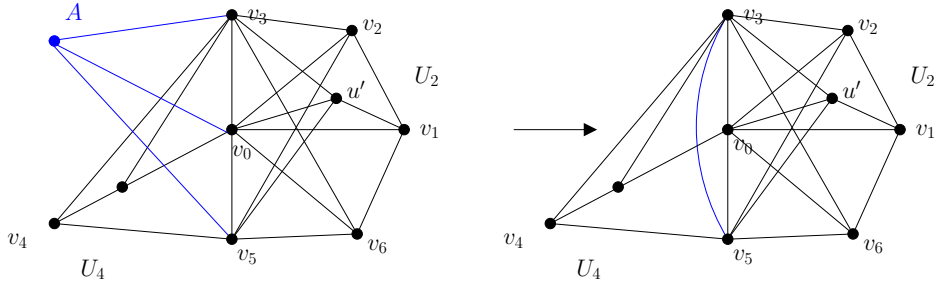}
\caption{Case (b)(i), 1.1.1.2: Example of Reduction \ref{r1}A being performed when $P_{3}$ has no internal vertices.}
\label{w7case1_1112_2}
\end{center}
\end{figure}

If $|(A\cup U_{4})\setminus S_{2}| = 4$, then a type 3 edge-vertex-cutset can be formed from $v_{3}$ and the four edges joining $\{v_{0}, v_{5}\}$ to $(A\cup U_{4})\setminus S_{2}$. Assume then that $|(A\cup U_{4})\setminus S_{2}| \le 3$.

Let $S_{3} = \{v_{0}, v_{1}, v_{3}, v_{5}\}$. Suppose there exists some path in $U_{2}$ such that $u'$, $v_{2}$, and $v_{6}$ are not each in separate bridges of $G|S_{3}$. It is straightforward to check that a $W_{7}$-subdivision then exists in $G$ (see Figure \ref{w7case1_1112_3} for an example). Suppose then that no such path exists, that is, each of $u'$, $v_{2}$, and $v_{6}$ are in separate bridges of $G|S_{3}$. Call these bridges $T_{u'}$, $T_{2}$, and $T_{6}$ respectively. Note that $U_{4}$ and $A$ also form bridges of $G|S_{3}$.

\begin{figure}[htbp]
\begin{center}
\includegraphics[width=0.8\textwidth]{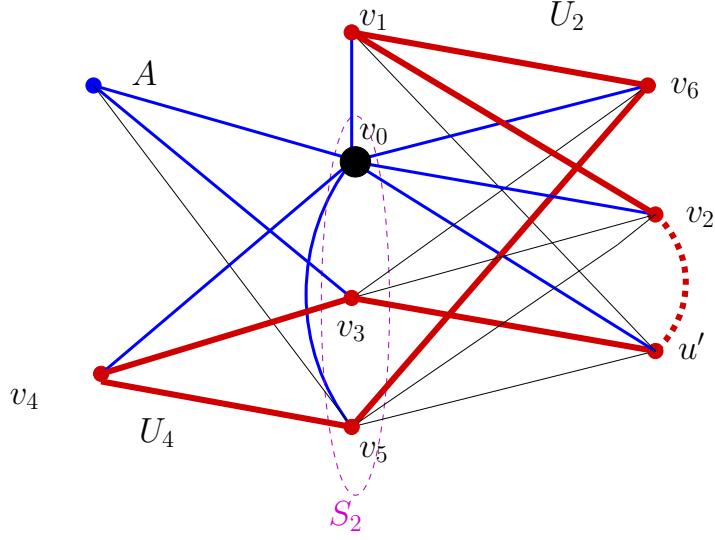}
\caption{Case (b)(i), 1.1.1.2: Example of path disjoint from $S_{3}$ joining $u'$ and $v_{2}$ creating a $W_{7}$-subdivision.}
\label{w7case1_1112_3}
\end{center}
\end{figure}

Let $i\in \{v_{0}, v_{5}\}$. Suppose there exists some bridge $U$ of $G|S_{3}$ such that $U\setminus S_{3}$ contains more than one neighbour of $i$. Then by Lemma \ref{twonbrs}, there exists a $W_{7}$-subdivision centred on $v_{0}$ (see Table \ref{t_w7i1112_3}).

\begin{table}[!h]
\begin{tabular}{p{0.35\textwidth}|p{0.3\textwidth}|p{0.3\textwidth}}
\hline
\textbf{Required in Lemma \ref{twonbrs}} & $S = \{t, u, v, w\}$ & Bridges $X$, $Y$, $Z$ of $G|S$ \\
\hline
\textbf{Equivalent construct in $G$} &  $S_{3} = \{v_{0}, v_{1}, v_{3}, v_{5}\}$ & $T_{u'}$, $T_{2}$, $U_{4}$  \\
\hline \hline
\textbf{Required in Lemma \ref{twonbrs}}  & $v$ & $P_{u}$, $P_{w}$, $P_{t}$ \\
\hline
\textbf{Equivalent construct in $G$}  & $i\in\{v_{0}, v_{5}\}$ & $P_{5}$, path from $i$ to $v_{3}$ in $\langle A\rangle$, path from $i$ to $v_{1}$ in $\langle T_{6}\rangle $. \\
\hline

\end{tabular}
\caption{Case (b)(i), 1.1.1.2: Applying Lemma \ref{twonbrs} to $G$, where some bridge $U\setminus S_{3}$ contains more than one neighbour of $v_{0}$.}
\label{t_w7i1112_3}
\end{table}

Assume then that no such bridge $U$ exists, that is, each bridge of $G|S_{3}$ contains at most one neighbour of $v_{0}$ not in $S_{3}$, and at most one neighbour of $v_{5}$ not in $S_{3}$.

Suppose then there exists some bridge $U$ of $G|S_{3}$ such that $U\cap S_{3} = S_{3}$, and $U\setminus S_{3}$ contains at least four neighbours of some vertex $i$, where $i\in \{v_{1}, v_{3}\}$. Then by Lemma \ref{fournbrs}, there exists a $W_{7}$-subdivision centred on $i$ (see Table \ref{t_w7i1112_4}).

\begin{table}[!h]
\begin{tabular}{p{0.35\textwidth}|p{0.3\textwidth}|p{0.3\textwidth}}
\hline
\textbf{Required in Lemma \ref{fournbrs}} & $S = \{t, u, v, w\}$ & Bridges $X$, $Y$, $Z$ of $G|S$ \\
\hline
\textbf{Equivalent construct in $G$} & $S_{3} = \{v_{0}, v_{1}, v_{3}, v_{5}\}$ & $T_{u'}$, $T_{2}$, $U_{4}$ \\
\hline \hline
\textbf{Required in Lemma \ref{fournbrs}}  & $v$, $u$  & $P_{u}$ \\
\hline
\textbf{Equivalent construct in $G$}  & $i\in \{v_{1}, v_{3}\}$, $v_{0}$  & path from $v_{0}$ to $i$ in $\langle T_{6}\rangle $\\
\hline

\end{tabular}
\caption{Case (b)(i), 1.1.1.2: Applying Lemma \ref{fournbrs} to $G$, where some bridge $U\setminus S_{3}$ contains at least four neighbours of $i\in \{v_{1}, v_{3}\}$.}
\label{t_w7i1112_4}
\end{table}

Assume then that no such bridge $U$ exists, that is, any bridge containing all vertices in $S_{3}$ contains at most three neighbours not in $S_{3}$ of each of $v_{1}$ and $v_{3}$. 

Suppose there exists some vertex $i\in \{v_{1}, v_{3}\}$ with degree $\ge 7$ such that some bridge $U$ of $G|S_{3}$ contains at least three neighbours of $i$ not in $S_{3}$. If there exists some bridge $U'$ of $G|S_{3}$ such that $U'\cap S_{3} = \{v_{1}, v_{3}, v_{5}\}$, then by Lemma \ref{threenbrs} a $W_{7}$-subdivision exists centred on $i$ (see Table \ref{t_w7i1112_6}). Assume then that no such bridge $U'$ exists, that is, $\{v_{1}, v_{3}, v_{5}\}$ is not a separating set of $G$. If some bridge of $G|S_{3}$ other than $U$ contains at least two neighbours of $i$ not in $S_{3}$, then by Lemma \ref{threetwonbrs} a $W_{7}$-subdivision can be formed centred on $i$ (see Table \ref{t_w7i1112_7}). Suppose then that all bridges of $G|S_{3}$ other than $U$ contain at most one neighbour of $i$ not in $S_{3}$. Then, since $i$ has degree $\ge 7$, there must either exist some fifth bridge $X$ of $G|S_{3}$ which contains $i$, or $i$ must be adjacent to some other vertex $j$ in $S_{3}$. Thus, Lemma \ref{threenbrs} can be applied to show that a $W_{7}$-subdivision exists centred on $i$ (see Table \ref{t_w7i1112_8}). Assume then that no such vertex $i$ exists in $\{v_{1}, v_{3}\}$.

\begin{table}[!h]
\begin{tabular}{p{0.35\textwidth}|p{0.3\textwidth}|p{0.3\textwidth}}
\hline
\textbf{Required in Lemma \ref{threenbrs}} & $S = \{t, u, v, w\}$ & Bridges $X$, $Y$, $Z$ of $G|S$ \\
\hline
\textbf{Equivalent construct in $G$} & $S_{3} = \{v_{0}, v_{1}, v_{3}, v_{5}\}$ & $T_{u'}$, $T_{2}$, $U_{4}$ \\
\hline \hline
\textbf{Required in Lemma \ref{threenbrs}}  & $v$ & $P_{u}$, $P_{w}$ \\
\hline
\textbf{Equivalent construct in $G$}  & $i\in \{v_{1}, v_{3}\}$ & path in $\langle T_{6}\rangle$, path in $\langle U'\rangle$\\
\hline

\end{tabular}
\caption{Case (b)(i), 1.1.1.2: Applying Lemma \ref{threenbrs} to $G$, where some bridge $U\setminus S_{3}$ contains at least three neighbours of $i\in \{v_{1}, v_{3}\}$ and some bridge $U'$ exists where $U'\cap S_{3} = \{v_{1}, v_{3}, v_{5}\}$.}
\label{t_w7i1112_6}
\end{table}

\begin{table}[!h]
\begin{tabular}{p{0.35\textwidth}|p{0.3\textwidth}|p{0.3\textwidth}}
\hline
\textbf{Required in Lemma \ref{threetwonbrs}} & $S = \{t, u, v, w\}$ & Bridges $X$, $Y$, $Z$ of $G|S$ \\
\hline
\textbf{Equivalent construct in $G$} & $S_{3} = \{v_{0}, v_{1}, v_{3}, v_{5}\}$ & $T_{u'}$, $T_{2}$, $U_{4}$ \\
\hline \hline
\textbf{Required in Lemma \ref{threetwonbrs}}  & $v$ & $P_{u}$ \\
\hline
\textbf{Equivalent construct in $G$}  & $i\in \{v_{1}, v_{3}\}$ & path in $\langle T_{6}\rangle $\\
\hline

\end{tabular}
\caption{Case (b)(i), 1.1.1.2: Applying Lemma \ref{threetwonbrs} to $G$, where some bridge $U\setminus S_{3}$ contains at least three neighbours of $i\in \{v_{1}, v_{3}\}$, and some other bridge contains at least two neighbours of $i$ not in $S_{3}$.}
\label{t_w7i1112_7}
\end{table}

\begin{table}[!h]
\begin{tabular}{p{0.35\textwidth}|p{0.3\textwidth}|p{0.3\textwidth}}
\hline
\textbf{Required in Lemma \ref{threenbrs}} & $S = \{t, u, v, w\}$ & Bridges $X$, $Y$, $Z$ of $G|S$ \\
\hline
\textbf{Equivalent construct in $G$} & $S_{3} = \{v_{0}, v_{1}, v_{3}, v_{5}\}$ & $T_{u'}$, $T_{2}$, $U_{4}$ \\
\hline \hline
\textbf{Required in Lemma \ref{threenbrs}}  & $v$ & $P_{u}$, $P_{w}$ \\
\hline
\textbf{Equivalent construct in $G$}  & $i\in \{v_{0}, v_{3}\}$ & path in $\langle X\rangle$ or $ij$ edge; path in $\langle T_{6}\rangle $\\
\hline

\end{tabular}
\caption{Case (b)(i), 1.1.1.2: Applying Lemma \ref{threenbrs} to $G$, where some bridge $U\setminus S_{3}$ contains at least three neighbours of $i\in \{v_{1}, v_{3}\}$.}
\label{t_w7i1112_8}
\end{table}

Suppose then there exists some vertex $i\in \{v_{1}, v_{3}\}$ with degree $\ge 7$ such that two bridges of $G|S_{3}$, $U'$ and $U''$, each contain two neighbours of $i$ not in $S_{3}$. If some bridge of $G|S_{3}$ other than $U'$ and $U''$ contains at least two neighbours of $i$ not in $S_{3}$, then by Lemma \ref{twotwotwonbrs} a $W_{7}$-subdivision can be formed centred on $i$ (see Table \ref{t_w7i1112_10}). Suppose then that all bridges of $G|S_{3}$ other than $U'$ and $U''$ contain at most one neighbour of $i$ not in $S_{3}$. Then, since $i$ has degree $\ge 7$, there must either exist some fifth bridge $X$ of $G|S_{3}$ which contains $i$, or $i$ must be adjacent to some other vertex $j$ in $S_{3}$. Thus, Lemma \ref{twotwonbrs} can be applied again to show that a $W_{7}$-subdivision exists centred on $i$ (see Table \ref{t_w7i1112_11}). Assume then that no such vertex $i$ exists in $\{v_{1}, v_{3}\}$.

\begin{table}[!h]
\begin{tabular}{p{0.35\textwidth}|p{0.3\textwidth}|p{0.3\textwidth}}
\hline
\textbf{Required in Lemma \ref{twotwotwonbrs}} & $S = \{t, u, v, w\}$ & Bridges $X$, $Y$, $Z$ of $G|S$ \\
\hline
\textbf{Equivalent construct in $G$} & $S_{3} = \{v_{0}, v_{1}, v_{3}, v_{5}\}$ & $T_{u'}$, $T_{2}$, $U_{4}$ \\
\hline \hline
\textbf{Required in Lemma \ref{twotwotwonbrs}}  & $v$ & $P_{u}$ \\
\hline
\textbf{Equivalent construct in $G$}  & $i\in \{v_{1}, v_{3}\}$ & path in $T_{6}$ \\
\hline

\end{tabular}
\caption{Case (b)(i), 1.1.1.2: Applying Lemma \ref{twotwotwonbrs} to $G$, where three bridges each contain at least two neighbours of $i\in \{v_{1}, v_{3}\}$ not in $S_{3}$.}
\label{t_w7i1112_10}
\end{table}

\begin{table}[!h]
\begin{tabular}{p{0.35\textwidth}|p{0.3\textwidth}|p{0.3\textwidth}}
\hline
\textbf{Required in Lemma \ref{twotwonbrs}} & $S = \{t, u, v, w\}$ & Bridges $X$, $Y$, $Z$ of $G|S$ \\
\hline
\textbf{Equivalent construct in $G$} & $S_{3} = \{v_{0}, v_{1}, v_{3}, v_{5}\}$ & $T_{u'}$, $T_{2}$, $U_{4}$ \\
\hline \hline
\textbf{Required in Lemma \ref{twotwonbrs}}  & $v$ & $P_{u}$, $P_{w}$ \\
\hline
\textbf{Equivalent construct in $G$}  & $i\in \{v_{1}, v_{3}\}$ & path in $\langle X\rangle$ or $ij$ edge; path in $T_{6}$ \\
\hline

\end{tabular}
\caption{Case (b)(i), 1.1.1.2: Applying Lemma \ref{twotwonbrs} to $G$, where two bridges each contain at least two neighbours of $i\in \{v_{1}, v_{3}\}$ not in $S_{3}$.}
\label{t_w7i1112_11}
\end{table}

Suppose then there exists some vertex $i\in \{v_{1}, v_{3}\}$ with degree $\ge 7$ such that some bridge $U$ of $G|S_{3}$ contains two neighbours of $i$ not in $S_{3}$. Since $i$ has at most five neighbours in $T_{u'}\cup T_{2}\cup T_{6}\cup U_{4}$, there must be some fifth bridge $X$ of $G|S_{3}$ that contains $i$, and some sixth bridge $Y$ of $G'|S_{3}$ that also contains $i$. Thus, Lemma \ref{twonbrs} can be applied to show there exists a $W_{7}$-subdivision in $G'$ (see Table \ref{t_w7i1112_12}). Suppose then that no such vertex exists in $S_{3}$.

\begin{table}[!h]
\begin{tabular}{p{0.35\textwidth}|p{0.3\textwidth}|p{0.3\textwidth}}
\hline
\textbf{Required in Lemma \ref{twonbrs}} & $S = \{t, u, v, w\}$ & Bridges $X$, $Y$, $Z$ of $G|S$ \\
\hline
\textbf{Equivalent construct in $G$} & $S_{3} = \{v_{0}, v_{1}, v_{3}, v_{5}\}$ & $T_{u'}$, $T_{2}$, $U_{4}$ \\
\hline \hline
\textbf{Required in Lemma \ref{twonbrs}}  & $v$ & $P_{u}$, $P_{w}$, $P_{t}$ \\
\hline
\textbf{Equivalent construct in $G$}  & $i\in \{v_{1}, v_{3}\}$ & paths in $\langle T_{6}\rangle$, $\langle X\rangle$, and $\langle Y\rangle$ \\
\hline

\end{tabular}
\caption{Case (b)(i), 1.1.1.2: Applying Lemma \ref{twonbrs} to $G$, where some bridge $U\setminus S_{3}$ contains at least two neighbours of $i\in \{v_{1}, v_{3}\}$.}
\label{t_w7i1112_12}
\end{table}

Thus, each vertex in $S_{3}$ in $G$ with degree $\ge 7$ has no more than one neighbour not in $S_{3}$ in each bridge of $G|S_{3}$. Reduction \ref{r8} can then be applied to $G$.

(A)(ii) Suppose now that if $A$ contains internal vertices on one of the paths $P_{3}$ or $P_{5}$, such vertices are on $P_{5}$. The same arguments used in (A)(i) can be applied to show that $G$ contains a $W_{7}$-subdivision.

(B) Suppose then there is no such bridge $A$, that is, there exist only two bridges of $G|S_{2}$: $U_{2}$ and $U_{4}$. Thus, the paths $P_{3}$ and $P_{5}$ are single edges. By Lemma \ref{lemma3}, then, a $W_{7}$-subdivision exists in $G$.

\textbf{1.1.2.} Suppose now that no such path $R_{1}$ exists, that is, $S_{1}$ forms a separating set in $G$. Let $U_{6}$ be the bridge of $G|S_{1}$ containing $v_{6}$. Let $U_{2}$ be the bridge of $G|S_{1}$ containing $v_{2}$. Observe that $v_{0}$ has at least four neighbours in $U_{2}$. (See Figure \ref{w7case1_112}.)

\begin{figure}[htbp]
\begin{center}
\includegraphics[width=0.8\textwidth]{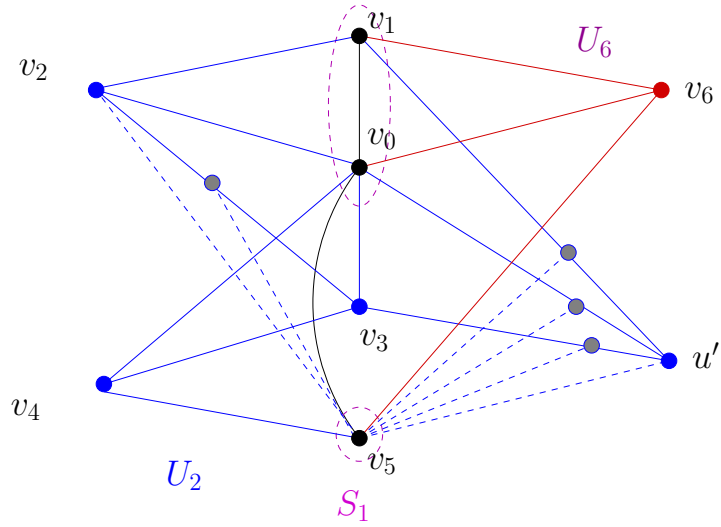}
\caption{Case (b)(i), 1.1.2.}
\label{w7case1_112}
\end{center}
\end{figure}

Suppose there exists some internal vertex on one of the paths $P_{1}$ or $P_{5}$ such that this vertex is contained in either $U_{2}$ or $U_{6}$. It is routine to check that the existence of such a vertex will result in a $W_{7}$-subdivision, regardless of which other vertices it is adjacent to in its containing bridge. (See Figure \ref{w7case1_112_1} for an example of such a situation.)

\begin{figure}[htbp]
\begin{center}
\includegraphics[width=0.8\textwidth]{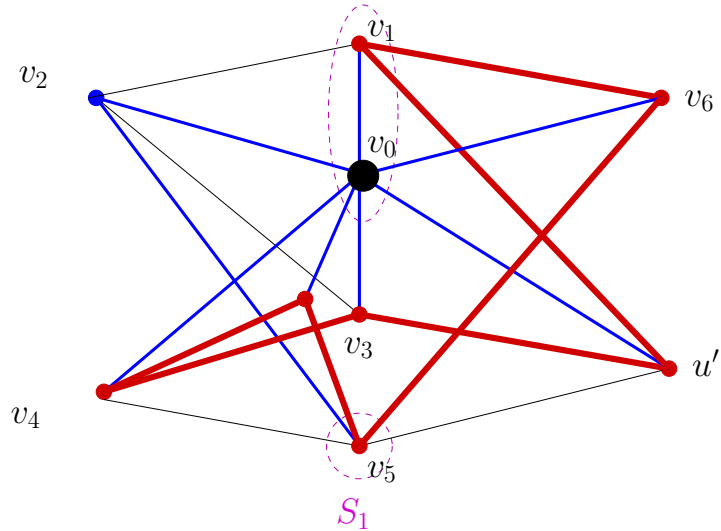}
\caption{Case (b)(i), 1.1.2: Example of an internal vertex on $P_{5}$ contained in $U_{2}$ resulting in a $W_{7}$-subdivision.}
\label{w7case1_112_1}
\end{center}
\end{figure}

Suppose then that no internal vertices of $P_{1}$ or $P_{5}$ are contained in either $U_{2}$ or $U_{6}$.

(A) Suppose there exists some third bridge $A$ of $G|S_{1}$. Suppose that $A$ contains internal vertices on both the paths $P_{1}$ and $P_{5}$. Then it is straightforward to check that a $W_{7}$-subdivision exists in $G$. Assume then that $A$ contains internal vertices of at most one of these paths.

(A)(i) Suppose firstly that such vertices are on $P_{1}$, if they exist.

If there are any bridges of $G|S_{1}$ other than $A$, $U_{2}$, and $U_{6}$, then by Lemma \ref{threebridges}, a $W_{7}$-subdivision exists in $G$. Table \ref{t_w7i112_1} shows how Lemma \ref{threebridges} can be applied in this situation.

\begin{table}[!h]
\begin{tabular}{p{0.35\textwidth}|p{0.3\textwidth}|p{0.3\textwidth}}
\hline
\textbf{Required in Lemma \ref{threebridges}} & $S = \{u, v, w\}$ & Bridges $X$, $Y$ of $G|S$  \\
\hline
\textbf{Equivalent construct in $G$} & $S_{1} = \{v_{0}, v_{1}, v_{5}\}$ & $U_{6}$, $A$ \\
\hline \hline
\textbf{Required in Lemma \ref{threebridges}}  & Bridge $Z$ of $G|S$ containing $\ge 3$ neighbours of $v$ not in $S$ & $P_{u}$, $P_{w}$ \\
\hline
\textbf{Equivalent construct in $G$}  & $U_{2}$ contains $\ge 3$ neighbours of $v_{0}$ not in $S_{1}$ & $P_{5}$; path from $v_{0}$ to $v_{1}$ in some fourth bridge of $G|S$ other than $U_{2}$, $U_{6}$, or $A$ \\
\hline

\end{tabular}
\caption{Case (b)(i), 1.1.2: Applying Lemma \ref{threebridges} to $G$, where there are at least four bridges of $G|S_{1}$.}
\label{t_w7i112_1}
\end{table}

Assume then that only three bridges of $G|S_{1}$ exist: $U_{2}$, $U_{6}$, and $A$. Since none of these bridges contain internal vertices of $P_{5}$, it can be assumed that $P_{5}$ is a single edge.

If either $A\setminus S_{1}$ or $U_{6}\setminus S_{1}$ contains more than one neighbour of either $v_{0}$ or $v_{5}$, then by Lemma \ref{lemmaW7}, a $W_{7}$-subdivision exists in $G$ (see Table \ref{t_w7i112_2}).

\begin{table}[!h]
\begin{tabular}{p{0.35\textwidth}|p{0.3\textwidth}|p{0.3\textwidth}}
\hline
\textbf{Required in Lemma \ref{lemmaW7}} & $S = \{u, v, w\}$ & Bridge $X$ of $G|S$ with $\ge 3$ neighbours of $v$ not in $S$ \\
\hline
\textbf{Equivalent construct in $G$} & $S_{1} = \{v_{0}, v_{1}, v_{5}\}$ & $U_{2}\setminus S_{1}$ contains $\ge 3$ neighbours of $v_{0}$ \\
\hline \hline
\textbf{Required in Lemma \ref{lemmaW7}} & Bridge $Y$ of $G|S$ with $\ge 2$ neighbours of $v$ not in $S$ & $P_{u}$, $P_{w}$ \\
\hline
\textbf{Equivalent construct in $G$}  & Either $A$ or $U_{6}$ & $P_{5}$; path from $v_{0}$ to $v_{1}$ in either $\langle U_{6}\rangle$ or $\langle A\rangle$ \\
\hline

\end{tabular}
\caption{Case (b)(i), 1.1.2: Applying Lemma \ref{lemmaW7} to $G$, where $A\setminus S_{1}$ or $U_{6}\setminus S_{1}$ contains more than one neighbour of either $v_{0}$ or $v_{5}$.}
\label{t_w7i112_2}
\end{table}

Assume then that $A\setminus S_{1}$ and $U_{6}\setminus S_{1}$ each contain at most one neighbour of $v_{0}$ and at most one neighbour of $v_{5}$.

Suppose $|(A\cup U_{6})\setminus S_{1}| > 3$. Then a type 3 edge-vertex-cutset can be formed from vertex $v_{3}$ and the four edges joining $S_{1}$ to $(A\cup U_{6})\setminus S_{1}$. Assume then that $|(A\cup U_{6})\setminus S_{1}| \le 3$.

Let $S_{3} = \{v_{0}, v_{1}, v_{3}, v_{5}\}$. Suppose there exists some path in $U_{2}$ such that $u'$, $v_{2}$, and $v_{4}$ are not each in separate bridges of $G|S_{3}$. It is straightforward to check that a $W_{7}$-subdivision then exists in $G$ (see Figure \ref{w7case1_112_2} for an example). Suppose then that no such path exists, that is, each of $u'$, $v_{2}$, and $v_{4}$ are in separate bridges of $G|S_{3}$. Call these bridges $T_{u'}$, $T_{2}$, and $T_{4}$ respectively. Note that $U_{6}$ and $A$ also form bridges of $G|S_{3}$.

\begin{figure}[htbp]
\begin{center}
\includegraphics[width=0.8\textwidth]{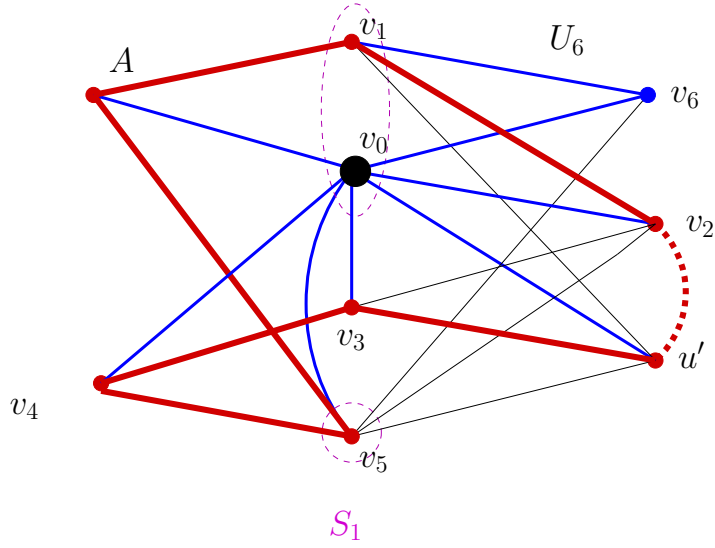}
\caption{Case (b)(i), 1.1.2: Example of a path disjoint from $S_{3}$ joining $u'$ and $v_{2}$ creating a $W_{7}$-subdivision.}
\label{w7case1_112_2}
\end{center}
\end{figure}

Suppose there exists some bridge $U$ of $G|S_{3}$ such that $U\setminus S_{3}$ contains more than one neighbour of either $v_{0}$ or $v_{5}$. Then by Lemma \ref{twonbrs}, there exists a $W_{7}$-subdivision in $G$ (see Table \ref{t_w7i112_3}).

\begin{table}[!h]
\begin{tabular}{p{0.35\textwidth}|p{0.3\textwidth}|p{0.3\textwidth}}
\hline
\textbf{Required in Lemma \ref{twonbrs}} & $S = \{t, u, v, w\}$ & Bridges $X$, $Y$, $Z$ of $G|S$ \\
\hline
\textbf{Equivalent construct in $G$} & $S_{3} = \{v_{0}, v_{1}, v_{3}, v_{5}\}$ & $T_{u'}$, $T_{2}$, $U_{6}$ \\
\hline \hline
\textbf{Required in Lemma \ref{twonbrs}}  & $v$ & $P_{u}$, $P_{w}$, $P_{t}$ \\
\hline
\textbf{Equivalent construct in $G$}  & $i\in \{v_{0}, v_{5}\}$ & $P_{5}$; path from $i$ to $v_{1}$ in $\langle A\rangle$, path from $i$ to $v_{3}$ in $\langle T_{4}\rangle$ \\
\hline

\end{tabular}
\caption{Case (b)(i), 1.1.2: Applying Lemma \ref{twonbrs} to $G$, where $U\setminus S_{3}$ contains more than one neighbour of either $v_{0}$ or $v_{5}$.} 
\label{t_w7i112_3}
\end{table}

Assume then that no such bridge $U$ exists, that is, each bridge of $G|S_{3}$ contains at most one neighbour of $v_{0}$ not in $S_{3}$ and at most one neighbour of $v_{5}$ not in $S_{3}$.

Suppose then there exists some bridge $U$ of $G|S_{3}$ such that $U\cap S_{3} = S_{3}$, and $U\setminus S_{3}$ contains at least four neighbours of some vertex $i$, where $i\in \{v_{1}, v_{3}\}$. Then by Lemma \ref{fournbrs}, there exists a $W_{7}$-subdivision centred on $i$ (see Table \ref{t_w7i112_4}).

\begin{table}[!h]
\begin{tabular}{p{0.35\textwidth}|p{0.3\textwidth}|p{0.3\textwidth}}
\hline
\textbf{Required in Lemma \ref{fournbrs}} & $S = \{t, u, v, w\}$ & Bridges $X$, $Y$, $Z$ of $G|S$ \\
\hline
\textbf{Equivalent construct in $G$} & $S_{3} = \{v_{0}, v_{1}, v_{3}, v_{5}\}$ & $T_{u'}$, $T_{2}$, $U_{6}$ or $T_{4}$ \\
\hline \hline
\textbf{Required in Lemma \ref{fournbrs}}  & $v$, $u$ & $P_{u}$ \\
\hline
\textbf{Equivalent construct in $G$}  & $i\in \{v_{1}, v_{3}\}$ & path in $\langle A\rangle$ (if $i = v_{1}$) or $P_{3}$ (if $i = v_{3}$) \\
\hline

\end{tabular}
\caption{Case (b)(i), 1.1.2: Applying Lemma \ref{fournbrs} to $G$, where $U\setminus S_{3}$ contains more than three neighbours of either $v_{1}$ or $v_{3}$.} 
\label{t_w7i112_4}
\end{table}

Assume then that no such bridge $U$ exists, that is, any bridge containing all vertices in $S_{3}$ contains at most three neighbours not in $S_{3}$ of each of $v_{1}$ and $v_{3}$. 

Suppose there exists some vertex $i\in \{v_{1}, v_{3}\}$ with degree $\ge 7$ such that some bridge $U$ of $G|S_{3}$ contains at least three neighbours of $i$ not in $S_{3}$. (Note that by the assumption of the previous paragraph, $U\setminus S_{3}$ must then contain \emph{only} three neighbours of $i$.) If some bridge of $G|S_{3}$ other than $U$ contains at least two neighbours of $i$ not in $S_{3}$, then by Lemma \ref{threetwonbrs} a $W_{7}$-subdivision can be formed centred on $i$ (see Table \ref{t_w7i112_5}). Suppose then that all bridges of $G|S_{3}$ other than $U$ contain at most one neighbour of $i$ not in $S_{3}$. Then, since $i$ has degree $\ge 7$, either there must exist some fifth bridge $X$ of $G|S_{3}$ which contains $i$, or $i$ must be adjacent to some other vertex $j$ in $S_{3}$. Thus, Lemma \ref{threenbrs} can be applied to show that a $W_{7}$-subdivision exists centred on $i$ (see Table \ref{t_w7i112_6}). Assume then that no such vertex $i$ exists in $\{v_{1}, v_{3}\}$.

\begin{table}[!h]
\begin{tabular}{p{0.35\textwidth}|p{0.3\textwidth}|p{0.3\textwidth}}
\hline
\textbf{Required in Lemma \ref{threetwonbrs}} & $S = \{t, u, v, w\}$ & Bridges $X$, $Y$, $Z$ of $G|S$ \\
\hline
\textbf{Equivalent construct in $G$} & $S_{3} = \{v_{0}, v_{1}, v_{3}, v_{5}\}$ & $T_{u'}$, $T_{2}$, $U_{6}$ or $T_{4}$ \\
\hline \hline
\textbf{Required in Lemma \ref{threetwonbrs}}  & $v$ & $P_{u}$ \\
\hline
\textbf{Equivalent construct in $G$}  & $i\in \{v_{1}, v_{3}\}$ & path in $\langle A\rangle$ (if $i = v_{1}$) or $P_{3}$ (if $i = v_{3}$)\\
\hline

\end{tabular}
\caption{Case (b)(i), 1.1.2: Applying Lemma \ref{threetwonbrs} to $G$, where $U\setminus S_{3}$ contains more than two neighbours of $i\in \{v_{1}, v_{3}\}$, and some other bridge contains more than one neighbour of $i$ not in $S_{3}$.} 
\label{t_w7i112_5}
\end{table}

\begin{table}[!h]
\begin{tabular}{p{0.35\textwidth}|p{0.3\textwidth}|p{0.3\textwidth}}
\hline
\textbf{Required in Lemma \ref{threenbrs}} & $S = \{t, u, v, w\}$ & Bridges $X$, $Y$, $Z$ of $G|S$ \\
\hline
\textbf{Equivalent construct in $G$} & $S_{3} = \{v_{0}, v_{1}, v_{3}, v_{5}\}$ & $T_{u'}$, $T_{2}$, $U_{6}$ or $T_{4}$ \\
\hline \hline
\textbf{Required in Lemma \ref{threenbrs}}  & $v$ & $P_{u}$, $P_{w}$ \\
\hline
\textbf{Equivalent construct in $G$}  & $i\in \{v_{1}, v_{3}\}$ & path in $\langle X\rangle$ or edge $ij$; path in $\langle A\rangle$ (if $i = v_{1}$) or $P_{3}$ (if $i = v_{3}$)\\
\hline

\end{tabular}
\caption{Case (b)(i), 1.1.2: Applying Lemma \ref{threenbrs} to $G$, where $U\setminus S_{3}$ contains more than two neighbours of either $v_{1}$ or $v_{3}$.} 
\label{t_w7i112_6}
\end{table}

Suppose then there exists some vertex $i\in \{v_{1}, v_{3}\}$ with degree $\ge 7$ such that two bridges of $G|S_{3}$, $U'$ and $U''$, each contain two neighbours of $i$ not in $S_{3}$. If some bridge of $G|S_{3}$ other than $U'$ and $U''$ contains at least two neighbours of $i$ not in $S_{3}$, then by Lemma \ref{twotwotwonbrs} a $W_{7}$-subdivision can be formed centred on $i$ (see Table \ref{t_w7i112_7}). Suppose then that all bridges of $G|S_{3}$ other than $U'$ and $U''$ contain at most one neighbour of $i$ not in $S_{3}$. Then, since $i$ has degree $\ge 7$, either there must exist some fifth bridge of $X$ $G|S_{3}$ which contains $i$, or $i$ must be adjacent to some other vertex $j$ in $S_{3}$. Thus, Lemma \ref{twotwonbrs} can be applied again to show that a $W_{7}$-subdivision exists centred on $i$ (see Table \ref{t_w7i112_8}). Assume then that no such vertex $i$ exists in $\{v_{1}, v_{3}\}$.

\begin{table}[!h]
\begin{tabular}{p{0.35\textwidth}|p{0.3\textwidth}|p{0.3\textwidth}}
\hline
\textbf{Required in Lemma \ref{twotwotwonbrs}} & $S = \{t, u, v, w\}$ & Bridges $X$, $Y$, $Z$ of $G|S$ \\
\hline
\textbf{Equivalent construct in $G$} & $S_{3} = \{v_{0}, v_{1}, v_{3}, v_{5}\}$ & $T_{u'}$, $T_{2}$, $U_{6}$ or $T_{4}$ \\
\hline \hline
\textbf{Required in Lemma \ref{twotwotwonbrs}}  & $v$ & $P_{u}$ \\
\hline
\textbf{Equivalent construct in $G$}  & $i\in \{v_{1}, v_{3}\}$ & path in $\langle A\rangle$ (if $i = v_{1}$) or $P_{3}$ (if $i = v_{3}$)\\
\hline

\end{tabular}
\caption{Case (b)(i), 1.1.2: Applying Lemma \ref{twotwotwonbrs} to $G$, where three bridges each contain two neighbours of $i\in\{v_{1}, v_{3}\}$ not in $S_{3}$.}
\label{t_w7i112_7}
\end{table}

\begin{table}[!h]
\begin{tabular}{p{0.35\textwidth}|p{0.3\textwidth}|p{0.3\textwidth}}
\hline
\textbf{Required in Lemma \ref{twotwonbrs}} & $S = \{t, u, v, w\}$ & Bridges $X$, $Y$, $Z$ of $G|S$ \\
\hline
\textbf{Equivalent construct in $G$} & $S_{3} = \{v_{0}, v_{1}, v_{3}, v_{5}\}$ & $T_{u'}$, $T_{2}$, $U_{6}$ or $T_{4}$ \\
\hline \hline
\textbf{Required in Lemma \ref{twotwonbrs}}  & $v$ & $P_{u}$, $P_{w}$ \\
\hline
\textbf{Equivalent construct in $G$}  & $i\in \{v_{1}, v_{3}\}$ & path in $\langle X\rangle$ or edge $ij$; path in $\langle A\rangle$ (if $i = v_{1}$) or $P_{3}$ (if $i = v_{3}$)\\
\hline

\end{tabular}
\caption{Case (b)(i), 1.1.2: Applying Lemma \ref{twotwonbrs} to $G$, where two bridges each contain two neighbours of $i\in\{v_{1}, v_{3}\}$ not in $S_{3}$.}

\label{t_w7i112_8}
\end{table}

Suppose then there exists some vertex $i\in \{v_{1}, v_{3}\}$ with degree $\ge 7$ such that some bridge $U$ of $G|S_{3}$ contains two neighbours of $i$ not in $S_{3}$. Since $i$ has only five neighbours in $T_{u'}\cup T_{2}\cup T_{4}\cup U_{6}$, there must exist two more bridges $X$ and $Y$ of $G|S_{3}$ that each contain $i$. Thus, Lemma \ref{twonbrs} can be applied to show there exists a $W_{7}$-subdivision in $G'$ (see Table \ref{t_w7i112_9}).

\begin{table}[!h]
\begin{tabular}{p{0.35\textwidth}|p{0.3\textwidth}|p{0.3\textwidth}}
\hline
\textbf{Required in Lemma \ref{twonbrs}} & $S = \{t, u, v, w\}$ & Bridges $X$, $Y$, $Z$ of $G|S$ \\
\hline
\textbf{Equivalent construct in $G$} & $S_{3} = \{v_{0}, v_{1}, v_{3}, v_{5}\}$ & $T_{u'}$, $T_{2}$, $U_{6}$ or $T_{4}$ \\
\hline \hline
\textbf{Required in Lemma \ref{twonbrs}}  & $v$ & $P_{u}$, $P_{w}$, $P_{t}$ \\
\hline
\textbf{Equivalent construct in $G$}  & $i\in \{v_{1}, v_{3}\}$ & path in $\langle X\rangle$; path in $\langle Y\rangle$; path in $\langle A\rangle$ (if $i = v_{1}$) or $P_{3}$ (if $i = v_{3}$)\\
\hline

\end{tabular}
\caption{Case (b)(i), 1.1.2: Applying Lemma \ref{twonbrs} to $G$, where $U\setminus S_{3}$ contains more than one neighbour of either $v_{1}$ or $v_{3}$.} 
\label{t_w7i112_9}
\end{table}

Suppose then that no such vertex $i$ exists in $S_{3}$.

Thus, each vertex in $S_{3}$ in $G$ with degree $\ge 7$ has no more than one neighbour not in $S_{3}$ in each bridge of $G|S_{3}$. Reduction \ref{r8} can then be applied to $G$.

(A)(ii) Suppose now that if $A$ contains internal vertices on one of the paths $P_{1}$ or $P_{5}$, such vertices are on $P_{5}$. The same arguments used in (A)(i) can be applied to show that $G$ contains a $W_{7}$-subdivision.

(B) Suppose then there are only two bridges of $G|S_{1}$: $U_{2}$ and $U_{6}$. Thus, the paths $P_{1}$ and $P_{5}$ are single edges. By Lemma \ref{lemma3}, then, a $W_{7}$-subdivision exists in $G$.

\textbf{1.2.} Suppose now that no such path $R$ exists, that is, $W$ forms a separating set in $G$ such that $u'$ and $v_{2}$ are in separate bridges of $G|W$. Let $U_{2}$ be the bridge of $G|W$ containing $v_{2}$. Recall that $U(u)$ is the bridge of $G|W$ containing $u'$. (See Figure \ref{w7case1_112}.)

\begin{figure}[htbp]
\begin{center}
\includegraphics[width=0.8\textwidth]{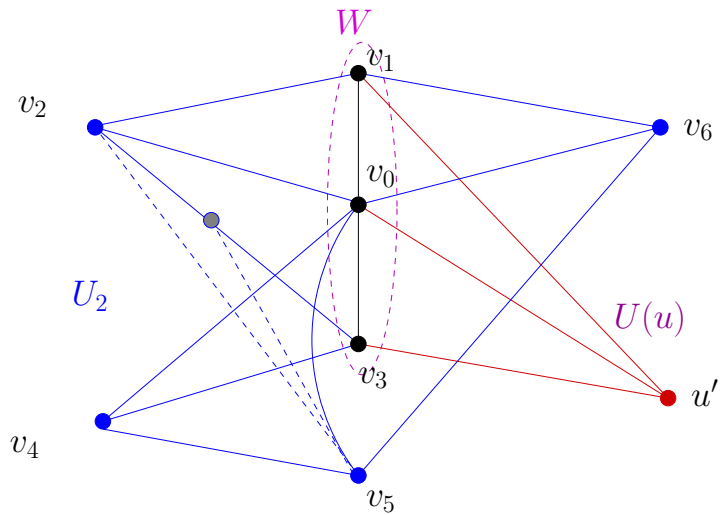}
\caption{Case (b)(i), 1.2.}
\label{w7case1_12}
\end{center}
\end{figure}

Suppose there exists some internal vertex on one of the paths $P_{1}$ or $P_{3}$ such that this vertex is contained in either $U_{2}$ or $U(u)$. The existence of such a vertex will result in a $W_{7}$-subdivision --- Figure \ref{w7case1_12_1} shows an example of such a situation. Assume then that no internal vertices of $P_{1}$ or $P_{3}$ are contained in either $U_{2}$ or $U(u)$.

\begin{figure}[htbp]
\begin{center}
\includegraphics[width=0.8\textwidth]{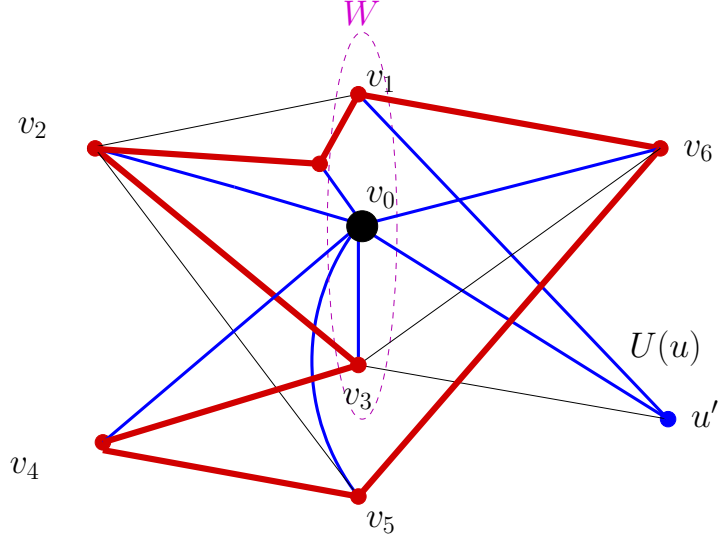}
\caption{Case (b)(i), 1.2: Example of an internal vertex on $P_{1}$ contained in $U_{2}$ resulting in a $W_{7}$-subdivision.}
\label{w7case1_12_1}
\end{center}
\end{figure}

(A) Suppose there exists some third bridge $A$ of $G|W$.

Suppose that $A$ contains internal vertices on both the paths $P_{1}$ and $P_{3}$. Then it is straightforward to check that a $W_{7}$-subdivision exists in $G$. Assume then that $A$ contains internal vertices of at most one of these paths --- without loss of generality, assume that such vertices are on $P_{1}$, if they exist.

If there are any bridges of $G|W$ other than $A$, $U_{2}$, or $U(u)$, then by Lemma \ref{threebridges}, a $W_{7}$-subdivision exists in $G$ (see Table \ref{t_w7i12_1}).

\begin{table}[!h]
\begin{tabular}{p{0.35\textwidth}|p{0.3\textwidth}|p{0.3\textwidth}}
\hline
\textbf{Required in Lemma \ref{threebridges}} & $S = \{u, v, w\}$ & Bridges $X$ and $Y$ of $G|S$ \\
\hline
\textbf{Equivalent construct in $G$} & $W = \{v_{0}, v_{1}, v_{3}\}$ & $U(u)$, $A$ \\
\hline \hline
\textbf{Required in Lemma \ref{threebridges}}  & Bridge $Z$ of $G|S$ containing $\ge 3$ neighbours of $v$ not in $S$  & $P_{u}$, $P_{w}$ \\
\hline
\textbf{Equivalent construct in $G$}  & $U_{2}$ contains $\ge 3$ neighbours of $v_{0}$ & $P_{3}$, path in some fourth bridge of $G|W$ other than $U_{2}$, $U(u)$, or $A$ \\
\hline

\end{tabular}
\caption{Case (b)(i), 1.1.2: Applying Lemma \ref{threebridges} to $G$, where there are at least four bridges of $G|W$.}
\label{t_w7i12_1}
\end{table}

Assume then that only three bridges of $G|W$ exist: $A$, $U_{2}$, and $U(u)$. Since $P_{3}$ contains no internal vertices in any of these three bridges, it can be assumed that $P_{3}$ is a single edge.

If either $A\setminus W$ or $U(u)\setminus W$ contains more than one neighbour of $v_{0}$, then by Lemma \ref{lemmaW7}, a $W_{7}$-subdivision exists in $G$. Assume then that $A\setminus W$ and $U(u)\setminus W$ each contain at most one neighbour of $v_{0}$.

Suppose $|(A\cup U(u))\setminus W| > 3$. Then a type 4 edge-vertex-cutset can be formed from $v_{1}$, $v_{3}$, and the two edges joining $v_{0}$ to $(A\cup U(u))\setminus W$. Assume then that $|(A\cup U(u))\setminus W| \le 3$. Thus, one of $A\setminus W$, $U(u)\setminus W$ contains at most one vertex, while the other contains at most two vertices. Without loss of generality, suppose that $|A\setminus W| = 1$, and $|U(u)\setminus W| \le 2$.

Suppose that $v_{1}$ has degree $\ge 7$. Since $v_{1}$ can have at most three neighbours in $(A\cup U(u))\setminus W$, there must be at least four neighbours of $v_{1}$ in $U_{2}\setminus W$. Thus, there exist two neighbours of $v_{1}$ in $U_{2}\setminus W$, say $x_{1}$ and $x_{2}$, such that $x_{1}, x_{2} \notin N_{H}(v_{1})$. By 3-connectivity, there must exist paths in $\langle U_{2}\setminus W\rangle$ joining $x_{1}$ and $x_{2}$ to $H\cap \langle U_{2}\rangle$. Such paths either result in a $W_{7}$-subdivision, or create a graph that is equivalent to one of those analysed previously in case 1.1.2.

Assume then that $v_{1}$ has degree $< 7$. By the same arguments, assume that $v_{3}$ also has degree $< 7$. Thus, Reduction \ref{r2}A can be performed on $G$.

(B) Suppose then there are only two bridges of $G|W$: $U_{2}$ and $U(u)$. Since neither of these bridges contain internal vertices on either of the paths $P_{1}$ or $P_{3}$, both of the paths $P_{1}$ and $P_{3}$ must be single edges. Thus, by Lemma \ref{lemma3}, a $W_{7}$-subdivision exists in $G$.

\textbf{2.} Assume then that there is no such path $Q$ from $H_{2}$ to $H_{4}$ as tested for in case 1. Thus, there exist at least three bridges of $G|W$: $U_{2}$, $U_{4}$, and $U(u)$, where $U_{2}$ and $U_{4}$ are the bridges containing $H_{2}$ and $H_{4}$ respectively. (See Figure \ref{w7case1_2}.)

\begin{figure}[htbp]
\begin{center}
\includegraphics[width=0.6\textwidth]{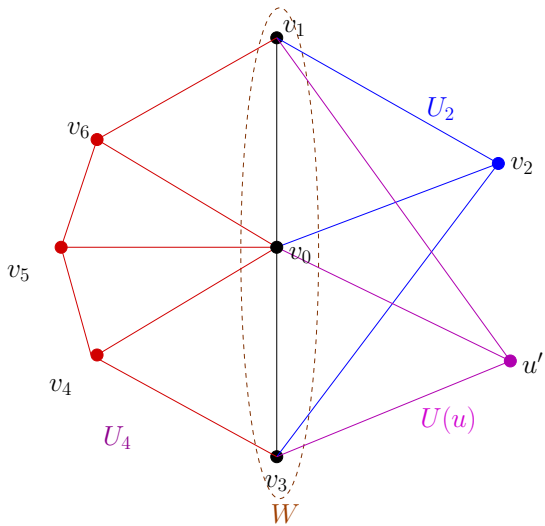}
\caption{Case (b)(i), 2.}
\label{w7case1_2}
\end{center}
\end{figure}

Suppose that $U_{2}$ contains some internal vertex $x$ of $P_{1}$, and some internal vertex $y$ of $P_{3}$. Then, since $U_{2}\setminus W$ contains at least three neighbours of $v_{0}$ (along $P_{1}$, $P_{2}$, and $P_{3}$), a $W_{7}$-subdivision can be formed centred on $v_{0}$ by applying Lemma \ref{lemma2} to $U_{2}$ and $U_{4}$, and by using $U(u)$ to form a seventh spoke from $v_{0}$ to either $v_{1}$ or $v_{3}$.

Assume then that $U_{2}$ does not contain internal vertices of both $P_{1}$ and $P_{3}$. By symmetry of the graph, assume also that $U(u)$ does not contain internal vertices of both $P_{1}$ and $P_{3}$.

Suppose then that $U_{2}$ contains some vertex $x$ such that $x$ is an internal vertex of either $P_{1}$ or $P_{3}$. Then by Lemma \ref{lemmaW7} a $W_{7}$-subdivision exists centred on $v_{0}$. Suppose then that no internal vertices of $P_{1}$ or $P_{3}$ are contained in $U_{2}$. By symmetry of the graph, assume also that no internal vertices of $P_{1}$ or $P_{3}$ are contained in $U(u)$.

If $U_{4}$ contains no internal vertices of $P_{1}$ or $P_{3}$, then by Lemma \ref{threebridges} a $W_{7}$-subdivision exists in $G$. Assume then without loss of generality that $U_{4}$ contains some internal vertex of $P_{1}$, say $x$. By 3-connectivity, there exists some path $P_{x}$ contained in $U_{4}\setminus W$ joining $x$ to $H\cap U_{4}$ that meets $H$ only at its endpoints. It is straightforward to check that the existence of such a path will result in a $W_{7}$-subdivision, except where $P_{x}$ meets $H$ at $v_{5}$, as shown in Figure \ref{w7case1_2_1}. Suppose then that $G$ contains such a configuration.

\begin{figure}[htbp]
\begin{center}
\includegraphics[width=0.6\textwidth]{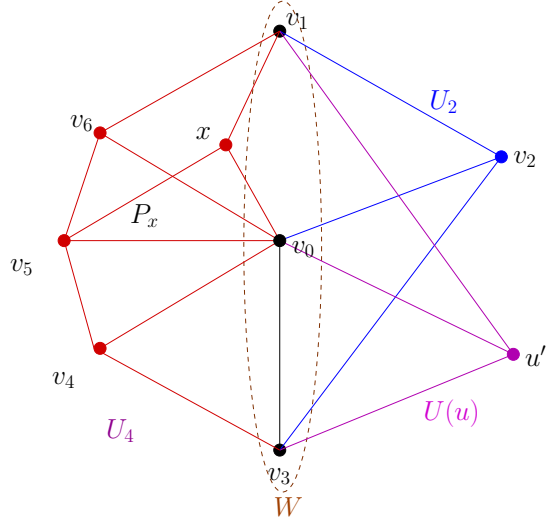}
\caption{Case (b)(i), 2: Internal vertex on $P_{1}$ contained in $U_{4}$}
\label{w7case1_2_1}
\end{center}
\end{figure}

Suppose there exist at least two more bridges of $G|W$ other than $U_{2}$, $U_{4}$, and $U(u)$. Then by Lemma \ref{threebridges}, a $W_{7}$-subdivision exists in $G$. Assume then that there are at most four bridges of $G|W$. Let $U'$ be the fourth bridge of $G|W$, if such a bridge exists.

If any bridge other than $U_{4}$ contains more than one neighbour of $v_{0}$ not in $W$, then by applying Lemma \ref{lemma1} to that bridge, a $W_{7}$-subdivision can be formed centred on $v_{0}$. Assume then that each of $U_{2}\setminus W$, $U(u)\setminus W$, and $U'\setminus W$ contain at most one neighbour of $v_{0}$.

Suppose that $|(U_{2}\cup U(u))\setminus W| > 3$. Then a type 4 edge-vertex-cutset can be formed from $v_{1}$, $v_{3}$, and the two edges joining $v_{0}$ to $(U_{2}\cup U(u))\setminus W$. Assume then that $|(U_{2}\cup U(u))\setminus W| \le 3$. By the same argument, assume that $|(U_{2}\cup U')\setminus W| \le 3$ and $|(U(u)\cup U')\setminus W| \le 3$. Therefore, $|(U_{2}\cup U(u)\cup U')\setminus W| \le 4$. Since $|V(G)| \ge 38$, then, $U_{4}$ must contain at least 34 vertices.

Let $S = \{v_{1}, v_{0}, v_{5}\}$.

\textbf{2.1.} Suppose that $x$, $v_{6}$, and $v_{4}$ are not in three separate bridges of $G|S$, but rather, there exists some path $Q$ disjoint from $S$ joining two of these vertices.

Searching and checking with the program shows that such a path results in a $W_{7}$-subdivision, except in the graphs of Figure \ref{w7case1_2_2}. Suppose that $G$ contains the configuration shown in one of these graphs. Then $G$ falls in to case 1. Figure \ref{w7case1_2_3} shows how a graph isomorphic to the type of graph analysed in case 1 (pictured in Figure \ref{w7case1_Q1}) is contained as a subdivision in $G$. The parts of the graph in bold are those parts also contained in the graph of case 1.

\begin{figure}[htbp]
\begin{center}
\includegraphics[width=\textwidth]{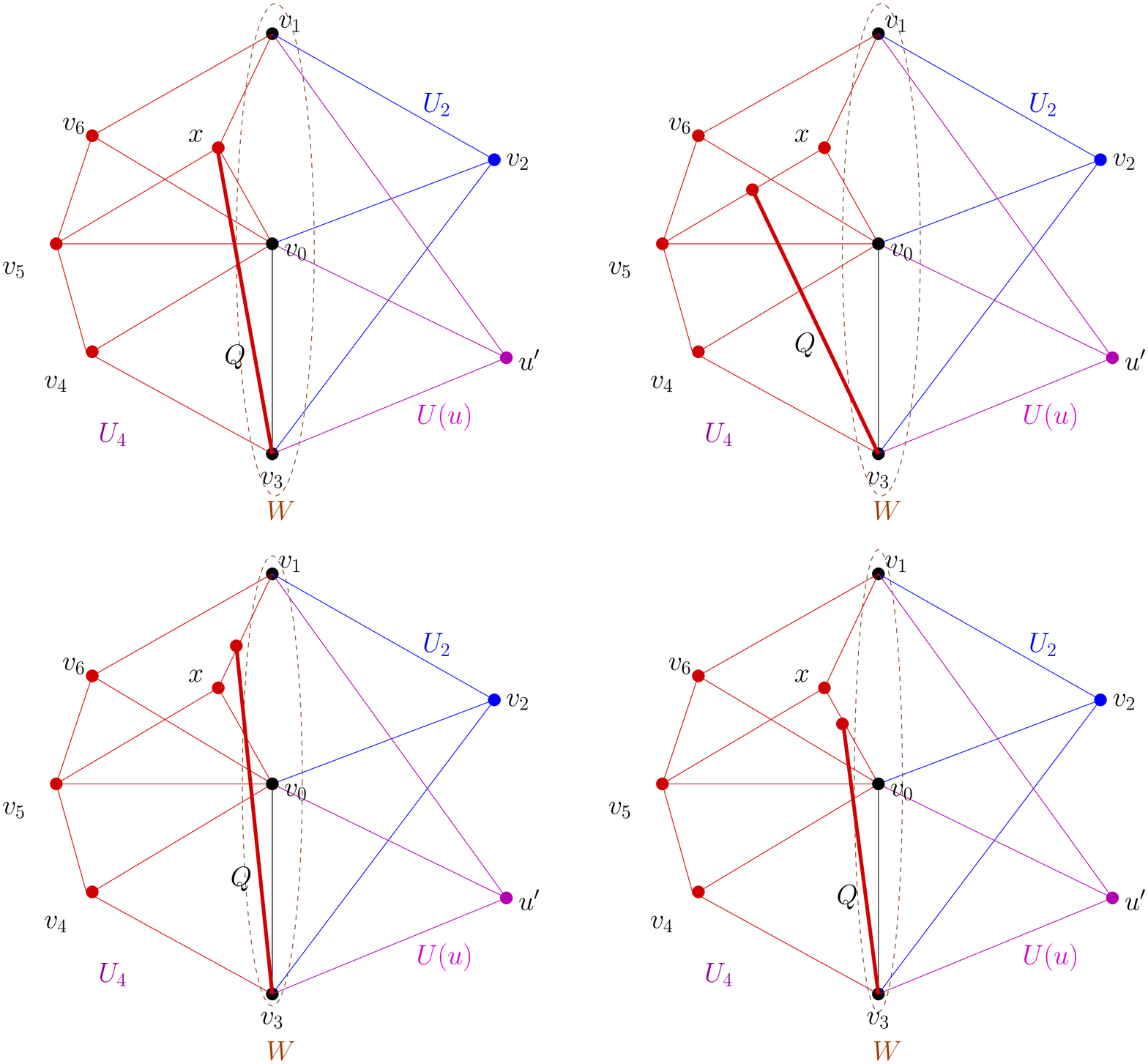}
\caption{Case (b)(i), 2: Path $Q$ exists such that $x$, $v_{6}$, and $v_{4}$ are not in three separate bridges of $G|S$.}
\label{w7case1_2_2}
\end{center}
\end{figure}

\begin{figure}[htbp]
\begin{center}
\includegraphics[width=\textwidth]{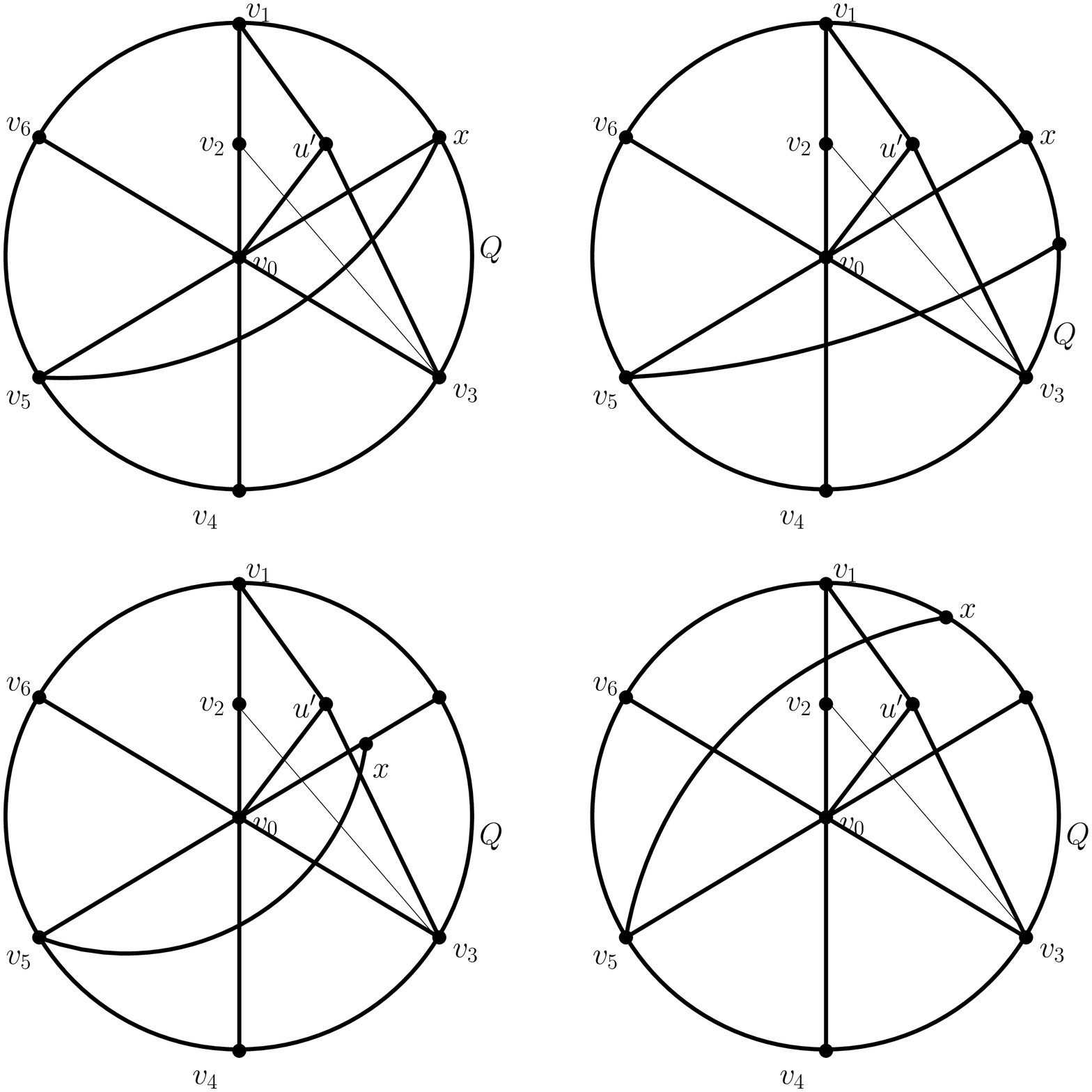}
\caption{Case (b)(i), 2: Graphs of Figure \ref{w7case1_2_2} equivalent to those of Figure \ref{w7case1_Q1} in Case 1.}
\label{w7case1_2_3}
\end{center}
\end{figure}

\textbf{2.2.} Suppose then that $x$, $v_{6}$, and $v_{4}$ are each in three separate bridges of $G|S$.

Let $T_{x}$ be the bridge of $G|S$ containing $x$. Let $T_{6}$ be the bridge of $G|S$ containing $v_{6}$. Let $T_{4}$ be the bridge of $G|S$ containing $v_{4}$ and $G\setminus U_{4}$. By the same argument used previously for bridges $U_{2}$, $U(u)$, and $U'$, there can be at most four vertices in $G\setminus T_{4}$. Thus, $|T_{4}| \ge 34$. Since there are at most four vertices in $T_{4}$ but not in $U_{4}$, $|T_{4}\cap U_{4}| \ge 30$.

Let $S_{1} = \{v_{3}, v_{0}, v_{5}\}$.

\textbf{2.2.1.} Suppose that $v_{6}$ and $v_{4}$ are not in separate bridges of $G|S_{1}$, but rather, there exists some path disjoint from $S_{1}$ joining these vertices.

Searching and checking with the program shows that such a path results in a $W_{7}$-subdivision, except in the graphs of Figure \ref{w7case1_2_4}. Suppose that $G$ contains the configuration shown in one of these graphs.

\begin{figure}[htbp]
\begin{center}
\includegraphics[width=\textwidth]{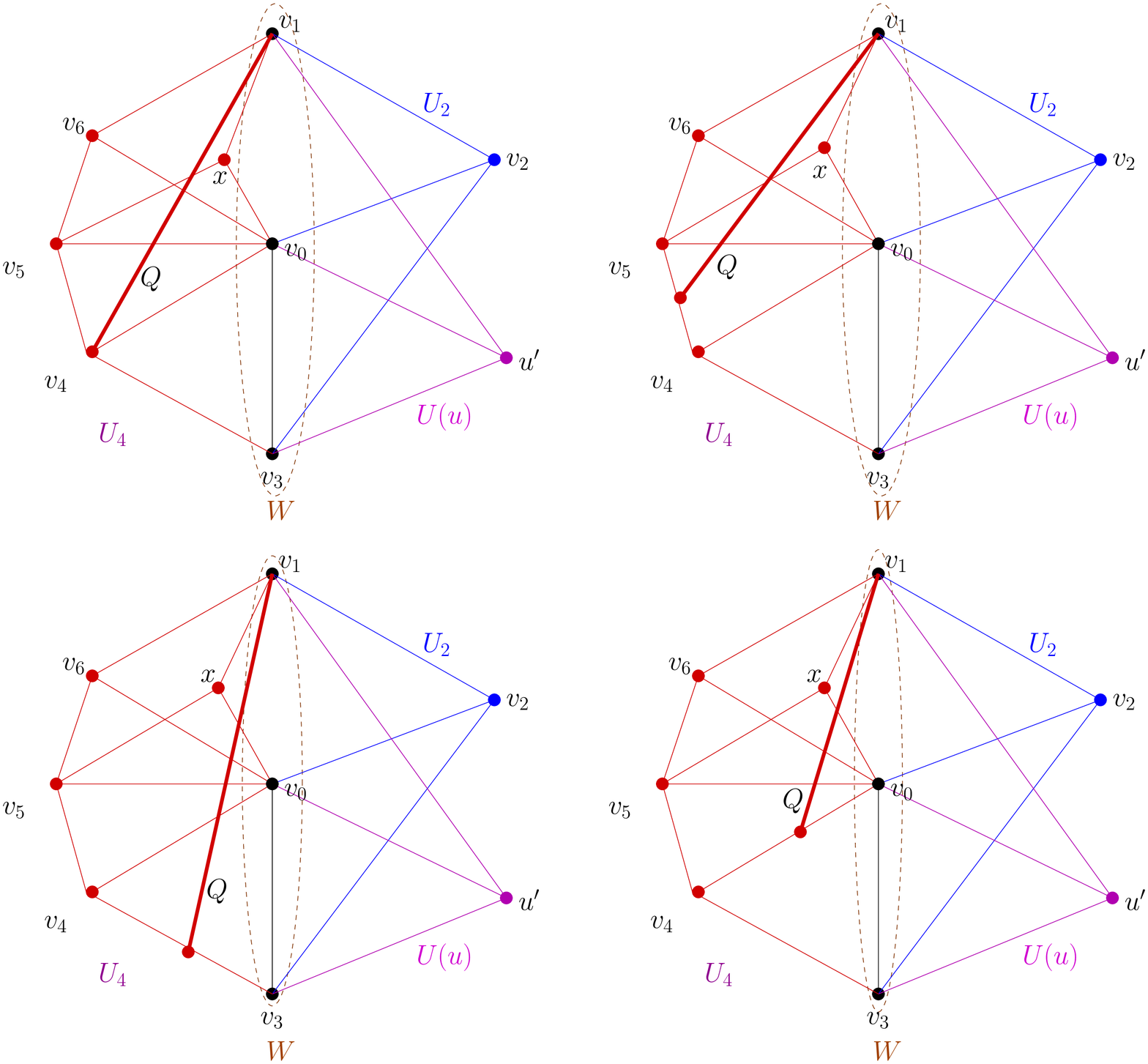}
\caption{Case (b)(i), 2: Path exists such that $v_{6}$ and $v_{4}$ are not in separate bridges of $G|S_{1}$.}
\label{w7case1_2_4}
\end{center}
\end{figure}

Suppose there exists some internal vertex $y$ on the path $P_{5}$. If $y$ is contained in some bridge of $G|S$ other than $T_{4}$, then by applying Lemma \ref{lemma1} to that bridge, a $W_{7}$-subdivision exists centred on $v_{0}$. Assume then that $y$ is contained in $T_{4}$. By 3-connectivity, there exists some path $P_{y}$ contained in $T_{4}\setminus S$ joining $y$ to $H\cap T_{4}$ that meets $H$ only at its endpoints. Searching and checking by the program shows that all possible placements of such a path result in the existence of a $W_{7}$-subdivision, except where $P_{y}$ meets $H$ at $v_{3}$. However, if $G$ contains such a path, then $G$ falls into case 1. Figure \ref{w7case1_2_5} shows how a graph isomorphic to the type of graph analysed in case 1 (pictured in Figure \ref{w7case1_Q1}) is contained as a subdivision in $G$, where $P_{y}$ meets $H$ at $v_{3}$.

\begin{figure}[htbp]
\begin{center}
\includegraphics[width=\textwidth]{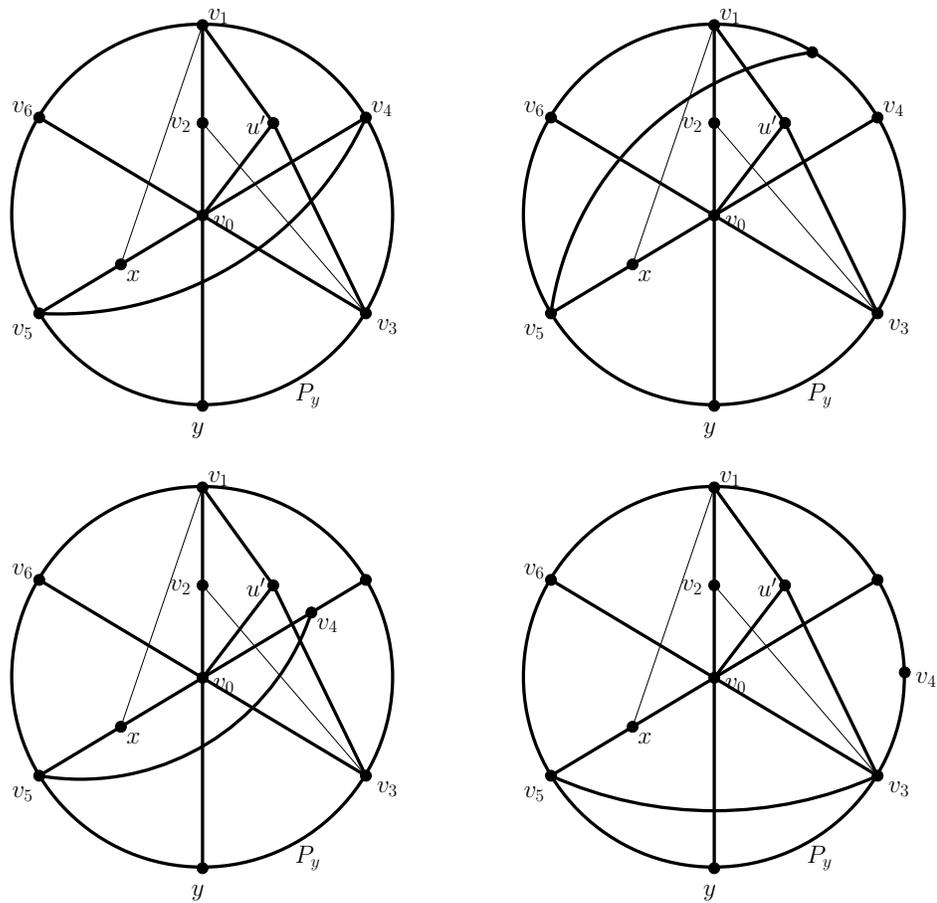}
\caption{Case (b)(i), 2: Graphs of Figure \ref{w7case1_2_4} plus path $P_{y}$ equivalent to graphs of Figure \ref{w7case1_Q1} in Case 1.}
\label{w7case1_2_5}
\end{center}
\end{figure}

Assume then that $P_{5}$ is a single edge.

Note that $|(T_{x}\cup T_{6})\setminus S| \le 3$, since otherwise a type 4 edge-vertex-cutset can be formed from $v_{1}$, $v_{5}$, and the two edges joining $v_{0}$ to $(T_{x}\cup T_{6})\setminus S$. Thus, one of the bridges $T_{x}$, $T_{6}$ contains only one vertex not in $S$. Assume without loss of generality that this bridge is $T_{x}$. Thus, each vertex in $S$ contains exactly one neighbour in $T_{x}\setminus S$. Reduction \ref{r1_big} can therefore be performed on $G$. Table \ref{t_w7i2_1} shows how Reduction \ref{r1_big} can be applied.

\begin{table}[!h]
\begin{tabular}{p{0.4\textwidth}|p{0.25\textwidth}|p{0.35\textwidth}}
\hline
\textbf{Required in Reduction \ref{r1_big}} & $S = \{t, u, v, w\}$ & Bridge $X$ of $G|S$ \\
\hline
\textbf{Equivalent construct in $G$} & $\{v_{0}, v_{1}, v_{3}, v_{5}\}$ & $T_{x}$ \\
\hline \hline
\textbf{Required in Reduction \ref{r1_big}} & Bridge $Y$ of $G|S$ & Bridge $Z$ of $G|S$ \\
\hline
\textbf{Equivalent construct in $G$}  & $T_{6}$ & The bridge of $G|\{v_{0}, v_{1}, v_{3}, v_{5}\}$ containing $v_{4}$ \\
\hline \hline
\textbf{Required in Reduction \ref{r1_big}} & Edge $vw$ & Bridges $A$ and $B$ \\
\hline
\textbf{Equivalent construct in $G$}  & $P_{5}$ & Bridges $U_{2}$ and $U(u)$ \\
\hline

\end{tabular}
\caption{Case (b)(i), 2.2.1: Applying Reduction \ref{r1_big} to $G$.}
\label{t_w7i2_1}
\end{table}

\textbf{2.2.2.} Suppose then that $v_{6}$ and $v_{4}$ are each in separate bridges of $G|S_{1}$.

Let $A$ be the bridge of $G|S_{1}$ containing $v_{6}$ and $G\setminus U_{4}$. Let $B$ be the bridge of $G|S_{1}$ containing $v_{4}$. If there are at least two neighbours of $v_{0}$ in $B\setminus S_{1}$, then by applying Lemma \ref{lemma1} to $B$, it is straightforward to check that a $W_{7}$-subdivision exists in $G$. Assume then that there is at most one neighbour of $v_{0}$ in $B\setminus S_{1}$. Thus, if $|B\setminus S_{1}| > 3$, a type 2 edge-vertex-cutset can be formed from $v_{5}$, $v_{3}$, and the edge joining $v_{0}$ to $B\setminus S_{1}$. Assume then that $|B\setminus S_{1}| \le 3$.

By the same argument used previously, if there exists some third bridge of $G|S_{1}$ other than $A$ and $B$, then there can be at most four vertices in $G\setminus A$. Regardless of the number of bridges of $G|S_{1}$, then, $|G\setminus A| \le 4$. However, since the set $G\setminus A = T_{4}\cap U_{4}$, this contradicts the conclusion drawn  at the start of case 2.2, where it is determined that $|T_{4}\cap U_{4}| \ge 30$.

\vspace{0.2in}
\noindent \textbf{Case (b)(ii): $u_{1} = v_{1}, u_{2} = v_{4}$} (Figure \ref{w7case2})

\begin{figure}[htbp]
\begin{center}
\includegraphics[width=0.4\textwidth]{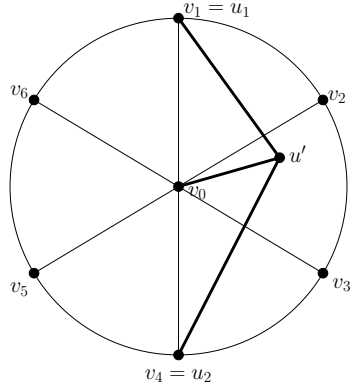}
\caption{Case (b)(ii): $u_{1} = v_{1}, u_{2} = v_{4}$}
\label{w7case2}
\end{center}
\end{figure}

Let $W = \{v_{0}, v_{1}, v_{4}\}$. Let $H_{2}$ be the subgraph consisting of the part of the rim from $v_{1}$ to $v_{4}$ that passes through $v_{2}$ and $v_{3}$, not including endpoints, and all of $P_{2}$ and $P_{3}$ except for $v_{0}$. Let $H_{5}$ be the subgraph consisting of the part of the rim from $v_{1}$ to $v_{4}$ that passes through $v_{5}$ and $v_{6}$, not including endpoints, and all of $P_{5}$ and $P_{6}$ except for $v_{0}$. Recall that $U(u)$ is the bridge of $G|V(H)$ which contains $u'$.

Suppose there exists some path $Q$ in $G$ such that $H_{2}$ is contained in $U(u)$. Testing with the program shows that all possible configurations of such a path result in the presence of a $W_{7}$-subdivision. Assume then that no such path exists. By symmetry of the graph, assume also that no path exists in $G$ such that $H_{5}$ is contained in $U(u)$.

\textbf{1.} Suppose there exists some path $Q$ from some point in $H_{2}$ to some point in $H_{5}$, such that $Q$ is disjoint from $W$. (By the previous paragraph, it can also be assumed that such a path must also be disjoint from all vertices in $U(u)$.) All but four of the possible configurations contain a $W_{7}$-subdivision. The four exceptions are shown in Figure \ref{w7case2_Q1}. Suppose that $G$ contains the configuration shown in one of these graphs.

\begin{figure}[htbp]
\begin{center}
\includegraphics[width=0.8\textwidth]{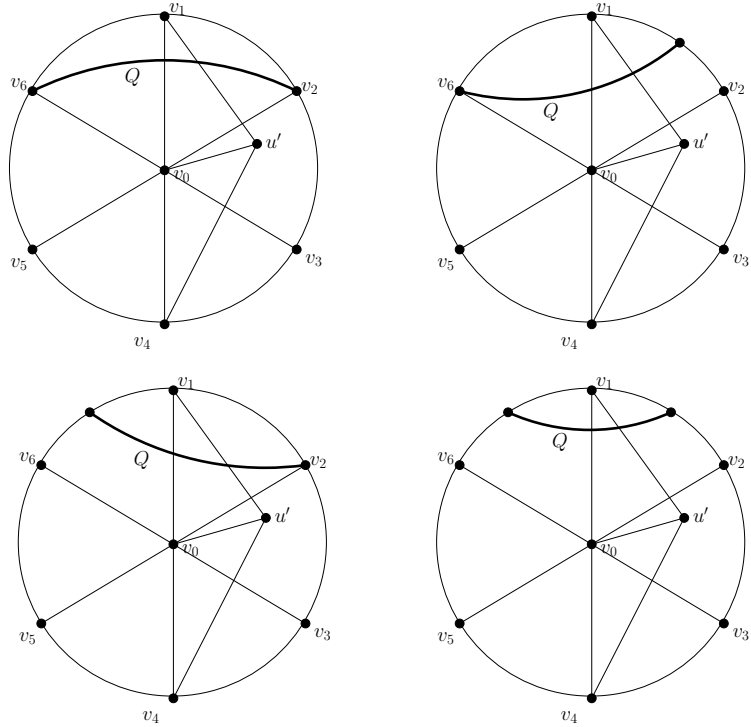}
\caption{Case (b)(ii), path $Q$ from $H_{2}$ to $H_{5}$.}
\label{w7case2_Q1}
\end{center}
\end{figure}

We know then that $U(u)$ forms a bridge of $G|W$, and that $H_{2}\cup H_{5}$ is in some bridge of $G|W$ other than $U(u)$. Call this bridge $U_{2}$.

\textbf{1.1.} Suppose there exists some internal vertex $x$ on one of the paths $P_{1}$ or $P_{4}$.

\textbf{1.1.1.} Suppose $x$ is contained in the bridge $U_{2}$. Thus, there exists some path $P_{x}$ from $x$ to $H_{2}\cup H_{5}\cup Q$ such that $P_{x}$ is contained in $U_{2}\setminus W$ and meets $H_{2}\cup H_{5}$ only at its endpoint, say, $x'$. Such a path results in a $W_{7}$-subdivision existing in $G$, unless $x$ is an internal vertex of the path $P_{4}$ and $x'\in \{v_{2}, v_{6}\}$, regardless of which of the configurations of Figure \ref{w7case2_Q1} is contained in $G$. Suppose then that $x' \in \{v_{2}, v_{6}\}$. Then $G$ falls into Case (b)(i). Figure \ref{w7case2_case1} shows how a graph isomorphic to the type of graph analysed in Case (b)(i) is contained as a subdivision in $G$ if the path $P_{x}$ exists as described. The parts of the graph in bold are those parts also contained in the graph of Case (b)(i). The dashed curves represent the four possible different placements of path $Q$.

\begin{figure}[!h]
\begin{center}
\includegraphics[width=0.6\textwidth]{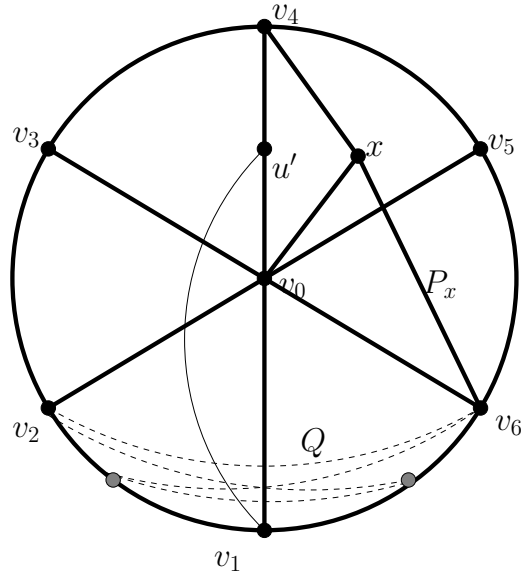}
\caption{Case (b)(ii), path $P_{x}$ places $G$ in Case (b)(i). Compare Figure \ref{w7case1}.}
\label{w7case2_case1}
\end{center}
\end{figure}

\textbf{1.1.2.} Suppose $x$ is contained in the bridge $U(u)$. 

Suppose $x$ lies on the path $P_{4}$. Since $x \in U(u)$, the neighbour of $v_{0}$ along $P_{4}$ is also in $U(u)$. Thus, $v_{0}$ has at least two neighbours in $U(u)\setminus W$. Lemma \ref{lemma1} can be applied to bridge $U(u)$, then, and a $W_{7}$-subdivision can thus be formed, as shown in Figure \ref{w7case2_1212_1}.

\begin{figure}[!h]
\begin{center}
\includegraphics[width=0.6\textwidth]{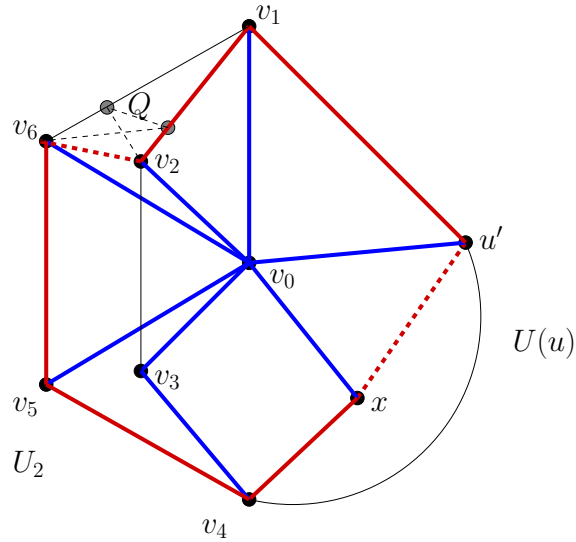}
\caption{Case (b)(ii), 1.1.2: internal vertex on $P_{4}$ contained in $U(u)$ results in $W_{7}$-subdivision}
\label{w7case2_1212_1}
\end{center}
\end{figure}

Assume then that $x$ is an internal vertex of the path $P_{1}$.

Suppose there exists some internal vertex $y$ on the path $P_{4}$. If $y \in U(u)$, then by Lemma \ref{lemma2}, a $W_{7}$-subdivision exists. If $y \in U_{2}$, then the graph falls into case 1.1.1 above. Assume then that no internal vertex of $P_{4}$ is contained in $U_{2}$ or $U(u)$.

Suppose there exists some bridge $A$ of $G|W$ other than $U_{2}$ and $U(u)$. Then a $W_{7}$-subdivision can be formed in $G$, with two spokes in $\langle U(u)\rangle$ (by Lemma \ref{lemma1}), four spokes in $\langle U_{2}\rangle$, and one spoke in $\langle A\rangle$ (from $v_{0}$ to $v_{1}$). Figure \ref{w7case2_1212_2} illustrates such a situation. 

\begin{figure}[!h]
\begin{center}
\includegraphics[width=0.6\textwidth]{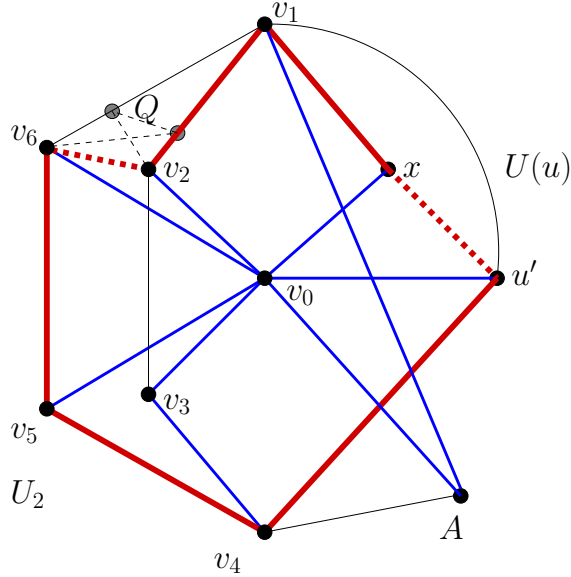}
\caption{Case (b)(ii), 1.1.2: third bridge $A$ of $G|W$ results in $W_{7}$-subdivision}
\label{w7case2_1212_2}
\end{center}
\end{figure}

Assume then that $U_{2}$ and $U(u)$ are the only bridges of $G|W$, and as such $P_{4}$ is a single edge.

Suppose $v_{0}$ has $\ge 3$ neighbours in $U(u)\setminus W$. Then by Lemma \ref{lemma2}, a $W_{7}$-subdivision exists centred on $v_{0}$. Assume then that $v_{0}$ has at most two neighbours in $U(u)\setminus W$.

Suppose $|U(u)\setminus W| > 3$. Then a type 4 edge-vertex-cutset can be formed from $v_{1}$, $v_{4}$, and the two edges joining $v_{0}$ to $U(u)\setminus W$. Assume then that $|U(u)\setminus W| \le 3$.

\textbf{1.1.2.1.} Suppose $v_{1}$ has at most two neighbours in $U_{2}\setminus W$, say, $x_{1}$ and $x_{2}$.

If $|U(u)\setminus W| = 3$, a type 4 edge-vertex-cutset can be formed from the edges $v_{1}x_{1}$ and $v_{1}x_{2}$, and the vertices $v_{0}$ and $v_{4}$. Assume then that $U(u)\setminus W$ contains only the two vertices $x$ and $u'$. This implies that $v_{0}$ is adjacent to $x$.

Suppose that $G$ contains the edges $v_{1}v_{4}$ and $v_{4}x$, and that $v_{4}$ has degree $\ge 7$. Thus, $v_{4}$ contains exactly four neighbours not in $U_{2}\setminus W$, and so must have at least three neighbours in $U_{2}\setminus W$. Therefore, by Lemma \ref{lemmaW7}, a $W_{7}$-subdivision exists centred on $v_{4}$.

Assume then that either $v_{4}$ has degree $< 7$, or that at least one of the edges $v_{1}v_{4}$, $v_{4}x$ does not exist in $G$.

Suppose the edge $v_{0}v_{1}$ exists in $G$. Then by Lemma \ref{lemmaW7}, a $W_{7}$-subdivision exists centred on $v_{0}$. Assume then that such an edge does not exist. Therefore, $v_{0}$ has exactly two neighbours in the set $\{v_{1}, x, u'\}$, and $v_{4}$ either has at most two neighbours in this set, or has degree $< 7$. Thus, a type 4a edge-vertex-cutset can be formed from the edges $v_{1}x_{1}$ and $v_{2}x_{2}$, and the vertices $v_{4}$ and $v_{0}$.

\textbf{1.1.2.2.} Suppose then that $v_{1}$ has some third neighbour $y$ in $U_{2}\setminus W$, such that $y\notin N_{H}(v_{1})$.

By 3-connectivity, there must exist some path $Y$ in $U_{2}\setminus W$ joining $y$ to $H\cup Q$, such that $Y$ meets $H\cup Q$ only at its endpoint, say $y'$. Using the program to generate and check all possible such paths $Y$, it is found that a $W_{7}$-subdivision exists in $G$ for each case, except where $y'$ is an internal vertex on the path $Q$. Assume then that this is the case for all such paths $Y$.

Consider the set $S = \{v_{1}, v_{2}, v_{6}\}$. Suppose $S$ does not form a separating set in $G$. Then there exists some path disjoint from $S$ joining the two components of $(H\cup U(u)\cup Q) - S$.  Using the program to generate and check all possible such paths, it is found that a $W_{7}$-subdivision exists in each case. Assume then that $S$ forms a separating set in $G$. Let $T'$ be the bridge of $G|S$ containing $y$. Let $T''$ be the bridge of $G|S$ containing $v_{3}$, $v_{5}$, and $U(u)$.

Suppose $v_{1}$ has degree $\ge 7$. Thus, there exist at least two neighbours of $v_{1}$, say $a_{1}$ and $a_{2}$, such that $a_{1}, a_{2} \neq y$ and $a_{1}, a_{2}\notin N_{H\cup P_{u1}}(v_{1})$. By 3-connectivity, there must exist paths $A_{1}$ and $A_{2}$ joining $H$ to $a_{1}$ and $a_{2}$ respectively. Using the program to generate and check all possible such paths $A_{1}$ and $A_{2}$, it is found that a $W_{7}$-subdivision exists in $G$ for each case. Assume then that $v_{1}$ has degree $< 7$. Thus, $v_{1}$ has at most four neighbours in $U_{2}\setminus W$.

Suppose $|U(u)\setminus W| = 3$. Suppose also that $v_{4}$ has at most two neighbours in $U(u)\setminus W$, say $b_{1}$ and $b_{2}$ (if a second neighbour exists). Since $v_{1}$ has degree $< 7$, and $v_{0}$ has only two neighbours in $U(u)\setminus W$, then a type 2a or 4a edge-vertex-cutset can be formed from $v_{0}$, $v_{1}$, $v_{4}b_{1}$, and $v_{4}b_{2}$ (if $b_{2}$ exists). Assume then that $v_{4}$ is adjacent to all three vertices in $U(u)\setminus W$. If $v_{4}$ has $\ge 3$ neighbours in $U_{2}\setminus W$, then, by Lemma \ref{lemma2}, a $W_{7}$-subdivision exists centred on $v_{4}$. Assume then that $v_{4}$ has at most two neighbours, say $c_{1}$ and $c_{2}$, in $U_{2}\setminus W$. Then a type 4 edge-vertex-cutset can be formed from $v_{1}$, $v_{0}$, $v_{4}c_{1}$, and $v_{4}c_{2}$.

Assume then that $|U(u)\setminus W| = 2$.

Since $|V(G)| \ge 38$, then, we know that $U_{2}$ contains at least 36 vertices. In the remainder of this case, the structure of $U_{2}$ is more closely examined. Various sets of size 3 contained in $U_{2}$ are identified to be separating sets. For each such separating set $U^{*}$, all but one of the bridges of $G|U^{*}$ are shown to be limited in size to some small number of vertices, otherwise some forbidden edge-vertex-cutset exists. It is then shown that the intersection of each of the `large' bridges can contain at most two vertices, which results in a contradiction.

\textbf{Step 1: Bounding $|V(G) \setminus T''|$.} Recall that $S = \{v_{1}, v_{2}, v_{6}\}$, that $T'$ is the bridge of $G|S$ containing $y$, and that $T''$ is the bridge of $G|S$ containing $v_{3}$, $v_{5}$, and $U(u)$.

Suppose $|V(G) \setminus T''| \ge 3$. Recall that for any neighbour $y$ of $v_{1}$ where $y\in U_{2}\setminus W$ but $y\notin N_{H}(v_{1})$, all paths in $U_{2}\setminus W$ joining $y$ to $H\cup Q$ must first meet $H\cup Q$ at an internal vertex of the path $Q$. Thus, any neighbour of $v_{1}$ in $U_{2}\setminus W$ is also in the bridge $T'$. Any neighbours of $v_{1}$ that are not in $T'$, then, must be in $U(u)\setminus W$. Since $|U(u)\setminus W| = 2$, there can be only two such neighbours of $v_{1}$. Thus, a type 4 edge-vertex-cutset can be formed from $v_{2}$, $v_{6}$, and the edges joining $v_{1}$ to $U(u)\setminus W$.

Assume then that $|V(G) \setminus T''| \le 2$.

Let $X_{1} = (T''\cap U_{2})\setminus W$. Since $|U_{2}\setminus W| \ge 33$, and $|V(G) \setminus T''| \le 2$, $X_{1}$ must contain at least 31 vertices.

Consider now the set $S_{1} = \{v_{2}, v_{0}, v_{4}\}$. Suppose that $S_{1}$ is not a separating set, but rather, there exists some path disjoint from $S_{1}$ joining $v_{3}$ to $v_{1}$. Using the program to check all possible placements of such a path shows that a $W_{7}$-subdivision exists in each case. Suppose then that $S_{1}$ forms a separating set in $G$, with at least two bridges: $T_{3}$, which contains the vertex $v_{3}$, and $T_{1}$, which contains the vertices $v_{1}$, $v_{5}$, $v_{6}$, $u'$, and $x$.

\textbf{Step 2: Bounding $|V(G)\setminus T_{1}|$.} Suppose $v_{0}$ has more than two neighbours in $T_{3}\setminus S_{1}$. Then a $W_{7}$-subdivision exists in $G$, by applying Lemma \ref{lemma2} to $T_{1}$ and $T_{3}$, and using the edge $v_{0}v_{4}$ as a seventh spoke. Assume then that $v_{0}$ has at most two neighbours in $T_{3}\setminus S_{1}$. Then, if $|T_{3}\setminus S_{1}| > 3$, a type 2 or 4 edge-vertex-cutset can be formed from $v_{2}$, $v_{4}$, and the edge or edges joining $S_{1}$ to $T_{3}\setminus S_{1}$. Assume then that $|T_{3}\setminus S_{1}| \le 3$.

Suppose there exists some bridge of $G|S_{1}$ other than $T_{1}$ and $T_{3}$. Then $G$ falls into Case (b)(i), as illustrated in Figure \ref{w7case2_case1_2}. Assume then that there are only two bridges of $G|S_{1}$. Thus, there are at most three vertices in $V(G)\setminus T_{1}$.

\begin{figure}[!h]
\begin{center}
\includegraphics[width=0.6\textwidth]{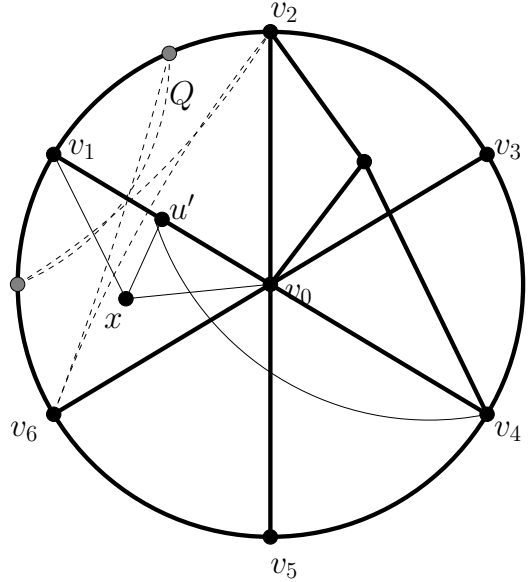}
\caption{Case (b)(ii), third bridge of $G|S_{1}$ places $G$ in Case (b)(i).}
\label{w7case2_case1_2}
\end{center}
\end{figure}

Recall $X_{1} = (T''\cap U_{2})\setminus W$, and $|X_{1}| \ge 31$.

Let $X_{2} = X_{1}\cap T_{1}$. Since $|X_{1}| \ge 31$ and $|V(G)\setminus T_{1}| \le 3$, $X_{2}$ must contain at least 28 vertices.

\textbf{Step 3: Excluding vertices on $P_{2}$.} Suppose there exists some internal vertex $p_{2}$ on the path $P_{2}$.

By 3-connectivity, there exists some path $Q_{2}$ from $p_{2}$ to $H - P_{2}$ such that $Q_{2}$ meets $H - P_{2}$ only at its endpoint, say, $q_{2}$. Using the program to generate and check all possible placements of $Q_{2}$, it is found that the existence of such a path results in a $W_{7}$-subdivision in $G$, unless $q_{2}$ is contained in the bridge $T_{3}$, or $q_{2} = v_{4}$. If the former is true for any such path $Q_{2}$, then all internal vertices on the path $P_{2}$ are contained in the bridge $T_{3}$, and thus are not in the set $X'$. Suppose then that $q_{2} = v_{4}$ for all such paths $Q_{2}$. This, however, would mean that all such vertices $q_{2}$ are contained in some third bridge of $G|S_{1}$ other than $T_{1}$ or $T_{3}$, and we have already deduced in Step 2 that no such bridge exists.

Assume then that $X_{2}$ does not contain any internal vertices on the path $P_{2}$.

Consider the set $S_{2} = \{v_{6}, v_{0}, v_{4}\}$. Suppose that $S_{2}$ is not a separating set, but rather, there exists some path disjoint from $S_{2}$ joining $v_{5}$ to $v_{1}$. Using the program to check all possible placements of such a path shows that a $W_{7}$-subdivision exists in each case. Suppose then that $S_{2}$ forms a separating set in $G$, with at least two bridges: $Y_{5}$, which contains the vertex $v_{5}$, and $Y_{1}$, which contains $v_{1}$, $v_{2}$, $v_{3}$, and $U(u)$.

\textbf{Step 4: Bounding $|V(G)\setminus Y_{1}|$.} Suppose $v_{0}$ has more than two neighbours in $Y_{5}\setminus S_{2}$. Then a $W_{7}$-subdivision exists in $G$, by applying Lemma \ref{lemma2} to $Y_{1}$ and $Y_{5}$, and using the edge $v_{0}v_{4}$ as a seventh spoke. Assume then that $v_{0}$ has at most two neighbours in $Y_{5}\setminus S_{2}$. Then, if $|Y_{5}\setminus S_{2}| > 3$, a type 2 or 4 edge-vertex-cutset can be formed from $v_{6}$, $v_{4}$, and the edge or edges joining $S_{2}$ to $Y_{5}\setminus S_{2}$. Assume then that $|Y_{5}\setminus S_{2}| \le 3$.

Suppose there exists some bridge of $G|S_{2}$ other than $Y_{5}$ and $Y_{1}$. Then $G$ falls into Case (b)(i), as illustrated in Figure \ref{w7case2_case1_3}. Assume then that there are only two bridges of $G|S_{2}$. Thus, there are at most three vertices in $V(G)\setminus Y_{1}$.

\begin{figure}[!h]
\begin{center}
\includegraphics[width=0.6\textwidth]{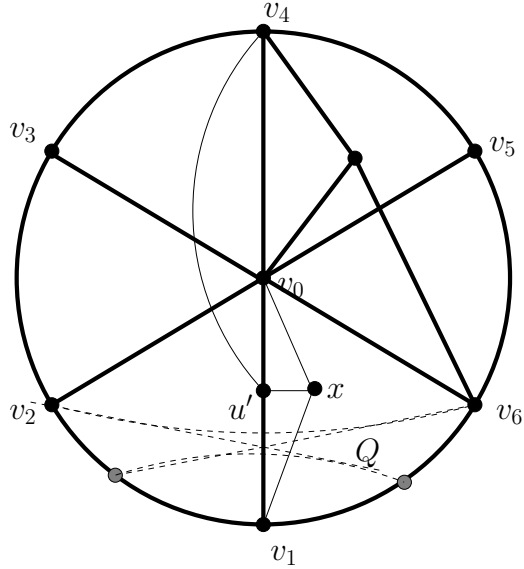}
\caption{Case (b)(ii), third bridge of $G|S_{2}$ places $G$ in Case (b)(i).}
\label{w7case2_case1_3}
\end{center}
\end{figure}

Recall $X_{2} = (U_{2}\cap T'' \cap T_{1})\setminus W$, and $|X_{2}| \ge 28$.

Let $X_{3} = X_{2}\cap Y_{1}$. Since $|X_{2}| \ge 28$, and $|V(G)\setminus Y_{1}| \le 3$, $X_{3}$ must contain at least 25 vertices.

\textbf{Step 5: Excluding vertices on $P_{6}$.} Suppose there exists some internal vertex $p_{6}$ on the path $P_{6}$.

By 3-connectivity, there exists some path $Q_{6}$ from $p_{6}$ to $H - P_{6}$ such that $Q_{6}$ meets $H - P_{6}$ only at its endpoint, say, $q_{6}$. Using the program to generate and check all possible placements of $Q_{6}$, it is found that the existence of such a path results in a $W_{7}$-subdivision in $G$, unless $q_{6}$ is contained in the bridge $Y_{5}$, or $q_{6} = v_{4}$. If the former is true for any such path $Q_{6}$, then all internal vertices on the path $P_{6}$ are contained in the bridge $Y_{5}$, and thus are not in the set $X''$. Suppose then that $q_{6} = v_{4}$ for all such paths $Q_{6}$. This, however, would mean that all such vertices $q_{6}$ are contained in some third bridge of $G|S_{2}$ other than $Y_{1}$ or $Y_{5}$, and we have already shown in Step 4 that no such bridge exists.

Assume then that $X_{3}$ does not contain any internal vertices on the path $P_{6}$.

\textbf{Step 6: Bounding vertices on $v_{1}Cv_{2}$.} Suppose there exists some internal vertex $p'_{2}$ on the path $v_{1}Cv_{2}$.

If $G$ contains the configuration shown in either the second or fourth graphs of Figure \ref{w7case2_Q1}, then one of the endpoints of $Q$ forms such a vertex, and so any internal vertices on $v_{1}Cv_{2}$ are contained in the bridge $T'$ (and thus are not contained in $X_{3}$).

Suppose then that $G$ contains the configuration shown in either the first or third graphs of Figure \ref{w7case2_Q1}. By 3-connectivity, there exists some path $Q'_{2}$ from $p'_{2}$ to $H - v_{1}Cv_{2}$ such that $Q'_{2}$ meets $H$ only at its endpoint, say, $q'_{2}$. Using the program to generate and check all possible placements of $Q'_{2}$, it is found that the existence of such a path results in a $W_{7}$-subdivision in $G$, unless $q'_{2}$ is contained in the bridge $T'$, or $q'_{2} = v_{4}$. If the former is true for any such path $Q'_{2}$, then all internal vertices on the path $v_{1}Cv_{2}$ are contained in the bridge $T'$, and thus are not in the set $X_{3}$. Suppose then that $q'_{2} = v_{4}$ for all such paths $Q'_{2}$.

Let $a_{1}$ be the vertex closest to $v_{1}$ along $v_{1}Cv_{2}$, and let $a_{n}$ be the vertex closest to $v_{2}$ along $v_{1}Cv_{2}$. Note that the removal of $v_{1}$, $v_{2}$, and $v_{4}$ disconnects the graph, placing $a_{1}$ and $a_{n}$ in a separate component from the other vertices in $H$. Let $A$ be the bridge of $G|\{v_{1}, v_{2}, v_{4}\}$ containing $a_{1}$, $a_{n}$, and the other internal vertices along $v_{1}Cv_{2}$. If $v_{4}$ contains at least three neighbours in $A\setminus \{v_{1}, v_{2}, v_{4}\}$, then by applying Lemma \ref{lemma2} to the bridge $A$, a $W_{7}$-subdivision can be formed in $G$. Suppose then that $v_{4}$ contains at most two neighbours in $A\setminus \{v_{1}, v_{2}, v_{4}\}$, say, $a'_{1}$ and $a'_{2}$ (if a second neighbour exists). Then, if $|A\setminus \{v_{1}, v_{2}, v_{4}\}| > 3$, a type 2 or 4 edge-vertex-cutset can be formed from $v_{1}$, $v_{2}$, and the edges $v_{4}a'_{1}$ and $v_{4}a'_{2}$ (if $a'_{2}$ exists).

Assume then that $|A\setminus \{v_{1}, v_{2}, v_{4}\}| \le 3$.

Recall $X_{3} = (U_{2}\cap T''\cap T_{1}\cap Y_{1})\setminus W$, and $|X_{3}| \ge 25$.

Let $X_{4} = X_{3} \setminus (A\setminus \{v_{1}, v_{2}, v_{4}\})$. Since $|X_{3}| \ge 25$ and $|A\setminus \{v_{1}, v_{2}, v_{4}\}| \le 3$, $X_{4}$ must contain at least 22 vertices.

\textbf{Step 7: Bounding vertices on $v_{6}Cv_{1}$.} Suppose there exists some internal vertex $p'_{6}$ on the path $v_{6}Cv_{1}$.

If $G$ contains the configuration shown in either the third or fourth graphs of Figure \ref{w7case2_Q1}, then one of the endpoints of $Q$ forms such a vertex, and so any internal vertices on $v_{6}Cv_{1}$ are contained in the bridge $T'$ (and thus are not contained in $X_{4}$).

Suppose then that $G$ contains the configuration shown in either the first or second graphs of Figure \ref{w7case2_Q1}. By 3-connectivity, there exists some path $Q'_{6}$ from $p'_{6}$ to $H - v_{6}Cv_{1}$ such that $Q'_{6}$ meets $H$ only at its endpoint, say, $q'_{6}$. Using the program to generate and check all possible placements of $Q'_{6}$, it is found that the existence of such a path results in a $W_{7}$-subdivision in $G$, unless $q'_{6}$ is contained in the bridge $T'$, or $q'_{6} = v_{4}$. If the former is true for any such path $Q'_{6}$, then all internal vertices on the path $v_{6}Cv_{1}$ are contained in the bridge $T'$, and thus are not in the set $X_{4}$. Suppose then that $q'_{6} = v_{4}$ for all such paths $Q'_{6}$.

Let $b_{1}$ be the vertex closest to $v_{1}$ along $v_{6}Cv_{1}$, and let $b_{n}$ be the vertex closest to $v_{6}$ along $v_{6}Cv_{1}$. Note that the removal of $v_{1}$, $v_{6}$, and $v_{4}$ disconnects the graph, placing $b_{1}$ and $b_{n}$ in a separate component from the other vertices in $H$. Let $B$ be the bridge of $G|\{v_{1}, v_{6}, v_{4}\}$ containing $b_{1}$, $b_{n}$, and the other internal vertices along $v_{6}Cv_{1}$.

By the same argument used in the previous paragraph for the bridge $A$, it can be assumed that $|B\setminus \{v_{1}, v_{6}, v_{4}\}| \le 3$.

Let $X_{5} = X_{4} \setminus (B\setminus \{v_{1}, v_{6}, v_{4}\})$. Since $|X_{4}| \ge 22$ and $|B\setminus \{v_{1}, v_{6}, v_{4}\}| \le 3$, $X_{5}$ must contain at least 19 vertices.

Consider now the set $S' = \{v_{4}, v_{2}, v_{6}\}$.  Suppose there exists a second bridge of $G|S'$ other than that containing $v_{1}$. Let $Z_{1}$ be the bridge of $G|S'$ that contains $v_{1}$.

\textbf{Step 8: Bounding $|V(G)\setminus Z_{1}|$.} Suppose $v_{4}$ has at least three neighbours not in $Z_{1}$, say, $a_{1}$, $a_{2}$, and $a_{3}$. If $a_{1}$, $a_{2}$, and $a_{3}$ are all contained in one bridge of $G|S'$, then by Lemma \ref{lemma2}, a $W_{7}$-subdivision exists centred on $v_{4}$ (see Figure \ref{w7case2_1212_3}). If two of these vertices, say $a_{1}$ and $a_{2}$, are in one bridge, while $a_{3}$ is in a separate bridge, then a $W_{7}$-subdivision can be formed by applying Lemma \ref{lemma1} to the bridge containing $a_{1}$ and $a_{2}$, and using the bridge containing $a_{3}$ to create a spoke-meets-rim vertex at $v_{2}$ (see Figure \ref{w7case2_1212_4}). Suppose then that $a_{1}$, $a_{2}$, and $a_{3}$ are each in separate bridges of $G|S'$, say, $A_{1}$, $A_{2}$, $A_{3}$. To avoid the possibility of Reduction \ref{r1}A, each of $A_{1}\setminus S'$, $A_{2}\setminus S'$, $A_{3}\setminus S'$ must contain at least two vertices. Thus, to avoid an internal 4-edge-cutset, there exists some vertex $v'$ in $S'$ such that $v'$ has at least two neighbours in at least one of $A_{1}\setminus S'$, $A_{2}\setminus S'$, $A_{3}\setminus S'$. Since each vertex in $S'$ also has at least three neighbours in $Z_{1}\setminus S'$, by Lemma \ref{lemmaW7}, a $W_{7}$-subdivision can be formed in $G$ centred on $v'$ (see Figure \ref{w7case2_1212_5} for an example).

\begin{figure}[!h]
\begin{center}
\includegraphics[width=0.9\textwidth]{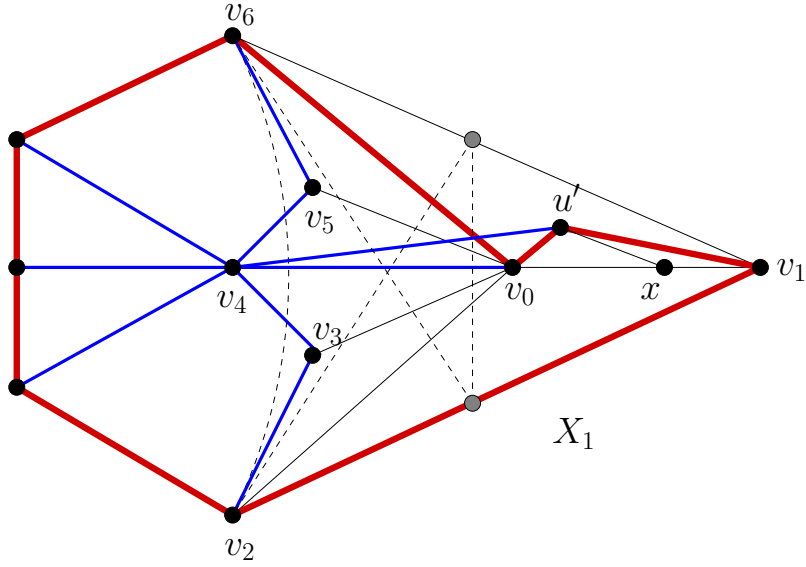}
\caption{Case (b)(ii), $W_{7}$-subdivision exists when $v_{4}$ has three neighbours not in $X_{1}$, all contained in the one bridge of $G|S'$}
\label{w7case2_1212_3}
\end{center}
\end{figure}

\begin{figure}[!h]
\begin{center}
\includegraphics[width=0.9\textwidth]{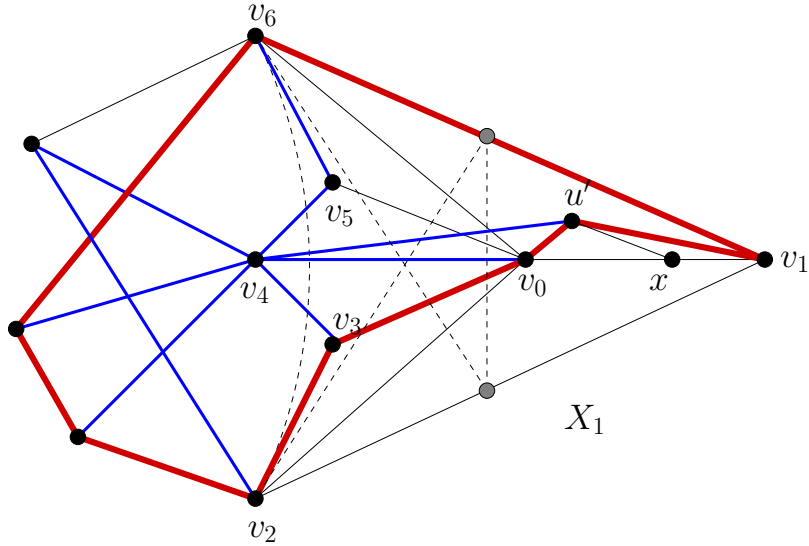}
\caption{Case (b)(ii), $W_{7}$-subdivision exists when $v_{4}$ has three neighbours not in $X_{1}$, contained in two bridges of $G|S'$}
\label{w7case2_1212_4}
\end{center}
\end{figure}

\begin{figure}[!h]
\begin{center}
\includegraphics[width=0.9\textwidth]{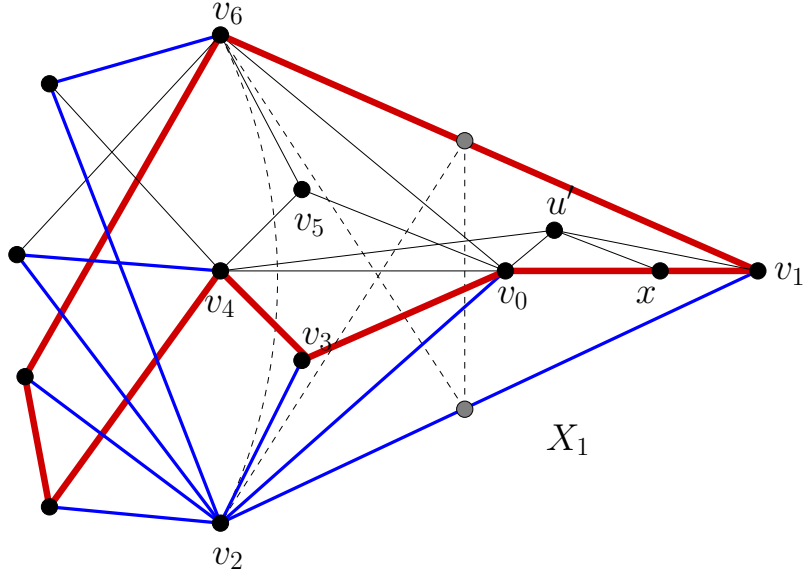}
\caption{Case (b)(ii), $W_{7}$-subdivision exists when $v_{4}$ has three neighbours not in $X_{1}$, each in a separate bridge of $G|S'$}
\label{w7case2_1212_5}
\end{center}
\end{figure}

Suppose then that $v_{4}$ has at most two neighbours not in $Z_{1}$. Therefore, unless $|V(G)\setminus Z_{1}| \le 3$, a type 2 or 4 edge-vertex-cutset is formed from $v_{2}$, $v_{6}$, and the edge or edges joining $v_{4}$ to $V(G)\setminus Z_{1}$. Suppose then that $|V(G)\setminus Z_{1}| \le 3$.

Recall $X_{5} = ((U_{2}\cap T''\cap T_{1}\cap Y_{1})\setminus W)\setminus (A\setminus \{v_{1}, v_{2}, v_{4}\})\setminus (B\setminus \{v_{1}, v_{6}, v_{4}\})$, and that $|X_{5}| \ge 19$.

Let $X_{6} = X_{5}\cap Z_{1}$. Since $|X_{5}|\ge 19$ and $|V(G)\setminus Z_{1}| \le 3$, $X_{6}$ must contain at least 16 vertices.

\textbf{Step 9: Proving $|X_{6}| \ge 16$ is a contradiction.} The vertices $v_{2}$ and $v_{6}$ are contained in the set $X_{6}$, but no other vertices in $H$ can be contained in $X_{6}$. To preserve 3-connectivity, though, there must exist some path $P_{X}$ disjoint from $\{v_{2}, v_{6}\}$ joining $X_{6}\setminus \{v_{2}, v_{6}\}$ to $H$. Let $p_{X}$ be the vertex where $P_{X}$ first meets $H$. To avoid creating a $W_{7}$-subdivision, it must be the case that $p_{X} \in \{v_{1}, v_{4}\}$. However, it has already been argued earlier in this case that any neighbour of $v_{1}$ in $U_{2}\setminus W$  is also in the bridge $T'$. Since any path joining $T'$ to $X_{6}$ must pass through some vertex in $\{v_{1}, v_{2}, v_{6}\}$, it cannot be the case that $p_{X} = v_{1}$. Assume then that $p_{X} = v_{4}$, and that any path joining $X_{6}$ to $H$ must pass through one of $v_{2}$, $v_{4}$, or $v_{6}$, that is, the set $S'$. However, (from Step 8) $X_{6}$ only contains vertices in the bridge $Z_{1}$ of $G|S'$. Thus, $X_{6}$ can only contain the two vertices $v_{2}$ and $v_{6}$, and so $|X_{6}| = 2$, which is a contradiction.

\textbf{1.1.3.} Suppose then that $P_{1}$ and $P_{4}$ do not contain any internal vertices in the bridges $U_{2}$ or $U(u)$, but rather, $x$ is contained in some third bridge $A$ of $G|W$ other than $U_{2}$ or $U(u)$. If $A$ contains internal vertices of both $P_{1}$ and $P_{4}$, then a $W_{7}$-subdivision can be formed in $G$. Suppose then that $A$ contains internal vertices of only one of these paths.

If there exists some fourth bridge of $G|W$ other than $U_{2}$, $U(u)$, and $A$, then by Lemma \ref{threebridges}, a $W_{7}$-subdivision exists in $G$. (See Table \ref{t_w7ii_1}.)

\begin{table}[!h]
\begin{tabular}{p{0.35\textwidth}|p{0.3\textwidth}|p{0.3\textwidth}}
\hline
\textbf{Required in Lemma \ref{threebridges}} & $S = \{u, v, w\}$ & Bridges $X$, $Y$ of $G|S$ \\
\hline
\textbf{Equivalent construct in $G$} & $W = \{v_{0}, v_{1}, v_{4}\}$ & Bridges $U(u)$, $A$ of $G|W$ \\
\hline \hline
\textbf{Required in Lemma \ref{threebridges}}  & Bridge $Z$ of $G|S$ containing $\ge 3$ neighbours of $v$ not in $S$ & $P_{u}$, $P_{w}$ \\
\hline
\textbf{Equivalent construct in $G$}  & Bridge $U_{2}$ of $G|W$ & One of $P_{1}$ or $P_{4}$ (whichever of these has no vertices contained in $A\setminus W$); path in fourth bridge of $G|W$ other than $U_{2}$, $U(u)$, and $A$. \\
\hline

\end{tabular}
\caption{Case (b)(ii), 1.1.3: Applying Lemma \ref{threebridges} to $G$, where there are at least four bridges of $G|W$.} 
\label{t_w7ii_1}
\end{table}

Suppose then that $U_{2}$, $U(u)$, and $A$ are the only bridges of $G|W$.

Suppose $x$ is on the path $P_{4}$. Then $G$ falls into Case (b)(i). Figure \ref{w7case2_1213_1} shows how a graph isomorphic to the type of graph analysed in Case (b)(i) is contained as a subdivision in $G$ in this situation.

\begin{figure}[!h]
\begin{center}
\includegraphics[width=0.6\textwidth]{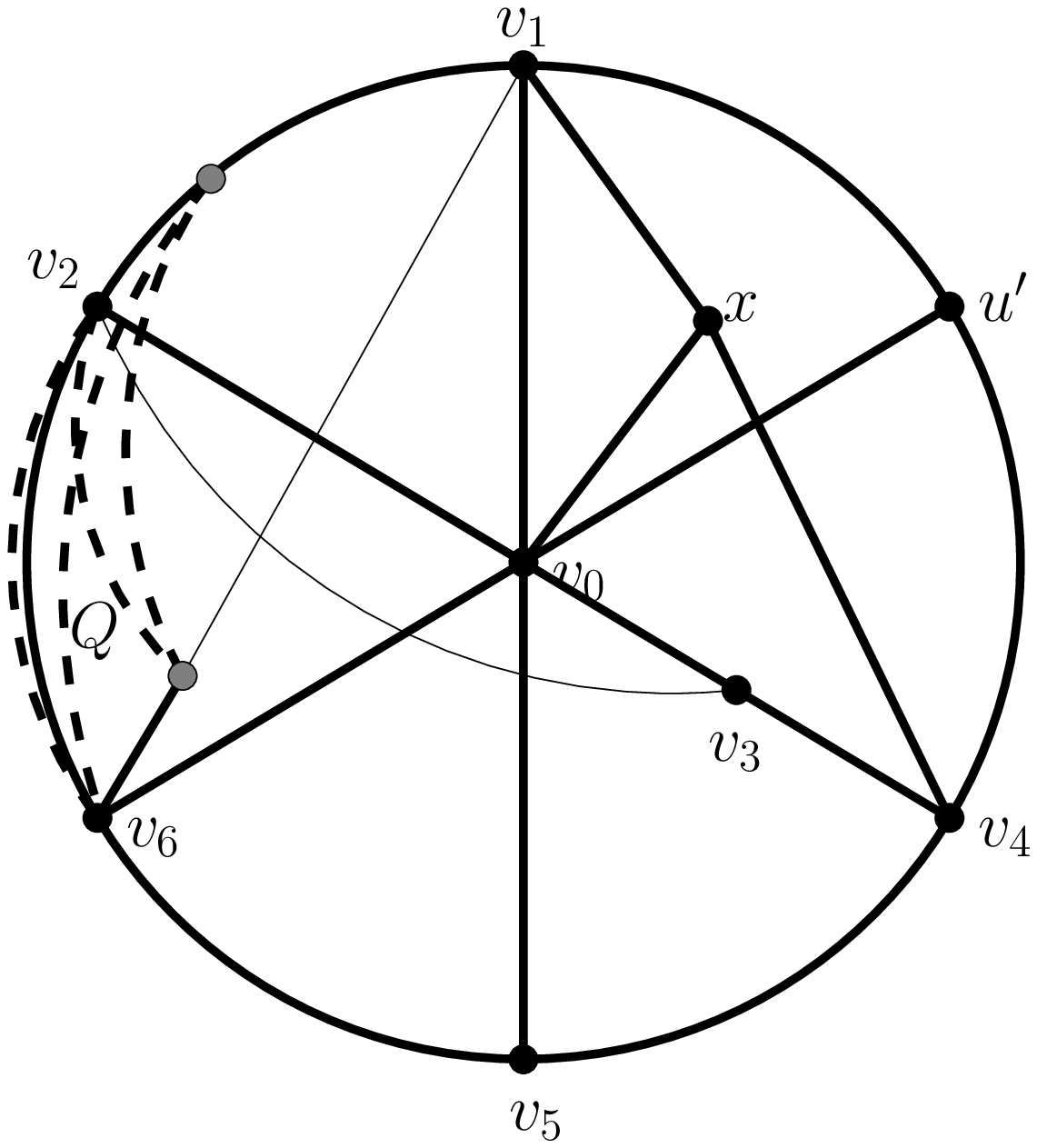}
\caption{Case (b)(ii), vertex $x$ on path $P_{4}$ places $G$ in Case (b)(i). Compare Figure \ref{w7case1}.}
\label{w7case2_1213_1}
\end{center}
\end{figure}

Suppose then that $x$ is on the path $P_{1}$, and $P_{4}$ is a single edge.

If $v_{0}$ contains more than one neighbour in $U(u)\setminus W$ or in $A\setminus W$, then by Lemma \ref{lemmaW7} a $W_{7}$-subdivision exists in $G$. Assume then that $v_{0}$ contains exactly one neighbour in $U(u)\setminus W$, and exactly one neighbour in $A\setminus W$. Thus, if $|(U(u)\cup A)\setminus W| > 3$, a type 4 edge-vertex-cutset can be formed from $v_{1}$, $v_{4}$, and the edges from $v_{0}$ to the two neighbours of $v_{0}$ in $(U(u)\cup A)\setminus W$.

Assume then that $|(U(u)\cup A)\setminus W| \le 3$. Thus, one of $U(u)\setminus W$, $A\setminus W$ contains only one vertex.

If either of the edges $v_{1}v_{4}$ or $v_{0}v_{1}$ exist in $G$, then, Reduction \ref{r1}A is possible. Assume then that these edges do not exist in $G$.

Suppose $v_{1}$ has degree $\ge 7$. Then there exist three neighbours of $v_{1}$, say, $a_{1}$, $a_{2}$, and $a_{3}$, such that $a_{i} \notin N_{H\cup P_{u1}}(v_{1})$ for all $1 \le i \le 3$.

By 3-connectivity, there exist at least two paths disjoint from $v_{1}$ joining $\{a_{1}, a_{2}, a_{3}\}$ to $H\cup Q\cup U(u)\cup A$, such that these paths are also vertex-disjoint from each other. Call these paths $P_{a_{1}}$ and $P_{a_{2}}$. Let $a'_{1}$ and $a'_{2}$ be the vertices where $P_{a_{1}}$ and $P_{a_{2}}$ first meet $H\cup Q\cup U(u)\cup A$ respectively.

Suppose that $\{v_{1}, a'_{1}, a'_{2}\}$ forms a separating set in $G$, the removal of which places $\{a_{1}, a_{2}, a_{3}\}$ and $H\cup Q\cup U(u)\cup A$ in different components. Then, since the bridge of $G|\{v_{1}, a'_{1}, a'_{2}\}$ containing $a_{1}$, $a_{2}$, $a_{3}$ contains at least three neighbours of $v_{1}$, and the bridge of $G|\{v_{1}, a'_{1}, a'_{2}\}$ containing $U(u)$ and $A$ contains at least two neighbours of $v_{1}$, Lemma \ref{lemmaW7} applies to show that $G$ contains a $W_{7}$-subdivision.

Suppose then that $\{v_{0}, a'_{1}, a'_{2}\}$ does not form a separating set in $G$. Thus, there exists some path joining $\{a_{1}, a_{2}, a_{3}\}$ to $H\cup Q\cup U(u)\cup A$, say $P_{a'_{3}}$, such that $P_{a'_{3}}$ is vertex-disjoint from both $P_{a'_{1}}$ and $P_{a'_{2}}$, and $P_{a'_{3}}$ first meets $H\cup Q\cup U(u)\cup A$ at some vertex $a'_{3}$. If $v_{4}\in \{a'_{1}, a'_{2}, a'_{3}\}$ or $v_{0}\in \{a'_{1}, a'_{2}, a'_{3}\}$, then Reduction \ref{r1}A can be performed on $G$. Assume then that this is not the case. Using the program to search and check all other possible placements of $a'_{1}$, $a'_{2}$ and $a'_{3}$ shows that a $W_{7}$-subdivision exists in each case.

Assume then that $v_{1}$ has degree $< 7$.

If $v_{4}$ also has degree $< 7$, then Reduction \ref{r2}A can be performed on $G$. Assume then that $v_{4}$ has degree $\ge 7$.

(A) Suppose $v_{1}$ has at most two neighbours in $U_{2}\setminus W$, say, $x_{1}$ and $x_{2}$.

If $|(U(u)\cup A)\setminus W| = 3$, a type 4 edge-vertex-cutset can be formed from the edges $v_{1}x_{1}$ and $v_{1}x_{2}$, and the vertices $v_{0}$ and $v_{4}$. Assume then that $U(u)\setminus W$ contains only $u'$, and $A\setminus W$ contains only $x$.

Since we know that the edges $v_{1}v_{4}$ and $v_{0}v_{1}$ do not exist in $G$, $v_{0}$ and $v_{4}$ must each have exactly two neighbours in the set $\{v_{1}, x, u'\}$. Therefore, a type 4a edge-vertex-cutset can be formed from the edges $v_{1}x_{1}$ and $v_{2}x_{2}$, and the vertices $v_{4}$ and $v_{0}$.

(B) Suppose then that $v_{1}$ has some third neighbour $y$ in $U_{2}\setminus W$, such that $y\notin N_{H}(v_{1})$.

By 3-connectivity, there must exist some path $Y$ in $U_{2}\setminus W$ joining $y$ to $H\cup Q$, such that $Y$ meets $H\cup Q$ only at its endpoint, say $y'$. Using the program to generate and check all possible such paths $Y$, it is found that a $W_{7}$-subdivision exists in $G$ for each case, except where $y'$ is an internal vertex on the path $Q$. Assume then that this is the case for all such paths $Y$.

Consider the set $S = \{v_{1}, v_{2}, v_{6}\}$. Suppose $S$ does not form a separating set in $G$. Then there exists some path disjoint from $S$ joining the two components of $(H\cup U(u)\cup A\cup Q) - S$.  Using the program to generate and check all possible such paths, it is found that a $W_{7}$-subdivision exists in each case. Assume then that $S$ forms a separating set in $G$. Let $T'$ be the bridge of $G|S$ containing $y$. Let $T''$ be the bridge of $G|S$ containing $v_{3}$, $v_{5}$, $U(u)$, and $A$.

Suppose $|(U(u)\cup A)\setminus W| = 3$. Suppose also that $v_{4}$ has at most two neighbours in $(U(u)\cup A)\setminus W$, say $b_{1}$ and $b_{2}$ (if a second neighbour exists). Since $v_{1}$ has degree $< 7$, and $v_{0}$ has only two neighbours in $(U(u)\cup A)\setminus W$, then a type 2a or 4a edge-vertex-cutset can be formed from $v_{0}$, $v_{1}$, $v_{4}b_{1}$, and $v_{4}b_{2}$ (if $b_{2}$ exists). Assume then that $v_{4}$ is adjacent to all three vertices in $(U(u)\cup A)\setminus W$. Since $v_{4}$ has degree $\ge 7$, $v_{4}$ must have $\ge 3$ neighbours in $U_{2}\setminus W$. Thus, by Lemma \ref{lemmaW7}, a $W_{7}$-subdivision exists centred on $v_{4}$.

Assume then that $|(U(u) \cup A)\setminus W| = 2$.

Suppose $|V(G) \setminus T''| \ge 3$. Recall that for any neighbour $y$ of $v_{1}$ where $y\in U_{2}\setminus W$ but $y\notin N_{H}(v_{1})$, all paths in $U_{2}\setminus W$ joining $y$ to $H\cup Q$ must first meet $H\cup Q$ at an internal vertex of the path $Q$. Thus, any neighbour of $v_{1}$ in $U_{2}\setminus W$ is also in the bridge $T'$. Any neighbours of $v_{1}$ that are not in $T'$, then, must be in $U(u)\setminus W$ or $A\setminus W$. Since $|(U(u)\cup A)\setminus W| = 2$, there can be only two such neighbours of $v_{1}$. Thus, a type 4 edge-vertex-cutset can be formed from $v_{2}$, $v_{6}$, and the edges $v_{1}x$ and $v_{1}u'$.

Assume then that $|V(G) \setminus T''| < 3$.

Let $X = (T''\cap U_{2})\setminus W$. Since $|V(G)| \ge 38$, $|U(u)\cup A| = 5$ (including $W$), and $|V(G) \setminus T''| \le 2$, $X$ must contain at least 31 vertices.

Consider the set $S_{1} = \{v_{2}, v_{0}, v_{4}\}$. Suppose that $S_{1}$ is not a separating set, but rather, there exists some path disjoint from $S_{1}$ joining $v_{3}$ to $v_{1}$. Using the program to check all possible placements of such a path shows that a $W_{7}$-subdivision exists in each case. Suppose then that $S_{1}$ forms a separating set in $G$, with at least two bridges: $T_{3}$, which contains the vertex $v_{3}$, and $T_{1}$, which contains the vertices $v_{1}$, $v_{5}$, $v_{6}$, $u'$, and $x$.

By the same arguments used in Case 1.1.2, it can be assumed that $|T_{3}\setminus S_{1}| \le 3$ and that there are no bridges of $G|S_{1}$ other than $T_{1}$ and $T_{3}$. Thus, $|V(G)\setminus T_{1}| \le 3$. Let $X' = X\cap T_{1}$. Since $|X| \ge 31$, $X'$ must contain at least 28 vertices.

Suppose there exists some internal vertex $p_{2}$ on the path $P_{2}$. By 3-connectivity, there exists some path $Q_{2}$ from $p_{2}$ to $H - P_{2}$ such that $Q_{2}$ meets $H - P_{2}$ only at its endpoint, say, $q_{2}$. Using the program to generate and check all possible placements of $Q_{2}$, it is found that the existence of such a path results in a $W_{7}$-subdivision in $G$, unless $q_{2}$ is contained in the bridge $T_{3}$, or $q_{2} = v_{4}$. Thus, by the same argument used in Case 1.1.2, it can be assumed that $X'$ does not contain any internal vertices on the path $P_{2}$.

Consider the set $S_{2} = \{v_{6}, v_{0}, v_{4}\}$. Suppose that $S_{2}$ is not a separating set, but rather, there exists some path disjoint from $S_{2}$ joining $v_{5}$ to $v_{1}$. Using the program to check all possible placements of such a path shows that a $W_{7}$-subdivision exists in each case. Suppose then that $S_{2}$ forms a separating set in $G$, with at least two bridges: $Y_{5}$, which contains the vertex $v_{5}$, and $Y_{1}$, which contains $v_{1}$, $v_{2}$, $v_{3}$, and $U(u)$.

By the same arguments used in Case 1.1.2, it can be assumed that $|Y_{5}\setminus S_{2}| \le 3$ and that there are no bridges of $G|S_{2}$ other than $Y_{5}$ and $Y_{1}$. Thus, $|V(G)\setminus Y_{1}| \le 3$. Let $X'' = X'\cap Y_{1}$. Since $|X'| \ge 28$, $X''$ must contain at least 25 vertices.

Suppose there exists some internal vertex $p_{6}$ on the path $P_{6}$. By 3-connectivity, there exists some path $Q_{6}$ from $p_{6}$ to $H - P_{6}$ such that $Q_{6}$ meets $H - P_{6}$ only at its endpoint, say, $q_{6}$. Using the program to generate and check all possible placements of $Q_{6}$, it is found that the existence of such a path results in a $W_{7}$-subdivision in $G$, unless $q_{6}$ is contained in the bridge $Y_{5}$, or $q_{6} = v_{4}$. Thus, by the same argument used in Case 1.1.2, it can be assumed that $X''$ does not contain any internal vertices on the path $P_{6}$.

Suppose there exists some internal vertex $p'_{2}$ on the path $v_{1}Cv_{2}$. If $G$ contains the configuration shown in either the second or fourth graphs of Figure \ref{w7case2_Q1}, then one of the endpoints of $Q$ forms such a vertex, and so any internal vertices on $v_{1}Cv_{2}$ are contained in the bridge $T'$. Suppose then that $G$ contains the configuration shown in either the first or third graphs of Figure \ref{w7case2_Q1}. By 3-connectivity, there exists some path $Q'_{2}$ from $p'_{2}$ to $H - v_{1}Cv_{2}$ such that $Q'_{2}$ meets $H$ only at its endpoint, say, $q'_{2}$. Using the program to generate and check all possible placements of $Q'_{2}$, it is found that the existence of such a path results in a $W_{7}$-subdivision in $G$, unless $q'_{2}$ is contained in the bridge $T'$, or $q'_{2} = v_{4}$. Thus, by the same arguments used in Case 1.1.2, it can be assumed either that all internal vertices on the path $v_{1}Cv_{2}$ are contained in the bridge $T'$, and thus are not in the set $X''$, or that all internal vertices on the path $v_{1}Cv_{2}$ are contained in some bridge $A'$ of $G|\{v_{1}, v_{2}, v_{4}\}$ such that $A'$ contains no vertices in $H - v_{4} - V(v_{1}Cv_{2})$, and $|A'\setminus \{v_{1}, v_{2}, v_{4}\}| \le 3$. Let $Z = X'' \setminus (A'\setminus \{v_{1}, v_{2}, v_{4}\})$. Since $|X''| \ge 25$, $|Z| \ge 22$.

Suppose there exists some internal vertex $p'_{6}$ on the path $v_{6}Cv_{1}$. If $G$ contains the configuration shown in either the third or fourth graphs of Figure \ref{w7case2_Q1}, then one of the endpoints of $Q$ forms such a vertex, and so any internal vertices on $v_{6}Cv_{1}$ are contained in the bridge $T'$. Suppose then that $G$ contains the configuration shown in either the first or second graphs of Figure \ref{w7case2_Q1}. By 3-connectivity, there exists some path $Q'_{6}$ from $p'_{6}$ to $H - v_{6}Cv_{1}$ such that $Q'_{6}$ meets $H$ only at its endpoint, say, $q'_{6}$. Using the program to generate and check all possible placements of $Q'_{6}$, it is found that the existence of such a path results in a $W_{7}$-subdivision in $G$, unless $q'_{6}$ is contained in the bridge $T'$, or $q'_{6} = v_{4}$. Thus, by the same arguments used in Case 1.1.2, it can be assumed either that all internal vertices on the path $v_{6}Cv_{1}$ are contained in the bridge $T'$, and thus are not in the set $Z$, or that all internal vertices on the path $v_{6}Cv_{1}$ are contained in some bridge $B$ of $G|\{v_{1}, v_{6}, v_{4}\}$ such that $B$ contains no vertices in $H - v_{4} - V(v_{6}Cv_{1})$, and $|B\setminus \{v_{1}, v_{6}, v_{4}\}| \le 3$. Let $Z' = Z - (B\setminus \{v_{1}, v_{6}, v_{4}\})$. Since $|Z| \ge 22$, $|Z'| \ge 19$.

Let $S' = \{v_{4}, v_{2}, v_{6}\}$.  Suppose there exists a second bridge of $G|S'$ other than that containing $v_{1}$. Let $X_{1}$ be the bridge of $G|S'$ that contains $v_{1}$.

By the same argument used in Case 1.1.2, assume that $v_{4}$ has at most two neighbours not in $X_{1}$. Therefore, unless $|V(G)\setminus X_{1}| \le 3$, a type 2 or 4 edge-vertex-cutset is formed from $v_{2}$, $v_{6}$, and the edge or edges joining $v_{4}$ to $V(G)\setminus X_{1}$. Suppose then that $|V(G)\setminus X_{1}| \le 3$. Let $Z'' = Z'\cap X_{1}$. Since $|Z'| \ge 19$, $|Z''| \ge 16$. However, again using the same arguments of Case 1.1.2, $Z''$ can only contain the two vertices $v_{2}$ and $v_{6}$, which is a contradiction.

\textbf{1.2.} Suppose then that $P_{1}$ and $P_{4}$ are single edges. Then by Lemma \ref{lemma3}, a $W_{7}$-subdivision exists in $G$.

\textbf{2.} Assume now there is no such path $Q$ from $H_{2}$ to $H_{5}$.

Thus, there exist at least three bridges of $G|W$: $U_{2}$, $U_{5}$, and $U(u)$, where $U_{2}$ and $U_{5}$ are the bridges containing the subgraphs $H_{2}$ and $H_{5}$ respectively.

Suppose that $U_{2}\setminus W$ contains some internal vertex $a$ on the path $P_{1}$, and some internal vertex $b$ on the path $P_{4}$. There must exist paths $P_{a}$ and $P_{b}$ in $\langle U_{2}\setminus W\rangle$ joining $a$ and $b$ to $H\cap \langle U_{2}\setminus W\rangle$ respectively. Using the program to generate and check all possible placements of such paths shows that a $W_{7}$-subdivision exists in each case. Assume then that $U_{2}\setminus W$ does not contain internal vertices on both $P_{1}$ and $P_{4}$. By symmetry of the graph, assume also that $U_{5}\setminus W$ does not contain internal vertices on both $P_{1}$ and $P_{4}$.

Suppose there exists some bridge $A$ of $G|W$ such that $A\setminus W$ contains at least three neighbours of $v_{0}$. If $A \notin \{U_{2}, U_{5}\}$, then Lemma \ref{lemmaW7} applies to show that a $W_{7}$-subdivision exists. Suppose then without loss of generality that $A = U_{2}$. Thus, given the conclusion drawn in the previous paragraph, $A$ does not contain internal vertices of both $P_{1}$ and $P_{4}$. Without loss of generality, suppose that there are no internal vertices of $P_{4}$ contained in $A$. If $U_{5}$ does not contain any internal vertices of $P_{4}$, then Lemma \ref{lemmaW7} applies to show that a $W_{7}$-subdivision exists. If $U_{5}$ does contain internal vertices of $P_{4}$, then a $W_{7}$-subdivision is shown to exist by applying Lemma \ref{lemma2} to both $U_{2}$ and $U_{5}$, and using $U(u)$ to form a path from $v_{0}$ to either $v_{1}$ or $v_{4}$ as a seventh spoke.

Assume then that each bridge of $G|W$ contains at most two neighbours of $v_{0}$ not in $W$.

Let $A$ be some bridge of $G|W$. Let $a_{1}$ be some neighbour of $v_{0}$ in $A\setminus W$, and let $a_{2}$ be the second neighbour of $v_{0}$ in $A\setminus W$ if such a vertex exists. If $|A\setminus W| > 3$, then a type 2 or 4 edge-vertex-cutset can be formed from $v_{1}$, $v_{4}$, $v_{0}a_{1}$, and $v_{0}a_{2}$ (if $a_{2}$ exists).

Assume then that each bridge of $G|W$ contains at most three vertices not in $W$. Then, since $|V(G)| \ge 38$, there must be at least 12 bridges of $G|W$.

If there exists some bridge of $G|W$ that contains only one vertex not in $W$, then Reduction \ref{r1}A can be performed on $G$. Assume then that each bridge of $G|W$ contains at least two vertices not in $W$.

Suppose each bridge of $G|W$ contains exactly two vertices not in $W$. Then, since $|V(G)| \ge 38$, there must exist at least 18 bridges of $G|W$. At least one bridge of $G|W$ must be contained as a subdivision in two others, therefore Reduction \ref{r1}B can be performed on $G$.

Assume then that there exists some bridge of $G|W$, say, $U'$, such that $|U'\setminus W| = 3$. To avoid an internal 4-edge-cutset, there must be at least five edges joining $W$ to $U'\setminus W$.

\textbf{2.1.} Suppose there exists some vertex $i \in \{v_{1}, v_{4}\}$ such that $i$ has at least three neighbours in $U'\setminus W$. If there exists some bridge $U''$ such that $U'' \neq U'$ and $i$ has at least two neighbours in $U''\setminus W$, then by Lemma \ref{lemmaW7} a $W_{7}$-subdivision exists centred on $i$. Assume then that every bridge of $G|W$ other than $U'$ (of which there are at least 11) contains exactly one neighbour of $i$ not in $W$. Let $j$ be the vertex in $\{v_{1}, v_{4}\}$ other than $i$.

\textbf{2.1.1.} Suppose there exists some bridge $U_{j}$ of $G|W$ such that $U_{j}\setminus W$ contains at least three neighbours of $j$. If there exists some bridge $U'_{j}$ such that $U'_{j} \neq U_{j}$ and $j$ has at least two neighbours in $U'_{j}\setminus W$, then by Lemma \ref{lemmaW7} a $W_{7}$-subdivision exists centred on $j$. Assume then that every bridge of $G|W$ other than $U_{j}$ contains exactly one neighbour of $j$ not in $W$. Thus, there are at least ten bridges of $G|W$ that contain only one neighbour of $i$ not in $W$ and only one neighbour of $j$ not in $W$. Since it has already been assumed that each bridge of $G|W$ contains at most two neighbours of $v_{0}$ not in $W$, it must be the case that Reduction \ref{r1}B can be performed on $G$.

\textbf{2.1.2.} Assume then that each bridge of $G|W$ contains at most two neighbours of $j$ not in $W$. Thus, there are at least 11 bridges of $G|W$ each with exactly one neighbour of $i$ not in $W$, at most two neighbours of $j$ not in $W$, and at most two neighbours of $v_{0}$ not in $W$. Either Reduction \ref{r1}B or Reduction \ref{r1}C can thus be performed on $G$.

\textbf{2.2.} Suppose then that there are two vertices in $W$, say $x$ and $y$, such that $x$ and $y$ each have exactly two neighbours in $U'\setminus W$. Then a type 2a or 4a edge-vertex-cutset can be formed from $x$, $y$, and the edge or edges joining the third vertex in $W$ to $U'\setminus W$.

\vspace{0.2in}
\noindent \textbf{Case (b)(iii): $u_{1} = v_{1}$, $u_{2} \in P_{4} \setminus \{v_{0}, v_{4}\}$} (Figure \ref{w7case3})

\begin{figure}[!h]
\begin{center}
\includegraphics[width=0.4\textwidth]{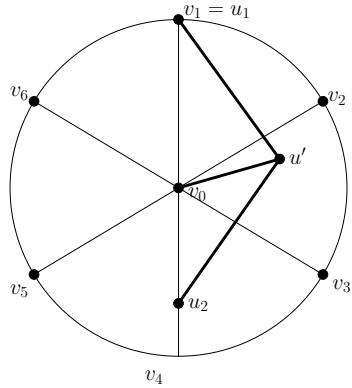}
\caption{Case (b)(iii): $u_{1} = v_{1}$, $u_{2} \in P_{4} \setminus \{v_{0}, v_{4}\}$}
\label{w7case3}
\end{center}
\end{figure}

Let $W = \{v_{0}, v_{1}, v_{4}\}$. Let $H_{2}$ be the subgraph consisting of the path from $v_{1}$ to $v_{4}$ that passes through $v_{2}$ and $v_{3}$, not including endpoints, and all of $P_{2}$ and $P_{3}$ except for $v_{0}$. Let $H_{5}$ be the subgraph consisting of the path from $v_{1}$ to $v_{4}$ that passes through $v_{5}$ and $v_{6}$, not including endpoints, and all of $P_{5}$ and $P_{6}$ except for $v_{0}$. Recall that $U(u)$ is the bridge of $G|V(H)$ which contains $u'$.

\textbf{1.} Suppose there exists some path $Q$ from some point in $H_{2}$ to some point in $H_{5}$. All but four of the possible configurations contain a $W_{7}$-subdivision. The four exceptions are shown in Figure \ref{w7case3_Q1}. Suppose that $G$ contains the configuration shown in one of these graphs.

\begin{figure}[!h]
\begin{center}
\includegraphics[width=0.8\textwidth]{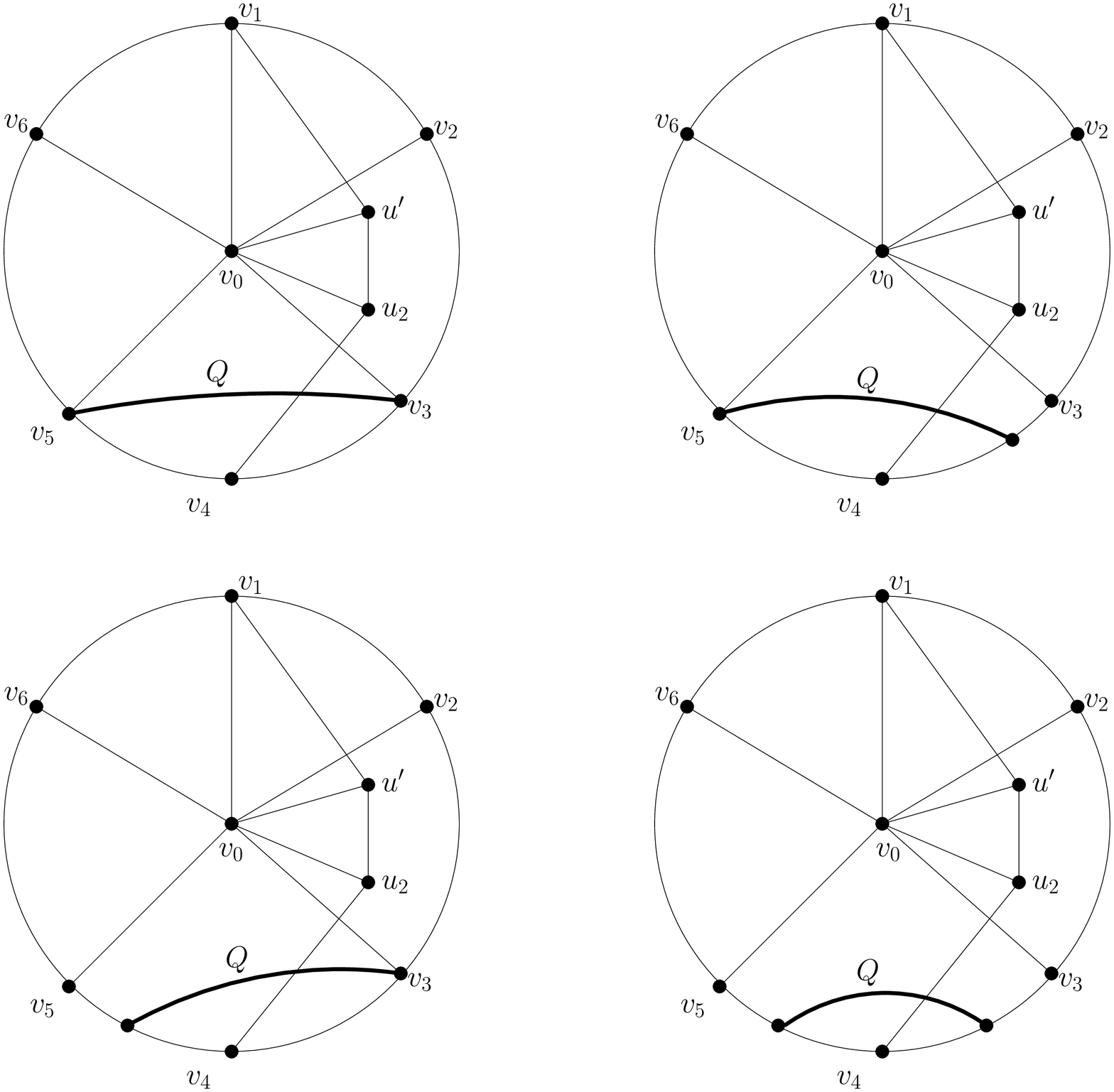}
\caption{Case (b)(iii), path $Q$ from $H_{2}$ to $H_{5}$.}
\label{w7case3_Q1}
\end{center}
\end{figure}

\textbf{1.1.} Suppose there exists some path $R$ in $G$ such that $W$ is not a separating set.

Using the program to generate and check all possible placements of such a path, it is found that the existence of such a path $R$ results in the existence of a $W_{7}$-subdivision in $G$.

\textbf{1.2.} Suppose that no such path $R$ exists in $G$, that is, $U(u)$ forms a bridge of $G|W$, and $H_{2}\cup H_{5}$ is in some bridge of $G|W$ other than $U(u)$. Call this bridge $U_{2}$. The same argument as in Case (b)(ii) 1.2.1.2 can be applied to show that $G$ contains a $W_{7}$-subdivision.

\textbf{2.} Assume now there is no such path $Q$ from $H_{2}$ to $H_{5}$. By symmetry of the graph, we can similarly assume that neither $H_{2}$ nor $H_{5}$ are contained in the bridge $U(u)$.

Thus, there exist at least three bridges of $G|W$: $U_{2}$, $U_{5}$, and $U(u)$, where $U_{2}$ and $U_{5}$ are the bridges containing the subgraphs $H_{2}$ and $H_{5}$ respectively. The same argument as in Case (b)(ii) 2 can be applied to show that $G$ contains a $W_{7}$-subdivision.
\end{proof}

\section{Algorithm}
\label{algorithm}

Theorem \ref{theorem} forms the basis for the following algorithm for solving SHP($W_{7}$).

\subsubsection*{Algorithm 1}

\begin{itemize}
\item[1.] Input: Graph $G$.
\item[2.] If $G$ is 3-connected, go to Step 4; otherwise:
		\begin{itemize}
		\item[(a)] If $G$ is not connected, apply the algorithm recursively to each connected component.
		\item[(b)] If $G$ is not 2-connected, apply the algorithm recursively to each block.
		\item[(c)] Find a separating set $V_{0}$ for $G$ of size 2. Form $G'$ by adding an edge between the two members of $V_{0}$ if none exists already.
		\item[(d)] Find the bridges $U_{1}, \ldots , U_{k}$ of $G'|V_{0}$, and apply the algorithm recursively to each $\langle U_{i} \rangle$, $1 \le i \le k$. If any $\langle U_{i} \rangle$ is accepted, accept $G$; otherwise reject $G$.
		\end{itemize}
\item[3.] If $G$ has an internal 3-edge-cutset, separate $G$ into parts along its 3-edge cutset as described in Algorithm 2 of \cite{Farr88}, and apply the algorithm recursively to each part.
\item[4.] If $G$ has an internal 4-edge-cutset, separate $G$ into parts along its 4-edge-cutset as described in Algorithm 1 of \cite{Robinson08}, and apply the algorithm recursively to each part.
\item[5.] If $G$ has a type 1 or 1a edge-vertex-cutset, separate $G$ into parts along its type 1 or 1a edge-vertex-cutset as follows:
		\begin{itemize}
		\item[(a)] Let $S = \{e_{1}, e_{2}, v\}$ be a type 1 or 1a edge-vertex-cutset of $G$.
		\item[(b)] Let $G_{1}$, $G_{2}$ be the two components of $G - S$.
		\item[(c)] Form $G'_{1}$ from $G$ by replacing $G_{2}$ as described in Theorem \ref{edgevertex} (if $S$ is a type 1 edge-vertex-cutset) or Theorem \ref{edgevertex1a} (if $S$ is a type 1a edge-vertex-cutset), and similarly, form $G'_{2}$ from $G$ by replacing $G_{1}$ as described in Theorem \ref{edgevertex} (if $S$ is a type 1 edge-vertex-cutset) or Theorem \ref{edgevertex1a} (if $S$ is a type 1a edge-vertex-cutset).
		\item[(d)] Apply the algorithm recursively to $G'_{1}$ and $G'_{2}$. If either is accepted, accept $G$; otherwise reject $G$.
		\end{itemize}
\item[6.] If $G$ has a type 2 or 2a edge-vertex-cutset, separate $G$ into parts along its type 2 or 2a edge-vertex-cutset as follows:
		\begin{itemize}
		\item[(a)] Let $S = \{e, v_{1}, v_{2}\}$ be a type 2 or 2a edge-vertex-cutset of $G$.
		\item[(b)] Let $G_{1}$, $G_{2}$ be the two components of $G - S$.
		\item[(c)] Form $G'_{1}$ from $G$ by replacing $G_{2}$ as described in Theorem \ref{edgevertex2} (if $S$ is a type 2 edge-vertex-cutset) or Theorem \ref{edgevertex2a} (if $S$ is a type 2a edge-vertex-cutset), and similarly, form $G'_{2}$ from $G$ by replacing $G_{1}$ as described in Theorem \ref{edgevertex2} (if $S$ is a type 2 edge-vertex-cutset) or Theorem \ref{edgevertex2a} (if $S$ is a type 2a edge-vertex-cutset).
		\item[(d)] Apply the algorithm recursively to $G'_{1}$ and $G'_{2}$. If either is accepted, accept $G$; otherwise reject $G$.
		\end{itemize}
\item[7.] If $G$ has a type 3 or 3a edge-vertex-cutset, separate $G$ into parts along its type 3 or 3a edge-vertex-cutset as follows:
		\begin{itemize}
	\item[(a)] Let $S = \{v, e_{1}, e_{2}, e_{3}, e_{4}\}$ be a type 3 or 3a edge-vertex-cutset of $G$.
		\item[(b)] Let $G_{1}$, $G_{2}$ be the two components of $G - S$, such that $G_{1}$ is the component of $G - S$ that contains exactly two vertices incident with $e_{1}, \ldots, e_{4}$.
		\item[(c)] Form $G'_{1}$ from $G$ by replacing $G_{2}$ as described in Theorem \ref{edgevertex3} (if $S$ is a type 3 edge-vertex-cutset) or Theorem \ref{edgevertex3a} (if $S$ is a type 3a edge-vertex-cutset), and similarly, form $G'_{2}$ from $G$ by replacing $G_{1}$ as described in Theorem \ref{edgevertex3} (if $S$ is a type 3 edge-vertex-cutset) or Theorem \ref{edgevertex3a} (if $S$ is a type 3a edge-vertex-cutset).
		\item[(d)] Apply the algorithm recursively to $G'_{1}$ and $G'_{2}$. If either is accepted, accept $G$; otherwise reject $G$.
		\end{itemize}
\item[8.] If $G$ has a type 4 or 4a edge-vertex-cutset, separate $G$ into parts along its type 4 or 4a edge-vertex-cutset as follows:
		\begin{itemize}
	\item[(a)] Let $S = \{v_{1}, v_{2}, e_{1}, e_{2}\}$ be a type 4 or 4a edge-vertex-cutset of $G$.
		\item[(b)] Let $G_{1}$, $G_{2}$ be the two components of $G - S$, such that $G_{1}$ is the component of $G - S$ that contains exactly one vertex incident with $e_{1}, e_{2}$.
		\item[(c)] Form $G'_{1}$ from $G$ by replacing $G_{2}$ as described in Theorem \ref{edgevertex4} (if $S$ is a type 4 edge-vertex-cutset) or Theorem \ref{edgevertex4a} (if $S$ is a type 4a edge-vertex-cutset), and similarly, form $G'_{2}$ from $G$ by replacing $G_{1}$ as described in Theorem \ref{edgevertex4} (if $S$ is a type 4 edge-vertex-cutset) or Theorem \ref{edgevertex4a} (if $S$ is a type 4a edge-vertex-cutset).
		\item[(d)] Apply the algorithm recursively to $G'_{1}$ and $G'_{2}$. If either is accepted, accept $G$; otherwise reject $G$.
		\end{itemize}
\item[9.] If $G$ has an internal $(1,1,1,1)$-cutset, separate $G$ into parts along its internal $(1,1,1,1)$-cutset as follows:
		\begin{itemize}
		\item[(a)] Let $E' = \{e_{1}, e_{2}, e_{3}, e_{4}\}$ be an internal $(1,1,1,1)$-cutset of $G$.
		\item[(b)] Let $G_{1}$, $G_{2}$ be the two components of $G - S$. Let $u_{1}, u_{2}, u_{3}, u_{4}$ be the endpoints of $e_{1}, \ldots, e_{4}$ in $G_{1}$, and let $v_{1}, v_{2}, v_{3}, v_{4}$ be the endpoints of $e_{1}, \ldots, e_{4}$ in $G_{2}$.
		\item[(c)] Form $G'_{1}$ from $G$ by replacing $G_{2}$ with the subgraph $X$, where $X$ contains only the four vertices $v_{1}, \ldots, v_{4}$, all of which are made adjacent to one another if they were not already. Form $G'_{2}$ from $G$ by replacing $G_{1}$ with the subgraph $Y$, where $Y$ contains only the four vertices $u_{1}, \ldots, u_{4}$, all of which are made adjacent to one another if they were not already.
		\item[(d)] Apply the algorithm recursively to $G'_{1}$ and $G'_{2}$. If either is accepted, accept $G$; otherwise reject $G$.
		\end{itemize}
\item[10.] If some reduction $R$ (where $R$ is one of Reductions \ref{r1}A, \ref{r1}B, \ref{r1}C, \ref{r2}A, \ref{r2}B, \ref{r6}, \ref{r7}, \ref{r8}, or \ref{r1_big}) can be performed on $G$, put $G' = R(G)$. Apply the algorithm recursively to $G'$. If $G'$ is accepted, accept $G$; otherwise, reject $G$.
\item[11.] If $|V(G)| < 38$, perform an exhaustive search of $G$ for a $W_{7}$-subdivision. If such a subdivision is found, accept $G$; otherwise reject $G$.
\item[12.] If $G$ has some vertex of degree at least 7, accept $G$; otherwise, reject $G$.
\end{itemize}

Note that while certain steps in this algorithm must be performed in the order given, in other cases the order can be varied with no effect on the algorithm's correctness. For example, step 10 (performing reductions on $G$) could be executed before any of the steps 3 through to 9 without affecting the outcome of the algorithm. However, steps 5, 6, 7 and 8, all of which deal with edge-vertex-cutsets, must be performed in the order given, since, for example, the theorem given in Section \ref{separatingsets} regarding type 2 edge-vertex-cutsets only applies to graphs with no type 1 edge-vertex-cutsets. Note also that while Step 11 could be performed earlier in the algorithm without altering its correctness, for the purposes of efficiency it is more desirable that this step be performed later.

Steps 2, 3, 4, and 12 use the same techniques as the algorithms presented in \cite{Farr88} and \cite{Robinson08} for solving SHP($W_{4}$), SHP($W_{5}$), and SHP($W_{6}$). Step 5 uses the same technique as in the algorithm of \cite{Robinson08} to deal with internal 4-edge-cutsets. Steps 5 to 9 deal with the new forbidden separating sets defined in Section \ref{separatingsets}, and the correctness of these steps follows from the theorems given in that section (Theorems \ref{edgevertex} to \ref{bigfourcutset}). Step 10 deals with the forbidden reductions defined in Section \ref{reductions}, and the correctness of this step follows from the theorems given in that section (Theorems \ref{reduction1} to \ref{reduction1big}).

Using the same arguments given in \cite{Robinson08}, the most complex steps in this algorithm (those involving finding sets of four edges; i.e., Steps 4, 7, and 9) have a worst case complexity of $O(m^{5})$. Thus, assuming an imbalanced division of $G$ at each recursion of the algorithm, as with the algorithm for SHP($W_{6}$) in \cite{Robinson08}, this algorithm's total complexity is $O(m^{6})$, and therefore runs in polynomial time.

\section{Concluding remarks}
\label{conclusion}

It is hoped that further work in this area may lead to a characterization for the general wheel, $W_{k}$. As mentioned in the Introduction to this paper, such a characterization may not be complete, but rather rely on parameterization of the input, taking $k$ as the parameter, and yielding an algorithm that is fixed-parameter tractable. Certain features of the algorithm in Section \ref{algorithm} may lend themselves to a parameterized algorithm. In particular, the overall structure of the algorithm --- breaking down the input graph into smaller, manageable components, until its size is bounded by some constant, then performing exhaustive search on the remaining, constant-sized input --- strongly resembles the parameterized algorithmic technique of reducing to a problem kernel, as described in \cite{D&F99}.

It is, however, probable that certain difficulties will arise in looking at characterizations beyond $W_{7}$. Given the large increase in length and difficulty between the proofs of the $W_{5}$ \cite{Farr88} and $W_{6}$ cases \cite{Robinson08}, and even more so between the $W_{6}$ and $W_{7}$ case, it seems likely that a characterization involving $W_{8}$ would be extremely complex. While some of the techniques used in this paper --- all of the Reductions defined in Section \ref{reductions}, for example --- are generalizable to higher cases, others are not. Each of the edge-vertex-cutsets defined in Section \ref{separatingsets} is useful in the algorithm for solving SHP($W_{7}$) because of an associated theorem that applies only to the $W_{7}$ case --- these theorems rely on the fact that for $k \le 7$, if $G$ contains some $W_{k}$-subdivision $H$ such that $H$ is centred on some vertex $v$ in an edge-vertex-cutset $S$, there exists some component $G_{1}$ of $G - S$ such that $G_{1}$ contains at most three neighbours of $v$ in $H$. This no longer holds when $k = 8$. Thus, in dealing with pattern graphs $W_{k}$ for $k\ge 8$, it would be necessary to find new techniques to replace edge-vertex-cutsets, or to develop new theorems that broaden their usefulness.

It may also be possible to use some of the techniques presented in this paper to work towards characterizing pattern graphs other than wheels. For example, developing a characterization of graphs containing no subdivisions of $K_{5}$ would be worth investigating.

\end{document}